\documentclass[11pt]{article}
\usepackage{amsmath,amssymb,amsfonts,amsthm,latexsym}
\usepackage{graphicx,booktabs,multirow}
\usepackage{enumerate,setspace,geometry,changepage,lipsum,xr}
\usepackage[toc,title,titletoc,header]{appendix}
\usepackage{titletoc}
\usepackage{etoc}
\usepackage[bottom]{footmisc}
\usepackage[authoryear]{natbib}
\defcitealias{MillerMolinariStoye2024WP}{MMS}
\newcommand{\MMS}{\citetalias{MillerMolinariStoye2024WP}}
\defcitealias{OulhajElGhouchHolman2019}{OEH}
\newcommand{\OEH}{\citetalias{OulhajElGhouchHolman2019}}

\usepackage{xstring} 

\numberwithin{equation}{section}
\numberwithin{table}{section}
\numberwithin{figure}{section} 

\usepackage[pdfborder={0 0 0},colorlinks=true,linkcolor=magenta,citecolor=blue,filecolor=magenta]{hyperref}
\newcommand{\myref}[2]{\hyperref[#1]{#2}}

\usepackage{pgfplots}
\pgfplotsset{compat=1.17}
\usepackage{xcolor}
\definecolor{babyblue}{rgb}{0.54, 0.81, 0.94}

\usepackage{lineno,etoolbox}
\setpagewiselinenumbers

\allowdisplaybreaks

\def\qed{\rule{2mm}{2mm}}

\newtheorem{theorem}{Theorem}
\newtheorem{lemma}{Lemma}
\newtheorem{corollary}{Corollary}
\newtheorem{assumption}{Assumption}

\newtheorem{example}{Example}
\newtheorem{remark}{Remark}
\newtheorem{proposition}{Proposition}

\newtheorem*{proposition*}{Proposition}

\newcounter{assumptionM}
\newcounter{assumptionA}
\def\theassumptionM{M.\arabic{assumptionM}}
\def\theassumptionA{A.\arabic{assumptionA}}
\AtEndEnvironment{remark}{~\qed}%
\AtEndEnvironment{example}{~\qed}%

\title{Testing Sign Agreement}

\author{Deborah Kim\thanks{I am deeply grateful to Ivan Canay, Joel Horowitz, Eric Auerbach, and Federico Bugni for their invaluable guidance on this project. I have also benefited from the comments and suggestions of Yong Cai, Filip Obradovic, Ahnaf Rafi, Neal Myungkou Shin, and Amilcar Velez, as well as from participants at numerous seminars, workshops, and conferences where earlier versions of this paper were presented. This research was supported in part by the computational resources and staff contributions provided by the high performance computing facilities at the University of Warwick. The previous \href{https://drive.google.com/open?id=1hkERE82w_aF9Vx_SdnL1OErTzSmIkmT4&usp=drive_fs}{JMP version} of this article is available on the author's website. All errors are my own.}\\
Department of Economics \\ 
University of Warwick\\ 
\url{deborah.kim@warwick.ac.uk}}

\begin{document}
\maketitle
\vspace{-10pt}
\begin{spacing}{1.1}
\begin{abstract}
This article studies the problem of testing sign agreement among a finite number of parameters, with applications to dynamic treatment effects, instrument validity, and labor-market discrimination. We propose least favorable and two-step tests that accommodate correlated estimators. Both have parametric and nonparametric bootstrap implementations and uniform asymptotic validity. Although neither test dominates the other in general, local power analysis shows that the two-step test can gain power when several parameters provide strong directional information. Simulations show substantial power gains for the two-step test over the least favorable and competing procedures. Empirical applications examine unilateral divorce laws and the college-proximity instrument.
\end{abstract}
\end{spacing}

\clearpage
\section{Introduction}
This article studies the problem of testing whether a finite set of parameters share a common sign. Writing $\mu=(\mu_1,\ldots,\mu_k)$ for the vector of parameters of interest, the hypothesis of \emph{sign agreement} is
\begin{align}\label{def:intro-null}
H_0:\ \mu_j\geq 0 \text{ for all } j=1,\ldots,k
\quad\text{or}\quad
\mu_j\leq 0 \text{ for all } j=1,\ldots,k,
\end{align}
tested against the alternative
\begin{align*}
H_1: \text{there exist } i\neq j \text{ such that }
\mu_i>0 \text{ and } \mu_j<0.
\end{align*}
Rejection is direct evidence that the $k$ parameters do not all point the same way: some are positive while others are negative. The parameters may include regular estimands such as ordinary least squares, instrumental-variable, two-way fixed-effects, or difference-in-differences coefficients.

The hypothesis in \eqref{def:intro-null} arises across many empirical settings. When the $k$ coordinates index post-treatment horizons, sign agreement asks whether the dynamic treatment response reverses over event time. When the coordinates instead index subgroups of a treated population, sign disagreement indicates that the treatment benefits one subgroup while harming another, a phenomenon known as \emph{qualitative interaction}. Sign agreement also underlies a testable implication of instrument validity for the local average treatment effect \citep{ImbensAngrist1994ECTA-LATE,Kitagawa2021JoE}. Furthermore, when the parameters are differences between two cumulative distribution functions (CDFs) evaluated at ordered support points, sign agreement characterizes whether the distributions are ordered by stochastic dominance. In the discrimination model of \citet{Bharadwaj/Deb/Renou:2024}, rejecting this ordering establishes that the two labor-market outcome distributions cross, a configuration consistent with statistical discrimination. Section~\ref{sec:examples} develops these examples in detail. A common feature of these settings is that estimators of the parameters are typically correlated.

Existing procedures developed for sign agreement remain limited in number and scope. Developed primarily for qualitative interaction in biostatistics, many assume that estimators of the parameters are statistically independent \citep{GailSimon1985Biometrics,Berger1989JASA,Zelterman1990SPL,PanWolfe1997,Silvapulle2001Biometrics,LiChan2006detecting,kitsche2014testing,zhao2019JASA}. Several procedures allow correlation but remain restricted to two parameters. \citet{RussekSimon1993Biometrics} develop a least favorable test under a bivariate Gaussian model with known covariance. In economics, \citet{Kaplan2015} recovers the likelihood-ratio test of Gail and Simon and numerically constructs nearly boundary-similar tests allowing correlation. \citet*[hereafter \MMS{}]{MillerMolinariStoye2024WP} independently rediscover the test of Russek-Cohen and Simon and show that this procedure improves upon the Bonferroni-adjusted sign agreement test used by \citet{BrinchMogstadWiswall2017JPE}. Among procedures developed for sign agreement, the test of \citet*[hereafter \OEH{}]{OulhajElGhouchHolman2019} is, to the best of our knowledge, the only one that accommodates arbitrary correlation and any fixed finite number of parameters. 

We develop two tests of sign agreement, a least favorable test and a two-step test, that accommodate general correlation among the $k$ estimators at any fixed finite $k$. Each test can be based on one of three statistics: a maximum, a sum of squares, and a quasi-likelihood ratio. We provide both parametric and nonparametric bootstrap implementations of each test and establish their uniform asymptotic validity over a broad class of nonparametric distributions, together with consistency against every fixed alternative. These results require a standard uniform integrability condition and, where needed, uniform nonsingularity of the correlation matrix. Uniformity is particularly important here because the limiting null distribution of the test statistic changes discontinuously with the underlying parameter as sign restrictions become binding. In such settings, pointwise asymptotic validity alone does not rule out substantial size distortions, even in large samples \citep{mikusheva2007uniform,AndrewsGuggenberger2009ECTA-hybrid,AndrewsGuggenberger2010ET}.

Our first test takes a least favorable approach to choosing the critical value. Under the Gaussian approximation, the null distributions of our test statistics depend on the parameter $\mu$ and the correlation matrix $\Omega$. For each $\Omega$, we choose the largest critical value over all $\mu$ satisfying the null, thereby protecting against every null configuration. The resulting parametric test nests the procedures of \citet{GailSimon1985Biometrics,PiantadosiGail1993,RussekSimon1993Biometrics} in their respective special cases. Extending this approach beyond the independence and bivariate settings considered in these papers is challenging: the critical value generally requires optimization over an unbounded parameter space for each $\Omega$, with no general closed-form solution. We address this difficulty by providing computable lower and upper bounds that contain the exact critical value with a prespecified Monte Carlo probability and converge to it as the approximation is refined. We also derive an analytical characterization of the critical value for the max statistic when all pairwise correlations are nonnegative, eliminating the need for optimization in that case.

The two-step test is designed to improve power against a wide range of alternatives by using sample information to restrict the set of null parameter values used to determine the critical value. Following the two-step principle of \citet{RomanoShaikhWolf2014ECTA}, we first construct a joint confidence region for $\mu$ with asymptotic coverage $1-\tau$, where $\tau\in(0,\alpha)$. In the second step, we take the largest $(1-\alpha+\tau)$ quantile of the test statistic's approximating distribution over the intersection of this region with the null parameter space. The adjustment from $1-\alpha$ to $1-\alpha+\tau$ accounts for the possibility that the first-step region excludes the true parameter. If the region contains no null values, the test rejects immediately; if it lies entirely within the null, the test does not reject. Otherwise, rejection occurs when the test statistic exceeds the resulting critical value.

Our local asymptotic power analysis compares the two proposed parametric tests under known identity covariance. The least favorable test dominates the two-step test when $k=2$. For $k\geq3$, it remains more powerful when one coordinate has a local drift that is sufficiently large in magnitude, while the remaining coordinates have either zero drifts or sufficiently small drifts of the opposite sign. Meanwhile, the two-step test can gain power when several coordinates have large local drifts of the same sign. The power tradeoff reflects two opposing effects: restricting the set over which the supremum is taken can lower the critical value, while increasing the quantile level works in the opposite direction. The two-step test can therefore gain power when the confidence region excludes null parameter values that yield large critical values, lowering the supremum enough to outweigh the adjustment for first-step uncertainty.

Comparisons with existing procedures highlight advantages in both power and validity. Our parametric least favorable max test is uniformly more powerful than the parametric test of \OEH{} when both procedures use the same studentized estimators and estimated correlation matrix. \OEH{} also propose a nonparametric bootstrap procedure, which we show is not pointwise asymptotically valid: its rejection probability exceeds the nominal level at a fixed null distribution even as the sample size grows. In the bivariate setting, \citet{kowalski2023RES} uses a heuristic nonparametric bootstrap procedure under independence; \MMS{} show that this procedure need not be valid under general correlation. These examples underscore the difficulty of bootstrap inference for statistics involving nondifferentiable operations \citep{bickel1981some,bickel1997resampling,andrews2000inconsistency,FangSantos2019}. Our proposed nonparametric bootstrap procedures have uniform asymptotic validity under the stated conditions.

A Monte Carlo study examines finite-sample size and power across dimensions, correlation structures, and error distributions. The parametric least favorable and two-step tests have null rejection probabilities close to the nominal level $\alpha$ under symmetric errors but can overreject under skewed errors, where their nonparametric bootstrap counterparts have null rejection probabilities close to $\alpha$. The power comparisons broadly reflect our local asymptotic results, and substantial gains for the two-step test also arise under correlation, particularly with the max and sum-of-squares statistics. The parametric LF max test also exhibits higher power than OEH in the simulations, consistent with our theoretical dominance result.

We illustrate the proposed tests in two applications. First, we test the validity of the college-proximity instrument of \citet{card1993NBER} for the returns to college completion. Both proposed tests reject at the $1\%$ level with either parametric or nonparametric bootstrap, providing evidence against instrument validity. Second, we revisit the dynamic effect of unilateral divorce laws on divorce rates \citep{Wolfers2006AER}. Both proposed parametric tests reject sign agreement across eight post-reform intervals at the $5\%$ level under both the original and an alternative specification, providing evidence of sign reversal over event time.

The remainder of the article is organized as follows. Section~\ref{sec:examples} presents motivating examples.
Section~\ref{sec:setup} fixes notation and defines uniform asymptotic validity. Section~\ref{sec:LFCtest} develops the least favorable test, and Section~\ref{sec:two-step-test} the two-step test, together with their uniform validity results. Section~\ref{sec:power-comparison} provides local power results. Section~\ref{sec:MonteCarlo-Simulation} reports the simulation study and Section~\ref{sec:Empirical-Application} the two empirical illustrations. Section~\ref{sec:Conclusion} concludes. Proofs and further results are collected in the appendices and the supplementary material.

\section{Motivating Examples}\label{sec:examples}
This section presents empirical settings in which sign agreement concerns several parameters whose estimators are naturally correlated.

\begin{example}[Dynamic Treatment Effect]
\rm\label{example:Heterogeneity-RCT}
A treatment effect may reverse sign over time, so that its short-run and long-run consequences point in opposite directions. Such reversals matter economically because an intervention may generate short-run gains followed by long-run losses, or impose initial costs before yielding delayed benefits. The sign agreement hypothesis provides a framework for examining this possibility. Consider a setting with a binary treatment $D\in{0,1}$ and outcomes observed at $k$ post-treatment horizons $g_1,\ldots,g_k$. Let $Y_1(g_j)$ and $Y_0(g_j)$ denote the potential outcomes with and without treatment at horizon $g_j$, and let
$$\mu_j=\mathbb{E}[Y_{1}(g_j) - Y_{0}(g_j)] \qquad j=1,\ldots,k$$
denote the corresponding average treatment effect (ATE). The hypothesis in
\eqref{def:intro-null} then states that these effects do not have opposite signs across the observed horizons. Section~\ref{sec:Empirical-Application} provides a quasi-experimental illustration.

The same hypothesis applies to treatment effects across subgroups. Suppose each individual belongs to one of $k$ groups defined by an observable pretreatment variable $G\in\mathcal{G}\equiv\{g_1,\ldots,g_k\}$. Let
\[
\mu_j=\mathbb{E}[Y_1-Y_0\mid G=g_j],\qquad j=1,\ldots,k,
\]
denote the ATE for group $g_j$. The hypothesis in \eqref{def:intro-null} states that the group-level effects share a common weak sign, corresponding to the absence of what the biostatistics literature calls \emph{qualitative interaction}. Hypotheses of this form arise routinely in clinical trials; see, among others, \citet{rastogi2008preoperative}, \citet{NEJM2022}, \citet{Jama2019}, and \citet{Lancet2019}.
\end{example}

\begin{example}[Instrument Validity for Local Average Treatment Effect]\label{example:LATE-instrument-validity}\rm
\citet{ImbensAngrist1994ECTA-LATE} show that, under instrument relevance and validity, the canonical two-stage least squares (TSLS) estimand has a causal interpretation as a local average treatment effect (LATE). The sign agreement hypothesis provides a framework for testing an implication of instrument validity. Let $Y\in\mathcal{Y}$, $D\in\{0,1\}$, and $Z\in\{0,1\}$ denote the observed outcome, treatment status, and instrument, respectively. Let $Y_d$ denote the potential outcome under treatment status $d$, and let $D_z$ denote the potential treatment status under instrument value $z$. These potential treatment responses define \textit{compliers}, with $(D_0,D_1)=(0,1)$, and \textit{defiers}, with $(D_0,D_1)=(1,0)$. Instrument relevance requires a nonzero first stage, while instrument validity comprises three assumptions:
\[
\begin{aligned}
    &\text{Exclusion: } Z\text{ affects }Y\text{ only through }D;\\
    &\text{Exogeneity: } (D_0,D_1,Y_1,Y_0)\text{ is independent of }Z;\\
    &\text{Monotonicity: } D_1\geq D_0\text{ a.s.}\quad\text{or}\quad D_1\leq D_0\text{ a.s.}
\end{aligned}
\]
Under these assumptions, the TSLS estimand equals the average treatment effect among compliers when $D_1\geq D_0$ and among defiers when $D_1\leq D_0$.

Although instrument validity cannot be verified from the observed data, it has testable implications. Let $\mathcal{Y}_1,\ldots,\mathcal{Y}_k$ partition the support of $Y$ into $k$ categories, which may be individual support points of a discrete outcome or bins of a continuous one. Define
\begin{align}\label{eq:instrument-validity-parameters}
    \mu_j
    &=P\{Y\in\mathcal{Y}_j,\,D=1\mid Z=1\}
      -P\{Y\in\mathcal{Y}_j,\,D=1\mid Z=0\},
      \qquad j=1,\ldots,k,\notag\\
    \mu_{k+j}
    &=P\{Y\in\mathcal{Y}_j,\,D=0\mid Z=0\}
      -P\{Y\in\mathcal{Y}_j,\,D=0\mid Z=1\},
      \qquad j=1,\ldots,k.
\end{align}
The testable implication can then be formulated analogously to \eqref{def:intro-null}:
\begin{align}\label{eq:instrument-validity-null}
    H_0:\mu_j\geq0\quad\text{for all }j=1,\ldots,2k
    \quad\text{or}\quad
    \mu_j\leq0\quad\text{for all }j=1,\ldots,2k;
\end{align}
see \citet{Kitagawa2021JoE}. Because the underlying restrictions hold for every Borel set, \eqref{eq:instrument-validity-null} is a valid implication for any choice of partition. When $Y$ has finite support and each category contains a single support point, these inequalities are equivalent to the restrictions over all Borel sets. Rejection provides evidence against instrument validity and thus against the assumptions supporting the standard LATE interpretation of the TSLS estimand.\footnote{Several procedures are available for testing instrument validity under the conventional monotonicity condition $D_1\geq D_0$, which rules out defiers (e.g., \citealt{Kitagawa2015ECTA}; \citealt{Huber2015REStat-testingLATE}; \citealt{Mourifie2017REStat-testingLATE}; and \citealt{Machado2019JoE}). These procedures take the direction of monotonicity as given. In contrast, we consider the two-sided monotonicity condition, under which either $D_1\geq D_0$ or $D_1\leq D_0$ almost surely; that is, either defiers or compliers are absent.} Section \ref{sec:Empirical-Application} illustrates this application using the college-proximity instrument of \citet{card1993NBER}.
\end{example}

\begin{example}[Statistical versus Taste-Based Discrimination] \rm\label{example:discrimination}
Economists distinguish between two forms of labor-market discrimination. Statistical discrimination arises when employers imperfectly observe individual productivity and therefore rely on group-specific signals, whereas taste-based discrimination reflects preferences or animus toward a particular group \citep{Becker1957,Phelps1972}. \citet{Bharadwaj/Deb/Renou:2024} show that, when two groups have the same mean productivity, their observed wage distributions \(G_1\) and \(G_2\) can be rationalized without taste-based discrimination if and only if neither distribution \textit{strictly} first-order stochastically dominates the other. Strict stochastic dominance therefore provides evidence of taste-based discrimination. Suppose that the outcome is discrete with ordered support $w_1<\cdots<w_J$ where $J\geq 3,$ as may occur when wages are reported in intervals or when the outcome is an ordered non-wage measure, such as promotion rank. Define
\begin{align*}
\mu_j=G_1(w_j)-G_2(w_j),
\qquad j=1,\ldots,J-1.
\end{align*}
The two distributions are ordered by weak stochastic dominance if and only if the components of $\mu$ have a common weak sign. The sign agreement hypothesis in \eqref{def:intro-null} therefore includes both strict dominance in either direction and equality of the distributions. Rejecting the null provides evidence that the two distributions cross and hence that neither strictly stochastically dominates the other. Under the maintained assumptions of \citet{Bharadwaj/Deb/Renou:2024}, such distributions can be rationalized by statistical discrimination or no discrimination.
\end{example}

A common feature of these examples is that the estimators of $\mu_1,...,\mu_k$ can be correlated. In Example~\ref{example:Heterogeneity-RCT}, dynamic treatment effects are estimated from repeated observations of the same units, so their estimators can be correlated across horizons. In Example~\ref{example:LATE-instrument-validity}, the estimated cell probabilities within each instrument group are correlated through their multinomial sampling structure. In Example \ref{example:discrimination}, the empirical CDF differences are correlated across evaluation points because each group's empirical CDF uses the same observations to estimate nested cumulative probabilities.

\section{Notation and Setup}\label{sec:setup}
This section introduces the setup and presents the hypothesis of sign agreement, along with the test statistics. Additionally, section \ref{subsec:uniform-validity} explains the concept of the uniform asymptotic validity and elucidates why pursuing it is desirable in our problem.

We begin by describing the statistical setup and notation used throughout. Let \(\{W_i:i=1,\ldots,n\}\) be a random sample from \(P\), and let \(\hat P_n\) denote the empirical distribution. Write \(\mu(P)=(\mu_1(P),\ldots,\mu_k(P))=\mathbb E_P[W_i]\). For each \(j=1,\ldots,k\), let \(\sigma_j^2(P)\) denote the variance of the \(j\)th component of \(W_i\), and suppose throughout that \(0<\sigma_j^2(P)<\infty\). Let \(\Omega(P)\) denote the correlation matrix of \(P\). Sample analogues are defined by replacing \(P\) with \(\hat{P}_n\). In particular, \(\bar{W}_n=(\bar{W}_{1,n},...,\bar{W}_{k,n})=\mu(\hat{P}_n)=(\mu_1(\hat{P}_{n}),...,\mu_k(\hat{P}_{n}))\), \(\hat{\Omega}_n=\Omega(\hat{P}_n)\), and \(S_{j, n}^2=\sigma_j^2(\hat{P}_n)\) for all \(j=1,..., k\); Define a diagonal matrix of sample variances as \(S_n^2=\operatorname{diag}(S_{1, n}^2, .., S_{k, n}^2)\). Let \(\mathbf{R}^k\) denote the \(k\)-dimensional Euclidean space, and define \(\mathbf{R}^k_+=\{x\in\mathbf{R}^k: x_j \geq 0 \text{ for all }j=1,...,k\}\) and \(\mathbf{R}^k_-=\{x\in\mathbf{R}^k: x_j \leq 0 \text{ for all }j=1,...,k\}\). Finally, let \(\mathbf O\) denote the set of \(k\times k\) positive definite correlation matrices and \(\bar{\mathbf O}\) be its closure. Let \(I_k\) denote the $k\times k$ identity matrix, let \(0_k\) denote the \(k\)-dimensional zero vector, and let $[k]$ denote the set of index $\{1,...,k\}.$

We are interested in testing hypotheses:
\begin{equation}\label{def:null-hypo}
\begin{gathered}
H_0: P\in \mathbf{P}_0
\quad \text{versus}\quad
P\in \mathbf{P}\setminus\mathbf{P}_0, \\
\text{where}\quad
\mathbf{P}_0
=
\left\{
P\in\mathbf{P}:
\mathbb{E}_P[W_i]\leq 0_k
\quad\text{or}\quad
\mathbb{E}_P[W_i]\geq 0_k
\right\}
\end{gathered}
\end{equation}
where $\mathbf{P}$ is a large set of nonparametric distributions satisfying conditions in Theorems \ref{thm:LFtest-uniform-validity} and \ref{thm:two-step-uniform-validity} and the inequality applies elementwise. Throughout this article, $k$ is greater than or equal to 2, so $\mathbf P_0$ is nontrivial. For exposition, the main text develops the tests for the sign-agreement hypothesis concerning the population mean vector. This formulation is not essential and extends to a broad range of regular estimands; see Section \ref{sec:extension-regular-estimators} in the supplementary material.

The two proposed tests use the same test statistics. To introduce them, define three real-valued functions $T^{\ell}:\mathbf{R}^k \times\mathbf{O}\mapsto \mathbf{R}$, for $\ell=m,s,q$, by
\begin{align}\label{def:T-ell-function}
     \begin{aligned}
     T^{m}(x) &= \min\left\{\max_{1\leq j \leq k} \{x_j\}, \max_{1\leq j \leq k} \{-x_j\}\right\}
     \\
     T^s(x) &= \min\left\{\sum_{1\leq j \leq k} x_j^2 I\{x_j\geq0\}, \sum_{1\leq j \leq k} x_j^2 I\{x_j\leq0\}\right\}\\
     T^q(x, \Omega) &= \inf_{\mu\in\mathbf R^k_+\cup \mathbf R^k_-} (x-\mu)' \Omega^{-1}(x-\mu)
\end{aligned}
\end{align}
where the second argument is suppressed for $\ell=m,s$. The superscripts $m$, $s$, and $q$ stand for maximum, sum of squares, and quasi-likelihood ratio (QLR), respectively. 

All three functions share the property that  $T^{\ell}(x)>0$ whenever $x$ contains components with opposite signs. They differ in how they quantify the extent of sign disagreement. The function \(T^{m}(x)\) takes the smaller of the largest positive component and the largest negative component in absolute value. The function \(T^{s}(x)\) takes the smaller of the sums of squared positive and negative components. Finally, \(T^{q}(x,\Omega)\) measures the squared Mahalanobis distance from \(x\) to the nearest vector whose components have a common sign. Interestingly, when $k=2$, the three functions satisfy $T^s(x)=\bigl(T^m(x)\bigr)^2=T^q(x,\Omega)$ whenever the components of $x$ have opposite signs. Accordingly, they generate the same family of level curves for such $x$, as Figure \ref{fig:Test-Statistic-Level-Curve} illustrates.

Using these functions, we define the test statistics by evaluating them at the vector of studentized sample means, $\sqrt{n}S^{-1}_n\bar{W}_n=(\tfrac{\sqrt{n}\bar{W}_{1,n}}{S_{1,n}},...,\tfrac{\sqrt{n}\bar{W}_{k,n}}{S_{k,n}})$, and the sample correlation $\hat{\Omega}_n$. For each $\ell\in\{m,s,q\}$, define
\begin{align}\label{def:test-statistic}
    T^\ell_n = T^\ell(\sqrt{n}S^{-1}_n\bar{W}_n, \hat\Omega_n).
\end{align}
Large values of this test statistic provide evidence against the null hypothesis. Some of these statistics have been used in related testing problems. In particular, \cite{GailSimon1985Biometrics} and \cite{PiantadosiGail1993} use $T^s_n$ and $T^m_n$, respectively, for sign agreement tests. In the context of testing one-sided moment inequalities, \cite{Song2012JBES} employs $T^m_n$, the properties of which are further studied by \cite{Kim2023ET}.

\begin{figure}[tb]
    \begin{center}
        \begin{tikzpicture}[scale=0.90] 
        \fill[gray!40] (-3,-3) rectangle (0,0);
        \fill[gray!40] (0,0) rectangle (3,3);
    
        \draw[-latex] (-3,0) -- (3,0) node[right] {};
        \draw[-latex] (0,-3) -- (0,3) node[above] {};

        \draw[gray, thick, dashed] (1,-1) -- (3,-1) node[above] {$1$};
        \draw[gray, thick, dashed] (1,-1) -- (1,-3) node[above] {};
        \draw[gray, thick, dashed] (2,-2) -- (3,-2) node[above] {$2$};
        \draw[gray, thick, dashed] (2,-2) -- (2,-3) node[above] {};
    
        \draw[gray, thick, dashed] (-1,1) -- (-1,3) node[above] {};
        \draw[gray, thick, dashed] (-1,1) -- (-3,1) node[above] {$1$};
        \draw[gray, thick, dashed] (-2,2) -- (-3,2) node[above] {$2$};
        \draw[gray, thick, dashed] (-2,2) -- (-2,3) node[above] {};
    
        \node at (2.5, 0.5) {$H_0$};
        \end{tikzpicture}
    \end{center}
    \caption{Level curves of $T^{\ell}(x)$ in \eqref{def:T-ell-function} for $x\in\mathbf{R}^2$. \footnotesize{The dashed curves represent the level curves of $T^{\ell}(x)$ in $\mathbf{R}^2.$ The two L-shaped curves marked by $1$ correspond to the set of values satisfying $T^m(x)=T^s(x)=T^q(x,\Omega)=1$. Similarly, the two L-shaped curves marked by 2 correspond to the set of values satisfying $T^m(x)=2$ or $T^s(x)=T^q(x,\Omega)=4$.}}
    \label{fig:Test-Statistic-Level-Curve}
\end{figure}

\subsection{Uniform versus Pointwise Asymptotic Validity}\label{subsec:uniform-validity}
Consider a generic test $\phi_n=\phi(W_1,\ldots,W_n)$, where $\phi_n=1$ denotes rejection. The test is said to be \textit{uniformly asymptotically valid} at level $\alpha$ over $\mathbf{P}_0$ if
\begin{align}\label{def:uniform-validity}
    \limsup_{n\to\infty} \sup_{P\in\mathbf{P}_0}\mathbb{E}_P[\phi_n] \leq \alpha.
\end{align}
This condition ensures that the test's \textit{size}, defined as the supremum of its rejection probabilities over the null, is controlled less than or close to $\alpha$ in large samples. It is stronger than the commonly imposed requirement of \textit{pointwise asymptotic validity}, under which
\begin{align*}
    \limsup_{n\to\infty} \mathbb{E}_P[\phi_n] \leq \alpha \text{ for all }P\in\mathbf{P}_0.
\end{align*}
The limit is taken for each distribution $P$, so the pointwise asymptotic validity does not guarantee the size control even in large samples. For instance, there may exist a sequence of distributions $\{P_n\}^\infty_{n=1}\in\mathbf{P}_0$ such that $ \mathbb{E}_{P_n}[\phi_n] >\alpha$ for all $n\geq1$, despite $\phi_n$ being pointwise asymptotically valid. 

The distinction between pointwise and uniform asymptotic validity is largely technical in many econometric applications, where pointwise validity can often be strengthened to uniform validity under reasonable assumptions on $\mathbf{P}$. Here, however, the limiting null distribution changes discontinuously with the underlying parameter. Tests that are only pointwise asymptotically valid can therefore exhibit substantial size distortions under empirically relevant distributions, even with a large sample. This concern motivates the emphasis on uniform validity in less well-behaved problems, including inference in partially identified models and one-sided moment inequalities \citep{ImbensManski2004confidence,mikusheva2007uniform,AndrewsGuggenberger2009ECTA-hybrid,AndrewsGuggenberger2010ET,RomanoShaikh2008JSPI,RomanoShaikh2010ECTA}.

\section{Least Favorable Test}\label{sec:LFCtest}
This section develops the least favorable (LF) test and establishes one of the main results of the paper. Section~\ref{subsec:LF-approach} motivates
least favorable approach and explains the theoretical and computational
challenges. Section \ref{subsec:LFC-critical-values} constructs the parametric and nonparametric bootstrap LF procedures. Section~\ref{subsec:LF-result} establishes uniform asymptotic validity of the LF test. Section~\ref{subsec:LF-computation} studies computation of the critical values. Section~\ref{sec:extension-regular-estimators} in the supplementary material extends the parametric construction and its uniform validity result to other regular estimands. 

\subsection{Challenges in Least Favorable Approach}\label{subsec:LF-approach}
The construction of a uniformly asymptotically valid test based on the test statistic $T^{\ell}_n$ in \eqref{def:test-statistic} requires a critical value $\hat{c}^{\ell}_n$ satisfying
\begin{align}\label{def:section3-uniform-validity}
    \limsup_{n\to\infty} \sup_{P\in\mathbf{P}_0} P\left\{T^{\ell}_n > \hat{c}^{\ell}_n\right\} \leq \alpha
\end{align}
for $\alpha\in(0,1)$. A natural choice for $\hat c_n^\ell$ is an estimator of the $1-\alpha$ quantile of the distribution of $T_n^\ell$ under $P$. This choice, however, is infeasible because the distribution cannot be consistently estimated uniformly over $\mathbf{P}_0$. 

To see why, we borrow the framework of \citet{RomanoShaikhWolf2014ECTA} and \cite{CanayShaikh2017ARE}. Consider a distribution function indexed by \(\theta_1\in\mathbf{R}^k\) and $P$
\begin{align}\label{def:J_n-section3}
    J^{\ell}_n(x, \theta_1,P)= P\left\{T^{\ell}(\sqrt{n}S^{-1}_{n} (\bar{W}_n-\mu(P)) + S^{-1}_n\theta_1, \hat{\Omega}_n)\leq x\right\}.
\end{align}
For fixed $\theta_1$, this distribution can be approximated by standard methods. If $\sqrt{n}S^{-1}_{n} (\bar{W}_n-\mu(P))$ is asymptotically normal, bootstrap methods can approximate this, as illustrated in Section \ref{subsec:LFC-critical-values}, and continuity of function $T^{\ell}$ then ensures an approximation of $J^{\ell}_n(x, \theta_1,P)$. The distribution of \(T_n^\ell\) in \eqref{def:test-statistic}  corresponds to \eqref{def:J_n-section3} with the nuisance parameter \(\theta_1=\sqrt{n}\mu(P)\):
\begin{align}\label{def:dist-test-stat}
    J^{\ell}_n(x,\sqrt{n}\mu(P),P)
    =P\left\{T^{\ell}(\sqrt{n}S^{-1}_{n} (\bar{W}_n-\mu(P)) + S^{-1}_n\sqrt{n}\mu(P),\hat{\Omega}_n)\leq x\right\} 
    =P\{T^{\ell}_n\leq x\}.
\end{align}
The nuisance parameter $\sqrt{n}\mu(P)$ varies with $n$ and cannot be estimated consistently uniformly over $\mathbf{P}_0$. For example, when $\mu(P)=0_k$, its natural estimator $\sqrt{n}\bar{W}_n$ converges to a  normal distribution, whereas the nuisance parameter itself equals $0_k$. Substituting this estimator for the nuisance parameter does not, in general, yield a uniformly consistent approximation to \eqref{def:dist-test-stat}.

A natural way to handle this nuisance dependence is least favorable (LF) approach. Rather than estimating the particular nuisance value, the LF critical value uses the largest $1-\alpha$ quantile attainable over all nuisance values satisfying the null. Define
\begin{align}\label{def:c_n}
    c^{\ell}_n(1-\alpha,P)=\sup\{(J^{\ell}_n)^{-1}(1-\alpha, \theta_1, P): \theta_1\in\mathbf{R}^k_+\cup\mathbf{R}^k_-\}
\end{align}
where $ (J^{\ell}_n)^{-1}(\xi,\theta_1, P)=\inf\{x\in\mathbf{R}:J^{\ell}_n(x,\theta_1,P) \geq \xi\}$ for any $\xi\in(0,1)$. By construction, $ c^{\ell}_n(1-\alpha,P)$ is no smaller than the $1-\alpha$ quantile of $T^{\ell}_n$ under the null. Hence, the oracle test that rejects when $T^{\ell}_n> c^{\ell}_n(1-\alpha,P)$ controls the size. 

The challenge, however, arises from the optimization in \eqref{def:c_n}: the nuisance configurations determining the supremum, referred to as least favorable configurations, generally depend on the underlying distribution $P$. This contrasts with the least favorable approach used in the literature on testing one-sided moment inequalities where the least favorable configuration is the known boundary point $0_k$, regardless of $P$. Their optimization therefore reduces to evaluating the quantile at this fixed nuisance value; see \citet{CanayShaikh2017ARE}. This simplification is generally unavailable in our setting.

Previous applications of the least favorable approach to sign agreement have focused on special cases. Under independence, analytical solutions are available; see, for example, \citet{GailSimon1985Biometrics} and \citet{PiantadosiGail1993}. Correlated cases have been studied in the bivariate setting by \citet{RussekSimon1993Biometrics}, \citet{Kaplan2015}, and \MMS{}. Extending the approach beyond these settings requires handling optimization over an unbounded nuisance space. Section~\ref{subsec:LFC-critical-values} constructs the LF critical values and establishes their theoretical properties. Section~\ref{subsec:LF-computation} develops computable lower and upper bounds that contain the exact critical value with a prespecified Monte Carlo probability and converge to it as the approximation is refined.

\subsection{Least Favorable Critical Values} \label{subsec:LFC-critical-values}
To construct estimators for $c^{\ell}_n(1-\alpha, P)$ in \eqref{def:c_n}, we first need to consistently estimate the function $J^{\ell}_n(x, \theta_1, P)$ for fixed $\theta_1\in\mathbf{R}^k_+\cup\mathbf{R}^k_-$ and $x\in\mathbf{R}$. For convenience, we restate this here:
\begin{align}\label{def:Jn-long}
    J^{\ell}_n (x, \theta_1, P) =P\{T^{\ell}(\sqrt{n}S^{-1}_n (\bar{W}_n- \mu(P))+ S^{-1}_n \theta_1 ) \leq x\}.
\end{align}
Although the centered studentized sample mean $\sqrt{n}S^{-1}_{n} (\bar{W}_n-\mu(P))$ depends on the unknown mean $\mu(P)$, its distribution is asymptotically normal under the maintained assumption and can therefore be estimated by parametric and nonparametric bootstrap methods.

We begin with a critical value based on parametric bootstrap. Define the function \(J^\ell\) indexed by \(\theta_1\in\mathbf R^k\) and \(\theta_2\in\bar{\mathbf O}\) by replacing the asymptotically normal component with the normal vector \(\theta_2^{1/2}Z\) where \(Z\sim N(0_k,I_k)\) and by reparameterizing \(S_n^{-1}\theta_1\) as \(\theta_1\):
\begin{align}\label{def:J} 
    J^\ell(x, \theta_1,\theta_2) = P\left\{T^\ell(\theta^{1/2}_2Z+ \theta_1, \theta_2) \leq x\right\}
\end{align}
where the second argument of \(T^\ell\) is suppressed for \(\ell=m,s\). Analogously to \(c_n^\ell(1-\alpha,P)\) in \eqref{def:c_n}, define the critical value as the largest $1-\alpha$ quantile attainable over all $\theta_1$ satisfying the null:
\begin{align}\label{def:LF-cv-normal}
\begin{aligned}
     {c}^{\ell}(1-\alpha, \hat\Omega_n) 
    &=\sup\left\{(J^{\ell})^{-1}(1-\alpha, \theta_1, \hat{\Omega}_n) : \theta_1\in\mathbf{R}^k_+ \cup\mathbf{R}^k_-\right\}.
\end{aligned}   
\end{align}
By the symmetry of $T^\ell$ and of the centered normal distribution, the supremum in \eqref{def:LF-cv-normal} may equivalently be taken over $\theta_1\in\mathbf R_+^k$. For $k=2$, this critical value coincides with that studied by \citet{RussekSimon1993Biometrics} and \MMS{}. The critical value is well-defined as shown in Lemma~\ref{lem:parametric-cv-existence}--\ref{lem:argmax-critical-value-all}. Furthermore, the mapping $\Omega\mapsto c^\ell(1-\alpha,\Omega)$ continuous in $\bar{\mathbf{O}}$ as established in Lemma~\ref{lem:continuity-critical-value-all}, which allows us to replace the unknown population correlation matrix by a consistent estimator and ensures numerical stability. Because the least-favorable configuration solving \eqref{def:LF-cv-normal} depends on the correlation, tractable closed-form expressions are available only in special cases. Section~\ref{subsec:LF-computation} develops a procedure for computing this critical value numerically.

Under independence, the optimization admits an analytical solution. When $\hat\Omega_n = I_k$, 
\begin{align}
     {c}^{\ell}(1-\alpha, I_k) &=\left\{ 
    \begin{matrix}
        \Phi^{-1}((1-\alpha)^{\frac{1}{k-1}}) & \quad \ell=m,\\
        (1-\alpha)\text{-quantile of }\sum^{k-1}_{r=0}{k-1 \choose r}2^{-(k-1)}\mathcal L(\chi^2_r) & \quad \ell=s, q
    \end{matrix}\right.\label{def:LF-cv-independent}
\end{align}
\noindent where the second line denotes a chi-bar-square mixture that assigns weight \({k-1\choose r}2^{-(k-1)}\) to a chi-square distribution with \(r\) degrees of freedom, and \(\chi_0^2\) is degenerate at zero. This simplification arises because the least favorable configuration solving \eqref{def:LF-cv-normal} is $(\infty,0,\ldots,0)$, or their permutations, as shown by \cite{GailSimon1985Biometrics} and \cite{PiantadosiGail1993} under $T^s$ and $T^m$ respectively. The same critical value with $\ell=m$ under $k=2$ also appeared in \cite{Kaplan2015}.

For the max statistic, an analytical simplification extends beyond independence to arbitrary nonnegative pairwise correlations.
\begin{proposition}\label{prop:LF-max-nonnegative-correlation} Let $\Omega$ be a correlation matrix with nonnegative pairwise
correlations. Then
\begin{align}
c^m(1-\alpha,\Omega)
=
\max_{1\le j\le k}
\inf\left\{
c\in\mathbf R:
\Phi_{k-1}
\left(c\mathbf 1_{k-1};\Omega_{-j,-j}\right)
\ge 1-\alpha
\right\}
\label{eq:LF-max-nonnegative-correlation}
\end{align}
for $\alpha\in(0,1/2)$ where $\Phi_{k-1}(\cdot;\Omega)$ denotes the CDF of $N(0_{k-1},\Omega)$, $\mathbf 1_{k-1}$ is the $(k-1)$-dimensional vector of ones, and $\Omega_{-j,-j}$ is the principal submatrix obtained by deleting the $j$th row and column.
\end{proposition}
\noindent Thus, when $\hat\Omega_n$ has nonnegative pairwise correlations, $c^m(1-\alpha,\hat\Omega_n)$ can be evaluated without implementing the numerical optimization. The key observation is that the least favorable configuration is one of the permutations of $(\infty,0,\ldots,0)$. The resulting critical value is thus determined by the $(k-1)$-dimensional normal distribution of the remaining coordinates. In the bivariate case, \eqref{eq:LF-max-nonnegative-correlation} recovers the result of \citet{RussekSimon1993Biometrics} that the critical value under non-negative pairwise correlation coincides with that under independence; see \MMS{} for a self-contained proof. 

We next introduce the nonparametric bootstrap counterpart of \eqref{def:LF-cv-normal}. Let $\{W_1^*,\ldots,W_n^*\}$ be a sample drawn with replacement from the empirical distribution $\hat P_n$, and let $\bar W_n^*$ and $S_n^*$ denote the corresponding sample mean and diagonal matrix of sample standard deviations. Conditional on the data, the distribution of $T^\ell( \sqrt n\,S_n^{*-1}(\bar W_n^*-\bar W_n) +S_n^{*-1}\theta_1 )$ approximates $J_n^\ell(\cdot,\theta_1,P)$ for fixed $\theta_1$; see
\citet{RomanoShaikh2012AoS}. Define the nonparametric bootstrap LF critical value by
\begin{gather}
\begin{aligned}
    {c}^{\ell*}_n(1-\alpha,\hat{P}_n) 
    &= \sup\left\{(J^{\ell}_n)^{-1}(1-\alpha, \theta_1, \hat{P}_n): \theta_1\in\mathbf{R}^k_+ \cup\mathbf{R}^k_- \right\}.
\end{aligned}\label{def:LFC-bootstrap-critical-value}
\end{gather}

The nonparametric bootstrap critical value in \eqref{def:LFC-bootstrap-critical-value} differs from the critical value obtained by naively bootstrapping $T^\ell_n$. Given a bootstrap sample $\{W^*_1,\ldots,W^*_n\}$, analysts might compute the bootstrap analogue
\begin{align*}
T^{\ell*}_n=T^{\ell}\left(\left(\tfrac{\sqrt{n}\bar{W}^*_{1,n}}{S^*_{1,n}},...,\tfrac{\sqrt{n}\bar{W}^*_{k,n}}{S^*_{k,n}}\right),\hat\Omega^*_n\right) 
\end{align*}
and use its $1-\alpha$ quantile as the critical value. The resulting bootstrap distribution, however, need not consistently estimate the relevant null distribution. Such failures are well known for nonregular statistics involving nondifferentiable operations such as maxima and minima; see, for example, \cite{bickel1981some}, \citet{bickel1997resampling}, and \cite{andrews2000inconsistency}. In the bivariate setting, \citet{kowalski2023RES} uses a heuristic nonparametric bootstrap procedure under independence, which \MMS{} show to be invalid under general correlation. In a general setting, \OEH{} propose a nonparametric bootstrap procedure that also fail to control size, as we show in Example \ref{ex:oulhaj-bootstrap-size}.

With the test statistic $T^\ell_n$ in \eqref{def:test-statistic} and two critical values $ c^{\ell}_n(1-\alpha, \hat{\Omega}_n)$ and $ c^{\ell*}_n(1-\alpha, \hat P_n)$ in \eqref{def:LF-cv-normal} and \eqref{def:LFC-bootstrap-critical-value}, we define the LF test by
\begin{align}\label{def:LFtest}
    \phi^{\ell}_n=I\left\{T^{\ell}_n> \hat{c}^{\ell}_n\right\} \quad\text{ where }\quad \hat{c}^{\ell}_n={c}^{\ell}(1-\alpha,\hat{\Omega}_n) \quad\text{or}\quad c^{\ell*}_n(1-\alpha,\hat{P}_n)
\end{align}
for each $\ell=m,s,q$. Analysts may choose among the three test statistics $\ell=m,s,q$ and between parametric and nonparametric bootstrap critical values.

Simulation results in Section~\ref{sec:MonteCarlo-Simulation} indicate that the nonparametric bootstrap critical value delivers null rejection probabilities closer to the nominal level than its parametric counterpart when the underlying distribution is skewed. This finding is consistent with \citet{andrewsbarwick2012ecta}, who report improved finite-sample size performance of bootstrap critical values for one-sided moment inequality tests.\footnote{The simulation evidence favors the bootstrap critical value despite the fact that the bootstrap should not be expected to produce an asymptotic refinement in their context.} 


\begin{remark}\label{remark:regularized-correlation}\normalfont
For \(\ell=q\), if \(\hat\Omega_n\) is not positive definite, replace it by $\tilde\Omega_n=
({\hat\Omega_n+\varepsilon_n I_k})/({1+\varepsilon_n}),$ where \(\varepsilon_n>0\) and \(\varepsilon_n\to0\), in both the test statistic \(T_n^q\) and the critical value \(\hat c_n^q\).
\end{remark}

\begin{remark}\label{rem:OEH-comparison}
\normalfont
Among existing procedures developed for sign agreement, the test of \OEH{} is, to the best of our knowledge, the only one that accommodates arbitrary correlation at any fixed finite $k$. Their Gaussian procedure combines two one-sided moment-inequality max tests with least favorable critical values through the intersection--union principle \citep{Berger1982Technometrics}, rejecting only when both tests reject. This test is conservative in the Gaussian experiment with a positive definite correlation matrix. Importantly, our LF procedure accounts directly for joint rejection under the union null and yields a weakly smaller critical value. Using the same studentized estimators and estimated correlation matrix, our parametric max test is therefore uniformly at least as powerful as theirs; see Proposition~\ref{prop:oulhaj-dominance} in Section~\ref{sec:oulhaj-comparison}. This construction also suggests potential alternatives based on other one-sided moment-inequality tests; Remark~\ref{rem:IUT-RSW} considers one such alternative.
\end{remark}


\subsection{Uniform Asymptotic Validity of the Least Favorable Test}\label{subsec:LF-result}
We now establish the asymptotic validity of the LF tests. We first state two regularity conditions used in the analysis.

\begin{assumption}\label{ass:uniform-integrability}
Let $W_i$, $i=1,\ldots,n$, be an i.i.d. sequence of random vectors with distribution $P\in\mathbf P$ on $\mathbf R^k$. The class $\mathbf P$ satisfies
\begin{align}\label{def:uniform-integrability}
\lim_{\lambda\to\infty}
\sup_{P\in\mathbf P}
\mathbb E_P\left[
\left(
\frac{W_{j,1}-\mu_j(P)}{\sigma_j(P)}
\right)^2
I\left\{
\left|
\frac{W_{j,1}-\mu_j(P)}{\sigma_j(P)}
\right|>\lambda
\right\}
\right]
=0
\qquad
\text{for }j=1,\ldots,k.
\end{align}
\end{assumption}

\begin{assumption}\label{ass:uniform-eigenvalue}
The smallest eigenvalue of the correlation matrix is bounded away from zero uniformly over $\mathbf P$:
\begin{align}\label{def:uniform-eigenvalue}
\inf_{P\in\mathbf P}
\lambda_{\min}\{\Omega(P)\}>0.
\end{align}
\end{assumption}
The uniform integrability condition in \eqref{def:uniform-integrability} is mild and standard. It is satisfied, for example, if a moment higher than the second exists uniformly over $\mathbf P$, i.e., for some $\delta>0$,
\begin{align*}
\sup_{P\in\mathbf P}
\mathbb E_P\left[
\left|
\frac{W_{j,1}-\mu_j(P)}{\sigma_j(P)}
\right|^{2+\delta}
\right]
<\infty
\qquad
\text{for }j=1,\ldots,k.
\end{align*}
This provides the uniform tail control needed for the uniform central limit theorem and uniform consistency of the sample variances and correlations used below. Conditions of this form are standard in uniform inference; see, among others, \cite{Romano2004SJS-uniformity}, Romano and Shaikh (\citeyear{RomanoShaikh2008JSPI}, \citeyear{RomanoShaikh2010ECTA}, \citeyear{RomanoShaikh2012AoS}), \cite{AndrewsGuggenberger2009ET}, and \cite{AndrewsSoares2010ECTA}. Assumption~\ref{ass:uniform-eigenvalue} is needed only for the QLR statistic. It keeps the correlation uniformly away from singularity, ensuring stability of the inverse correlation entering $T^q$.

Under these conditions, we establish uniform asymptotic validity of the LF test.

\begin{theorem}\label{thm:LFtest-uniform-validity}
Suppose Assumption~\ref{ass:uniform-integrability} holds. Then, for $\alpha\in(0,\tfrac12)$ and $\ell\in\{m,s\}$, the least favorable test $\phi_n^\ell$ in \eqref{def:LFtest} satisfies
\begin{align}\label{def:LFtest-uniform-validity}
\limsup_{n\to\infty}
\sup_{P\in\mathbf P_0}
\mathbb E_P[\phi_n^\ell]
\leq \alpha.
\end{align}
If, in addition, Assumption~\ref{ass:uniform-eigenvalue} holds, then \eqref{def:LFtest-uniform-validity} also holds for $\ell=q$.
\end{theorem}

Theorem~\ref{thm:LFtest-uniform-validity} bounds the asymptotic size of the LF test by $\alpha$. This bound need not be attained for an arbitrary class $\mathbf P$. The LF critical value accounts for the entire null nuisance space, whereas $\mathbf P$ may exclude sequences of distributions approaching a least favorable configuration. The critical value may then be asymptotically larger than the relevant quantiles of the test statistic, making the test conservative. The bound is sharp, however, when $\mathbf P$ is sufficiently rich. The following corollary shows that allowing location shifts is sufficient.

\begin{corollary}\label{cor:LFCtest-uniform-validity-equality}
Suppose Assumption~\ref{ass:uniform-integrability} holds and
$\mathbf P$ contains all location shifts of some mean-zero
distribution with identity correlation matrix. For $\alpha\in(0,\tfrac12)$, the LF tests in \eqref{def:LFtest} with the parametric critical values in \eqref{def:LF-cv-normal} satisfy \eqref{def:LFtest-uniform-validity} with equality for $\ell\in\{m,s\}$. If, in addition, Assumption~\ref{ass:uniform-eigenvalue} holds, the same conclusion holds for $\ell=q$.
\end{corollary}

To our knowledge, Theorem~\ref{thm:LFtest-uniform-validity} provides the first uniform asymptotic validity result for tests of sign agreement that allows an arbitrary finite number of parameters, unrestricted dependence, and a broad class of nonparametric distributions. Classical results instead establish finite-sample validity within restricted classes of normal distributions, typically assuming independence \citep{GailSimon1985Biometrics,PiantadosiGail1993,Silvapulle2001Biometrics} or, in the bivariate case, a known correlation coefficient \citep{RussekSimon1993Biometrics}. In contemporaneous and independently developed work, \MMS{} establish uniform asymptotic validity for the bivariate problem under high-level uniform convergence conditions.

The existing finite-sample validity results under normal distributions or pointwise asymptotic normality do not directly extend to the present setting. In particular, the maximum and minimum operations make the asymptotic distribution sensitive to means that are within $O(n^{-1/2})$ of zero. Different local drifts toward the boundary can therefore induce different limiting distributions, a phenomenon related to the boundary problem studied by \citet{andrews2000inconsistency}. We instead draw on techniques developed for one-sided moment inequality tests, particularly \cite{AndrewsSoares2010ECTA}, \cite{AndrewsGuggenberger2009ET}, and \citet{RomanoShaikhWolf2014ECTA}. Separately, continuity of the critical value established in Lemma~\ref{lem:continuity-critical-value-all} plays a key role showing nonconservativeness in Corollary \ref{cor:LFCtest-uniform-validity-equality}.

Although Theorem \ref{thm:LFtest-uniform-validity} is stated for a vector of sample means, this formulation is not essential. Section \ref{sec:extension-regular-estimators} in the supplementary material extends the LF test, as well as the two-step test introduced in Section \ref{sec:two-step-test}, to general regular estimators. In particular, the same uniform validity conclusions hold when the estimator admits a uniformly
asymptotically linear representation with uniformly integrable standardized influence functions and consistently estimated marginal scales and correlation matrix; see Corollary \ref{cor:general-estimator-uniform-validity}. Under standard regularity conditions, this framework includes commonly used estimands such as OLS, strongly identified IV and two-stage least-squares, and difference-in-differences coefficients.

\begin{remark}\normalfont
The proposed test need not similar or unbiased at boundary. In the independent bivariate Gaussian setting, \citet{Berger1989JASA} constructs a test uniformly more powerful than the likelihood-ratio test, and \citet{Zelterman1990SPL} numerically constructs an approximately boundary-similar test. \citet{Kaplan2015} numerically constructs nearly boundary-similar and nearly unbiased tests allowing correlation, while \MMS{} exhibit an exactly boundary-similar and unbiased test under independence. These constructions sacrifice rejection monotonicity: a test may reject some estimates yet fail to reject after one or both coordinate increase in absolute value without changing their signs.
\end{remark}

\subsection{Computing the Least Favorable Critical Value}
\label{subsec:LF-computation}
Beyond the special cases admitting analytical characterizations, computing the parametric LF critical value in \eqref{def:LF-cv-normal} generally requires numerical optimization. The solution depends on the underlying distribution, and the optimization is over an unbounded nuisance space. We exploit the structure of the test statistic to construct computable lower and upper bounds that contain the exact critical value with a prespecified Monte Carlo probability.

We construct the lower and upper bounds by treating a compact region and its complement separately. First, for a truncation level $A>0$, partition $[0,A]^k$ into finitely many cells and choose a representative nuisance value in each cell. The largest $(1-\alpha)$-quantile at these representatives gives a lower bound on the supremum over the compact region. To obtain an upper bound, a Gaussian-shift comparison raises the quantile level at each representative sufficiently to bound the $(1-\alpha)$-quantile throughout its cell. The largest of these adjusted quantiles then bounds the compact region from above. Second, outside $[0,A]^k$, we obtain an upper bound using $k$ one-sided envelope statistics, one for each coordinate that can exceed $A$. A lower bound for this region is obtained from the limiting quantiles as one coordinate diverges while the others remain zero. Taking the maximum of the two regional lower bounds, and likewise of the two regional upper bounds, yields bounds on the LF critical value over the entire nuisance space. This construction reduces the unbounded optimization to finitely many quantile calculations; Appendix~\ref{appendix:computation-lfcv} provides the details.

The required quantiles are themselves evaluated by simulation. Using independent draws and order-statistic bounds, we attach a probabilistic guarantee to the resulting computational bounds. In particular, the lower and upper endpoints contain the exact parametric critical value in \eqref{def:LF-cv-normal} with a prespecified Monte Carlo probability and converge to it as the computational approximation is refined.

\begin{proposition}\label{prop:LF-computational-accuracy}
Fix $\alpha\in(0,1/2)$, $\Omega\in\mathbf O$, and
$\ell\in\{m,s,q\}$. For each computational refinement level $r$, let $\hat c_{L,r}^\ell$ and $\hat c_{U,r}^\ell$ denote the lower and upper bounds constructed in Appendix~\ref{appendix:computation-lfcv}, using
independent draws and an overall Monte Carlo failure probability
$\delta_r$. Then
\begin{align}
P_{\mathrm{MC}}\left\{
    \hat c_{L,r}^\ell
    \leq c^\ell(1-\alpha,\Omega)
    \leq \hat c_{U,r}^\ell
\right\}
\geq 1-\delta_r ,
\label{eq:LF-computational-bracket}
\end{align}
where $P_{\mathrm{MC}}$ denotes probability with respect to the Monte Carlo draws. Moreover, if the truncation region expands, the partition mesh converges to zero, the Monte Carlo sample size increases sufficiently fast, and
$\delta_r\to0$, then
\begin{align}
\hat c_{L,r}^\ell
\overset{P_{\mathrm{MC}}}{\longrightarrow}
c^\ell(1-\alpha,\Omega),
\qquad
\hat c_{U,r}^\ell
\overset{P_{\mathrm{MC}}}{\longrightarrow}
c^\ell(1-\alpha,\Omega).
\label{eq:LF-computational-convergence}
\end{align}
\end{proposition}

Proposition~\ref{prop:LF-computational-accuracy} therefore provides more
than a point approximation to the optimization in
\eqref{def:LF-cv-normal}: it gives probabilistically valid computational
bounds and shows that their width can be made arbitrarily small. When a
single critical value is desired, we use the midpoint
\begin{align}
    \hat c_{M,r}^\ell
    :=
    \frac{\hat c_{L,r}^\ell+\hat c_{U,r}^\ell}{2}.
\end{align}
Alternatively, the upper endpoint may be used as a conservative computational choice. The probability in \eqref{eq:LF-computational-bracket} concerns only the Monte Carlo computation for a fixed $\Omega$. It is distinct from the sampling probability.

For the nonparametric bootstrap critical value, we use a numerical adaptation of the parametric computation. Conditional on the observed sample, we replace the Gaussian-shift comparisons with bounds evaluated separately for each bootstrap draw and search both null orthants, since the bootstrap distribution need not be sign symmetric. The largest bootstrap $(1-\alpha)$-quantile evaluated at selected null nuisance values provides a numerical lower estimate. Bounds covering all nuisance values within each finite cell and the remaining unbounded region provide a numerical upper estimate. We use the midpoint of these estimates as the critical value.

\section{Two-step Test}\label{sec:two-step-test}
The two-step test is designed to improve power by incorporating sample information about the parameter into the choice of critical value. Section~\ref{subsec:two-step-test} presents its parametric and nonparametric bootstrap implementations. Section~\ref{subsec:two-step-uniform-validity} establishes uniform asymptotic validity. Section~\ref{sec:extension-regular-estimators} in the supplementary material extends the parametric construction and its uniform validity result to other regular estimands.

\subsection{Two-Step Procedure}\label{subsec:two-step-test}
The procedure consists of two steps. First, we construct a confidence region for the mean vector. Next, we maximize the critical value only over the values in this confidence region that satisfy the null hypothesis. We begin with the parametric bootstrap implementation.

Fix \(\tau\in(0,\alpha)\). The first step constructs a \(1-\tau\) confidence region for the mean vector. Define
\begin{align}\label{def:two-step-normal-CS}
{\mathbf C}_n(1-\tau)
:=
\Bigg\{
\mu\in\mathbf R^k:
n
\left(\bar W_n-\mu\right)
^{\top}
\hat\Sigma_n^{-1}
\left(\bar W_n-\mu\right)
\leq
\chi^2_{k,1-\tau}
\Bigg\},
\end{align}
where \(\chi^2_{k,1-\tau}\) denotes the \(1-\tau\) quantile of a chi-squared distribution with \(k\) degrees of freedom and $\hat\Sigma_n=S_n^{\top} \hat\Omega_n S_n$ is the sample variance-covariance matrix.

For the next step, recall from \eqref{def:J} that
\begin{align*}
J^\ell(x,\theta_1,\theta_2)
=
P\left\{
T^\ell(\theta_2^{1/2}Z+\theta_1,\theta_2)
\leq x
\right\},
\qquad Z\sim N(0_k,I_k).
\end{align*}
The location parameter \(\theta_1\) in \(J^\ell\) is expressed on the same studentized \(\sqrt n\)-scale as the test statistic, so for a candidate mean vector \(\mu\), the corresponding location parameter is $\theta_1=\sqrt nS_n^{-1}\mu.$ Whenever the intersection $\mathbf C_n(1-\tau)\cap\left(\mathbf R_+^k\cup\mathbf R_-^k\right)$ is nonempty, define the parametric two-step critical value as
\begin{align}\label{def:two-step-normal-cv}
\tilde c_n^\ell
:=
\sup_{\mu\in
\mathbf C_n(1-\tau)
\cap
(\mathbf R_+^k\cup\mathbf R_-^k)}
\left(J^\ell\right)^{-1}
\left(
1-\alpha+\tau,\,
\sqrt n\,S_n^{-1}\mu,\,
\hat\Omega_n
\right),
\qquad
\ell\in\{m,s,q\}.
\end{align}
Relative to the LF critical value in \eqref{def:LF-cv-normal}, the maximization in \eqref{def:two-step-normal-cv} is restricted to the null values admitted by the first-step confidence region, and the quantile level is increased from \(1-\alpha\) to \(1-\alpha+\tau\). The latter adjustment accounts for the probability \(\tau\) allocated to the first step. Under alternatives for which the first-step confidence region tends to exclude null configurations yielding large quantiles, restricting the maximization can lower the critical value. Power gains are possible when this reduction outweighs the increase in the quantile level, a tradeoff examined in the local asymptotic power analysis in Section~\ref{sec:power-comparison}.

The first-step confidence region gives rise to three cases. If it does not intersect the null parameter space, the test rejects. If it lies entirely within the null parameter space, the test does not reject. Otherwise, the test uses the restricted critical value in \eqref{def:two-step-normal-cv}. Therefore, the two-step test is
\begin{align}\label{def:two-step-normal-test}
\psi_n^{\ell}
:=
\begin{cases}
1,
&
\mathbf C_n(1-\tau)
\cap
(\mathbf R_+^k\cup\mathbf R_-^k)
=
\varnothing,
\\[5pt]
0,
&
\mathbf C_n(1-\tau)
\subseteq
\mathbf R_+^k\cup\mathbf R_-^k,
\\[5pt]
I\left\{
T_n^\ell>
\tilde c_n^\ell
\right\},
&
\text{otherwise},
\end{cases}
\qquad
\ell\in\{m,s,q\}.
\end{align}

This construction follows the two-step approach of \citet{RomanoShaikhWolf2014ECTA}, adapting their method for one-sided moment inequalities to sign agreement. Related ideas that use confidence regions to restrict nuisance values and adjust for possible noncoverage appear in \citet{BergerBoos1994JASA} and \citet{silvapulle1996JASA}.

We next introduce the nonparametric bootstrap procedure. Let \(\{W_1^*,\ldots,W_n^*\}\) be obtained by drawing \(n\) observations independently with replacement from the empirical distribution \(\hat P_n\). Let \(\bar W_n^*\) and \(\hat\Sigma_n^*\) denote the bootstrap analogues of \(\bar W_n\) and \(\hat\Sigma_n\), respectively. Let \(\hat P_n^*\) denote probability conditional on the original sample. The first-step bootstrap confidence region is
\begin{align}\label{def:two-step-bootstrap-CS}
\mathbf C_n^*(1-\tau)
:=
\Bigg\{
\mu\in\mathbf R^k:
n
\left(\bar W_n-\mu\right)
^{\top}
\hat\Sigma_n^{-1}
\left(\bar W_n-\mu\right)
\leq
r_{n,\tau}^*
\Bigg\}
\end{align}
where 
\begin{align}\label{def:two-step-bootstrap-radius}
r_{n,\tau}^*
:=
\inf\left\{
r\in\mathbf R:
\hat P_n^*\{n (\bar W^*_n - \bar W_n)^{\top}(\hat{\Sigma}^{*}_n)^{-1}(\bar W^*_n - \bar W_n)\leq r\}
\geq1-\tau
\right\}
\end{align}
denotes the conditional \(1-\tau\) quantile of nonparametric bootstrap counterpart of the chi-squared distribution of degree $k$ conditional on the sample. 

We define the second step critical value similarly. Given the nonparametric bootstrap distribution function \(J_n^\ell(x,\theta_1,\hat P_n)\), where $J_n^\ell$ is in \eqref{def:Jn-long}, define the nonparametric bootstrap two-step critical value as
\begin{align}\label{def:two-step-bootstrap-cv}
\tilde c_n^{\ell *}
:=
\sup_{\mu\in
\mathbf C_n^*(1-\tau)
\cap
(\mathbf R_+^k\cup\mathbf R_-^k)}
\left(J_n^\ell\right)^{-1}
\left(
1-\alpha+\tau,\,
\sqrt n\mu,\,
\hat P_n
\right),
\qquad
\ell\in\{m,s,q\}.
\end{align}
The nonparametric bootstrap two-step test is then
\begin{align}\label{def:two-step-bootstrap-test}
\psi^{\ell*}_n
:=
\begin{cases}
1,
&
\mathbf C_n^*(1-\tau)
\cap
(\mathbf R_+^k\cup\mathbf R_-^k)
=
\varnothing,
\\[5pt]
0,
&
\mathbf C_n^*(1-\tau)
\subseteq
\mathbf R_+^k\cup\mathbf R_-^k,
\\[5pt]
I\left\{
T_n^\ell>
\tilde c_n^{\ell *}
\right\},
&
\text{otherwise},
\end{cases}
\qquad
\ell\in\{m,s,q\}.
\end{align}
Thus, for both implementations, the second-step critical value needs to be calculated only when the first-step confidence region contains both null and alternative parameter values.

The first-step confidence region need not take the ellipsoidal form used above. It may instead be replaced by other simultaneous confidence regions such as those based on a maximum or sum-of-squares statistic. These alternatives avoid covariance inversion in the first step and are therefore relevant when $\Omega(P)$ is singular. More generally, the statistic used to construct the first-step confidence region need not coincide with the statistic \(T^\ell\) used in the second-step sign-agreement test.

\begin{remark}\label{rem:IUT-RSW}
\normalfont
Following the discussion in Remark~\ref{rem:OEH-comparison}, a potential alternative to our two-step test combines two one-sided nonparametric bootstrap tests of \citet{RomanoShaikhWolf2014ECTA}, one for each of the two orthants in the null. We denote this test IUT--RSW and include it in the simulations in Section~\ref{sec:MonteCarlo-Simulation}. Both the parametric and nonparametric bootstrap versions of our two-step test have higher estimated average power than IUT--RSW when compared using the same statistic.
\end{remark}

\subsection{Uniform Asymptotic Validity}
\label{subsec:two-step-uniform-validity}

The two-step test is uniformly asymptotically valid under the same conditions in Theorem \ref{thm:LFtest-uniform-validity}.


\begin{theorem}\label{thm:two-step-uniform-validity}
Suppose Assumptions \ref{ass:uniform-integrability} and \ref{ass:uniform-eigenvalue} hold. Then, for every fixed $\alpha\in(0,1/2)$, $\tau\in(0,\alpha)$, and
$\ell\in\{m,s,q\}$, the parametric bootstrap two-step test
$\psi_n^\ell$ in \eqref{def:two-step-normal-test} satisfies
\begin{align}\label{eq:two-step-normal-uniform-validity}
    \limsup_{n\to\infty}\sup_{P\in\mathbf P_0}
    \mathbb E_P[\psi_n^\ell]\le\alpha.
\end{align}
The nonparametric bootstrap two-step test $\psi_n^{\ell *}$ in
\eqref{def:two-step-bootstrap-test} likewise satisfies
\begin{align}\label{eq:two-step-bootstrap-uniform-validity}
    \limsup_{n\to\infty}\sup_{P\in\mathbf P_0}
    \mathbb E_P[\psi_n^{\ell *}]\le\alpha.
\end{align}
\end{theorem}
The two-step test is uniformly asymptotically valid over a broad nonparametric class under the same uniform integrability as the LF test, together with the nonsingularity requirement. Because the ellipsoidal first-step confidence region in \eqref{def:two-step-normal-CS} involves \(\hat\Sigma_n^{-1}\), we require the smallest eigenvalue of \(\Omega(P)\) to be bounded away from zero uniformly over $\mathbf{P}$. Assumption \ref{ass:uniform-eigenvalue} can be dispensed with when both the first-step confidence region and the second-step test use max or sum-of-squares statistics, which do not require inversion of the covariance.

The role of \(\tau\) can be seen directly from the validity argument. The probability that the first-step confidence region fails to contain the true mean is asymptotically bounded by \(\tau\). When the true mean is contained in the confidence region, it is included among the nuisance values over which the second-step critical value is maximized, and the \(1-\alpha+\tau\) quantile corresponds to an upper-tail probability of \(\alpha-\tau\). Combining these two sources of error yields an overall asymptotic rejection probability no greater than \(\alpha\).

The uniform asymptotic validity result holds for every fixed \(\tau\in(0,\alpha)\), although finite-sample performance can depend on its choice. A larger \(\tau\) produces a smaller first-step confidence region and hence a more restricted nuisance-parameter search, but it also raises the second-step quantile from \(1-\alpha\) to \(1-\alpha+\tau\). Thus, the choice of \(\tau\) trades off tighter first-step localization against a more stringent second-step critical value. Following \citet[Remark~2.5]{RomanoShaikhWolf2014ECTA}, we set \(\tau=0.1\alpha\) in our numerical implementation.

\section{Local Power Comparison}\label{sec:power-comparison}
In this section, we compare the rejection probabilities of the two proposed tests under local alternatives. For tractability, we focus on their parametric implementations and assume that the population covariance matrix is known and equal to \(I_k\).

For \(A,B\geq0\) and integers \(h,r\geq1\) satisfying \(h+r\leq k\), consider sequences \(P_n\in\mathbf P\) such that
\begin{align}\label{eq:local-alternative}
\sqrt n\,\mu(P_n)
\longrightarrow
\theta^{A,B,h,r}
:=
\left(
\underbrace{A,\ldots,A}_{h\text{ times}},
\underbrace{-B,\ldots,-B}_{r\text{ times}},
\underbrace{0,\ldots,0}_{k-h-r\text{ times}}
\right)^{\top}.
\end{align}
For fixed \(A,B>0\), \eqref{eq:local-alternative} describes \(n^{-1/2}\)-local alternatives. Under the condition \(\Sigma(P_n)=\Omega(P_n)=I_k\), the constants \(A\) and \(B\) measure the positive and negative mean components in units of their standard errors. In the analysis below, we consider large values of \(A\) to study how increasingly strong evidence about some positive signs affects the ability to detect negative means that remain local to zero.

To state the comparison, for $d\geq1$ and $\xi\in(0,1/2)$ define
\begin{align}\label{eq:power-one-sided-quantiles}
\begin{split}
a_d^m(\xi)
:=\Phi^{-1}\bigl((1-\xi)^{1/d}\bigr),
\qquad a_d^s(\xi)=a_d^q(\xi)
:=\inf\bigg\{
t\in\mathbf R:
\sum_{j=0}^d 2^{-d}\binom dj F_{\chi_j^2}(t)\geq1-\xi
\bigg\}
\end{split}
\end{align}
where $F_{\chi_j^2}$ denotes the chi-squared distribution function with
$j$ degrees of freedom and $F_{\chi_0^2}(t)=I\{t\geq0\}$. These are the
$1-\xi$ quantiles of $\max_{j\leq d}(-Z_j)$ and $\sum_{j=1}^d Z_j^2I\{Z_j\leq0\}$, respectively, for $Z\sim N_d(0_d,I_d)$. In particular, $c^\ell(1-\alpha,I_k)=a_{k-1}^\ell(\alpha),$ for $\ell\in\{m,s,q\}.$ The following proposition considers three power regimes.

\begin{proposition}\label{prop:power-comparison}
Suppose that Assumption \ref{ass:uniform-integrability} holds and  $\Sigma(P_n)=\Omega(P_n)=I_k$. Fix $\alpha\in(0,1/2)$ and $\tau\in(0,\alpha)$. The least favorable test $\phi_n$ in \eqref{def:LFtest} and the two-step test $\psi_n$ in \eqref{def:two-step-normal-test} with $\hat{\Sigma}_n=\hat{\Omega}_n=I_k$ satisfy the following:
\begin{enumerate}
\item[(i)] For $k=2$, $\psi_n^\ell\leq\phi_n^\ell$ for every
sample realization and every $\ell\in\{m,s,q\}$. Moreover, for every
sequence with $\sqrt n\,\mu(P_n)\to\theta\in\mathbf R^2$,
\begin{align*}
\liminf_{n\to\infty}
\left\{
\mathbb E_{P_n}[\phi_n^\ell]
-
\mathbb E_{P_n}[\psi_n^\ell]
\right\}>0.
\end{align*}

\item[(ii)] Let $k\geq3$, $h=1$, and $\ell\in\{m,s,q\}$.
There exists $B_0>0$ such that, for every $B\in(0,B_0]$
and every $r\in\{1,\ldots,k-1\}$, the LF test is
asymptotically more powerful than the two-step test for
all sufficiently large $A$. Specifically, along
sequences satisfying \eqref{eq:local-alternative}, with
$A$ and $B$ held fixed as $n\to\infty$,
\begin{align}\label{eq:power-LF-wins}
\lim_{n\to\infty}\mathbb E_{P_n}[\phi_n^\ell]
>
\limsup_{n\to\infty}\mathbb E_{P_n}[\psi_n^\ell].
\end{align}

\item[(iii)] Let $k\geq3$, $2\leq h\leq k-1$, and
$\ell\in\{m,s,q\}$. Suppose that
\begin{align}\label{def:cv-comparison}
a_{k-h}^\ell(\alpha-\tau)<c^\ell(1-\alpha,I_k).
\end{align}
Then, for every $B>0$ and every $r\in\{1,\ldots,k-h\}$,
the two-step test is asymptotically more powerful than
the LF test for all sufficiently large $A$. Specifically,
along sequences satisfying \eqref{eq:local-alternative},
with $A$ and $B$ held fixed as $n\to\infty$,
\begin{align}\label{eq:power-two-step-wins}
\liminf_{n\to\infty}\mathbb E_{P_n}[\psi_n^\ell]
>
\lim_{n\to\infty}\mathbb E_{P_n}[\phi_n^\ell].
\end{align}
\end{enumerate}
\end{proposition}

Proposition~\ref{prop:power-comparison} highlights the tradeoff between restricting the nuisance values used to determine the critical value and reducing the second-step significance level. When $k=2$, the two-step rejection region is contained in the LF rejection region, giving a strictly positive asymptotic power gap at every finite local alternative.

Under independence, the LF critical value $a_{k-1}^\ell(\alpha)$ is approached along null vectors $(A,0,\ldots,0)$ as $A\to\infty$. With one large positive local drift and small negative drifts, the first-step region often retains null vectors close to this configuration. Restricting the nuisance set then provides too little benefit to offset the smaller second-step level $\alpha-\tau$, even after accounting for immediate first-step rejection.

The two-step test can instead gain power when several coordinates have large positive local drifts. As $A$ increases, the first-step region excludes the negative orthant with probability approaching one and requires every feasible null mean to have large positive values in all $h$ coordinates. This excludes configurations with only one large positive coordinate and all others zero. Consequently, the two-step critical value can be bounded above by a value arbitrarily close to $a_{k-h}^\ell(\alpha-\tau)$ when $A$ is sufficiently large. Condition~\eqref{def:cv-comparison} ensures that this bound is smaller than the LF critical value, allowing the restriction of the nuisance set to outweigh the smaller second-step level.

Condition~\eqref{def:cv-comparison} is a sufficient condition, and its failure does not imply that the LF test is more powerful. It can be readily checked numerically using the quantiles in \eqref{eq:power-one-sided-quantiles}. For the max statistic, it has the explicit form
\begin{align*}
\frac{k-h}{k-1}
<
\frac{\log(1-\alpha+\tau)}{\log(1-\alpha)}.
\end{align*}
For $\tau=0.1\alpha$, parts (ii) and (iii), with $h=k-1$ in the latter, show that neither procedure uniformly dominates the other in local asymptotic power for any fixed $k\geq3$.

\section{Monte Carlo Simulation Study}\label{sec:MonteCarlo-Simulation}
This section examines the finite-sample size and power of the LF and two-step (TS) tests and their nonparametric-bootstrap counterparts, denoted NPLF and NPTS. Each procedure uses the max, sum-of-squares or QLR statistic. We also include the max test of \OEH{} and the adaptive conditional Bonferroni test of \cite{zhao2019JASA}, denoted ZSS, which assumes known independence. For comparison with our two-step procedures as discussed in Remark \ref{rem:IUT-RSW}, we construct an intersection--union benchmark, denoted IUT--RSW, from two one-sided nonparametric bootstrap tests of \cite{RomanoShaikhWolf2014ECTA} using the same three statistics. Both experiments use $n=250$, $k\in\{3,5,8\}$, significance level $\alpha=0.05$, and 5,000 Monte Carlo replications per configuration. The two-step tuning parameter is $\tau=\alpha/10=0.005$.

Observations are generated independently as
\begin{equation}\label{eq:simulation-dgp}
 W_i=\mu+L\varepsilon_i,\qquad LL'=\Omega,\qquad i=1,\ldots,n,
\end{equation}
where $L$ is the lower-triangular Cholesky factor and the coordinates of $\varepsilon_i$ are independent with mean zero and variance one. We consider three population correlation matrices: $\Omega_{\mathrm{Zero}}=I_k$, and the Toeplitz matrices $\Omega_{\mathrm{Positive}}$ and $\Omega_{\mathrm{Negative}}$, whose $(r,s)$ elements are $\rho^{|r-s|}$ with $\rho=0.9$ and $\rho=-0.9$, respectively. Marginal variances are estimated throughout. Tests use the studentized mean $X_n=\sqrt n S_n^{-1}\bar W_n$, where $S_n$ contains the sample marginal standard deviations. Except for ZSS, calibration in the reported comparisons uses the unrestricted sample correlation matrix $\widehat\Omega_n$. Bootstrap samples resample whole observation rows, with marginal variances and correlations re-estimated. 

For LF max, we use \eqref{eq:LF-max-nonnegative-correlation} whenever the correlation matrix supplied to the test has nonnegative off-diagonal entries. Other LF and NPLF critical values use the numerical procedure in Section~\ref{appendix:computation-lfcv}, taking the midpoint of the resulting bracket. Numerical calibration uses 2,500 pilot and 50,000 draws. For TS and NPTS, we numerically approximate the required critical-value suprema over the intersection of the first-step confidence region and the null space. Both NPTS steps use bootstrap reference distributions, with $2^{15}$ draws per replication, including the first $2^{11}$ used for the numerical search. IUT--RSW tests use 1,000 bootstrap resamples per replication.

\subsection{Maximum null rejection probabilities}

For the size comparison, the coordinates of $\varepsilon_i$ in \eqref{eq:simulation-dgp} follow $N(0,1)$, $t_3/\sqrt{3}$, or $(\chi^2_3-3)/\sqrt{6}$. We report maximum null rejection probabilities (MNRPs). The empirical null grid is
\begin{equation}\label{eq:simulation-null-grid-style}
 \mathcal M_k
 =\{0_k\}\cup
 \left\{\frac{s M}{\sqrt n}e_j:
   s\in\{-1,1\},\ M\in\{1,5,10,15\},\ j=1,\ldots,k\right\},
\end{equation}
where $e_j$ is the $j$th unit vector. The grid contains $1+8k=25,41,65$ configurations for $k=3,5,8$, with the origin included once. For each test, distribution and correlation design, we first calculate the rejection frequency at every configuration and then take its maximum over $\mathcal M_k$.

\begin{table}[!htbp]
\centering
\caption{Empirical maximum null rejection probabilities (\%)}
\label{tab:mnrp-all}
\label{tab:mnrp-normal}\label{tab:mnrp-t3}\label{tab:mnrp-chi2}
\scriptsize
\setlength{\tabcolsep}{3.5pt}
\renewcommand{\arraystretch}{0.90}
\begin{tabular}{@{}l*{9}{r}@{}}
\toprule
& \multicolumn{3}{c}{$\Omega_{\mathrm{Positive}}$} & \multicolumn{3}{c}{$\Omega_{\mathrm{Zero}}$} & \multicolumn{3}{c}{$\Omega_{\mathrm{Negative}}$} \\
\cmidrule(lr){2-4}\cmidrule(lr){5-7}\cmidrule(lr){8-10}
Test & $k=3$ & $k=5$ & $k=8$ & $k=3$ & $k=5$ & $k=8$ & $k=3$ & $k=5$ & $k=8$ \\

\midrule
\multicolumn{10}{l}{A. Normal errors} \\
\addlinespace[1pt]
LF max & 5.16 & 5.50 & 5.28 & 5.00 & 5.24 & 5.40 & 4.76 & 5.74 & 5.10 \\
LF sum & 5.26 & 5.36 & 5.42 & 5.22 & 5.00 & 5.24 & 4.74 & 4.92 & 4.20 \\
LF QLR & 5.16 & 5.52 & 5.34 & 5.30 & 5.06 & 5.74 & 5.00 & 5.80 & 4.74 \\
\addlinespace[1pt]
NPLF max & 5.02 & 5.22 & 4.98 & 4.74 & 4.80 & 5.04 & 4.50 & 5.42 & 4.76 \\
NPLF sum & 5.04 & 5.12 & 5.14 & 4.92 & 4.74 & 4.86 & 4.50 & 4.60 & 4.02 \\
NPLF QLR & 4.62 & 5.08 & 4.80 & 5.00 & 4.72 & 4.94 & 4.66 & 5.18 & 4.14 \\
\addlinespace[1pt]
TS max & 4.48 & 4.98 & 4.80 & 4.86 & 5.00 & 5.02 & 5.16 & 5.26 & 4.86 \\
TS sum & 4.50 & 4.96 & 4.82 & 4.80 & 5.00 & 5.04 & 5.16 & 5.24 & 4.82 \\
TS QLR & 4.44 & 4.90 & 4.90 & 4.78 & 5.16 & 5.38 & 4.86 & 5.68 & 5.76 \\
\addlinespace[1pt]
NPTS max & 4.54 & 4.90 & 4.62 & 4.50 & 4.58 & 4.80 & 4.54 & 5.04 & 4.50 \\
NPTS sum & 4.66 & 4.86 & 4.86 & 4.84 & 4.50 & 4.74 & 4.14 & 4.80 & 4.36 \\
NPTS QLR & 4.50 & 4.96 & 4.62 & 4.90 & 4.46 & 4.96 & 4.48 & 5.28 & 4.46 \\
\addlinespace[1pt]
IUT--RSW max & 4.60 & 4.98 & 4.66 & 4.50 & 4.64 & 4.92 & 4.54 & 5.14 & 4.54 \\
IUT--RSW sum & 4.72 & 4.88 & 5.12 & 4.86 & 4.62 & 4.84 & 4.26 & 4.86 & 4.36 \\
IUT--RSW QLR & 4.58 & 4.96 & 4.80 & 4.94 & 4.70 & 5.08 & 4.52 & 5.38 & 4.56 \\
\addlinespace[1pt]
OEH max & 4.46 & 5.12 & 5.14 & 3.62 & 4.60 & 4.98 & 4.56 & 5.18 & 5.02 \\
ZSS & \mbox{-{}-} & \mbox{-{}-} & \mbox{-{}-} & 4.82 & 5.06 & 5.14 & \mbox{-{}-} & \mbox{-{}-} & \mbox{-{}-} \\
\midrule
\multicolumn{10}{l}{B. Standardized $t_3$ errors} \\
\addlinespace[1pt]
LF max & 5.38 & 5.22 & 5.02 & 5.50 & 5.44 & 5.16 & 5.02 & 5.02 & 5.04 \\
LF sum & 5.34 & 5.40 & 5.10 & 5.66 & 4.94 & 4.90 & 5.10 & 4.76 & 4.72 \\
LF QLR & 5.30 & 5.24 & 5.06 & 5.72 & 5.20 & 5.28 & 5.02 & 5.08 & 4.78 \\
\addlinespace[1pt]
NPLF max & 4.18 & 4.24 & 4.12 & 3.84 & 3.78 & 3.06 & 4.26 & 4.32 & 4.32 \\
NPLF sum & 4.30 & 4.42 & 4.42 & 3.94 & 3.44 & 3.28 & 4.14 & 4.00 & 4.22 \\
NPLF QLR & 3.98 & 3.98 & 4.04 & 3.94 & 3.46 & 3.14 & 3.84 & 3.28 & 2.90 \\
\addlinespace[1pt]
TS max & 4.88 & 4.88 & 4.46 & 4.64 & 4.40 & 3.96 & 4.74 & 4.66 & 4.82 \\
TS sum & 5.04 & 4.92 & 4.60 & 4.64 & 4.62 & 4.04 & 4.56 & 4.64 & 4.78 \\
TS QLR & 4.94 & 4.88 & 4.52 & 4.76 & 4.80 & 4.38 & 4.62 & 5.02 & 5.40 \\
\addlinespace[1pt]
NPTS max & 5.54 & 5.34 & 5.00 & 5.60 & 5.04 & 4.18 & 5.42 & 4.38 & 4.36 \\
NPTS sum & 5.40 & 5.48 & 5.04 & 5.66 & 4.98 & 4.52 & 5.28 & 4.78 & 4.72 \\
NPTS QLR & 5.54 & 5.28 & 5.02 & 5.64 & 5.00 & 4.52 & 5.52 & 4.64 & 4.56 \\
\addlinespace[1pt]
IUT--RSW max & 5.74 & 5.44 & 5.18 & 5.64 & 5.26 & 4.38 & 5.48 & 4.42 & 4.40 \\
IUT--RSW sum & 5.42 & 5.62 & 5.16 & 5.78 & 5.10 & 4.66 & 5.34 & 5.02 & 4.74 \\
IUT--RSW QLR & 5.76 & 5.50 & 5.18 & 5.74 & 5.08 & 4.64 & 5.44 & 4.82 & 4.62 \\
\addlinespace[1pt]
OEH max & 4.86 & 5.16 & 4.86 & 3.62 & 3.98 & 3.92 & 3.90 & 4.54 & 4.98 \\
ZSS & \mbox{-{}-} & \mbox{-{}-} & \mbox{-{}-} & 5.20 & 5.26 & 4.84 & \mbox{-{}-} & \mbox{-{}-} & \mbox{-{}-} \\
\midrule
\multicolumn{10}{l}{C. Standardized $\chi^2_3$ errors} \\
\addlinespace[1pt]
LF max & 6.62 & 6.66 & 6.84 & 7.56 & 7.88 & 9.20 & 5.22 & 5.64 & 5.32 \\
LF sum & 6.66 & 6.42 & 6.28 & 7.64 & 7.70 & 8.80 & 5.82 & 5.86 & 4.92 \\
LF QLR & 6.54 & 6.74 & 6.82 & 7.66 & 7.82 & 9.22 & 6.48 & 7.02 & 7.24 \\
\addlinespace[1pt]
NPLF max & 5.16 & 5.40 & 5.08 & 5.34 & 5.12 & 5.08 & 4.70 & 4.94 & 4.64 \\
NPLF sum & 5.48 & 5.28 & 5.32 & 5.38 & 4.86 & 4.90 & 4.90 & 5.12 & 4.58 \\
NPLF QLR & 4.96 & 5.08 & 4.90 & 5.42 & 4.84 & 4.98 & 4.86 & 4.54 & 4.08 \\
\addlinespace[1pt]
TS max & 5.80 & 6.08 & 5.88 & 6.32 & 7.72 & 8.40 & 5.32 & 5.48 & 5.36 \\
TS sum & 5.68 & 5.94 & 5.76 & 6.06 & 7.58 & 8.40 & 4.98 & 5.40 & 5.14 \\
TS QLR & 5.78 & 6.06 & 5.82 & 6.04 & 7.76 & 8.76 & 5.62 & 6.52 & 8.34 \\
\addlinespace[1pt]
NPTS max & 5.06 & 4.94 & 4.58 & 5.04 & 4.86 & 4.70 & 4.50 & 4.44 & 4.60 \\
NPTS sum & 4.98 & 4.84 & 4.86 & 5.20 & 4.72 & 4.76 & 4.72 & 5.00 & 4.70 \\
NPTS QLR & 5.06 & 4.82 & 4.54 & 5.20 & 4.78 & 4.78 & 4.58 & 4.48 & 4.88 \\
\addlinespace[1pt]
IUT--RSW max & 4.98 & 4.98 & 4.76 & 5.14 & 4.82 & 4.88 & 4.58 & 4.62 & 4.76 \\
IUT--RSW sum & 5.02 & 4.94 & 4.94 & 5.42 & 4.72 & 4.92 & 4.72 & 5.18 & 4.80 \\
IUT--RSW QLR & 4.98 & 4.96 & 4.70 & 5.40 & 4.76 & 4.90 & 4.94 & 4.58 & 4.98 \\
\addlinespace[1pt]
OEH max & 5.62 & 6.30 & 6.40 & 4.76 & 7.16 & 8.36 & 4.24 & 5.20 & 5.20 \\
ZSS & \mbox{-{}-} & \mbox{-{}-} & \mbox{-{}-} & 7.28 & 7.56 & 8.68 & \mbox{-{}-} & \mbox{-{}-} & \mbox{-{}-} \\
\bottomrule
\end{tabular}
\par\vspace{4pt}
\begin{minipage}{\linewidth}\footnotesize
Note: Entries are maxima of empirical rejection frequencies over the signed null grid in \eqref{eq:simulation-null-grid-style}, using 5,000 replications per configuration and $n=250$. $\Omega_{\mathrm{Positive}}$ and $\Omega_{\mathrm{Negative}}$ denote Toeplitz correlations with $\rho=0.9$ and $\rho=-0.9$; $\Omega_{\mathrm{Zero}}$ denotes $I_k$. Errors are standardized before the Cholesky transformation. LF and NPLF use midpoint critical values. IUT--RSW uses 1,000 nonparametric bootstrap resamples per replication. ZSS assumes known independence and is reported only under $\Omega_{\mathrm{Zero}}$; \mbox{-{}-} indicates unavailable comparisons. No size correction is applied.
\end{minipage}
\end{table}

Table~\ref{tab:mnrp-all} reports the empirical MNRPs. Small deviations from the nominal $5\%$ level should be interpreted cautiously, since taking the maximum of estimated rejection frequencies can amplify Monte Carlo variation. Under Normal errors, the Gaussian LF and TS tests and their nonparametric bootstrap counterparts have MNRPs broadly close to the nominal level. Under $t_3$ errors, LF, TS, and NPLF have MNRPs close to or below the nominal level, with NPLF showing greater conservatism. NPTS and IUT--RSW remain broadly close to nominal. Under both symmetric error distributions, OEH tends to be conservative under $\Omega_{\mathrm{Zero}}$, whereas ZSS remains close to $5\%$.

Under skewed $\chi^2_3$ errors, the parametric procedures exhibit more pronounced overrejection, especially under $\Omega_{\mathrm{Zero}}$ as $k$ increases. At $k=8$, LF and TS using the QLR statistic have MNRPs of $9.22\%$ and $8.76\%$, respectively. OEH and ZSS also overreject in this setting. By comparison, NPLF, NPTS, and IUT--RSW remain close to nominal across correlations and dimensions, with MNRPs ranging from $4.08\%$ to $5.48\%$, $4.44\%$ to $5.20\%$, and $4.58\%$ to $5.42\%$, respectively. These results are consistent with the nonparametric bootstrap capturing finite-sample asymmetry that is missed by the Gaussian approximation.

\subsection{Power comparisons}

For the power comparison, we use $\varepsilon_i\sim N_k(0,I_k)$ and $\mu=\theta/\sqrt n$ in \eqref{eq:simulation-dgp}. We do not apply size correction, so comparisons preserve the conservativeness built into TS, NPTS, OEH and IUT-RSW.

Let $1_r$ and $0_r$ denote vectors of ones and zeros of length $r$, respectively, and let $m=\lfloor k/2\rfloor$. We consider three configurations of the standardized mean:
\begin{equation}\label{eq:simulation-power-designs}
\begin{aligned}
 \text{Design 1:}\quad
 \theta&=(A,-2,0_{k-2} ^\top),\\
 \text{Design 2:}\quad
 \theta&=(A1_{k-1} ^\top,-3),\\
 \text{Design 3:}\quad
 \theta&=(A1_m ^\top,0_{k-2m} ^\top,-A1_m ^\top)
\end{aligned}
\quad \text{where}\quad
\begin{gathered}
 A\in\{0,0.5,1,1.5,2,2.5,3,4,6\}.
\end{gathered}
\end{equation}
Design 1 has one varying positive coordinate, one fixed negative coordinate, and otherwise zero coordinates. Design 2 has $k-1$ equal positive coordinates and one fixed negative coordinate. These configurations mirror designs in Proposition~\ref{prop:power-comparison}(ii) and (iii), respectively. For Design 1, part (ii) predicts an LF advantage for sufficiently large $A$, provided that the fixed negative drift is sufficiently small. For Design 2, the condition in \eqref{def:cv-comparison} holds for all three statistics and $k=3,5,8$, so part (iii) predicts a TS advantage for sufficiently large $A$. Design 3 strengthens both signs simultaneously, with equal numbers of positive and negative coordinates and one zero coordinate when $k$ is odd. The coordinates are ordered as displayed in \eqref{eq:simulation-power-designs}.
Design 3 uses the same grid without $A=6$, as all the rejection rates reach 1. In every design, $A>0$ belongs to the alternative.

\begin{table}[!htbp]
\centering
\caption{Average empirical power (\%): estimated correlation}
\label{tab:average-rejection-estimated}\label{tab:average-power}
\scriptsize
\setlength{\tabcolsep}{3.5pt}
\renewcommand{\arraystretch}{0.95}
\begin{tabular}{@{}l*{9}{r}@{}}
\toprule
& \multicolumn{3}{c}{$\Omega_{\mathrm{Positive}}$} & \multicolumn{3}{c}{$\Omega_{\mathrm{Zero}}$} & \multicolumn{3}{c}{$\Omega_{\mathrm{Negative}}$} \\
\cmidrule(lr){2-4}\cmidrule(lr){5-7}\cmidrule(lr){8-10}
Test & $k=3$ & $k=5$ & $k=8$ & $k=3$ & $k=5$ & $k=8$ & $k=3$ & $k=5$ & $k=8$ \\
\midrule
\multicolumn{10}{l}{A. Design 1} \\
\addlinespace[2pt]
LF max & 26.18 & 20.39 & 15.74 & 31.45 & 23.17 & 18.54 & 38.68 & 32.67 & 27.49 \\
LF sum & 12.84 & 3.50 & 1.34 & 30.60 & 21.42 & 15.74 & 30.95 & 19.92 & 12.63 \\
LF QLR & 25.74 & 19.59 & 14.48 & 30.62 & 21.62 & 16.29 & 55.68 & 64.45 & 58.59 \\
\addlinespace[2pt]
NPLF max & 25.31 & 19.53 & 14.80 & 30.28 & 22.07 & 17.71 & 37.85 & 31.79 & 26.61 \\
NPLF sum & 12.10 & 3.31 & 1.22 & 29.58 & 20.50 & 14.89 & 30.08 & 18.99 & 12.21 \\
NPLF QLR & 23.96 & 18.17 & 13.37 & 29.42 & 20.25 & 14.58 & 54.92 & 63.01 & 56.77 \\
\addlinespace[2pt]
TS max & 30.46 & 22.22 & 16.56 & 29.82 & 21.99 & 17.66 & 49.95 & 48.64 & 40.84 \\
TS sum & 20.94 & 6.35 & 1.87 & 29.18 & 20.53 & 15.10 & 46.77 & 42.68 & 33.42 \\
TS QLR & 30.36 & 21.58 & 15.62 & 29.22 & 20.78 & 15.80 & 55.61 & 65.09 & 60.37 \\
\addlinespace[2pt]
NPTS max & 29.79 & 21.47 & 15.91 & 29.09 & 21.20 & 16.86 & 49.06 & 46.86 & 38.29 \\
NPTS sum & 19.96 & 5.86 & 1.67 & 28.61 & 19.78 & 14.43 & 45.89 & 40.70 & 30.25 \\
NPTS QLR & 29.67 & 20.87 & 14.82 & 28.58 & 19.65 & 14.23 & 55.12 & 64.28 & 58.83 \\
\addlinespace[2pt]
IUT--RSW max & 23.72 & 18.38 & 14.09 & 26.32 & 19.97 & 16.34 & 35.32 & 30.23 & 25.73 \\
IUT--RSW sum & 8.15 & 2.65 & 1.01 & 23.99 & 17.35 & 12.95 & 30.80 & 18.06 & 11.43 \\
IUT--RSW QLR & 23.44 & 17.80 & 12.98 & 23.83 & 17.19 & 12.84 & 51.19 & 59.92 & 55.00 \\
\addlinespace[2pt]
OEH max & 24.18 & 19.48 & 15.19 & 25.37 & 20.38 & 17.41 & 35.46 & 31.16 & 27.00 \\
ZSS & \mbox{-{}-} & \mbox{-{}-} & \mbox{-{}-} & 30.01 & 22.22 & 17.80 & \mbox{-{}-} & \mbox{-{}-} & \mbox{-{}-} \\
\midrule
\multicolumn{10}{l}{B. Design 2} \\
\addlinespace[2pt]
LF max & 52.53 & 50.60 & 49.52 & 59.43 & 58.97 & 55.84 & 63.44 & 59.79 & 55.77 \\
LF sum & 44.01 & 21.72 & 6.71 & 59.61 & 54.99 & 42.31 & 56.80 & 46.68 & 26.12 \\
LF QLR & 52.31 & 50.17 & 48.56 & 59.71 & 55.43 & 42.50 & 83.25 & 76.99 & 60.93 \\
\addlinespace[2pt]
NPLF max & 51.66 & 49.46 & 48.39 & 58.44 & 58.00 & 54.29 & 62.59 & 58.89 & 54.60 \\
NPLF sum & 42.90 & 20.80 & 6.13 & 58.75 & 53.94 & 40.80 & 55.89 & 45.14 & 25.19 \\
NPLF QLR & 50.42 & 48.57 & 46.90 & 58.63 & 53.57 & 40.09 & 82.44 & 75.26 & 58.72 \\
\addlinespace[2pt]
TS max & 53.91 & 52.91 & 52.31 & 59.34 & 61.59 & 60.53 & 75.73 & 79.17 & 77.07 \\
TS sum & 48.76 & 34.88 & 23.70 & 60.49 & 61.95 & 56.43 & 76.84 & 75.15 & 59.99 \\
TS QLR & 53.90 & 52.89 & 52.04 & 60.50 & 62.14 & 56.79 & 84.61 & 86.54 & 83.32 \\
\addlinespace[2pt]
NPTS max & 53.39 & 52.17 & 51.22 & 58.60 & 60.66 & 59.11 & 74.82 & 77.56 & 74.71 \\
NPTS sum & 48.03 & 33.47 & 21.90 & 59.94 & 60.89 & 54.84 & 76.02 & 73.83 & 57.08 \\
NPTS QLR & 53.41 & 52.17 & 50.83 & 59.87 & 60.73 & 54.57 & 84.25 & 86.04 & 82.42 \\
\addlinespace[2pt]
IUT--RSW max & 51.92 & 51.50 & 50.93 & 56.02 & 58.54 & 57.04 & 60.77 & 59.68 & 57.02 \\
IUT--RSW sum & 39.64 & 29.68 & 21.08 & 55.75 & 55.64 & 47.37 & 58.54 & 49.79 & 35.88 \\
IUT--RSW QLR & 51.80 & 51.20 & 50.50 & 55.67 & 55.46 & 46.73 & 82.18 & 81.38 & 72.73 \\
\addlinespace[2pt]
OEH max & 50.45 & 49.40 & 48.85 & 53.24 & 55.97 & 53.96 & 59.62 & 58.23 & 55.14 \\
ZSS & \mbox{-{}-} & \mbox{-{}-} & \mbox{-{}-} & 60.92 & 65.17 & 66.27 & \mbox{-{}-} & \mbox{-{}-} & \mbox{-{}-} \\
\midrule
\multicolumn{10}{l}{C. Design 3} \\
\addlinespace[2pt]
LF max & 34.45 & 35.30 & 38.10 & 37.07 & 46.70 & 55.55 & 35.91 & 60.42 & 62.06 \\
LF sum & 23.86 & 25.69 & 31.84 & 36.54 & 50.03 & 62.47 & 31.07 & 55.56 & 62.29 \\
LF QLR & 34.18 & 34.83 & 37.56 & 36.70 & 50.15 & 62.69 & 76.07 & 89.95 & 95.76 \\
\addlinespace[2pt]
NPLF max & 33.88 & 34.56 & 37.37 & 36.43 & 45.99 & 54.67 & 35.24 & 59.70 & 61.34 \\
NPLF sum & 23.32 & 25.10 & 31.22 & 35.91 & 49.34 & 61.95 & 30.22 & 54.76 & 61.89 \\
NPLF QLR & 32.98 & 33.73 & 36.59 & 35.91 & 49.16 & 61.64 & 75.54 & 89.46 & 95.25 \\
\addlinespace[2pt]
TS max & 38.29 & 39.79 & 41.73 & 36.13 & 46.33 & 56.01 & 59.61 & 82.96 & 87.96 \\
TS sum & 36.89 & 38.82 & 42.83 & 35.60 & 49.59 & 62.69 & 59.06 & 82.75 & 87.98 \\
TS QLR & 38.29 & 39.80 & 41.99 & 35.72 & 49.71 & 62.83 & 75.52 & 89.88 & 96.14 \\
\addlinespace[2pt]
NPTS max & 37.88 & 39.28 & 41.19 & 35.67 & 45.64 & 55.14 & 58.79 & 82.42 & 87.44 \\
NPTS sum & 36.42 & 38.39 & 42.31 & 35.22 & 49.03 & 62.11 & 58.25 & 82.17 & 87.45 \\
NPTS QLR & 37.88 & 39.28 & 41.43 & 35.15 & 48.89 & 61.88 & 75.18 & 89.56 & 95.71 \\
\addlinespace[2pt]
IUT--RSW max & 32.90 & 34.53 & 37.55 & 33.20 & 43.99 & 53.33 & 32.92 & 57.84 & 60.16 \\
IUT--RSW sum & 18.46 & 23.71 & 31.49 & 30.98 & 46.33 & 59.81 & 31.63 & 51.09 & 59.94 \\
IUT--RSW QLR & 32.68 & 34.21 & 37.18 & 30.89 & 46.12 & 59.53 & 71.69 & 88.22 & 94.54 \\
\addlinespace[2pt]
OEH max & 33.08 & 34.52 & 37.65 & 33.06 & 44.71 & 54.45 & 32.83 & 59.18 & 61.68 \\
ZSS & \mbox{-{}-} & \mbox{-{}-} & \mbox{-{}-} & 36.45 & 47.47 & 57.09 & \mbox{-{}-} & \mbox{-{}-} & \mbox{-{}-} \\
\bottomrule
\end{tabular}
\par\vspace{4pt}
\begin{minipage}{\linewidth}\footnotesize
Note: Entries are equal-weight averages of empirical rejection probabilities over the alternative configurations $A>0$. Each configuration uses 5,000 replications and $n=250$. Marginal variances are estimated for all procedures; correlation matrices are estimated for all procedures except ZSS. \mbox{-{}-} indicates that ZSS is not reported under $\Omega_{\mathrm{Positive}}$ or $\Omega_{\mathrm{Negative}}$.
\end{minipage}
\end{table}

Table~\ref{tab:average-rejection-estimated} reports equal-weight averages of empirical power over the alternative configurations $A>0$ for $k=3,5,8$, grouped by correlation structure. Corresponding power curves appear in Appendix~\ref{app:power-additional}. 

Under $\Omega_{\mathrm{Zero}}$, LF and TS have similar average power in Designs 1 and 3, with LF ahead by about $1\%$ point in Design 1. In Design 2, TS outperforms LF for every statistic at $k=5,8$. These comparisons are qualitatively consistent with Proposition~\ref{prop:power-comparison}(ii) and (iii), respectively. Under $\Omega_{\mathrm{Positive}}$, TS has higher average power in all three designs. Under $\Omega_{\mathrm{Negative}}$, its advantages are particularly pronounced for the max and sum-of-squares statistics, reaching approximately $23\%$ points in Design 1 and $33.9\%$ points in Design 2, the latter for sum-of-squares at $k=8$. In Design 2, TS outperforms LF for every statistic and all three dimensions under both correlation structures. Bootstrap calibration has only a modest effect on power in these Gaussian experiments. NPLF has average power close to LF, with shortfalls of at most $2.41\%$ points across all statistics, designs and correlations. NPTS retains the main power patterns of TS, with average power at most $3.2\%$ points lower. The LF--TS pattern largely carries over to the bootstrap procedures.

The simulations reveal complementary power advantages of the proposed procedures over the benchmarks. LF max is at least as high as OEH throughout all configurations, consistent with the power dominance result in Proposition \ref{prop:oulhaj-dominance}. TS and NPTS achieve higher average power than IUT--RSW in every comparison using the same statistic. These gains are most pronounced under $\Omega_{\mathrm{Negative}}$: the largest average advantages of TS are approximately $28\%$, $32\%$ and $11\%$ points for max, sum-of-squares and QLR, respectively, with slightly smaller gains for NPTS. LF also generally outperforms IUT--RSW in Designs 1 and 3, although IUT--RSW often improves on LF in Design 2 as $k$ increases. Under $\Omega_{\mathrm{Zero}}$, ZSS additionally assumes that independence is known. Despite this informational advantage, it does not uniformly outperform the proposed procedures: LF max has higher average power in Design 1, and the highest average power among LF and TS exceeds that of ZSS in Design 3. ZSS leads in Design 2, with its advantage becoming more pronounced as $k$ increases.

The relative power of max, sum-of-squares and QLR depends on the mean configuration and correlation structure. Under $\Omega_{\mathrm{Zero}}$, max performs well in Design 1, while sum-of-squares and QLR gain with $k$ in Design 3. Sum-of-squares and QLR coincide when correlation is known and yield similar LF and TS powers when it is estimated. Under $\Omega_{\mathrm{Positive}}$, TS max and QLR outperform sum in Designs 1 and 2; under $\Omega_{\mathrm{Negative}}$, QLR leads within each proposed procedure.

\section{Empirical Illustration}\label{sec:Empirical-Application}
This section applies the two proposed tests in empirical settings where estimators of parameters are correlated by construction. In the first,  we test whether the dynamic treatment response of divorce rates to unilateral divorce laws \citep{Wolfers2006AER} maintains a single sign over event time. In the second, we test the validity of the college-proximity instrument of \citet{card1993NBER}.

\subsection{Dynamic Treatment Effects of Unilateral Divorce Law}
\citet{Wolfers2006AER} examines the effect of unilateral divorce laws on the divorce rate in the United States. Exploiting their staggered adoption by twenty-nine states between 1968 and 1988, he studies how the effect of reform evolves over time. We test the null that the dynamic average treatment effects share a common weak sign across horizons. Sign reversals matter economically because unilateral divorce may operate through opposing channels. By allowing either spouse to leave without the other's consent, reform may increase marital dissolution. Conversely, more cautious matching among couples marrying under the new regime may reduce subsequent divorce. The relative importance of these channels may vary over time.

We consider the original specification by Wolfers (Model 1) as well as the following specification by \citet{SunAbraham2021}:
\begin{equation}\label{eq:wolfers-sun-abraham}
Y_{st}=\alpha_s+\lambda_t
+\sum_{g\in\mathcal{G}}\sum_{\ell\in\mathcal{L}_g\setminus\{-1\}}
\delta_{g,\ell}\,\mathbf{1}\{G_s=g\}\,\mathbf{1}\{t-g=\ell\}
+\varepsilon_{st}.
\end{equation}
Here, $Y_{st}$ is the divorce rate in state $s$ and calendar year $t$, $\alpha_s$ and $\lambda_t$ are state and year fixed effects, and $G_s$ is the year in which state $s$ adopts unilateral divorce. The set $\mathcal{G}$ contains the included adoption cohorts, and $\mathcal{L}_g$ contains the observed event years for cohort $g$. Event year $\ell=0$ denotes the adoption year. The coefficients in Wolfers's original specification can be causally interpreted only under additional restrictions such as treatment effect homogeneity across cohorts and constancy within event time bins. The specification above instead identifies cohort-specific dynamic effects $\delta_{g,\ell}$ under the standard assumptions.

We estimate \eqref{eq:wolfers-sun-abraham} by population-weighted least squares, retaining the weighting used by \citet{Wolfers2006AER} while adopting the interacted specification of \citet{SunAbraham2021}. We include leads and lags and omit $\ell=-1$, the year immediately before adoption. Never-treated states have $G_s=\infty$ and serve as the comparison group. Excluding Alaska and Oklahoma, which adopted before the observation window, leaves 1,565 observations across 49 states.  Once estimated, we aggregate these effects into Wolfers's eight bins: years 1--2, 3--4, \ldots, 13--14, and 15 onward, counting the adoption year as year~1. Let $\mu_k$ denote the average effect in bin $k$, using nonnegative weights proportional to the population of each observed cohort--event-year cell within that bin and normalized to sum to one. Each $\mu_k$ measures an effect on annual divorces per 1,000 persons. The estimators $\hat\mu_1,\ldots,\hat\mu_8$ are correlated because they
reuse observations from the same states, share never-treated comparison states, and are estimated jointly with common fixed effects.

We implement both parametric LF and two-step tests using three proposed statistics. We omit their nonparametric bootstrap versions because adoption cohorts are sparse: four adoption years have only one adopting state each. Resampling whole states can omit these cohorts, leaving the corresponding cohort-specific effects unidentified in some bootstrap draws. These features prevent us from invoking a standard joint bootstrap justification for this application. We also implement the test of \OEH{}.

Table~\ref{tab:wolfers-results} compares our estimates with the replicated estimates from Table~2, column~1, of \citet{Wolfers2006AER}. Both specifications yield positive estimates for years 1--8 and negative estimates thereafter.
For the Sun--Abraham specification, the LF procedure gives $p$-values of $1.35\%$, $3.83\%$, and $1.70\%$ for $T^m$, $T^s$, and $T^q$, respectively; the two-step procedure produces $p$-values of $0.76\%$, $1.25\%$, and $0.80\%$, lower than their least favorable counterparts. The OEH test gives $p$-values of $1.39\%$, larger than that of LF $T^m$, as predicted in Proposition \ref{prop:oulhaj-dominance}. The $p$-values under Wolfers Model 1 are smaller, all below under 1\%. The evidence therefore points to opposite-signed short- and longer-run bin averages under both specifications.

\begin{table}[t]
\centering
\caption{Dynamic divorce-rate estimates and sign-agreement tests for \citet{Wolfers2006AER}.}
\label{tab:wolfers-results}
\small
\setlength{\tabcolsep}{3pt}
\begin{tabular}{@{}lccccccc@{}}
\toprule
\multicolumn{8}{l}{\emph{Panel A: Dynamic estimates, Wolfers's bins $(k=8)$}}\\
\multicolumn{2}{l}{Years after reform} & \multicolumn{3}{c}{Sun--Abraham} & \multicolumn{3}{c}{Wolfers Model 1}\\
\multicolumn{2}{l}{} & \multicolumn{3}{c}{$\hat\mu_k$ (SE)} & \multicolumn{3}{c}{Estimate (SE)}\\
\midrule
\multicolumn{2}{l}{1--2} & \multicolumn{3}{c}{$\phantom{-}0.295$ (0.072)} & \multicolumn{3}{c}{$\phantom{-}0.267$ (0.085)}\\
\multicolumn{2}{l}{3--4} & \multicolumn{3}{c}{$\phantom{-}0.202$ (0.089)} & \multicolumn{3}{c}{$\phantom{-}0.210$ (0.085)}\\
\multicolumn{2}{l}{5--6} & \multicolumn{3}{c}{$\phantom{-}0.142$ (0.111)} & \multicolumn{3}{c}{$\phantom{-}0.164$ (0.085)}\\
\multicolumn{2}{l}{7--8} & \multicolumn{3}{c}{$\phantom{-}0.144$ (0.133)} & \multicolumn{3}{c}{$\phantom{-}0.158$ (0.084)}\\
\multicolumn{2}{l}{9--10} & \multicolumn{3}{c}{$-0.135$ (0.138)} & \multicolumn{3}{c}{$-0.121$ (0.084)}\\
\multicolumn{2}{l}{11--12} & \multicolumn{3}{c}{$-0.328$ (0.156)} & \multicolumn{3}{c}{$-0.324$ (0.083)}\\
\multicolumn{2}{l}{13--14} & \multicolumn{3}{c}{$-0.468$ (0.187)} & \multicolumn{3}{c}{$-0.461$ (0.084)}\\
\multicolumn{2}{l}{15+} & \multicolumn{3}{c}{$-0.581$ (0.215)} & \multicolumn{3}{c}{$-0.507$ (0.080)}\\
\midrule
\multicolumn{8}{l}{\emph{Panel B: Joint $p$-values (\%)}}\\
Specification & \shortstack{LF\\$T^m$} & \shortstack{LF\\$T^s$} & \shortstack{LF\\$T^q$} & \shortstack{Two-step\\$T^m$} & \shortstack{Two-step\\$T^s$} & \shortstack{Two-step\\$T^q$} & OEH\\
\midrule
Sun--Abraham & $1.35$ & $3.83$ & $1.70$ & $0.76$ & $1.25$ & $0.80$ & $1.39$\\
Wolfers Model 1 & $0.57$ & $0.44$ & $0.27$ & $0.46$ & $0.10$ & $0.15$ & $0.65$\\
\bottomrule
\end{tabular}

\begin{minipage}{0.96\textwidth}
\vspace{2pt}
\footnotesize
\emph{Notes:}
Panel~A reports estimated effects on annual divorces per 1,000 persons, with standard errors in parentheses. The Sun--Abraham specification uses 1,565 observations across 49 states, population-weighted estimation and aggregation, and state-clustered standard errors. Wolfers Model~1 reproduces Table~2, column~1, of \citet{Wolfers2006AER}, using 1,631 observations and the original model-based standard errors. Panel~B reports joint $p$-values in percent for the parametric tests. LF denotes the least favorable procedure. The two-step procedure uses a first-step confidence ellipsoid with $\tau=0.1\alpha$. LF calculations for $T^s$ and $T^q$ use 1 million draws; LF $T^m$ and OEH $p$-values are evaluated by numerical multivariate normal integration. Final two-step calculations use $2^{20}$ Gaussian quasi-Monte Carlo draws, increased to $2^{22}$ for $T^s$ under the Sun--Abraham specification.
\end{minipage}
\end{table}

\subsection{Validity of the College-proximity Instrument}
\citet{card1993NBER} uses geographic variation in college proximity to learn the effect of schooling on log hourly wages in the United States. We define the binary treatment as $D=1\{\text{years of schooling}\geq16\}$, a proxy for college completion. The outcome $Y$ is log hourly wages in 1976, and the binary instrument $Z$ indicates whether an individual lived near a four-year college in 1966. As illustrated in Example~\ref{example:LATE-instrument-validity}, the canonical instrumental variable (IV) estimand is a local average treatment effect under instrument validity and relevance. In particular, monotonicity requires the absence of either defiers or compliers.

A nearby college can lower attendance costs by allowing students to live at home. If this encourages some students to complete college without discouraging completion for others, then defiers are absent. Proximity may, however, also alter which college a student attends. Some students may choose a local institution instead of an alternative at which they would have been more likely to graduate, potentially reducing their completed schooling;
see \citet{rouse1995JBES} and \citet{CohodesGoodman2014AEJ}.
If proximity reduces completion for some students and increases it for none, then compliers are absent. If it increases completion for some students and reduces it for others, neither direction of monotonicity holds, and the usual LATE interpretation of the
IV estimand need not hold.

We test the testable implication of instrument validity with the null hypothesis in \eqref{eq:instrument-validity-null} using the National Longitudinal Survey of Young Men. This null differs from that considered by \citet{Kitagawa2015ECTA}, whose monotonicity assumption rules out defiers. The sample contains $3{,}010$ men with wages and schooling, treated as i.i.d.\ observations: $2{,}053$ have $Z=1$ and $957$ have $Z=0$. We discretize $Y$ at its pooled-sample quartiles. This implies the estimators $\hat{\mu}$ satisfy $\sum_{j=1}^k\hat\mu_j=\sum_{j=1}^k\hat\mu_{k+j}$ and so their covariance is a singular positive semidefinite matrix. To accommodate this singularity, the two-step procedure uses a rectangular first-step confidence set. We report the max statistic $T^m$ and sum-of-squares statistic $T^s$, without the QLR statistic, under both the LF and the two-step ($\tau=0.1\alpha$) procedures.
We also report the test of \OEH{} as well.

Panel~A of Table~\ref{tab:card-results} shows positive moments
throughout the college-completion block. The two highest-wage
moments in the $D=0$ block are negative and individually significant at the $5\%$ level, indicating that high-wage men below the college-completion threshold are more prevalent among those who grew up near a college. Panel~B reports parametric and nonparametric bootstrap
$p$-values, expressed in percent to two decimal places.
For the quartile specification, the parametric procedure yields $p$-values of $0.02\%$ and $0.03\%$ for $T^m$ and $T^s$, respectively, under both the LF and two-step procedures. Under nonparametric bootstrap procedure, the
LF procedure yields $0.04\%$ and $0.05\%$, while
the two-step procedure yields $0.02\%$ and $0.03\%$.
Both procedures therefore reject sign agreement at the $1\%$ significance level. The OEH test also rejects, with a $p$-value of $0.02\%$. The rejection survives coarsening to terciles: all reported $p$-values are at most $0.02\%$. Consequently, the maintained model does not justify interpreting TSLS as either a complier or a defier average treatment effect.

\begin{table}[t]
\centering\small
\caption{Sign agreement tests for the
\citet{card1993NBER} college-proximity instrument.}
\label{tab:card-results}
\begin{tabular}{lccccc}
\toprule
\multicolumn{6}{l}{\emph{Panel A: Moment estimates,
quartile bins $(2k=8)$}}\\
\multicolumn{2}{l}{Wage bin}
& \multicolumn{2}{c}{$D=1$ block}
& \multicolumn{2}{c}{$D=0$ block}\\
\multicolumn{2}{l}{}
& \multicolumn{2}{c}{$\hat\mu_j$ (SE)}
& \multicolumn{2}{c}{$\hat\mu_{4+j}$ (SE)}\\
\midrule
\multicolumn{2}{l}{$(-\infty,5.977]$}
& \multicolumn{2}{c}{$0.003$ (0.008)}
& \multicolumn{2}{c}{$\phantom{-}0.131$ (0.017)}\\
\multicolumn{2}{l}{$(5.977,6.287]$}
& \multicolumn{2}{c}{$0.014$ (0.009)}
& \multicolumn{2}{c}{$\phantom{-}0.028$ (0.016)}\\
\multicolumn{2}{l}{$(6.287,6.564]$}
& \multicolumn{2}{c}{$0.001$ (0.010)}
& \multicolumn{2}{c}{$-0.037$ (0.015)}\\
\multicolumn{2}{l}{$(6.564,\infty)$}
& \multicolumn{2}{c}{$0.051$ (0.011)}
& \multicolumn{2}{c}{$-0.052$ (0.013)}\\
\midrule
\multicolumn{6}{l}{\emph{Panel B: Joint $p$-values (\%)}}\\
Specification
& LF $T^m$ & LF $T^s$
& Two-step $T^m$ & Two-step $T^s$
& OEH\\
\midrule
\multicolumn{6}{l}{\emph{Parametric (Gaussian)}}\\
Quartiles ($2k=8$)
& $0.02$ & $0.03$ & $0.02$ & $0.03$ & $0.02$\\
Terciles ($2k=6$)
& $<0.01$ & $0.01$ & $<0.01$ & $<0.01$ & $<0.01$\\
\addlinespace
\multicolumn{6}{l}{\emph{Nonparametric bootstrap}}\\
Quartiles ($2k=8$)
& $0.04$ & $0.05$ & $0.02$ & $0.03$ & \textemdash\\
Terciles ($2k=6$)
& $<0.01$ & $0.02$ & $<0.01$ & $<0.01$ & \textemdash\\
\bottomrule
\end{tabular}

\begin{minipage}{0.96\textwidth}
\vspace{2pt}
\footnotesize
\emph{Notes:}
Panel~A reports estimates of the probability differences in \eqref{eq:instrument-validity-parameters}, with standard errors in parentheses. The sample contains 3,010 men (2,053 with $Z=1$ and 957 with $Z=0$), treated as unweighted i.i.d.\ observations. Both specifications use $D=1\{\text{years of schooling}\geq16\}$, with wage bins defined by pooled-sample quantiles of log hourly wages.
Panel~B reports joint $p$-values in percent for the parametric and nonparametric bootstrap tests. LF denotes the least favorable procedure. The two-step procedure uses a first-step confidence rectangle with $\tau=0.1\alpha$. For both implementations, LF calculations use 1 million draws, increased to 10 million for the tercile $T^m$ calculation. All reported two-step calculations use 10 million second-step draws. These counts refer to Gaussian draws for the parametric tests and bootstrap resamples for the nonparametric tests. OEH is reported only in its parametric version; \textemdash{} indicates an unreported result. $p$-values below $0.01\%$ are shown as $<0.01$.
\end{minipage}
\end{table}

\section{Conclusion}\label{sec:Conclusion}
This article develops least favorable and two-step tests of sign agreement among a fixed finite collection of parameters, allowing general correlation among their estimators. We establish uniform asymptotic validity and consistency against every fixed alternative for their parametric and nonparametric bootstrap implementations under the stated conditions. We also provide computable bounds for the Gaussian LF critical value with a prespecified probability guarantee and an analytical characterization for the max statistic under nonnegative pairwise correlations.

The two proposed tests are not uniformly ordered in power in general. Our local power analysis identifies configurations favoring each procedure, reflecting the tradeoff between restricting the nuisance set and reducing the second-step significance level. The simulations nevertheless generally favor the two-step test in terms of average power, with substantial gains under correlation. Our parametric LF max test also uniformly weakly dominates the existing alternative, the parametric test of \OEH{}, when both use the same studentized estimators and estimated correlation matrix. Under skewed errors, the nonparametric bootstrap implementations improve size accuracy, with null rejection probabilities close to the nominal level.

The empirical illustrations provide evidence of sign reversal in the dynamic effects of unilateral divorce laws and reject a testable implication of instrument validity in the college-proximity application. Extensions of interest include allowing the number of parameters to grow, testing sign agreement over a continuum, and identifying positive and negative parameters with simultaneous error control. Further analytical characterizations and data-dependent choices of the first-step level $\tau$ could also improve practical implementation.


\phantomsection
\bibliography{unif2sided}

\addcontentsline{toc}{section}{References}

\appendices
\renewcommand{\theequation}{\Alph{section}-\arabic{equation}}
\setcounter{equation}{0}
\small 

\section[\hspace{2mm}Computation of Least Favorable Critical Value]{Computation of Least Favorable Critical Value}\label{appendix:computation-lfcv}

The LF parametric critical value in \eqref{def:LF-cv-normal} is defined by maximizing the
$(1-\alpha)$-quantile over the nuisance parameter vector. This section provides details on how to numerically approximate this quantity. In particular, we construct lower and upper bounds which contain the LF critical value with a user-specified Monte Carlo confidence level.

Although the definition in \eqref{def:LF-cv-normal} involves the supremum over the unbounded set $[0,\infty)^k$, it can be bounded above by finitely many quantiles. The construction proceeds by splitting this set at a truncation level $A_r>0$ into the box $[0,A_r]^k$ and its complement, and by treating the two regions with different arguments. The box is compact, so it can be partitioned into finitely many cells whose diameters are small. 
The complement of the box can be covered by the $k$ regions. Lemma~\ref{lem:LF-outside-box} supplies, for each such region, a single envelope statistic that dominates the statistic uniformly over that region. The supremum over an unbounded set is therefore bounded above by finitely many quantiles, one for each cell and one for each coordinate. The same split organizes the lower bound.

To make this precise, we first explain how we compute quantile estimators of a generalized $p$-quantile for a generic random variable $X$,
$q(p; X):= \inf\{ x\in\mathbf{R}: P\{X\leq x\}\geq p\}.$ For a numerical step \(r\), choose a validation sample size \(M_r\) and a numerical error probability \(\eta_r\in(0,1)\). For \(M_r\) independent draws \(X_1,\ldots,X_{M_r}\) from \(X\), write \(X_{(1)}\leq\cdots\leq X_{(M_r)}\), with \(X_{(0)}=-\infty\) and \(X_{(M_r+1)}=+\infty\), and define
\begin{align*}
m_{L,r}(p)
&:=
\max\left\{
m\in\{0,\ldots,M_r\}:
P\{\operatorname{Binomial}(M_r,p)\leq m-1\}\leq\eta_r
\right\},\\
m_{U,r}(p)
&:=
\min\left\{
m\in\{1,\ldots,M_r+1\}:
P\{\operatorname{Binomial}(M_r,p)\geq m\}\leq\eta_r
\right\},
\end{align*}
and
\begin{align}
\hat q_{L,r}(p;X)
:=
X_{(m_{L,r}(p))},
\qquad
\hat q_{U,r}(p;X)
:=
X_{(m_{U,r}(p))}.
\label{def:LF-MC-quantile-bounds}
\end{align}
This construction guarantees that $[\hat q_{L,r}(p;X),\infty)$ and $(-\infty,\hat q_{U,r}(p;X)]$ contain the target quantile at least with probability $1-\eta_r$ (see \citet{ZielinskiZielinski2005} for details):
\begin{align}
P_{\mathrm{MC}}\{
\hat q_{L,r}(p;X)\leq q(p;X)
\}
\geq1-\eta_r,
\qquad
P_{\mathrm{MC}}\{
q(p;X)\leq\hat q_{U,r}(p;X)
\}
\geq1-\eta_r.
\label{eq:LF-MC-quantile-statements}
\end{align}
Furthermore, for each \(r\), fix a truncation level \(A_r>0\). Let \(\mathcal B_r\) be a finite rectangular partition of \([0,A_r]^k\), let \(\mathcal G_r\subset[0,A_r]^k\) be a finite set of candidate nuisance values, and choose \(g_B\in B\) for each \(B\in\mathcal B_r\) so that $\{g_B:B\in\mathcal B_r\}\subseteq\mathcal G_r$. 

For the lower bound, every \(g\in\mathcal G_r\) yields a feasible nuisance-specific quantile $q(1-\alpha,T^{\ell}(\Omega^{1/2}Z+g, \Omega))$ and hence a lower bound on \(c^\ell(1-\alpha,\Omega)\). To account for boundary configurations not approximated by the finite grid, we include the singleton-infinity limits \(R_j^\ell\) from Lemma~\ref{lem:LF-boundary-reduction}. We define the lower lound $\hat c_{L,r}^\ell$ by 
\begin{align}
\hat c_{L,r}^\ell
:=
\max\left\{
\max_{g\in\mathcal G_r}
\hat q_{L,r}
\left(
1-\alpha;
T^\ell(\Omega^{1/2}Z+g,\Omega)
\right),
~
\max_{j\in[k]}
\hat q_{L,r}
\left(
1-\alpha;
R_j^\ell
\right)
\right\},
\label{def:improved-lower-cv}
\end{align}
where \(Z\sim N(0,I_k)\), and the covariance argument of \(T^\ell\) is suppressed for \(\ell=m,s\).

For the upper bound, we require a bound that holds uniformly over each cell
\(B\in\mathcal B_r\). Define
\begin{align}
\rho_B
:=
\textstyle\sup_{\theta_1\in B}
\{(\theta_1-g_B)'\Omega^{-1}(\theta_1-g_B)\}^{1/2},
\qquad
\alpha_B
:=
\Phi\{\Phi^{-1}(\alpha)-\rho_B\}.
\label{def:improved-cell-radius}
\end{align}
Lemma~\ref{lem:LF-sharp-Gaussian-shift} implies that
\[
q(1-\alpha;
T^\ell(\Omega^{1/2}Z+\theta_1,\Omega))
\le
q(1-\alpha_B;
T^\ell(\Omega^{1/2}Z+g_B,\Omega))
\qquad
\forall\theta_1\in B.
\]
Thus the quantile at the representative point \(g_B\), evaluated at the
enlarged probability level \(1-\alpha_B\), provides an upper bound for the
entire cell. The set \(\mathcal E_{B,r}^\ell\) collects any additional
statistic-specific envelope random variables whose \((1-\alpha)\)-quantiles
also bound the cell uniformly. By Lemma~\ref{lem:parametric-cv-existence},
choose a finite constant \(C_\ell\) such that
\(c^\ell(1-\alpha,\Omega)<C_\ell\). We then define
\begin{align}
\hat u_{B,r}^\ell
:=
\min\Bigg\{
&C_\ell,\,
\hat q_{U,r}
\left(
1-\alpha_B;
T^\ell(\Omega^{1/2}Z+g_B,\Omega)
\right),
\left\{
\hat q_{U,r}(1-\alpha;E):
E\in\mathcal E_{B,r}^\ell
\right\}
\Bigg\},
\label{def:improved-cell-upper}
\end{align}
where the last component is omitted if
\(\mathcal E_{B,r}^\ell=\varnothing\).

For the complement of \([0,A_r]^k\), Lemma~\ref{lem:LF-outside-box}
provides \(k\) envelope random variables. In particular,
\(V^\ell(\Omega^{1/2}Z+A_re_j,\Omega)\) yields an upper bound for every
nuisance vector whose \(j\)th coordinate is at least \(A_r\). Define
\begin{align}
\hat u_{A,j,r}^\ell
:=
\min\left\{
C_\ell,\,
\hat q_{U,r}
\left(
1-\alpha;
V^\ell(\Omega^{1/2}Z+A_re_j,\Omega)
\right)
\right\},
\label{def:improved-outside-upper}
\end{align}
and combine the cellwise and outside-box bounds as
\begin{align}
\hat c_{U,r}^\ell
:=
\max\left\{
\max_{B\in\mathcal B_r}\hat u_{B,r}^\ell,
~
\max_{j\in[k]}\hat u_{A,j,r}^\ell
\right\}.
\label{def:improved-upper-cv}
\end{align}

Finally, let \(H_r\) denote the total number of one-sided quantile bounds used in constructing \(\hat c_{L,r}^\ell\) and
\(\hat c_{U,r}^\ell\). Specifically,
\begin{align}
H_r
:=
|\mathcal G_r|
+k
+\textstyle\sum_{B\in\mathcal B_r}
\{1+|\mathcal E_{B,r}^\ell|\}
+k.
\label{def:LF-Hr}
\end{align}
The terms in \(H_r\) correspond, respectively, to the finite-point lower bounds, the singleton-infinity lower bounds, the cellwise upper bounds, and the outside-box upper bounds. Given a user-chosen overall numerical error probability \(\delta_r\in(0,1)\), set
\begin{align}
\eta_r:=\frac{\delta_r}{H_r}.
\label{def:LF-eta-r}
\end{align}
Thus each one-sided quantile bound is constructed with failure probability \(\eta_r\), so the union bound gives simultaneous validity of all the quantile bounds with probability at least \(1-\delta_r\).

The following proposition restates Proposition \ref{prop:LF-computational-accuracy} with more details. It states that the two numerical bounds bracket the
population critical value with the user-specified Monte Carlo confidence level $1-\delta_r$, and that the bracket closes as the grid refines and the simulation effort grows.

\begin{proposition}\label{prop:LF-improved-numerical-accuracy}
Fix \(\alpha\in(0,1/2)\), \(\Omega\in\mathbf O\), and \(\ell\in\{m,s,q\}\). For each numerical step \(r\), let \(\hat c_{L,r}^\ell\) and \(\hat c_{U,r}^\ell\) be constructed as in \eqref{def:improved-lower-cv} and \eqref{def:improved-upper-cv}, using the order-statistic quantile bounds in \eqref{def:LF-MC-quantile-bounds} based on \(M_r\) validation draws and the error allocation \(\eta_r=\delta_r/H_r\) in \eqref{def:LF-eta-r}. Suppose that
\(\mathcal G_r\), \(\mathcal B_r\), the representatives \(g_B\), and the
collections \(\mathcal E_{B,r}^\ell\) are fixed before the validation draws
are generated. Then
\begin{align}
P_{\mathrm{MC}}\left\{
\hat c_{L,r}^\ell
\leq
c^\ell(1-\alpha,\Omega)
\leq
\hat c_{U,r}^\ell
\right\}
\geq
1-\delta_r,
\label{eq:improved-cv-certificate}
\end{align}
where \(P_{\mathrm{MC}}\) denotes probability with respect to the validation
draws.

Suppose, in addition, that $A_r\to\infty$ and $\max_{B\in\mathcal B_r}\rho_B\to0$ where $\rho_B$ is in \eqref{def:improved-cell-radius}. Suppose also that $\delta_r\to0$ and ${\log(H_r/\delta_r)}/{M_r}\to0.$ Then
\begin{align}
\hat c_{L,r}^\ell
\overset{P_{\mathrm{MC}}}{\longrightarrow}
c^\ell(1-\alpha,\Omega),
\qquad
\hat c_{U,r}^\ell
\overset{P_{\mathrm{MC}}}{\longrightarrow}
c^\ell(1-\alpha,\Omega).
\label{eq:improved-cv-consistency}
\end{align}
If, moreover,
\(\sum_{r=1}^\infty\delta_r<\infty\) and the validation draws across
numerical steps are defined on a common probability space, then both
convergences in \eqref{eq:improved-cv-consistency} hold almost surely.
\end{proposition}

The implemented test uses the midpoint of the numerical bracket as its critical value,
\begin{align*}
\hat c_{M,r,n}^\ell
:=
\frac{
\hat c_{L,r,n}^\ell+\hat c_{U,r,n}^\ell
}{2},
\end{align*}
where \(\hat c_{L,r,n}^\ell\) and \(\hat c_{U,r,n}^\ell\) are computed using the realized correlation matrix and \(\hat c_n^\ell=c^\ell(1-\alpha,\hat\Omega_n)\) denotes the exact feasible critical value. Conditional on the sample, Theorem \ref{prop:LF-improved-numerical-accuracy} implies that the approximation error $|\hat c_{M,r,n}^\ell-\hat c_n^\ell|$ is smaller than the half-width of the numerical bracket $[\hat c_{L,r,n}, \hat c_{U,r,n}]$ with probability at least $1-\delta_r$. Moreover, by \eqref{eq:improved-cv-consistency}, the midpoint converges to the exact critical value as the numerical approximation is refined. Alternatively, the upper endpoint \(\hat c_{U,r,n}^\ell\) may be used when a conservative numerical implementation is preferred. In either case, the bracket concerns only the numerical approximation of \(c^\ell(1-\alpha,\hat\Omega_n)\); it does not account for the distinct finite-sample errors arising from estimating \(\Omega(P)\), studentization, or the Gaussian approximation.

The choice of \(\delta_r\) determines the confidence level of the numerical certificate, while \(M_r\) determines its precision. A convenient way to hoose \(M_r\) is to specify a probability scale Monte Carlo tolerance \(\varepsilon_{\mathrm{MC},r}>0\) and require
\begin{align*}
\left\{
\frac{\log(2H_r/\delta_r)}{2M_r}
\right\}^{1/2}
\leq
\varepsilon_{\mathrm{MC},r},
\end{align*}
This choice ensures that, with probability at least \(1-\delta_r\), all empirical CDFs entering the numerical bounds are within \(\varepsilon_{\mathrm{MC},r}\) of their population counterparts uniformly in their arguments. The resulting accuracy in the critical value scale depends on the distributions near their target quantiles and is measured directly by the observed width \(\hat c_{U,r}^\ell-\hat c_{L,r}^\ell\). Thus \(M_r\) may be
increased until this width is sufficiently small for the intended application.

\begin{remark}\label{rem:LF-cell-envelopes}\normalfont
The collection \(\mathcal E_{B,r}^\ell\) consists of additional
statistic-specific envelope random variables that provide valid upper bounds
for the cell \(B\). Its precise composition is not essential for the validity
of Theorem~\ref{prop:LF-improved-numerical-accuracy}: any collection may be
used, including \(\mathcal E_{B,r}^\ell=\varnothing\), provided that each
\(E\in\mathcal E_{B,r}^\ell\) satisfies
\[
\sup_{\theta_1\in B}
J^{\ell,-1}(1-\alpha,\theta_1,\Omega)
\leq
q(1-\alpha;E),
\]
and the collection is fixed before the validation draws are generated.
Including additional envelopes can tighten the numerical upper bound but may
increase computation. To keep the presentation concise, we do not list the
statistic-specific envelopes used in the implementation; computational details are available from the author upon request or from the python package.
\end{remark}

\section[\hspace{2mm}Comparison with Oulhaj et al.(2019)]{Comparison with Oulhaj et al.(2019)}
\label{sec:oulhaj-comparison}

This section establishes two comparisons with \OEH{}. Our parametric LF max test is uniformly at least as powerful as their parametric procedure when both use the same studentized estimators and estimated correlation matrix. Their nonparametric bootstrap procedure, however, need not be pointwise asymptotically valid.

\subsection{Power dominance of the parametric LF test}
\label{subsec:oulhaj-parametric-dominance}
Fix $\alpha\in(0,1/2)$ and let $Z\sim N(0_k,I_k)$. \OEH{} combine two one-sided max tests through the intersection--union principle. When both procedures use the same studentization and estimated correlation matrix, their parametric (Gaussian) critical value is
\begin{equation}
 c_{\mathrm{OEH}}(1-\alpha,\Omega)
 =(1-\alpha)\text{-quantile of }
 \max_{j\in[k]}(\Omega^{1/2}Z)_j,
 \qquad \Omega\in\bar{\mathbf O}.
 \label{eq:oulhaj-cv}
\end{equation}
\OEH{} obtain this critical value by evaluating each one-sided statistic at the zero mean vector (their equation (9), p.~155). Their truncation of the directional statistics at zero does not affect the critical value or rejection rule because $c_{\mathrm{OEH}}(1-\alpha,\Omega)
\geq\Phi^{-1}(1-\alpha)>0$. Their parametric test is therefore
\begin{equation}
 \phi_n^{\mathrm{OEH}}
 =I\{T_n^m>c_{\mathrm{OEH}}(1-\alpha,\hat\Omega_n)\}.
 \label{eq:oulhaj-test}
\end{equation}

\begin{proposition}\label{prop:oulhaj-dominance}
For every $\Omega\in\bar{\mathbf O}$,
the LF parametric max critical value in
\eqref{def:LF-cv-normal} satisfies
\begin{equation}
 c^m(1-\alpha,\Omega)
 \leq c_{\mathrm{OEH}}(1-\alpha,\Omega).
 \label{eq:oulhaj-cv-order}
\end{equation}
Consequently, the LF max test $\phi_n^m$
in \eqref{def:LFtest} satisfies
\[
\phi_n^{\mathrm{OEH}}\leq\phi_n^m
\]
for every sample with well-defined studentization.
In the Gaussian experiment $W\sim N(\theta_1,\Omega)$
with known $\Omega\in\mathbf O$, the LF test has strictly
higher power at every
$\theta_1\notin\mathbf R^k_+\cup\mathbf R^k_-$.
\end{proposition}

\begin{proof}
Fix $\Omega\in\bar{\mathbf O}$ and let $Y=\Omega^{1/2}Z$.
For every $\theta_1\in\mathbf R^k_+$,
\[
 T^m(Y+\theta_1)
 \leq\max_{j\in[k]}(-Y_j-\theta_{1,j})
 \leq\max_{j\in[k]}(-Y_j)
 \overset{d}{=}\max_{j\in[k]}Y_j.
\]
Gaussian symmetry and $T^m(-x)=T^m(x)$ give the same
distributional bound for $\theta_1\in\mathbf R^k_-$.
Taking $(1-\alpha)$-quantiles and then the supremum over
$\mathbf R^k_+\cup\mathbf R^k_-$ in
\eqref{def:LF-cv-normal} proves \eqref{eq:oulhaj-cv-order}.
Applying this inequality at the common estimated correlation
matrix $\hat\Omega_n$ and comparing \eqref{eq:oulhaj-test}
with the parametric version of \eqref{def:LFtest} gives
$\phi_n^{\mathrm{OEH}}\leq\phi_n^m$ for every sample
with well-defined studentization. The strict critical-value inequality and strict power dominance for known positive definite $\Omega$ are established in Lemma \ref{lem:oulhaj-strict-dominance}.
\end{proof}

The gain comes from choosing the critical value for the sign-agreement statistic directly. \OEH{} choose a critical value that controls each one-sided rejection
probability separately. The LF critical value instead controls the probability that \emph{both} one-sided statistics exceed the same threshold under the union null.

\subsection{Pointwise asymptotic invalidity of the bootstrap test}
\label{subsec:oulhaj-bootstrap-invalidity}

The bootstrap test of OEH is not pointwise asymptotically valid in general. Consider their test of $\mu(P)\geq0_k$, using unit weights, which their
framework permits (p.~155). Its statistic and bootstrap analogue are
\[
 \begin{aligned}
 T_{n,-}&=\max\{0,\max_{j\leq k}(-\sqrt n\bar W_{j,n})\},\\
 T_{n,-}^*&=\max\{0,\max_{j\leq k}[-\sqrt n(\bar W_{j,n}^*-\bar W_{j,n})-\sqrt n\max(\bar W_{j,n},0)]\}.
 \end{aligned}
\]
The other direction reverses signs. Let $\phi_n^{\mathrm{BT}}$ denote
the resulting intersection--union test, rejecting only when both
statistics strictly exceed their respective ideal $95\%$ bootstrap
quantiles. This is the bootstrap recentering displayed on p.~156 of OEH. While $\sqrt n(\bar W_{j,n}^*-\bar W_{j,n})$ is asymptotically normal, the other term $\sqrt n\max(\bar W_{j,n},0)$ diverges under a fixed null hypothesis. The following example shows that this recentering with by $\sqrt n\max(\bar W_{j,n},0)$ leads to the asymptotic size violation. 

\begin{example}\normalfont
\label{ex:oulhaj-bootstrap-size}
Take $k=3$ and fix $P=N((a,0,0),I_3)$ with $a>0$, so that
$P\in\mathbf P_0$ in \eqref{def:null-hypo}. Use the ordinary
empirical bootstrap. The positive-direction statistic diverges while
its critical value stays bounded in probability; hence rejection is
asymptotically determined by $T_{n,-}$.
Write $$X=-\sqrt n\bar W_{2,n}\quad Y=-\sqrt n\bar W_{3,n}\quad\text{so that}\quad T_{n,-}=\max(X,Y,0)$$
with probability tending to one.

The conditional bootstrap CLT gives a Gaussian approximation uniform over rectangles.
For $t\geq0$, the event $T_{n,-}^*\leq t$ requires each of the three terms inside its maximum to be at most $t$. Moving the estimated shifts to the right-hand side therefore gives
\begin{align}\label{def:boot_CDF}
 \begin{aligned}
 \hat P_n^*\{T_{n,-}^*\leq t\}
 &=\hat P_n^*\bigl\{
 -\sqrt n(\bar W_{j,n}^*-\bar W_{j,n})
 \leq t+\sqrt n\max(\bar W_{j,n},0)
 \text{ for all }j\bigr\}\\
 &=\textstyle\prod_{j=1}^3
 \Phi\bigl(t+\sqrt n\max(\bar W_{j,n},0)\bigr)+o_P(1),
 \end{aligned}
\end{align}
uniformly over $t\geq0$. Since $\bar W_{1,n}\to a>0$, the first factor tends to one uniformly over $t\geq0$. The remaining shifts satisfy $\sqrt n\max(\bar W_{2,n},0)=-\min(X,0)$ and $\sqrt n\max(\bar W_{3,n},0)=-\min(Y,0)$.
Evaluating \eqref{def:boot_CDF} at $t=T_{n,-}$, which
uniformity permits, gives
\[
 \hat P_n^*\{T_{n,-}^*>T_{n,-}\}=p(X,Y)+o_P(1).
\] on $\{\max(X,Y,0)>0\}$, where
\[
 p(X,Y)=1-\Phi(\max(X,Y,0)-\min(X,0))
             \Phi(\max(X,Y,0)-\min(Y,0)).
\]

Next, the limiting rejection probability is
\begin{equation}
 \lim_{n\to\infty}\mathbb E_P[\phi_n^{\mathrm{BT}}]
 =P\{\max(X,Y,0)>0,\ p(X,Y)<0.05\}.
 \label{eq:oulhaj-bootstrap-rejection-limit}
\end{equation}
Indeed, the other one-sided test rejects with probability
tending to one. For $T_{n,-}>0$, the uniform approximation
and the continuity and strict increase of the approximating
CDF determine the bootstrap-quantile decision, since
$P\{\max(X,Y,0)>0,\ p(X,Y)=0.05\}=0$.
At $T_{n,-}=0$, strict rejection is impossible because
the bootstrap critical value is nonnegative.

We now obtain a lower bound by identifying three disjoint subsets of
this limiting rejection event. Let
$c=\Phi^{-1}(\sqrt{0.95})\approx 1.9545$ and
$A:=\{\max(X,Y,0)>c\}$.
Since $\min(X,0)\leq0$ and $\min(Y,0)\leq0$, on $A$ we have $ p(X,Y)\leq1-\Phi(\max(X,Y,0))^2 <1-\Phi(c)^2=0.05.$ Thus $A$ is contained in the limiting rejection event. There is additional probability in the limiting rejection
event outside $A$. In particular, consider
$E:=\{1.7<X<1.95,\ Y<-1\}$.
On $E$, we have $\max(X,Y,0)=X<c$, so $E$ is disjoint
from $A$. Nevertheless, $\min(X,0)=0$ and $\min(Y,0)=Y$,
giving $p(X,Y)=1-\Phi(X)\Phi(X-Y)
<1-\Phi(1.7)\Phi(2.7) \approx0.047878<0.05.$ By symmetry, the same holds for $E':=\{1.7<Y<1.95,\ X<-1\}$. The three events are disjoint. Independence and \eqref{eq:oulhaj-bootstrap-rejection-limit} therefore give
\[
\begin{aligned}
\lim_{n\to\infty}\mathbb E_P[\phi_n^{\mathrm{BT}}]
&\geq P(A)+P(E)+P(E') =0.0560217\ldots>0.05.
\end{aligned}
\]
This violates pointwise asymptotic validity at a fixed
null distribution.
\end{example}

\addcontentsline{apc}{appendices}{\protect\numberline{\thesubsection}Title of Appendix A}


\newpage 

\clearpage
\begin{titlepage}
  \centering

  {\Large \textsc{Supplementary Appendix for}\\[0.5em]
  \textbf{``Testing Sign Agreement''}}\\[2em]

  {\large
  \begin{minipage}[t]{\textwidth}
    \centering
    Deborah Kim\\
    Department of Economics\\
    University of Warwick\\
    \url{deborah.kim@warwick.ac.uk}
  \end{minipage}
  \par
  }
  \vspace{3em}

  \begin{spacing}{1.1}
  \begin{abstract}
  \noindent
This document provides supplementary material for the paper
``Testing Sign Agreement.'' It presents additional simulation power figures, extends the parametric procedures to regular estimators, and collects proofs and supporting lemmas for the theoretical and computational results.
  \end{abstract}
  \end{spacing}

  \vspace{1em}
  \noindent\textbf{KEYWORDS:} union of moment inequalities, uniform asymptotic validity, bootstrap

  \vspace{0.5em}
  \noindent\textbf{JEL codes:} C12, C31, C35, C36.
  \vfill
\end{titlepage}

\renewcommand{\thesection}{S.\arabic{section}}
\renewcommand{\thesubsection}{S.\arabic{section}.\arabic{subsection}}
\renewcommand{\theequation}{S.\arabic{section}.\arabic{equation}}
\numberwithin{equation}{section}   
\renewcommand{\thetheorem}{S.\arabic{theorem}}
\renewcommand{\thelemma}{S.\arabic{lemma}}
\renewcommand{\thecorollary}{S.\arabic{corollary}}
\renewcommand{\theproposition}{S.\arabic{proposition}}
\renewcommand{\theassumption}{S.\arabic{assumption}}
\renewcommand{\thedefinition}{S.\arabic{definition}}
\renewcommand{\theexample}{S.\arabic{example}}
\renewcommand{\theremark}{S.\arabic{remark}}
\setcounter{section}{0}
\setcounter{theorem}{0}
\setcounter{lemma}{0}
\cleardoublepage
\setcounter{section}{0} 
\setcounter{lemma}{0}
\setcounter{equation}{0}


\startcontents[appendix]   
\printcontents[appendix]{l}{1}{
\section*{Supplementary Contents}}

Throughout this supplementary appendix, we use the notation of Section~\ref{sec:setup}, and inequalities between vectors are interpreted componentwise. For a vector $x$, write $\|x\|=(\sum_jx_j^2)^{1/2}$ and $\|x\|_\infty=\max_j|x_j|$, and define its positive and negative parts by $x_+=(\max\{x_j,0\})_j$ and $x_-=(\max\{-x_j,0\})_j$. For a matrix $A=(a_{ij})$, let $\|A\|_{\max}=\max_{i,j}|a_{ij}|$ and $\|A\|_{\mathrm{op}}=\sup_{\|x\|=1}\|Ax\|$; an unadorned matrix norm $\|A\|$ also denotes $\|A\|_{\max}$. Write $[k]=\{1,\ldots,k\}$ and $I^c=[k]\setminus I$. For $I,J\subseteq[k]$, $x_I$ denotes the subvector indexed by $I$, and $A_{I\times J}$ denotes the submatrix with rows in $I$ and columns in $J$. For a symmetric matrix $A$, $\lambda_{\min}(A)$ and $\lambda_{\max}(A)$ denote its smallest and largest eigenvalues, and $A\succ0$ means that $A$ is positive definite; when $A$ is positive semidefinite, $A^{1/2}$ denotes its symmetric positive semidefinite square root. We use $\Rightarrow$ for weak convergence and $\overset{P_n}{\to}$ for convergence in probability under $P_n$. The conditional bootstrap law given the original sample is denoted by $\hat P_n^*$; $\rightsquigarrow_{\hat P_n^*}$ and $\overset{\hat P_n^*}{\to}$ denote conditional weak convergence and conditional convergence in probability, respectively, both understood in $P_n$-probability. All quantiles are generalized left quantiles, $F^{-1}(u)=\inf\{x\in\mathbf R:F(x)\ge u\}$ for $u\in(0,1)$, with $\inf\varnothing=+\infty$.

We complete the definitions of the LF and two-step procedures on exceptional samples as follows. Whenever any marginal sample variance is zero, assign the value $+\infty$ to the statistic entering $J_n^\ell$, for every nuisance value and every $\ell\in\{m,s,q\}$; the same convention applies to bootstrap resamples. On such original samples, set $T_n^\ell$ and both LF critical values equal to $+\infty$ and define all tests not to reject. No undefined studentization is evaluated. When all marginal variances are positive, the LF max and sum procedures retain the estimated correlation matrix even if it is singular. For QLR, replace any singular original or bootstrap correlation matrix $M$ by $(M+\varepsilon_n I_k)/(1+\varepsilon_n)$, where $\varepsilon_n>0$ and $\varepsilon_n\to0$, using the resulting matrix consistently wherever that correlation matrix enters the statistic or reference distribution. For the ellipsoidal two-step procedures, if the original covariance matrix $\hat\Sigma_n$ is not positive definite, set $\psi_n^\ell=\psi_n^{\ell *}=0$ and both first-step confidence regions equal to $\mathbf R^k$; this convention takes precedence over the displayed rejection rules, and no covariance inverse is evaluated on these samples. If a bootstrap covariance matrix $\hat\Sigma_n^*$ is not positive definite, assign the bootstrap Wald statistic the value $+\infty$. All bootstrap resamples are retained in their conditional distributions. Generalized quantiles and suprema may therefore equal $+\infty$. In particular, on an original sample with positive definite covariance, $r_{n,\tau}^*=+\infty$ gives $\mathbf C_n^*(1-\tau)=\mathbf R^k$, and the second-step critical value is then computed over the entire null parameter space. A critical value of $+\infty$ implies nonrejection. Lemma~\ref{lem:centered-sampling-bootstrap} implies that zero marginal variances have vanishing sampling probability and vanishing conditional bootstrap probability. Under Assumption~\ref{ass:uniform-eigenvalue}, the same conclusions hold for singular covariance matrices. Thus these conventions do not affect the asymptotic conclusions under the assumptions imposed for each procedure.

\section[\hspace{2mm}Extension to Regular Estimators]{Extension to Regular Estimators}\label{sec:extension-regular-estimators}

The sample-mean formulation used in the main text is convenient but is not
essential for the proposed tests. This section shows that the least favorable (LF) test and the two-step test with their parametric critical values remain uniformly asymptotically valid for a general regular estimator. We also briefly discuss familiar estimators covered by the high-level conditions below. Throughout, $k\ge2$ is fixed, as in the main text.

Let $U_i$, $i=1,\ldots,n$, be an i.i.d. sequence
of random elements taking values in a measurable
space $(\mathcal X,\mathcal A)$, with common
distribution $P\in\mathbf P$, where $\mathbf P$
is a class of probability measures on
$(\mathcal X,\mathcal A)$. Let $\beta(P)=(\beta_1(P),\ldots,\beta_k(P))'\in
\mathbf R^k$ be the parameter of interest. The null hypothesis is
\begin{align*}
    H_0:\quad \beta(P)\in\mathbf R_+^k\cup\mathbf R_-^k,
\end{align*}
and write
\[
\mathbf P_{\beta,0}
:=
\{P\in\mathbf P:\beta(P)\in\mathbf R_+^k\cup\mathbf R_-^k\}.
\]
Suppose that $\hat\beta_n$ estimates $\beta(P)$ and admits an influence
function $\zeta_P:\mathcal X\to\mathbf R^k$ satisfying
$\mathbb E_P[\zeta_P(U_i)]=0_k$. Define
\begin{align*}
    V(P)&:=\mathbb E_P[\zeta_P(U_i)\zeta_P(U_i)'],\\
    \sigma_j^2(P)&:=V_{jj}(P),\qquad
    D(P):=\operatorname{diag}(\sigma_1(P),\ldots,\sigma_k(P)),\\
    \Omega(P)&:=D(P)^{-1}V(P)D(P)^{-1},
\end{align*}
where $0<\sigma_j^2(P)<\infty$ for every $P\in\mathbf P$ and
$j\in[k]$. Let $\hat V_n$ be a symmetric positive-semidefinite estimator of
$V(P)$. On the event that all diagonal elements of $\hat V_n$ are strictly
positive, let $\hat D_n$ be the diagonal matrix formed from their square roots
and define
\[
    \hat\Omega_n:=\hat D_n^{-1}\hat V_n\hat D_n^{-1}.
\]
On the complement of this event, set $\hat D_n=I_k$ and $\hat\Omega_n=I_k$ for notation and define both tests to be zerob. The consistency
condition below implies that this convention is asymptotically irrelevant.

For $\ell\in\{m,s,q\}$, define
\begin{align*}
    T_{n,\beta}^\ell
    :=T^\ell\left(
    \sqrt n\,\hat D_n^{-1}\hat\beta_n,
    \hat\Omega_n
    \right),
\end{align*}
where the second argument is suppressed for $\ell=m,s$. The extension of
the LF test considered in this section uses the Gaussian parametric critical
value:
\begin{align}\label{def:general-estimator-parametric-LF-test}
    \phi_{n,\beta}^{\ell}
    :=I\left\{
    T_{n,\beta}^\ell>c^\ell(1-\alpha,\hat\Omega_n)
    \right\}.
\end{align}
{For} the {LF} test with {$\ell=q$}, the {regularization convention in
Remark~\ref{remark:regularized-correlation}} is {applied} to
{$\hat\Omega_n$ in both the statistic} and the critical value{.

Fix $\tau\in(0,\alpha)$. On samples with $\hat V_n\succ0$}, {define
\begin{align}\label{def:general-estimator-normal-CS}
    \mathbf C_{n,\beta}(1-\tau)
    :=\left\{b\in\mathbf R^k:
    n(\hat\beta_n-b)'\hat V_n^{-1}(\hat\beta_n-b)
    \le\chi^2_{k,1-\tau}\right\}.
\end{align}
Whenever its intersection with} the null {parameter space is nonempty}, {let
\begin{align}\label{def:general-estimator-two-step-cv}
    \tilde c_{n,\beta}^\ell
    :=\sup_{b\in\mathbf C_{n,\beta}(1-\tau)
                   \cap(\mathbf R_+^k\cup\mathbf R_-^k)}
       (J^\ell)^{-1}
       (1-\alpha+\tau,\sqrt n\,\hat D_n^{-1}b,\hat\Omega_n).
\end{align}
The parametric two-step test is
\begin{align}\label{def:general-estimator-two-step-test}
    \psi_{n,\beta}^\ell
    :=\begin{cases}
    1,&\mathbf C_{n,\beta}(1-\tau)
         \cap(\mathbf R_+^k\cup\mathbf R_-^k)=\varnothing,\\[3pt]
    0,&\mathbf C_{n,\beta}(1-\tau)
         \subseteq\mathbf R_+^k\cup\mathbf R_-^k,\\[3pt]
    I\{T_{n,\beta}^\ell>\tilde c_{n,\beta}^\ell\},
       &\text{otherwise}.
    \end{cases}
\end{align}
If $\hat V_n$ is not positive definite, set
$\psi_{n,\beta}^\ell=0$ and,} for the {proof,
$\mathbf C_{n,\beta}(1-\tau)=\mathbf R^k$; no inverse} is {evaluated
on these samples}. {Thus this section considers Gaussian critical values
for both procedures.}

\begin{assumption}
\label{ass:regular-asymptotically-linear}
The following statements hold for every sequence $\{P_n\}\subset\mathbf P$.
\begin{enumerate}[(i)]
    \item The estimator is uniformly asymptotically linear after
    studentization:
    \begin{align}\label{eq:uniform-asymptotic-linearity-general}
        \sqrt n\,D(P_n)^{-1}
        \bigl(\hat\beta_n-\beta(P_n)\bigr)
        =
        \frac{1}{\sqrt n}\sum_{i=1}^n
        \varphi_{P_n}(U_i)
        +o_{P_n}(1),
    \end{align}
    where
\[
\varphi_P(u):=D(P)^{-1}\zeta_P(u),
\qquad
\varphi_{P,j}(u):=
\frac{\zeta_{P,j}(u)}{\sigma_j(P)},
\quad j\in[k],
\]
and $\zeta_{P,j}$ denotes the $j$th component of $\zeta_P$.

    \item For each $j\in[k]$, the squared normalized influence functions are
    uniformly integrable:
    \begin{align}\label{eq:UI-normalized-influence-function}
        \lim_{\lambda\to\infty}
        \sup_{P\in\mathbf P}
        \mathbb E_P\left[
        \varphi_{P,j}(U_i)^2
        I\{|\varphi_{P,j}(U_i)|>\lambda\}
        \right]
        =0.
    \end{align}

    \item The scale and correlation estimators are well defined with
    probability approaching one and are uniformly consistent in the
    sequential sense:
    \begin{align}\label{eq:studentization-general-well-defined}
        P_n\!\left\{\min_{j\in[k]}(\hat V_n)_{jj}>0\right\}\to1,
    \end{align}
    and
    \begin{align}\label{eq:studentization-general-consistency}
        \hat D_n^{-1}D(P_n)\overset{P_n}{\to}I_k,
        \qquad
        \|\hat\Omega_n-\Omega(P_n)\|_{\max}
        \overset{P_n}{\to}0.
    \end{align}
\end{enumerate}
\end{assumption}

Condition \eqref{eq:UI-normalized-influence-function} is the influence-function analogue of the standardized uniform integrability condition in \eqref{def:uniform-integrability}.
 Assumption \ref{ass:regular-asymptotically-linear} is satisfied by many
conventional fixed-dimensional estimators under their standard regularity
conditions. These include OLS estimators under suitable moment conditions and
uniform nonsingularity of the population regressor second-moment matrix;
linear instrument variable and two-stage least-squares estimators under suitable moment conditions and uniform strong identification; and repeated-cross-section or
regression difference-in-differences estimators under the corresponding
overlap, full-rank, moment, and variance-estimation conditions. More generally,
the assumption applies to regular estimators that are smooth functions of
finitely many population moments when the relevant derivatives are well
behaved, the asymptotic variances are nondegenerate, and the covariance
estimator is uniformly consistent.


\begin{corollary}\label{cor:general-estimator-uniform-validity}
Suppose Assumption \ref{ass:regular-asymptotically-linear} holds. Then, for
$\alpha\in(0,1/2)$ and $\ell\in\{m,s\}$, the LF test with the Gaussian
parametric critical value in \eqref{def:general-estimator-parametric-LF-test} satisfies
\begin{align}\label{eq:general-estimator-LF-validity}
    \limsup_{n\to\infty}
    \sup_{P\in\mathbf P_{\beta,0}}
    \mathbb E_P[\phi_{n,\beta}^{\ell}]
    \le\alpha.
\end{align}
If, in addition,
\begin{align}\label{eq:general-estimator-eigenvalue-condition}
    \inf_{P\in\mathbf P}\lambda_{\min}(\Omega(P))>0,
\end{align}
then \eqref{eq:general-estimator-LF-validity} also holds for $\ell=q$.
Under this additional condition, for every $\tau\in(0,\alpha)$ and
$\ell\in\{m,s,q\}$, the parametric two-step test in \eqref{def:general-estimator-two-step-test} satisfies
\begin{align}\label{eq:general-estimator-two-step-validity}
    \limsup_{n\to\infty}
    \sup_{P\in\mathbf P_{\beta,0}}
    \mathbb E_P[\psi_{n,\beta}^\ell]
    \le\alpha.
\end{align}
\end{corollary}

The proof is given in Subsection~\ref{subsec:supporting-regular-estimators}. The sample-mean framework in the main text is a special case. Taking
$\beta(P)=\mathbb E_P[W_i]$, $\hat\beta_n=\bar W_n$, and
$\zeta_P(W_i)=W_i-\mu(P)$ makes
\eqref{eq:uniform-asymptotic-linearity-general} an exact identity and reduces
\eqref{eq:UI-normalized-influence-function} to
\eqref{def:uniform-integrability}. Condition
\eqref{eq:studentization-general-consistency} then follows from Lemmas
\ref{alem:consistency-sample-variance} and
\ref{alem:consistency-sample-corr}.


\section[\hspace{2mm}Additional Power Figures]{Additional Power Figures}\label{app:power-additional}
Figures~\ref{fig:power-estimated-zero}, \ref{fig:power-estimated-positive}, and~\ref{fig:power-estimated-negative} report power under $\Omega_{\mathrm{Zero}}$, $\Omega_{\mathrm{Positive}}$, and $\Omega_{\mathrm{Negative}}$, respectively, using the same designs and estimated correlation matrices. The $\Omega_{\mathrm{Zero}}$ figure additionally includes ZSS, which assumes known independence. The complete comparison across statistics and bootstrap versions is reported in Table~\ref{tab:average-rejection-estimated}.

\begin{figure}[!p]
\centering
\includegraphics[width=\linewidth]{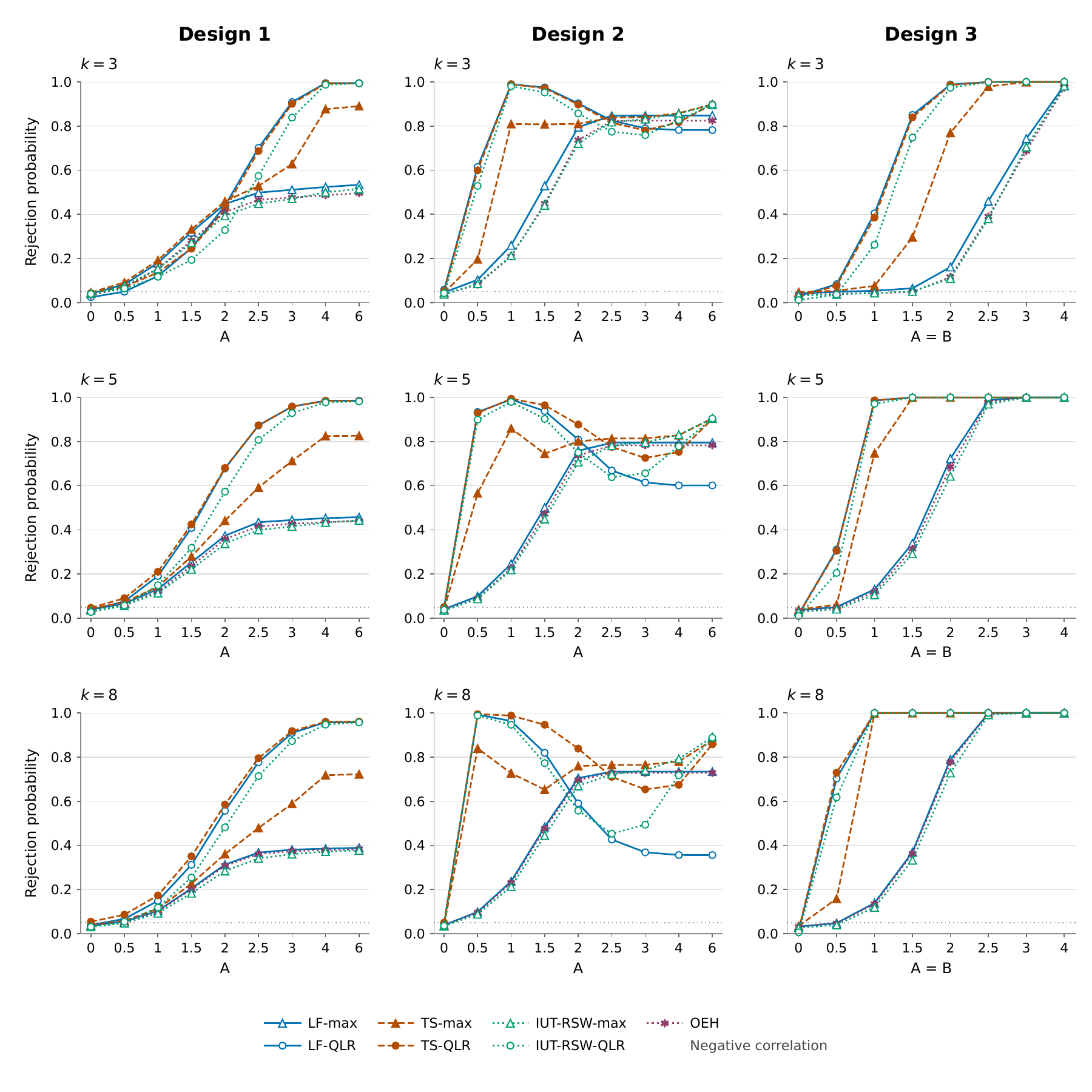}
\caption{Power under $\Omega_{\mathrm{Negative}}$ with estimated correlation matrices. Columns are Designs 1--3; rows are $k=3,5,8$. Curves show LF, TS and IUT--RSW using max and QLR, together with OEH. The horizontal dotted line marks $\alpha=0.05$. The null configuration $A=0$ is shown for reference but excluded from the table averages.}
\label{fig:power-estimated-negative}
\end{figure}

\begin{figure}[!p]
\centering
\includegraphics[width=\linewidth]{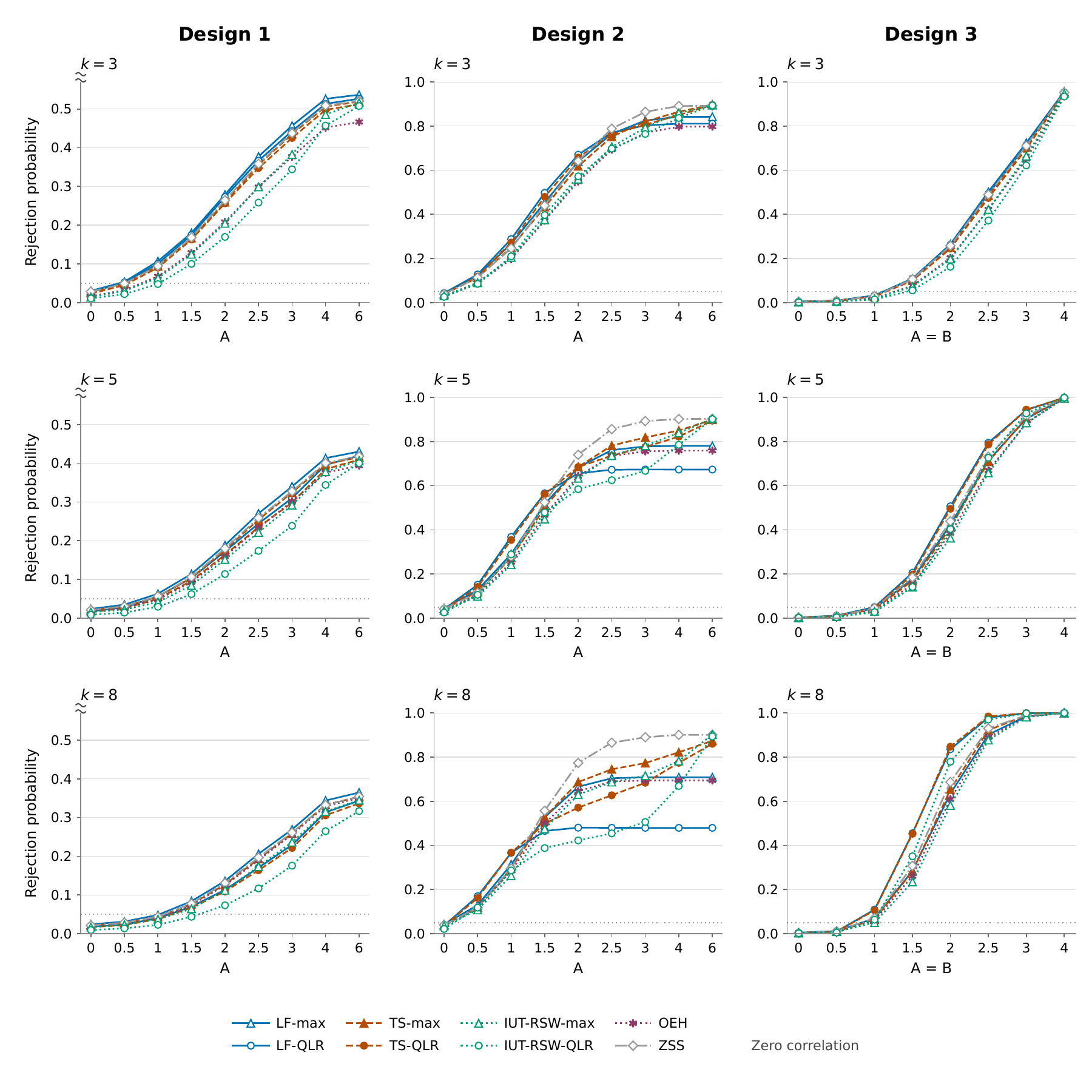}
\caption{Power under $\Omega_{\mathrm{Zero}}$ with estimated correlation matrices. Columns are Designs 1--3; rows are $k=3,5,8$. Curves show LF, TS and IUT--RSW using max and QLR, together with OEH and ZSS. ZSS uses known independence. The horizontal dotted line marks $\alpha=0.05$; $A=0$ is a null reference point.}
\label{fig:power-estimated-zero}
\end{figure}

\begin{figure}[!p]
\centering
\includegraphics[width=\linewidth]{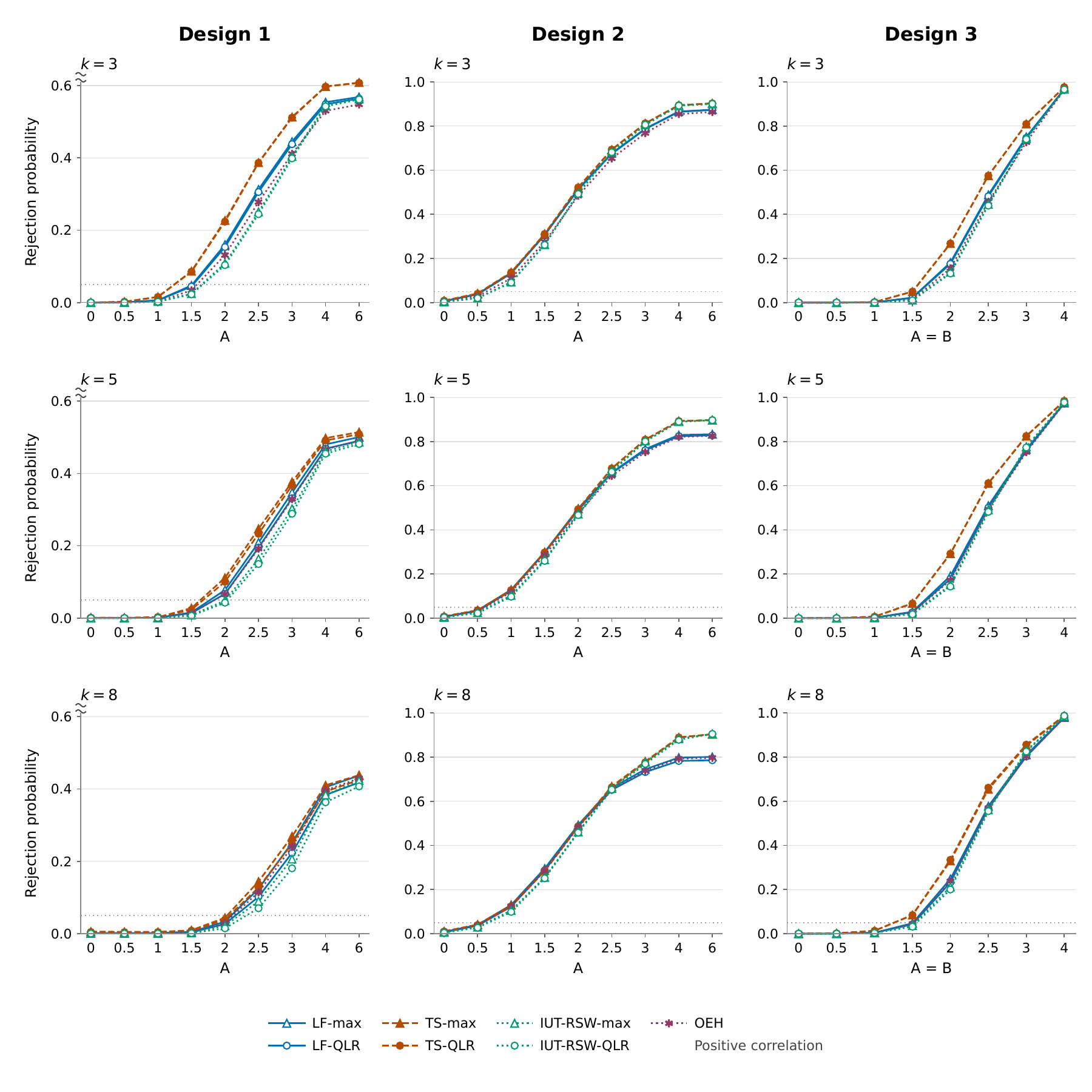}
\caption{Power under $\Omega_{\mathrm{Positive}}$ with estimated correlation matrices. Columns are Designs 1--3; rows are $k=3,5,8$. Curves show LF, TS and IUT--RSW using max and QLR, together with OEH. The horizontal dotted line marks $\alpha=0.05$; $A=0$ is a null reference point.}
\label{fig:power-estimated-positive}
\end{figure}

\section[\hspace{2mm}Proof of the main results]{Proof of the main results}\label{supp:proof-main-results}

Throughout this section, when $\ell=q$, with a slight abuse of notation,
$\hat\Omega_n$ denotes the matrix used in the test: it denotes the sample
correlation matrix when $\hat\Omega_n\succ0$, and otherwise denotes
\[
\tilde\Omega_n
:=
\frac{\hat\Omega_n+\varepsilon_n I_k}{1+\varepsilon_n},
\qquad
\varepsilon_n>0,\quad \varepsilon_n\to0.
\]
Accordingly, every occurrence of $T_n^q$, $J^q$, and $\hat c_n^q$ is
understood to use this same matrix.

\subsection{Proof of Theorem \ref{thm:LFtest-uniform-validity}}
\label{subsec:LF-validity-proof}
Here we provide the proof of the uniform asymptotic validity of the least favorable test with the parametric bootstrap critical value. The nonparametric bootstrap {counterpart follows the same two cases} and is {given below. Write $$\hat c_n^\ell(1-\alpha):=c^\ell(1-\alpha,\hat\Omega_n),\qquad\hat c_n^{\ell *}(1-\alpha):=c_n^{\ell *}(1-\alpha,\hat P_n).$$ On original samples with a zero marginal variance, the test does not reject by Remark~\ref{remark:regularized-correlation}. For the proof, set $T_n^\ell=0$ and both critical values equal} to {$+\infty$ there, and extend any reference CDF below by the point mass at zero}. {These samples have probability tending to zero by Lemma~\ref{lem:centered-sampling-bootstrap}.}

\noindent\textit{Parametric bootstrap critical value.} Suppose, by way of contradiction, that the upper bound fails for
\(\ell\in\{m,s,q\}\) {with the parametric critical value}.
Then there exist {\(\eta\in(0,1-\alpha)\)}, a subsequence \(n_l\), and a sequence
\(P_{n_l}\in\mathbf P_0\) such that
\begin{align}\label{contra-hypo:LFC-uniform-validity}
\mathbb E_{P_{n_l}}[\phi^\ell_{n_l}]
=
P_{n_l}\{T^\ell_{n_l}>\hat c^\ell_{n_l}(1-\alpha)\}
>
\alpha+\eta
\end{align}
for all \(l\). Since \(P_{n_l}\in\mathbf P_0\), after passing to a further subsequence, still denoted by
\(n_l\), we may assume that either $\mu(P_{n_l})\in\mathbf R_+^k$ for all {$l$} or $\mu(P_{n_l})\in\mathbf R_-^k$ for all {$l$}. We prove the result for the latter. The positive case follows by replacing \(W_i\) by \(-W_i\). By compactness of \(\bar{\mathbf O}\), after passing to a further subsequence, $\Omega(P_{n_l})\to\Omega^*$ for some \(\Omega^*\in\bar{\mathbf O}\). {For $\ell=q$, Assumption \ref{ass:uniform-eigenvalue} gives $\Omega^*\succ0$.}

By extracting a further subsequence if necessary, each coordinate of $\frac{\sqrt{n_l}\mu_j(P_{n_l})}{\sigma_j(P_{n_l})}$ has an extended limit in \([-\infty,0]\). There are two cases.

\medskip\noindent
\textbf{Case 1:} Suppose {$\tfrac{\sqrt{n_l}\mu_j(P_{n_l})}{\sigma_j(P_{n_l})}\to-\infty$} for all {$j$.} By Lemma{~}\ref{lem:statistic-degenerate-case}, {$T^m_{n_l}\overset{P_{n_l}}{\to}-\infty$} and $T^\ell_{n_l}\overset{P_{n_l}}{\to}0$ for {$\ell=s,q$.} We show that the least favorable critical value in \eqref{def:LF-cv-normal} is bounded away from {$-\infty$} for $\ell=m$ and from {zero} for $\ell=s,q$. The Gaussian comparisons below are on original samples with positive marginal variances. On samples with a zero marginal sample variance, we set the critical value to \(+\infty\) for the proof, so the stated lower bounds also hold on those samples.

For $\ell=m$, the lower bound for $\hat c^m_{n_l}(1-\alpha)$ follows from
\begin{align*}
\hat c^m_{n_l}(1-\alpha)
\overset{(a)}{\ge}
(J^m)^{-1}(1-\alpha,0_k,\hat\Omega_{n_l})
\overset{(b)}{\ge}
\Phi^{-1}\left(\tfrac{1-\alpha}{2}\right).
\end{align*}
Here, (a) holds by the definition of {$\hat c^m_{n_l}(1-\alpha)$}. For (b), let $Y\sim N(0_k,\hat\Omega_{n_l})$. Since $T^m(Y)\ge\min\{Y_1,-Y_1\}=-|Y_1|$,
\begin{align}\label{eq:CDF-m-upperbound}
    J^m(x,0_k,\hat\Omega_{n_l})
=P\{T^m(Y)\le x\}
\le P\{-|Y_1|\le x\}
=2\Phi(x)I\{x<0\}+I\{x\ge0\}
\end{align}
for all $x\in\mathbf R$.

For {$\ell=s$,} let $Y=(Y_2,\ldots,Y_k)'\sim N(0_{k-1},(\hat\Omega_{n_l})_{-1,-1})$. Choose $q^*>0$ such that
{\[
(k-1)\frac{\sqrt{q^*}}{\sqrt{2\pi}}<\frac12-\alpha.
\]}
Such a $q^*$ exists because $\alpha<1/2$. Then, for each fixed realized value of $\hat\Omega_{n_l}$,
\begin{align*}
\hat c^s_{n_l}(1-\alpha)
&\overset{(a)}{\ge}
\liminf_{M\to\infty}
(J^s)^{-1}(1-\alpha,(-M,0_{k-1})',\hat\Omega_{n_l})\\
&\overset{(b)}{\ge}
\inf\left\{x\in\mathbf R:
P\left\{\sum_{j=2}^kY_j^2I\{Y_j\ge0\}\le x\right\}
\ge1-\alpha\right\}
\overset{(c)}{\ge}q^*.
\end{align*}
Here, (a) holds by definition of the least favorable critical value. For (b), if $Z\sim N(0_k,I_k)$, then
\[
T^s(\hat\Omega_{n_l}^{1/2}Z+(-M,0_{k-1})')
\Rightarrow\sum_{j=2}^kY_j^2I\{Y_j\ge0\}
\]
as $M\to\infty$ for fixed $\hat\Omega_{n_l}$. Lower semicontinuity of generalized quantiles under weak convergence gives (b). For (c), note that, for any $t>0$,
\[
\begin{aligned}
P\left\{\sum_{j=2}^kY_j^2I\{Y_j\ge0\}\le t\right\}
&\le P\{Y\in\mathbf R_-^{k-1}\}
+\sum_{j=2}^kP\{0<Y_j\le\sqrt t\}
\le\frac12+(k-1)\frac{\sqrt t}{\sqrt{2\pi}}.
\end{aligned}
\]
The second inequality uses $P\{Y\in\mathbf R_-^{k-1}\}\le P\{Y_2\le0\}=1/2$ and $\phi(u)\le\phi(0)=1/\sqrt{2\pi}$. The choice of $q^*$ gives (c).

For $\ell=q$, let $E_{n_l}$ be the event that {all original marginal sample variances are positive and the unregularized sample correlation matrix $\Omega(\hat P_{n_l})$} is positive definite. On $E_{n_l}$, let $Y=(Y_2,\ldots,Y_k)'\sim N(0_{k-1},(\hat\Omega_{n_l})_{-1,-1})$. Then, for each such fixed $\hat\Omega_{n_l}$,
\begin{align*}
\hat c^q_{n_l}(1-\alpha)
&\ge\liminf_{M\to\infty}
(J^q)^{-1}(1-\alpha,(-M,0_{k-1})',\hat\Omega_{n_l})\\
&\overset{(a)}{\ge}
\inf\left\{x\in\mathbf R:
P\left\{\inf_{\nu\in\mathbf R_-^{k-1}}
(Y-\nu)'((\hat\Omega_{n_l})_{-1,-1})^{-1}(Y-\nu)\le x\right\}\ge1-\alpha\right\}\\
&\overset{(b)}{\ge}
\inf\left\{x\in\mathbf R:
P\left\{\frac1{k-1}\sum_{j=2}^kY_j^2I\{Y_j\ge0\}\le x\right\}
\ge1-\alpha\right\}
\overset{(c)}{\ge}q^\dagger>0.
\end{align*}
Here, (a) follows from Lemma~\ref{lem:orthant-distance-reduction}, with signs reversed, and the same quantile argument as above. For (b),
\[
\textstyle\inf_{\nu\in\mathbf R_-^{k-1}}(Y-\nu)'A^{-1}(Y-\nu)
\ge\textstyle\inf_{\nu\in\mathbf R_-^{k-1}}\frac1{k-1}\|Y-\nu\|^2
=\frac1{k-1}\textstyle\sum_{j=2}^kY_j^2I\{Y_j\ge0\},
\]
because
\[
A^{-1}\succeq\frac1{\lambda_{\max}(A)}I_{k-1}
\succeq\frac1{\operatorname{tr}(A)}I_{k-1}
=\frac1{k-1}I_{k-1}.
\]
Finally, (c) follows as above by choosing $q^\dagger>0$ small enough that
\[
\frac12+(k-1)\frac{\sqrt{(k-1)q^\dagger}}{\sqrt{2\pi}}<1-\alpha.
\]

Therefore, $P_{n_l}\{T^\ell_{n_l}>\hat c^\ell_{n_l}(1-\alpha)\}\to0$ for $\ell=m,s$. For $\ell=q$,
\[
\begin{aligned}
P_{n_l}\{T^q_{n_l}>\hat c^q_{n_l}(1-\alpha)\}
&\le P_{n_l}(E_{n_l}^c)
+P_{n_l}\{T^q_{n_l}>\hat c^q_{n_l}(1-\alpha),E_{n_l}\}\\
&\le P_{n_l}(E_{n_l}^c)+P_{n_l}\{T^q_{n_l}>q^\dagger\}
\to0,
\end{aligned}
\]
where $P_{n_l}(E_{n_l}^c)\to0$ follows from Lemmas~\ref{lem:centered-sampling-bootstrap} and~\ref{lem:sample-corr-positive-definite}. This contradicts \eqref{contra-hypo:LFC-uniform-validity}.

\medskip\noindent
\textbf{Case 2:} Suppose there exists a nonempty set {$I\subseteq[k]$} such that
{\[
\frac{\sqrt{n_l}\mu_j(P_{n_l})}{\sigma_j(P_{n_l})}
\to\delta_j\in(-\infty,0]\quad(j\in I),
\qquad
\frac{\sqrt{n_l}\mu_j(P_{n_l})}{\sigma_j(P_{n_l})}
\to-\infty\quad(j\notin I).
\]}
For $\ell=m,s$, Lemma~\ref{lem:conv-ip-to-Jn} implies that there exists $N=N(\eta)>0$ such that, for all $n_l\ge N$,
\[
P_{n_l}\left\{
\sup_{x\in\mathbf R}
\left|J^\ell(x,\sqrt{n_l}S_{n_l}^{-1}\mu(P_{n_l}),\hat\Omega_{n_l})
-J_{n_l}^\ell(x)\right|\le\frac\eta4
\right\}\ge1-\frac\eta4,
\]
where $J_{n_l}^\ell(x)=P_{n_l}\{T_{n_l}^\ell\le x\}$. Then, for all $n_l\ge N$,
\[
\begin{aligned}
P_{n_l}\{T^\ell_{n_l}>\hat c^\ell_{n_l}(1-\alpha)\}
&\le P_{n_l}\left\{T^\ell_{n_l}>
(J^\ell)^{-1}(1-\alpha,\sqrt{n_l}S_{n_l}^{-1}\mu(P_{n_l}),\hat\Omega_{n_l})\right\}
\le\alpha+\frac\eta2.
\end{aligned}
\]
The first inequality holds by definition of {the least favorable critical value}, since the true mean belongs to the null orthant. The second follows from Lemma~\ref{alem:random-quantile-transfer}, with $\gamma=\alpha$ and $\varepsilon=\eta/4$. This contradicts \eqref{contra-hypo:LFC-uniform-validity}.

For $\ell=q$, recall $E_{n_l}$ is defined as in Case~1, and extend the reference CDF by the point mass at zero on $E_{n_l}^c$ so that $(J^q)^{-1}(1-\alpha,\sqrt{n_l}S_{n_l}^{-1}\mu(P_{n_l}),\hat\Omega_{n_l})$ is defined everywhere. By definition, $\hat c^q_{n_l}(1-\alpha)\ge (J^q)^{-1}(1-\alpha,\sqrt{n_l}S_{n_l}^{-1}\mu(P_{n_l}),\hat\Omega_{n_l})$ on $E_{n_l}$. Therefore,
\[
\begin{aligned}
P_{n_l}\{T^q_{n_l}>\hat c^q_{n_l}(1-\alpha)\}
&\le P_{n_l}(E_{n_l}^c)+P_{n_l}\{T^q_{n_l}>Q_{n_l},E_{n_l}\}
\le P_{n_l}(E_{n_l}^c)+\alpha+\frac\eta2
\end{aligned}
\]
for all sufficiently large $n_l$. The second inequality follows from the $q$ conclusion of Lemma~\ref{lem:conv-ip-to-Jn} and Lemma~\ref{alem:random-quantile-transfer}, as above. Since $P_{n_l}(E_{n_l}^c)\le\eta/4$ for all sufficiently large $n_l$,
\[
P_{n_l}\{T^q_{n_l}>\hat c^q_{n_l}(1-\alpha)\}
\le\alpha+\frac{3\eta}4<\alpha+\eta,
\]
contradicting \eqref{contra-hypo:LFC-uniform-validity}.

\medskip\noindent
\textit{Nonparametric bootstrap critical value.}
Suppose {now that \eqref{contra-hypo:LFC-uniform-validity} holds with $\hat c_{n_l}^{\ell *}(1-\alpha)$ in place} of $\hat c_{n_l}^\ell(1-\alpha)$. Make the same subsequence selections and consider the same two cases.

\medskip\noindent
\textbf{Case 1:} the sampling limits are unchanged. All conditional comparisons in this case are on original samples with positive marginal variances; thecritical values equal $+\infty$ otherwise. For $\ell=m$, choose $b_m<\Phi^{-1}((1-\alpha)/2)$. Lemma~\ref{lem:centered-sampling-bootstrap} and continuity of $T^m$ give
\[
J_{n_l}^m(b_m,0_k,\hat P_{n_l})
\overset{P_{n_l}}{\longrightarrow}
J^m(b_m,0_k,\Omega^*)\le2\Phi(b_m)<1-\alpha.
\]
where the inequality follows from \eqref{eq:CDF-m-upperbound}. The limiting CDF is continuous by Lemma~\ref{lem:strictly_increasing_dist_T}. Since $\hat c_{n_l}^{m *}(1-\alpha)\ge(J_{n_l}^m)^{-1}(1-\alpha,0_k,\hat P_{n_l})$, it follows that
$P_{n_l}\{\hat c_{n_l}^{m *}(1-\alpha)> b_m\}\to1$.

For $\ell=s,q$, write
\[
Z_{n_l}^*:=\sqrt{n_l}S_{n_l}^{*-1}(\bar W_{n_l}^*-\bar W_{n_l}),
\qquad A_{n_l}^*:=(\hat\Omega_{n_l}^*)_{-1,-1},
\]
on resamples with positive marginal variances, {using} the {stated regularization for $q$. Define
\[
R_{n_l}^{s *}:=\sum_{j=2}^k(Z_{j,n_l}^*)^2I\{Z_{j,n_l}^*\ge0\},
\qquad
R_{n_l}^{q *}:=\inf_{\nu\in\mathbf R_-^{k-1}}
(Z_{-1,n_l}^*-\nu)'(A_{n_l}^*)^{-1}(Z_{-1,n_l}^*-\nu),
\]
and set $R_{n_l}^{\ell *}=+\infty$ on resamples with a} zero {marginal variance}. Let {$K_{n_l}^\ell(x):=\hat P_{n_l}^*\{R_{n_l}^{\ell *}\le x\}$}.
{For each fixed original} sample {with positive marginal} variances, the same {argument} as {in} the {parametric case gives
\[
\hat c_{n_l}^{\ell *}(1-\alpha)
\ge\liminf_{M\to\infty}
(J_{n_l}^\ell)^{-1}(1-\alpha,(-M,0_{k-1})',\hat P_{n_l})
\ge(K_{n_l}^\ell)^{-1}(1-\alpha),
\qquad\ell=s,q.
\]
Indeed}, {on each resample with positive marginal variances,} the first {coordinate} of {the shift $S_{n_l}^{*-1}(-M,0_{k-1})'$ tends} to {$-\infty$}, {and Lemma~\ref{lem:orthant-distance-reduction} gives} the {limiting statistic $R_{n_l}^{\ell *}$}. {Exceptional resamples retain} the {value $+\infty$ throughout}, {so} the {same generalized-quantile argument applies}.

{Let $A^*:=(\Omega^*)_{-1,-1}$ and $Y=(Y_2,\ldots,Y_k)'\sim N(0_{k-1},A^*)$. Lemma~\ref{lem:centered-sampling-bootstrap}}, the {continuous mapping theorem}, {and} Lemma{~\ref{lem:strictly_increasing_dist_T} give
\[
\begin{aligned}
K_{n_l}^s(q^*)&\overset{P_{n_l}}{\longrightarrow}
P\left\{\sum_{j=2}^kY_j^2I\{Y_j\ge0\}\le q^*\right\}<1-\alpha,\\
K_{n_l}^q(q^\dagger)&\overset{P_{n_l}}{\longrightarrow}
P\left\{\inf_{\nu\in\mathbf R_-^{k-1}}(Y-\nu)'(A^*)^{-1}(Y-\nu)\le q^\dagger\right\}<1-\alpha.
\end{aligned}
\]}
The {strict inequalities follow from} the {bounds used} to {choose $q^*$ and $q^\dagger$ above. The conditional probability} of {exceptional resamples tends to zero by Lemma~\ref{lem:centered-sampling-bootstrap}}. {Consequently},
{\[
P_{n_l}\{\hat c_{n_l}^{s *}(1-\alpha)\le q^*\}\to0,
\qquad
P_{n_l}\{\hat c_{n_l}^{q *}(1-\alpha)\le q^\dagger\}\to0.
\]
Combining these bounds} with {$T_{n_l}^m\to-\infty$ and $T_{n_l}^s,T_{n_l}^q\to0$ in} probability {gives
\[
P_{n_l}\{T_{n_l}^\ell>\hat c_{n_l}^{\ell *}(1-\alpha)\}\to0,
\]
contradicting} the assumed inequality in Case~1.

\medskip\noindent
\textbf{Case 2:} Lemma~\ref{lem:bootstrap-conv-to-Jn} gives, {for all sufficiently large $n_l$,
\[
P_{n_l}\left\{
\sup_{x\in\mathbf R}
\left|J_{n_l}^\ell(x,\sqrt{n_l}\mu(P_{n_l}),\hat P_{n_l})-J_{n_l}^\ell(x)\right|
\le\frac\eta4\right\}\ge1-\frac\eta4.
\]
For $\ell=m,s$}, the {definition} of the {nonparametric LF} critical value {and Lemma~\ref{alem:random-quantile-transfer} therefore imply
\[
\begin{aligned}
P_{n_l}\{T_{n_l}^\ell>\hat c_{n_l}^{\ell *}(1-\alpha)\}
&\le P_{n_l}\left\{T_{n_l}^\ell>
(J_{n_l}^\ell)^{-1}(1-\alpha,\sqrt{n_l}\mu(P_{n_l}),\hat P_{n_l})\right\}
\le\alpha+\frac\eta2.
\end{aligned}
\]
For $\ell=q$}, {extend this reference CDF by the point mass at zero on $E_{n_l}^c$. The same argument on $E_{n_l}$} gives
\[
P_{n_l}\{T_{n_l}^q>\hat c_{n_l}^{q *}(1-\alpha)\}
\le P_{n_l}(E_{n_l}^c)+\alpha+\frac\eta2
\le\alpha+\frac{3\eta}4
\]
for {all sufficiently large $n_l$. These bounds contradict} the {assumed inequality in Case~2. Thus \eqref{def:LFtest-uniform-validity} holds for both critical values}.

\subsection{Proof of Theorem~\ref{thm:two-step-uniform-validity}}
\label{subsec:two-step-validity-proof}

Fix $\alpha\in(0,1/2)$, $\tau\in(0,\alpha)$, and
$\ell\in\{m,s,q\}$. Suppose, by way of contradiction, that at least
one of the bounds in
\eqref{eq:two-step-normal-uniform-validity} and
\eqref{eq:two-step-bootstrap-uniform-validity} fails.
Then there exist $\eta\in(0,1-\alpha)$, a subsequence $n_l$,
and a sequence $P_{n_l}\in\mathbf P_0$ such that
\begin{align}\label{contra-hypo:two-step-uniform-validity}
    \max\left\{
        \mathbb E_{P_{n_l}}[\psi_{n_l}^{\ell}],
        \mathbb E_{P_{n_l}}[\psi_{n_l}^{\ell *}]
    \right\}
    >\alpha+\eta
\end{align}
for all $l$. Since $P_{n_l}\in\mathbf P_0$, after passing to a further subsequence, still denoted by $n_l$, we may assume that either $\mu(P_{n_l})\in\mathbf R_+^k$ for all $l$ or
$\mu(P_{n_l})\in\mathbf R_-^k$ for all $l$. We prove the result for the former. The negative-orthant case follows by reversing signs throughout. By compactness of $\bar{\mathbf O}$, after passing to a further subsequence, $\Omega(P_{n_l})\to\Omega^*$ for some $\Omega^*\in\bar{\mathbf O}$. Assumption~\ref{ass:uniform-eigenvalue} gives $\Omega^*\succ0$.
By extracting a further subsequence if necessary, each coordinate
of $\sqrt{n_l}\mu_j(P_{n_l})/\sigma_j(P_{n_l})$
has an extended limit in $[0,+\infty]$. Moreover, $    P_{n_l}\{\hat\Sigma_{n_l}\succ0\}\to1$ by marginal variance consistency in
Lemma~\ref{lem:centered-sampling-bootstrap} and Lemma~\ref{lem:sample-corr-positive-definite}. We use the conventions stated before the theorem on singular original samples. There are two cases.

\medskip\noindent
\textbf{Case 1: All standardized local means diverge.}
Suppose
$\sqrt{n_l}\mu_j(P_{n_l})/\sigma_j(P_{n_l})\to+\infty$
for all $j\in[k]$.
By \eqref{eq:statistic-degenerate-case-eq2},
\begin{align}\label{eq:two-step-all-diverging-observed-means}
    \min_{j\in[k]}
    \frac{\sqrt{n_l}\,\bar W_{j,n_l}}{S_{j,n_l}}
    \overset{P_{n_l}}{\longrightarrow}+\infty.
\end{align}
Let $r:=\chi^2_{k,1-\tau}>0$. On an original sample with $\hat\Sigma_{n_l}\succ0$, we have
\begin{align}\label{eq:confidence-set-inclusion}
    \mathbf C_{n_l}(1-\tau)=\big\{
        \mu:
        (\bar W_{n_l}-\mu)^\top
        \hat\Sigma_{n_l}^{-1}
        (\bar W_{n_l}-\mu)\le \frac{r}{n_l}
    \big\}
    \subset
    \mathop{\times}\limits_{j=1}^{k}\left[\bar W_{j,n_l}- S_{j,n_l}\sqrt{\frac{r}{n_l}}, \bar W_{j,n_l} + S_{j,n_l}\sqrt{\frac{r}{n_l}}\right]
\end{align}
because $\mu_j-\bar W_{j,n_l}= (\hat \Sigma^{1/2}_{n_l}e_j)^{\top} \hat \Sigma^{-1/2}(\mu_j-\bar W_{j,n_l})$ and applying Cauchy-Schwarz inequality gives
\begin{align*}
|\mu_j-\bar W_{j,n_l}|^2 
&\leq \|\hat \Sigma^{1/2}_{n_l}e_j\|^2_2 \cdot \|\hat \Sigma^{-1/2}_{n_l}(\mu_j-\bar W_{j,n_l})\|^2_2 
= e_j^{\top}\hat\Sigma_{n_l}e_j \cdot (\mu_j-\bar W_{j,n_l})^{\top} \hat\Sigma_{n_l}^{-1}(\mu_j-\bar W_{j,n_l}) \leq S^2_{j,n_l} \frac{r}{n_l}.
\end{align*}
Equations~\eqref{eq:two-step-all-diverging-observed-means}
and \eqref{eq:confidence-set-inclusion}, together with
$P_{n_l}\{\hat\Sigma_{n_l}\succ0\}\to1$, therefore imply
\[
    P_{n_l}\{
        \mathbf C_{n_l}(1-\tau)\subseteq\mathbf R_+^k
    \}\to1.
\]
For the nonparametric bootstrap confidence region, the same argument applies because $ r_{n_l,\tau}^*
    \overset{P_{n_l}}{\longrightarrow}\chi^2_{k,1-\tau}$ by \eqref{eq:two-step-bootstrap-radius-limit}. Hence
\[
    P_{n_l}\{
        \mathbf C_{n_l}^*(1-\tau)\subseteq\mathbf R_+^k
    \}\to1.
\]
The non-rejection branch in
\eqref{def:two-step-normal-test} and
\eqref{def:two-step-bootstrap-test} now gives
\[
    \mathbb E_{P_{n_l}}[\psi_{n_l}^{\ell}]\to0,
    \qquad
    \mathbb E_{P_{n_l}}[\psi_{n_l}^{\ell *}]\to0.
\]
This contradicts \eqref{contra-hypo:two-step-uniform-validity}.

\medskip\noindent
\textbf{Case 2: At least one standardized local mean has a finite limit.}
Let $I\subseteq[k]$ be the nonempty set of coordinates with
finite limits. Then, \eqref{eq:conv-ip-to-Jn-condition2} holds
along the selected subsequence.

On original samples with $\hat\Sigma_{n_l}\succ0$, define the
infeasible true-mean reference quantiles
\begin{align}\label{eq:two-step-true-mean-quantiles}
    q_{n_l}^{\ell}
    &:=(J^\ell)^{-1}
    \left(
        1-\alpha+\tau,\,
        \sqrt{n_l}\,S_{n_l}^{-1}\mu(P_{n_l}),\,
        \hat\Omega_{n_l}
    \right),\notag\\
    q_{n_l}^{\ell *}
    &:=(J_{n_l}^\ell)^{-1}
    \left(
        1-\alpha+\tau,\,
        \sqrt{n_l}\,\mu(P_{n_l}),\,
        \hat P_{n_l}
    \right).
\end{align}
For this argument, extend both reference CDFs by the point mass
at $+\infty$ on singular original samples, so both quantiles
equal $+\infty$ there. The tests are zero on these samples.

We first claim that
\begin{align}\label{eq:two-step-true-mean-tail-bounds}
    \limsup_{l\to\infty}
    P_{n_l}\{T_{n_l}^{\ell}>q_{n_l}^{\ell}\}
    \le\alpha-\tau, \qquad
    \limsup_{l\to\infty}
    P_{n_l}\{T_{n_l}^{\ell}>q_{n_l}^{\ell *}\}
    \le\alpha-\tau.
\end{align}
To verify these bounds, write $J_{n_l}^\ell(x):=P_{n_l}\{T_{n_l}^\ell\le x\}$. Lemmas~\ref{lem:conv-ip-to-Jn} and
\ref{lem:bootstrap-conv-to-Jn} give, respectively,
\begin{align*}
    \Delta_{n_l}^\ell
    &:=
    \sup_{x\in\mathbf R}
    \left|
        J^\ell\left(
            x,\sqrt{n_l}S_{n_l}^{-1}\mu(P_{n_l}),
            \hat\Omega_{n_l}
        \right)
        -J_{n_l}^\ell(x)
    \right|
    \overset{P_{n_l}}{\longrightarrow}0,\\
    \Delta_{n_l}^{\ell *}
    &:=
    \sup_{x\in\mathbf R}
    \left|
        J_{n_l}^\ell\left(
            x,\sqrt{n_l}\mu(P_{n_l}),\hat P_{n_l}
        \right)
        -J_{n_l}^\ell(x)
    \right|
    \overset{P_{n_l}}{\longrightarrow}0.
\end{align*}
Here, on singular original samples, both reference CDFs are understood to equal zero at every finite $x$, consistently with
the convention $q_{n_l}^\ell=q_{n_l}^{\ell *}=+\infty$. Since these samples have probability tending to zero, this extension does not affect either convergence in probability.

Fix $0<\varepsilon<1-\alpha+\tau$.
By the definitions in \eqref{eq:two-step-true-mean-quantiles},
$q_{n_l}^\ell$ and $q_{n_l}^{\ell *}$ are the respective
reference quantiles at level $1-\alpha+\tau$.
Applying Lemma~\ref{alem:random-quantile-transfer} to each
reference CDF, with $\gamma=\alpha-\tau$, therefore gives
\begin{align*}
    P_{n_l}\{T_{n_l}^\ell>q_{n_l}^\ell\}
    &\le
    \alpha-\tau+\varepsilon
    +P_{n_l}\{\Delta_{n_l}^\ell>\varepsilon\},\\
    P_{n_l}\{T_{n_l}^\ell>q_{n_l}^{\ell *}\}
    &\le
    \alpha-\tau+\varepsilon
    +P_{n_l}\{\Delta_{n_l}^{\ell *}>\varepsilon\}.
\end{align*}
The last probability in each inequality tends to zero.
Taking the limit superior as $l\to\infty$ thus bounds each
upper-tail probability by $\alpha-\tau+\varepsilon$.
Since $\varepsilon>0$ can be arbitrarily small, letting
$\varepsilon\downarrow0$ proves
\eqref{eq:two-step-true-mean-tail-bounds}.
The argument also applies when the limiting distribution has
an atom at zero: the CDF approximations are uniform, including
at zero, and the quantile-transfer bound does not require
continuity.

Now define the first-step coverage events, in mean units, by
\begin{align}\label{eq:two-step-coverage-events}
    E_{n_l}
    :=\{\mu(P_{n_l})\in\mathbf C_{n_l}(1-\tau)\},
    \qquad
    E_{n_l}^*
    :=\{\mu(P_{n_l})\in\mathbf C_{n_l}^*(1-\tau)\}.
\end{align}
Under the null, $E_{n_l}$ rules out an empty intersection of
$\mathbf C_{n_l}(1-\tau)$ with
$\mathbf R_+^k\cup\mathbf R_-^k$.
If the confidence region lies entirely in the null, the test
is zero. Otherwise, on a positive definite original sample,
the supremum in \eqref{def:two-step-normal-cv} includes
$\mu(P_{n_l})$, so
$\tilde c_{n_l}^{\ell}\ge q_{n_l}^{\ell}$.
Since the test is also zero on singular original samples,
\begin{align}\label{eq:two-step-normal-rejection-inclusion}
    \{\psi_{n_l}^{\ell}=1\}
    \subseteq
    E_{n_l}^c
    \cup
    \{T_{n_l}^{\ell}>q_{n_l}^{\ell}\}.
\end{align}
The same reasoning for the bootstrap confidence region yields
\begin{align}\label{eq:two-step-bootstrap-rejection-inclusion}
    \{\psi_{n_l}^{\ell *}=1\}
    \subseteq
    (E_{n_l}^*)^c
    \cup
    \{T_{n_l}^{\ell}>q_{n_l}^{\ell *}\}.
\end{align}

Combining Lemma~\ref{lem:two-step-first-step-coverage} with
\eqref{eq:two-step-true-mean-tail-bounds} and these inclusions gives
\begin{align*}
    \limsup_{l\to\infty}
    \mathbb E_{P_{n_l}}[\psi_{n_l}^{\ell}]
    &\le
    \limsup_{l\to\infty}P_{n_l}(E_{n_l}^c)
    +
    \limsup_{l\to\infty}
    P_{n_l}\{T_{n_l}^{\ell}>q_{n_l}^{\ell}\}
    \le\tau+(\alpha-\tau)=\alpha,\\
    \limsup_{l\to\infty}
    \mathbb E_{P_{n_l}}[\psi_{n_l}^{\ell *}]
    &\le
    \limsup_{l\to\infty}P_{n_l}((E_{n_l}^*)^c)
    +
    \limsup_{l\to\infty}
    P_{n_l}\{T_{n_l}^{\ell}>q_{n_l}^{\ell *}\}
    \le\tau+(\alpha-\tau)=\alpha.
\end{align*}
These bounds contradict
\eqref{contra-hypo:two-step-uniform-validity}.
Therefore, both \eqref{eq:two-step-normal-uniform-validity}
and \eqref{eq:two-step-bootstrap-uniform-validity} hold.

\subsection{Proof of Proposition \ref{prop:LF-max-nonnegative-correlation}}
Below we restate Proposition \ref{prop:LF-max-nonnegative-correlation} using the function $\bar J^{\ell, -1}$ defined in Lemma \ref{lem:argmax-critical-value-all}.

\noindent\textbf{Proposition \ref{prop:LF-max-nonnegative-correlation}.}\emph{
Fix $k\geq2$, $\alpha\in(0,\tfrac12)$, and $\Omega\in\bar{\mathbf O}$. Suppose
$\Omega_{ij}\geq0$ for all $i,j\in[k]$, and let $Z\sim N(0_k,\Omega)$.
Then, for every $c>0$,
\begin{align}
 c^m(1-\alpha,\Omega)
 &=
 \max_{j\in[k]}
 \bar J^{m,-1}(1-\alpha,\infty e_j,\Omega),
 \label{eq:LF-max-nonnegative-cv}
\end{align}
where $\bar J^{\ell, -1}$ is in Lemma \ref{lem:argmax-critical-value-all} and $\infty e_j$ denotes the compactified nuisance vector whose $j$th coordinate equals $+\infty$ and whose remaining coordinates equal zero.}

\begin{proof}
We first suppose that $\Omega \in \mathbf O$. The extension to singular $\Omega \in \bar{\mathbf O}$ is given at the end of the proof. For $c>0$ and $\theta_1\in\mathbf R_+^k$, write
\begin{align}\label{eq:reject_prob_m}
\pi_c(\theta_1)
:=
P\{T^m(Z+\theta_1)>c\},
\qquad Z\sim N(0_k,\Omega).
\end{align}
We first determine the upper envelope of $\pi_c$ by taking the supremum over the nuisance parameter $\theta_1\in\mathbf R_+^k$. We achieve this by showing that
$\pi_c$ is quasiconvex in each coordinate. Since a quasiconvex function on an interval cannot exceed the larger of its values at the two endpoints, maximizing $\pi_c$ over $[0,A]^k$ reduces to maximizing it over the vertices of the box. Letting $A\to\infty$ then reduces the optimization to the compactified vertices with coordinates in $\{0,+\infty\}$. We show that these vertices are dominated by those having exactly one coordinate equal to $+\infty$, which gives the desired tail envelope. The critical-value formula then follows by inverting this envelope.

\noindent \textit{Step 1: quasiconvexity of $\pi_c$ in coordinate $j$.} Fix $j\in[k]$ and hold $\theta_{1,-j}$ fixed. Since
$Z\sim N(0_k,\Omega)$ and $\Omega_{jj}=1$, the conditional distribution of
$Z_{-j}$ given $Z_j=u$ is
\begin{align}\label{eq:conditional-normal}
\{Z_{-j}\mid Z_j=u\} 
\sim
N\left(
\Omega_{-j,j}u,\,
\Omega_{-j,-j}-\Omega_{-j,j}\Omega_{j,-j}
\right)
\end{align}
so the covariance does not depend on $u$. For $u_2\geq u_1$, the conditional distributions satisfy
\[
\{Z_{-j}\mid Z_j=u_2\}
\overset{d}{=}
\{ \{Z_{-j}\mid Z_j=u_1\}
+
\Omega_{-j,j}(u_2-u_1)\}.
\]
Since $\Omega_{ij}\geq0$ for all $i,j$, the vector $\Omega_{-j,j}(u_2-u_1)$ is elementwise nonnegative, which implies that $\{Z_{-j}\mid Z_j=u_2\}$ stochatically dominates $\{Z_{-j}\mid Z_j=u_1\}$. For $u\in\mathbf R$, define
\begin{align*}
a_j(u)
&:=
P\left\{
\max_{h\ne j}(Z_h+\theta_{1,h})>c,\ 
\min_{h\ne j}(Z_h+\theta_{1,h})\geq-c
\ \middle|\ Z_j=u
\right\},\\
b_j(u)
&:=
P\left\{
\min_{h\ne j}(Z_h+\theta_{1,h})<-c,\ 
\max_{h\ne j}(Z_h+\theta_{1,h})\leq c
\ \middle|\ Z_j=u
\right\}.
\end{align*}
It follows that $a_j(u)$ is weakly increasing in $u$ since $a_j(u)$ is stochastically increasing in $Z_{-j}$, while $b_j(u)$ is weakly decreasing in $u$, as $b_j(u)$ is stochastically decreasing in $Z_{-j}$. Moreover, since $\Omega\succ0$, the conditional covariance matrix $\Omega_{-j,-j}-\Omega_{-j,j}\Omega_{j,-j}$ is positive definite. The conditional Gaussian distribution therefore has a continuous density, and the boundaries of the events above have probability zero. Hence $a_j$ and $b_j$ are continuous.

For $t\geq0$, let $\theta_1(t)$ be the vector obtained by setting the $j$th coordinate equal to $t$ and holding $\theta_{1,-j}$ fixed, and define
\[
q_j(t)
:=
\pi_c(\theta_1(t)).
\]
We claim that $q_j$ is quasiconvex in $t$. For $c>0$, the rejection event is given as $$\{T^m(x)>c\}
=
\{
\max_{h\in[k]}x_h>c,\ 
\min_{h\in[k]}x_h<-c
\}.$$ Conditional on $Z_j=u$, if the coordinates
$\{Z_h+\theta_{1,h}:h\ne j\}$ already contain a value above $c$ and a value below $-c$, rejection occurs regardless of $u+t$. If they contain a value above $c$ but no value below $-c$, rejection additionally requires $u+t<-c$. Similarly, if they contain a value below $-c$ but no value above
$c$, rejection additionally requires $u+t>c$. 
 Therefore,
\begin{align}
q_j(t)
=
q_{j,0}
+
\int_{-\infty}^{-c-t}a_j(u)\phi(u)\,du
+
\int_{c-t}^{\infty}b_j(u)\phi(u)\,du,
\label{eq:LF-max-coordinate-section}
\end{align}
where
\[
q_{j,0}
:=
\int_{-\infty}^{\infty}
P\left\{
\max_{h\ne j}(Z_h+\theta_{1,h})>c,\ 
\min_{h\ne j}(Z_h+\theta_{1,h})<-c
\ \middle|\ Z_j=u
\right\}
\phi(u)\,du
\]
does not depend on $t$. Differentiating gives
\begin{align}
q_j'(t)
=
b_j(c-t)\phi(c-t)
-
a_j(-c-t)\phi(c+t)
\quad\Leftrightarrow\quad
\frac{q_j'(t)}{\phi(c-t)}
=
b_j(c-t)
-
e^{-2ct}a_j(-c-t).
\label{eq:LF-max-coordinate-derivative}
\end{align}
While $q'_j$ is not monotone in $t$, $q'_j(t)/\phi(c-t)$ is. As $t$ increases, $b_j(c-t)$ is weakly increasing because $b_j$ is weakly decreasing, and $a_j(-c-t)$ and $e^{-2ct}$ are both weakly decreasing.
Hence
\[
t\mapsto
b_j(c-t)-e^{-2ct}a_j(-c-t)
\]
is weakly increasing. Since $\phi(c-t)>0$, $q_j'(t)$ can change sign at most once, and any sign change is from negative to positive. Therefore $q_j$ is
quasiconvex on $[0,\infty)$ and, for every $A>0$,
\begin{align}
\sup_{0\leq t\leq A}q_j(t)
=
\max\{q_j(0),q_j(A)\}.
\label{eq:LF-max-coordinate-endpoint}
\end{align}

\noindent \textit{Step 2: vertices dominated by $\infty e_j$'s.}  Applying \eqref{eq:LF-max-coordinate-endpoint} successively to the $k$ coordinates, while holding the remaining coordinates fixed at each step, shows that each coordinate can be replaced by either 0 or A without decreasing the maximum. Hence,
\begin{align}
\sup_{\theta_1\in[0,A]^k}\pi_c(\theta_1)
=
\max_{\delta\in\{0,A\}^k} \pi_c(\delta)
=
\max_{D\subseteq[k]}
\pi_c(A\mathbf 1_D),
\label{eq:LF-max-box-vertices}
\end{align}
where $\mathbf 1_D$ denotes the vector whose coordinates in $D$ equal one and whose remaining coordinates equal zero. Since
$\mathbf R_+^k=\bigcup_{A>0}[0,A]^k$, letting $A\to\infty$ gives
\[
\sup_{\theta_1\in\mathbf R_+^k}\pi_c(\theta_1)
=
\max_{D\subseteq[k]}
\lim_{A\to\infty}\pi_c(A\mathbf 1_D),
\]
where the maximum and the limit may be interchanged because there are
finitely many subsets $D$. By dominated convergence,
\[
\lim_{A\to\infty}\pi_c(A\mathbf 1_D)
=
\begin{cases}
\pi_c(0_k),
& D=\varnothing,\\[3pt]
P\left\{\max_{h\notin D}(-Z_h)>c\right\},
& \varnothing\subsetneq D\subsetneq[k],\\[3pt]
0,
& D=[k].
\end{cases}
\]
For every nonempty $D\subseteq[k]$ and every $j\in D$, $[k]\setminus D\subseteq[k]\setminus\{j\},$ so $P\{\max_{h\notin D}(-Z_h)>c\}
\leq
P\{\max_{h\ne j}(-Z_h)>c\}.$ Thus every compactified vertex with more than one coordinate equal to $+\infty$ is weakly dominated by a singleton-infinity vertex.

It remains to show that the origin is also dominated. Fix $j\in[k]$ and set
$\theta_{1,-j}=0_{k-1}$. From \eqref{eq:conditional-normal}, we have $\{Z_{-j}\mid Z_j=-u\}
\overset{d}{=}
\{-Z_{-j}\mid Z_j=u\}.$ Therefore, $a_j(-u)=b_j(u)$ for every $u \in \mathbf R$ because $\theta_{1,-j}=0_{k-1}$. Using this equality in \eqref{eq:LF-max-coordinate-derivative},
\[
\frac{q_j'(t)}{\phi(c-t)}
=
b_j(c-t)-e^{-2ct}b_j(c+t).
\]
Since $b_j$ is weakly decreasing, $c-t\leq c+t$, and $e^{-2ct}\leq1$, $b_j(c-t)
\geq
b_j(c+t)
\geq
e^{-2ct}b_j(c+t).$ Hence $q_j'(t)\geq0$ for all $t\geq0$, and therefore
\[
\pi_c(0_k)
\leq
\lim_{t\to\infty}\pi_c(te_j)
=
P\{
\max_{h\ne j}(-Z_h)>c
\}.
\]
Combining the preceding results gives
\begin{align}
\sup_{\theta_1\in\mathbf R_+^k}
P\{T^m(Z+\theta_1)>c\}
=
\max_{j\in[k]}
P\{\max_{h\ne j}(-Z_h)>c\}.
\label{eq:LF-max-nonnegative-tail-envelope}
\end{align}

\noindent \textit{Step 3: envelope quantile as the critical value.}
Define
\[
c_j
:=
\inf\{
x\in\mathbf R:
P\{
\max_{h\ne j}(-Z_h)\leq x
\}
\geq1-\alpha
\},
\qquad
c^*
:=
\max_{j\in[k]}c_j.
\]
The distribution of $\max_{h\ne j}(-Z_h)$ is continuous. Moreover,
\[
P\{
\max_{h\ne j}(-Z_h)\leq0
\}
=
P\{Z_h\geq0\ \forall h\ne j\}
\leq\tfrac12
<1-\alpha,
\]
so $c_j>0$ for every $j$ and hence $c^*>0$. By
\eqref{eq:LF-max-nonnegative-tail-envelope}, $\sup_{\theta_1\in\mathbf R_+^k}
P\{T^m(Z+\theta_1)>c^*\}
\leq\alpha.$ This implies that for every $\theta_1 \in \mathbf{R}^k_+$, $J^{m,-1} (1-\alpha, \theta_1)\leq c^*$, and hence $c^m(1-\alpha,\Omega):=\sup_{\theta_1\in\mathbf{R}^k_+}J^{m,-1} (1-\alpha, \theta_1)\leq c^*.$ Conversely, let $j^*\in[k]$ satisfy $c_{j^*}=c^*$. By the definition of the
compactified distribution, $\bar J^{m,-1}
(1-\alpha,\infty e_{j^*},\Omega)
=
c_{j^*}
=
c^*.$ Since the least favorable critical value is the maximum of the compactified
quantile, $c^m(1-\alpha,\Omega)
\geq
\bar J^{m,-1}
(1-\alpha,\infty e_{j^*},\Omega)
=
c^*.$ Combining the two inequalities proves \eqref{eq:LF-max-nonnegative-cv}. 

Now let $\Omega\in\bar{\mathbf O}$ be possibly singular and entrywise
nonnegative. For $\varepsilon>0$, define
\[
\Omega_\varepsilon:=\frac{\Omega+\varepsilon I_k}{1+\varepsilon}.
\]
Then $\Omega_\varepsilon\in\mathbf O$, $\Omega_\varepsilon\ge0$ elementwise,
and $\Omega_\varepsilon\to\Omega$. The result above gives
\[
c^m(1-\alpha,\Omega_\varepsilon)
=
\max_{j\in[k]}
\bar J^{m,-1}(1-\alpha,\infty e_j,\Omega_\varepsilon).
\]
By Lemmas \ref{lem:argmax-critical-value-all} and
\ref{lem:continuity-critical-value-all}, both sides are continuous in
$\Omega$. Letting $\varepsilon\downarrow0$ proves the result for
$\Omega\in\bar{\mathbf O}$.
\end{proof}

\subsection{Proof of Proposition \ref{prop:LF-improved-numerical-accuracy}}
\begin{proof}
Let \(Y:=\Omega^{1/2}Z\), where \(Z\sim N(0_k,I_k)\), and let \(\mathcal A_r\) denote the event that all the one-sided quantile inequalities in \eqref{eq:LF-MC-quantile-statements} used in constructing \(\hat c_{L,r}^\ell\) and \(\hat c_{U,r}^\ell\) hold. Since there are \(H_r\) such inequalities and each fails with probability at most \(\eta_r=\delta_r/H_r\), the union bound gives
\begin{align}
P_{\mathrm{MC}}(\mathcal A_r^c)
\leq H_r\eta_r=\delta_r,
\qquad
P_{\mathrm{MC}}(\mathcal A_r)\geq1-\delta_r.
\label{eq:LF-simultaneous-MC-event}
\end{align}

We establish the lower bound. On $\mathcal A_r$ and for every $g\in\mathcal G_r$, we have
\begin{align}
\hat q_{L,r}
(
1-\alpha;
T^\ell(Y+g,\Omega)
)
&\le
q_{1-\alpha}\{T^\ell(Y+g,\Omega)\} 
=
J^{\ell,-1}(1-\alpha,g,\Omega)
\le
c^\ell(1-\alpha,\Omega)
\label{eq:improved-lower-finite}
\end{align}
where the first inequality holds by \eqref{eq:LF-MC-quantile-statements} and the second inequality holds because $g\in\mathbf R_+^k$. The singleton-infinity terms satisfy the similar inequality. On $\mathcal A_r$ and for each $j\in[k]$, we have
\begin{align}
\hat q_{L,r}(1-\alpha;R_j^\ell)
\le
q_{1-\alpha}(R_j^\ell)
=
\lim_{a\to\infty}
J^{\ell,-1}(1-\alpha,ae_j,\Omega) \le
c^\ell(1-\alpha,\Omega)\qquad
\forall j\in[k]
\label{eq:improved-lower-boundary}
\end{align}
where the first inequality follows from \eqref{eq:LF-MC-quantile-statements}, the equality holds by Lemmas~\ref{lem:argmax-critical-value-all} and \ref{lem:LF-boundary-reduction}, and the last inequality holds as $ae_j\in\mathbf R_+^k$ for every finite $a$. Taking the maximum in \eqref{def:improved-lower-cv} and combining
\eqref{eq:improved-lower-finite}--\eqref{eq:improved-lower-boundary} yields
\begin{align}
\hat c_{L,r}^\ell
\le c_{L,r}^{\ell,\mathrm{pop}}
:=
\max\{
\max_{g\in\mathcal G_r}
J^{\ell,-1}(1-\alpha,g,\Omega),
~
\max_{j\in[k]}q_{1-\alpha}(R_j^\ell)\}
\le
c^\ell(1-\alpha,\Omega)
\qquad\text{on }\mathcal A_r.
\label{eq:improved-lower-valid}
\end{align}

We next establish the upper numerical bound. We divide the nuisance space $\mathbf R_+^k$ into the finite box $[0,A_r]^k$ and its complement. We first consider the finite box. Fix a cell $B\in\mathcal B_r$. We claim that 
\begin{align}
\sup_{\theta_1\in B}
J^{\ell,-1}(1-\alpha,\theta_1,\Omega)
\le
J^{\ell,-1}(1-\alpha_B,g_B,\Omega)
=
q_{1-\alpha_B}\{T^\ell(Y+g_B,\Omega)\}.
\label{eq:NP-quantile-cell-bound}
\end{align}
To see this, let $c_B:=J^{\ell,-1}(1-\alpha_B,g_B,\Omega)$ where $\alpha_B$ is in \eqref{def:improved-cell-radius}. By construction, $P\{T^\ell(Y+g_B,\Omega)>c_B\}\leq\alpha_B.$ For every $\theta_1\in B$, the Mahalanobis distance between $\theta_1$ and
$g_B$ is no greater than $\rho_B$. Applying
\eqref{eq:LF-sharp-Gaussian-shift} in Lemma \ref{lem:LF-sharp-Gaussian-shift} to the rejection event gives $P\{T^\ell(Y+\theta_1,\Omega)>c_B\}
\leq
\Phi\{\Phi^{-1}(\alpha_B)+\rho_B\}
=
\alpha.$ Hence
$J^{\ell,-1}(1-\alpha,\theta_1,\Omega)\leq c_B$ for every
$\theta_1\in B$, which proves \eqref{eq:NP-quantile-cell-bound}. Then, the corresponding upper Monte Carlo statement gives
\[
\sup_{\theta_1\in B}
J^{\ell,-1}(1-\alpha,\theta_1,\Omega)\leq q_{1-\alpha_B}\{T^\ell(Y+g_B,\Omega)\}
\le
\hat q_{U,r}
(
1-\alpha_B;
T^\ell(Y+g_B,\Omega)
)\quad \text{on}\quad \mathcal{A}_r.
\]
Each $E\in\mathcal E_{B,r}^\ell$ provides another valid population upper bound: $\sup_{\theta_1\in B}
J^{\ell,-1}(1-\alpha,\theta_1,\Omega)
\le
q_{1-\alpha}(E).$ Equation \ref{eq:LF-MC-quantile-statements} implies, 
$q_{1-\alpha}(E)
\le
\hat q_{U,r}(1-\alpha;E)$ on $\mathcal{A}_r$ for all $E \in \mathcal E_{B,r}^\ell.$ Finally, $\sup_{\theta_1\in B}
J^{\ell,-1}(1-\alpha,\theta_1,\Omega)
\le
c^\ell(1-\alpha,\Omega)
<
C_\ell.$ Thus, every term entering the minimum in \eqref{def:improved-cell-upper} is an upper bound for the cellwise least favorable quantile. Consequently,
\begin{align*}
\sup_{\theta_1\in B}
J^{\ell,-1}(1-\alpha,\theta_1,\Omega)
\le
\hat u_{B,r}^\ell
\qquad
\text{on }\mathcal A_r.
\end{align*}

Next, we consider the complement of $[0,A_r]^k$. Let the largest quantile outside of the box as 
\begin{align*}
c_{A_r,\mathrm{out}}^\ell
:=
\textstyle\sup_{\substack{\theta_1\in\mathbf R_+^k,~
\|\theta_1\|_\infty>A_r}}
J^{\ell,-1}(1-\alpha,\theta_1,\Omega)
\end{align*}
Then we have, on $\mathcal A_r$,
\begin{align*}
c_{A_r,\mathrm{out}}^\ell
&\le
\max_{j\in[k]}
q
(
1-\alpha;
V^\ell(Y+A_re_j,\Omega)
)
\leq
\max_{j\in[k]}
\hat q_{U,r}
(
1-\alpha;
V^\ell(Y+A_re_j,\Omega)
)
\end{align*}
where the first inequality follows from Lemma~\ref{lem:LF-outside-box} and the second inequality holds by \eqref{eq:LF-MC-quantile-statements}. As the same constant $C_\ell$ truncates every
outside-box bound, $c_{A_r,\mathrm{out}}^\ell\le c^\ell(1-\alpha,\Omega)<C_\ell$, we have $c_{A_r,\mathrm{out}}^\ell\leq \max_{j\in[k]}\hat u_{A,j,r}^\ell$ on $\mathcal{A}_r.$ The cells in $\mathcal B_r$ cover $[0,A_r]^k$, while $\max_{j\in[k]}\hat u_{A,j,r}^\ell$ covers its complement in
$\mathbf R_+^k$. Hence, on $\mathcal A_r$,
\begin{align*}
c^\ell(1-\alpha,\Omega)
\le
\max\{
\max_{B\in\mathcal B_r}\hat u_{B,r}^\ell,
~
\max_{j\in[k]}\hat u_{A,j,r}^\ell
\}
=
\hat c_{U,r}^\ell.
\end{align*}
Together with \eqref{eq:improved-lower-valid}, this gives
\[
\hat c_{L,r}^\ell
\le
c^\ell(1-\alpha,\Omega)
\le
\hat c_{U,r}^\ell
\qquad\text{on }\mathcal A_r.
\]
Combining this with \eqref{eq:LF-simultaneous-MC-event} proves
\eqref{eq:improved-cv-certificate}.

It remains to prove convergence of the numerical endpoints. We proceed in two steps. First, we show that the population counterparts of the lower and upper bounds converge to the exact least favorable critical value as the search region expands and the partition is refined. Second, we show that the Monte Carlo order-statistic bounds converge to these population counterparts as $M_r$ increases.

Consider the population counterpart of the lower numerical bound $c_{L,r}^{\ell,\mathrm{pop}}$ in \eqref{eq:improved-lower-valid}. Every term in this maximum corresponds to a feasible finite nuisance value or to a limiting feasible nuisance configuration. By Lemma~\ref{lem:argmax-critical-value-all}, there exists
$\theta_1^*\in[0,\infty]^k$ such that $c^\ell(1-\alpha,\Omega)=\bar J^{\ell,-1}(1-\alpha,\theta_1^*,\Omega).$ We distinguish whether $\theta_1^*$ is finite or lies on the boundary at infinity. 
Suppose first that $\theta_1^*\in[0,\infty)^k$. Since \(A_r\to\infty\),
we have \(\theta_1^*\in[0,A_r]^k\) for all sufficiently large \(r\).
Let \(B_r\in\mathcal B_r\) be a cell containing \(\theta_1^*\). By
construction, \(g_{B_r}\in\mathcal G_r\). Moreover,
\begin{align*}
\|g_{B_r}-\theta_1^*\|
&\leq
\sqrt{\lambda_{\max}(\Omega)}
\left\{
(g_{B_r}-\theta_1^*)'
\Omega^{-1}
(g_{B_r}-\theta_1^*)
\right\}^{1/2}\leq
\sqrt{\lambda_{\max}(\Omega)}\,\rho_{B_r}
\leq
\sqrt{\lambda_{\max}(\Omega)}
\max_{B\in\mathcal B_r}\rho_B
\to0
\end{align*}
and so \(g_{B_r}\to\theta_1^*\). By continuity of
\(\theta_1\mapsto J^{\ell,-1}(1-\alpha,\theta_1,\Omega)\), $J^{\ell,-1}(1-\alpha,g_{B_r},\Omega)
\to
J^{\ell,-1}(1-\alpha,\theta_1^*,\Omega)
=
c^\ell(1-\alpha,\Omega).$ Since $c_{L,r}^{\ell,\mathrm{pop}}
\geq
J^{\ell,-1}(1-\alpha,g_{B_r},\Omega),$ we obtain $\liminf_{r\to\infty}
c_{L,r}^{\ell,\mathrm{pop}}
\geq
c^\ell(1-\alpha,\Omega).$ Together with $c_{L,r}^{\ell,\mathrm{pop}}\leq c^\ell(1-\alpha,\Omega)$, this proves convergence in the finite-maximizer case. 
Suppose instead that $\theta_1^*$ lies on the boundary at infinity. Since $\theta_1^*$ attains the global maximum, Lemma~\ref{lem:LF-boundary-reduction} implies $c^\ell(1-\alpha,\Omega)=\max_{j\in[k]}q_{1-\alpha}(R_j^\ell).$ These $k$ boundary quantiles are included in $c_{L,r}^{\ell,\mathrm{pop}}$ at every numerical step. Hence, in this case, $c_{L,r}^{\ell,\mathrm{pop}} = c^\ell(1-\alpha,\Omega)$ for every $r$. Therefore, in either case,
\begin{align}
c_{L,r}^{\ell,\mathrm{pop}}
\to
c^\ell(1-\alpha,\Omega).
\label{eq:LF-pop-lower-convergence}
\end{align}

For the population upper bound, replace each Monte Carlo upper quantile in
\eqref{def:improved-cell-upper} and \eqref{def:improved-outside-upper} by its exact population quantile. Specifically, let
\begin{align*}
u_{B,r}^{\ell,\mathrm{pop}}
&:=
\min\{
C_\ell,\,
J^{\ell,-1}(1-\alpha_B,g_B,\Omega),\{
q_{1-\alpha}(E):
E\in\mathcal E_{B,r}^\ell
\}
\}\qquad\text{for each cell }B\in\mathcal{B}_r\\
u_{A,j,r}^{\ell,\mathrm{pop}}
&:=
\min\{
C_\ell,\,
q_{1-\alpha}
\{
V^\ell(\Omega^{1/2}Z+A_re_j,\Omega)
\}
\}\qquad\text{for each }j\in[k]\\
c_{U,r}^{\ell,\mathrm{pop}} 
&:=
\max\{
\textstyle\max_{B\in\mathcal B_r}u_{B,r}^{\ell,\mathrm{pop}},
~
\textstyle\max_{j\in[k]}u_{A,j,r}^{\ell,\mathrm{pop}}
\}
\end{align*}
where the last collection is omitted when $\mathcal E_{B,r}^\ell=\varnothing$, and $Z\sim N(0_k, I_k)$. The preceding argument gives 
\begin{align}\label{def:LF-pop-upper}
c^\ell(1-\alpha,\Omega)\leq c_{U,r}^{\ell,\mathrm{pop}} \qquad \text{for each }r.
\end{align}
Consider first the cellwise part. For each $r$, choose a cell $B_r$ attaining
$\max_{B\in\mathcal B_r}u_{B,r}^{\ell,\mathrm{pop}}$. By compactness of $[0,\infty]^k$, after passing to a subsequence if necessary, $g_{B_r}$ converges in the compactified nuisance space. Since the maximal Mahalanobis cell radius converges to zero, $\rho_{B_r}\to0$ and hence $\alpha_{B_r}\to\alpha$. If the limiting nuisance
value is finite, continuity of the nuisance-specific quantile gives $J^{\ell,-1}(1-\alpha_{B_r},g_{B_r},\Omega)
\to
J^{\ell,-1}(1-\alpha,g,\Omega)
\le
c^\ell(1-\alpha,\Omega).$ If the limiting nuisance value lies on the boundary at infinity, the continuous compactification in  Lemma~\ref{lem:argmax-critical-value-all}, together with  $\alpha_{B_r}\to\alpha$, gives the corresponding boundary quantile, which by Lemma~\ref{lem:LF-boundary-reduction} is also no greater
than $c^\ell(1-\alpha,\Omega)$. Thus $\limsup_{r\to\infty}
\max_{B\in\mathcal B_r}
u_{B,r}^{\ell,\mathrm{pop}}
\le
c^\ell(1-\alpha,\Omega).$ For the outside-box part, Lemma~\ref{lem:LF-outside-box} and its limiting statement imply, for every $j\in[k]$, $q_{1-\alpha} \{V^\ell(Y+A_re_j,\Omega)\}
\to
q_{1-\alpha}(R_j^\ell)$ as $A_r \to \infty.$ Therefore, $\limsup_{r\to\infty}
\max_{j\in[k]}u_{A,j,r}^{\ell,\mathrm{pop}}
\le
\max_{j\in[k]}q_{1-\alpha}(R_j^\ell)
\le
c^\ell(1-\alpha,\Omega).$
Combining these with \eqref{def:LF-pop-upper} gives
\begin{align}
c_{U,r}^{\ell,\mathrm{pop}}
\to
c^\ell(1-\alpha,\Omega).
\label{eq:LF-pop-upper-convergence}
\end{align}

It remains to prove \eqref{eq:improved-cv-consistency}. First we claim that, uniformly over the probability levels \(p\) used at step \(r\),
\begin{align}
\left|
\frac{m_{L,r}(p)}{M_r}-p
\right|
+
\left|
\frac{m_{U,r}(p)}{M_r}-p
\right|
=
O\left(
\sqrt{\frac{\log(H_r/\delta_r)}{M_r}}
+\frac1{M_r}
\right)
=o(1).
\label{eq:LF-rank-convergence}
\end{align}
To see this, let $B_r\sim \text{Binomial}(M_r,p).$ Recall $m_{U,r}(p)$ in \eqref{def:LF-MC-quantile-bounds} is the smallest integer $m$ satisfying $P\{B_r \geq m\}\leq \eta_r.$ The Hoeffding's inequality for binomial distribution with $e^{-2t_r^2/M_r}=\eta_r$ implies {$P\{B_r \geq \lceil M_rp+t_r\rceil\} \leq P\{B_r -M_rp \geq t_r\}\leq e^{-2t_r^2/M_r}=\eta_r$}, which gives $m_{U,r}(p) \leq M_rp+t_r+1$, implying $m_{U,r}(p)/M_r -p \leq \sqrt{\log(1/\eta_r)/2M_r} + 1/M_r.$ For the other direction, since $\eta_r \to0$, eventually $\eta_r<1/2$. We know $P\{B_r \geq\lfloor M_rp\rfloor\}\geq 1/2>\eta_r$, so $m_{U,r}(p)\geq \lfloor M_r p\rfloor + 1>M_r p$, implying $m_{U,r}(p)/M_r - p \geq -1/M_r$. Similar arguments apply to $m_{L,r}(p)$

Furthermore, for \(h=1,\ldots,H_r\), let \(X_{h,r}\) denote the random variable associated with the \(h\)th quantile bound used at numerical step \(r\). Let $F_{h,r}$ and $\hat F_{h,r}$ be its CDF and ECDF based on Monte Carlo draws, respectively. Let $\varepsilon_r
:=
\left\{
\frac{\log(2H_r/\delta_r)}{2M_r}
\right\}^{1/2}.$ The Dvoretzky--Kiefer--Wolfowitz inequality and the union bound give
\begin{align}
P_{\mathrm{MC}}\left\{
\max_{1\le h\le H_r}
\sup_x|\hat F_{h,r}(x)-F_{h,r}(x)|
>\varepsilon_r
\right\}
\leq\delta_r.
\label{eq:LF-uniform-DKW}
\end{align}
The assumed rate condition implies \(\varepsilon_r\to0\).

We first consider the lower endpoint. For each \(r\), choose \(X_{L,r}\) from $\left\{
T^\ell(Y+g,\Omega):g\in\mathcal G_r
\right\}
\cup
\left\{
R_j^\ell:j\in[k]
\right\}$ such that $q(1-\alpha;X_{L,r})=c_{L,r}^{\ell,\mathrm{pop}},$ which is possible because the collection is finite. Let \(F_{L,r}\) denote the CDF of \(X_{L,r}\), and write $q_r:=q(1-\alpha;X_{L,r}).$  We show that the lower order-statistic estimator $\hat q_{L,r}(1-\alpha;X_{L,r})$ converges to \(q_r\) in probability. To this end, cnsider an arbitrary subsequence. Associating \(X_{L,r}\) with its
finite nuisance value, or with \(\infty e_j\) when \(X_{L,r}=R_j^\ell\), compactness of \([0,\infty]^k\) yields a further subsequence converging to a compactified nuisance configuration. Along this subsequence, $X_{L,r}$ weakly converges, so \(F_{L,r}\) converges to the corresponding limiting CDF \(F^*\) pointwise at continuity points; the \((1-\alpha)\)-quantile of $F^*$ is \(c^\ell(1-\alpha,\Omega)=: c^*\). Fix \(\varepsilon \in (0,c^*)\), and choose a continuity point \(a\in(c^*-\varepsilon,c^*)\). Since $q_r\leq c^*$, $ q_r -\varepsilon< c^*-\varepsilon <a$ and so 
$$ F_{L,r}(q_r-\varepsilon) \leq F_{L,r}(a).$$
Moreover, since \(F^*\) is strictly increasing at \(c^*\), $F^*(a)<1-\alpha.$ Set $\kappa_\varepsilon := \tfrac{1-\alpha-F^*(a)}{2}>0.$ Since \(F_{L,r}(a)\to F^*(a)\) as $r\to\infty$, for all sufficiently large \(r\) along this subsequence, 
$$F_{L,r}(a)\leq 1-\alpha-\kappa_\varepsilon.$$
Recall $m_{L,r}(1-\alpha)/M_r \to 1-\alpha$ by \eqref{eq:LF-rank-convergence} and $\sup_{x}|\hat{F}_{L,r}(x) - F_{L,r}(x)|\overset{P_{MC}}{\to} 0$ by \eqref{eq:LF-uniform-DKW}. This implies that, $\hat F_{L,r}(q_r-\varepsilon) <{m_{L,r}(1-\alpha)}/{M_r}$ with probability approaching one. Hence fewer than \(m_{L,r}(1-\alpha)\) validation draws are less than or equal to \(q_r-\varepsilon\), and therefore,
\[
\hat q_{L,r}(1-\alpha;X_{L,r})
>
q_r-\varepsilon\quad \text{ with probability approaching one.}
\]
Since $P_{\mathrm{MC}}\{\hat q_{L,r}(1-\alpha;X_{L,r})
\leq
\hat c_{L,r}^\ell
\leq
q_r \leq c^\ell(1-\alpha,\Omega)\} \geq 1-\delta_r\to1$ and $q_r \to c^\ell(1-\alpha,\Omega)$ by \eqref{eq:LF-pop-lower-convergence}, the arbitrary subsequence argument implies $\hat c_{L,r}^\ell
-
c_{L,r}^{\ell,\mathrm{pop}}
\overset{P_{\mathrm{MC}}}{\longrightarrow}0.$ Together with \eqref{eq:LF-pop-lower-convergence},
\begin{align}
\hat c_{L,r}^\ell
\overset{P_{\mathrm{MC}}}{\longrightarrow}
c^\ell(1-\alpha,\Omega).
\label{eq:LF-numerical-lower-convergence}
\end{align}

We next consider the upper endpoint. Fix \(\varepsilon>0\). We claim that there exists \(\kappa_\varepsilon>0\) such that, for all sufficiently large \(r\),
\begin{align}
F_{B,r}(x)
:=P\{T^\ell(Y+g_B,\Omega)\leq c^* +\varepsilon\}
\geq
1-\alpha_B+\kappa_\varepsilon
\qquad
\forall B\in\mathcal B_r.
\label{eq:LF-upper-cell-separation}
\end{align}
Otherwise, there exist \(r_s\to\infty\) and \(B_s\in\mathcal B_{r_s}\) such that $F_{B_s,r_s}(c^*+\varepsilon)
<
1-\alpha_{B_s}+\frac1s.$ By compactness, after passing to a subsequence, \(g_{B_s}\to g^*\in[0,\infty]^k\). Moreover, \(\rho_{B_s}\to0\) by assumption, so \(\alpha_{B_s}\to\alpha\). By {Lemma~\ref{lem:argmax-critical-value-all}}, the {statistics} converge to the {extended Gaussian statistic} associated with $g^*${. Write $F^*$ for} its CDF {at finite arguments; when $g^*=(\infty,\ldots,\infty)$ and $\ell=m$, this means $F^*(x)=1$ for every finite $x$}. Since {$c^*+\varepsilon>0$}, {convergence} at {this value} gives
{\[
F_{B_s,r_s}(c^*+\varepsilon)\to F^*(c^*+\varepsilon)\leq1-\alpha.
\]
If every coordinate of $g^*$ is infinite,} the {limiting statistic equals $-\infty$ for $\ell=m$ and zero for $\ell=s,q$}, {so $F^*(c^*+\varepsilon)=1$, a contradiction. Otherwise, its $(1-\alpha)$-quantile is no greater than $c^*$ by the definition} of {the LF critical value. Lemma~\ref{lem:strictly_increasing_dist_T}, including its one-sided version on the finite coordinates, gives $F^*(c^*+\varepsilon)>1-\alpha$ unless this CDF value is already one. Either conclusion is} a contradiction.

By \eqref{eq:LF-rank-convergence}, \eqref{eq:LF-uniform-DKW}, and \eqref{eq:LF-upper-cell-separation}, with probability approaching one,
\[
\hat F_{B,r}(c^*+\varepsilon)
\geq
\frac{m_{U,r}(1-\alpha_B)}{M_r}
\qquad
\forall B\in\mathcal B_r,
\]
where \(\hat F_{B,r}\) denotes the empirical CDF of
\(T^\ell(Y+g_B,\Omega)\). Hence, with probability approaching one,
\[
\max_{B\in\mathcal B_r}\hat u_{B,r}^\ell
\leq
c^*+\varepsilon.
\]

For the outside-box terms, let $F_{A_r, j}$ be the CDF of $V^\ell(Y+A_re_j,\Omega)$ be and $F_{j}$ be the CDF of $R_j^\ell$. The proof of Lemma~\ref{lem:LF-outside-box} implies that $V^\ell(Y+A_re_j,\Omega)$ weakly converges to $R_j^\ell$ for all $j\in[k]$. Fix $\varepsilon>0$. Since $q(1-\alpha;R_j^\ell) \leq c^*$ and $F_j$ is continuous and strictly increasing at positive values, $F_j(c^*+\varepsilon) > 1-\alpha. $ Choose $\kappa_\varepsilon=\min_{j\in[k]} (F_j(c^*+\varepsilon) - (1-\alpha))/2$, and hence $F_j(c^*+\varepsilon) > 1-\alpha + \kappa_\epsilon$ for all $j.$ Since $F_{A_r, j}(c^*+\varepsilon)\to F_j(c^*+\varepsilon)$, for sufficiently large $r$, $F_{A_r,j}(c^*+\varepsilon) > 1-\alpha + \kappa_\epsilon$ for all $j$. By \eqref{eq:LF-rank-convergence} and \eqref{eq:LF-uniform-DKW}, for larger values of $r$, 
$$\hat{F}_{A_r, j} (c^*+ \varepsilon)>m_{U,r}(1-\alpha)/{M_r}\quad \forall j\in[k]$$
with probability approaching to 1. Therefore, $\hat{q}_{U,r}(1-\alpha;V^\ell(Y+A_re_j,\Omega)) \leq c^*+\varepsilon$ for all $j \in [k]$; accordingly,
\[
\max_{j\in[k]}\hat u_{A,j,r}^\ell
\leq
c^*+\varepsilon
\]
with probability approaching one. Therefore, $P_{\mathrm{MC}}\{
\hat c_{U,r}^\ell>c^*+\varepsilon
\}
\to0.$ On the other hand,
\(\hat c_{U,r}^\ell\geq c^*\) on \(\mathcal A_r\), whose probability
converges to one. Hence
\begin{align}
\hat c_{U,r}^\ell
\overset{P_{\mathrm{MC}}}{\longrightarrow}
c^\ell(1-\alpha,\Omega).
\label{eq:LF-numerical-upper-convergence}
\end{align}
Equations \eqref{eq:LF-numerical-lower-convergence} and
\eqref{eq:LF-numerical-upper-convergence} prove
\eqref{eq:improved-cv-consistency}.

Finally, suppose \(\sum_{r=1}^\infty\delta_r<\infty\) and that the validation draws across numerical steps are defined on a common probability space. By \eqref{eq:LF-simultaneous-MC-event} and \eqref{eq:LF-uniform-DKW}, the failure probabilities of both the certification event and the uniform empirical-CDF event are bounded by \(\delta_r\). The first Borel--Cantelli lemma therefore implies that both events hold eventually almost surely. Since \(\varepsilon_r\to0\) and the rank error in \eqref{eq:LF-rank-convergence} also converges to zero, the preceding lower- and upper-endpoint arguments apply eventually almost surely. It follows that $\hat c_{L,r}^\ell \to c^\ell(1-\alpha,\Omega),$ and $\hat c_{U,r}^\ell
\to c^\ell(1-\alpha,\Omega)$ almost surely.
\end{proof}

\subsection{Proof of Proposition~\ref{prop:power-comparison}}
\label{sec:two-step-power-proofs}
Throughout this section, write $\gamma:=\alpha-\tau.$ For convenience, we restate the local drift: 
\begin{align}\label{eq:local-alternative2}
\sqrt n\,\mu(P_n)
\longrightarrow
\theta^{A,B,h,r}
:=
\left(
\underbrace{A,\ldots,A}_{h\text{ times}},
\underbrace{-B,\ldots,-B}_{r\text{ times}},
\underbrace{0,\ldots,0}_{k-h-r\text{ times}}
\right)^{\top}.
\end{align}
For $x=\sqrt n\,\bar W_n$, the change of variables $\eta=\sqrt n\,\mu$ in the two-step critical value gives
\begin{align}\label{eq:power-proof-critical}
\tilde c^\ell(x)
:=
\sup \{
(J^\ell)^{-1}(1-\gamma,\eta,I_k): {\eta\in\mathbf R_+^k\cup\mathbf R_-^k, \quad\|\eta-x\|^2\leq \chi^2_{k,1-\tau}}  \}
\end{align}
The expression is defined when the feasible intersection is nonempty.
It equals $\tilde c_n^\ell$ in \eqref{def:two-step-normal-cv} at the realized $x$. Write $\phi^\ell(x)=I\{T^\ell(x,I_k)>a_{k-1}^\ell(\alpha)\}$ and
\begin{align}\label{eq:power-proof-rule}
\psi^\ell(x)
:=
\begin{cases}
1,
&T^q(x,I_k)=T^s(x)>\chi^2_{k,1-\tau},\\[3pt]
0,
&\{\eta\in\mathbf R^k:\|\eta-x\|^2\leq \chi^2_{k,1-\tau}\}
  \subseteq\mathbf R_+^k\cup\mathbf R_-^k,\\[3pt]
I\{T^\ell(x,I_k)>\tilde c^\ell(x)\},
&\text{otherwise}.
\end{cases}
\end{align}
Thus, $\phi_n^\ell=\phi^\ell(\sqrt n\,\bar W_n)$ and $\psi_n^\ell=\psi^\ell(\sqrt n\,\bar W_n)$.
 
\begin{proof}
We begin with some observations used throughout the proof.

\noindent\textit{Gaussian limit.}
Since $\Sigma(P_n)=\Omega(P_n)=I_k$, we have $D(P_n)=I_k$.
Lemma~\ref{lem:centered-sampling-bootstrap}, specifically
\eqref{eq:centered-population-CLT}, gives
$\sqrt n\{\bar W_n-\mu(P_n)\}\Rightarrow Z$ with $Z\sim N_k(0_k,I_k)$.
Combining this with \eqref{eq:local-alternative2} and Slutsky's
theorem yields
\begin{align}\label{eq:power-CLT}
\sqrt n\,\bar W_n
\Rightarrow
Z+\theta^{A,B,h,r}.
\end{align}
For the finite drift $\theta^{A,B,h,r}$,
Lemma~\ref{lem:strictly_increasing_dist_T} (i)--(ii) shows that
$T^m(Z+\theta^{A,B,h,r})$ is continuously distributed, while $T^s(Z+\theta^{A,B,h,r})$ and $T^q(Z+\theta^{A,B,h,r},I_k)$ may have an atom at zero but have no atoms at positive values. Since the LF critical value \(a_{k-1}^\ell(\alpha)\) is positive,
\begin{align}\label{eq:power-LF-Gaussian}
\lim_{n\to\infty}\mathbb E_{P_n}[\phi_n^\ell]
=
P\{T^\ell(Z+\theta^{A,B,h,r},I_k)>a_{k-1}^\ell(\alpha)\}
=:\pi_L^\ell(A,B,h,r).
\end{align}

\noindent\textit{Residual experiment for large \(A\).} For parts (ii) and (iii), put \(d:=k-h\) and write \(Z=(Z_H,Z_D)\), where \(Z_H\) contains the first \(h\) coordinates and \(Z_D\) contains the remaining \(d\)
coordinates. Since \(Z\sim N_k(0_k,I_k)\), these blocks are
independent, with \(Z_H\sim N_h(0_h,I_h)\) and \(Z_D\sim N_d(0_d,I_d)\).
Set
\[
Y:=Z_D+(-B\mathbf1_r,0_{d-r}), \quad\text{so that}\quad Z+\theta^{A,B,h,r}=(A\mathbf1_h+Z_H,Y).
\]
For a deterministic vector, use the corresponding
decomposition \(x=(x_H,y)\).

For \(y\in\mathbf R^d\), the one-sided statistics in Lemma~\ref{lem:orthant-distance-reduction} are
\begin{align*}
V^m(y)=\max_{j\leq d}(-y_j),
\qquad
V^s(y)=V^q(y,I_d)=\textstyle\sum_{j=1}^d( \max\{-y_j,0\})^2.
\end{align*}
By \eqref{eq:power-one-sided-quantiles}, $a_d^\ell(\xi)$ is the $(1-\xi)$-quantile of $V^\ell(Z_d,I_d)$ for $Z_d\sim N_d(0_d,I_d)$. These quantiles are positive for \(\xi<1/2\), and the corresponding CDFs are continuous and strictly increasing at their positive quantiles.

For almost every realization of \(Z\), as \(A\to\infty\),
the first \(h\) coordinates of \(Z+\theta^{A,B,h,r}\)
tend to \(+\infty\), while the remaining coordinates
equal \(Y\). Applying
\eqref{eq:orthant-partial-divergence} in
Lemma \ref{lem:orthant-distance-reduction} with retained
coordinates \(\{h+1,\ldots,k\}\) and matrix \(I_k\) gives
\[
T^\ell(Z+\theta^{A,B,h,r},I_k)
\overset{a.s.}{\longrightarrow}V^\ell(Y,I_d).
\]
The cutoff \(a_{k-1}^\ell(\alpha)\) is positive.
By Lemma \ref{lem:strictly_increasing_dist_T}(iii),
the limiting statistic has no atom at this cutoff
when \(d\geq2\). The same holds when \(d=1\), since
then \(V^m(Y)=-Y\) and
\(V^s(Y)=V^q(Y,1)=(Y_-)^2\) for a nondegenerate
normal variable \(Y\). Consequently,
\[
I\{T^\ell(Z+\theta^{A,B,h,r},I_k)
>a_{k-1}^\ell(\alpha)\}
\overset{a.s.}{\longrightarrow}
I\{V^\ell(Y,I_d)>a_{k-1}^\ell(\alpha)\}.
\]
By bounded convergence and \eqref{eq:power-LF-Gaussian},
\begin{align}\label{eq:power-LF-large-A}
\lim_{A\to\infty}\pi_L^\ell(A,B,h,r)
=
P\{V^\ell(Y,I_d)>a_{k-1}^\ell(\alpha)\}.
\end{align}

\noindent\textbf{Part (i).}
The pointwise statement that every two-step rejection is an LF rejection follows by Lemma~\ref{lem:power-k2-calibration}. It remains only to verify that this pointwise difference produces a strict asymptotic power difference under every finite local alternative. By Lemma~\ref{lem:power-k2-calibration}, there is a nonempty open set $O\subseteq\{x\in\mathbf R^2:
\phi^\ell(x)=1,\ \psi^\ell(x)=0\}$. Thus, for every \(n\), $\mathbb E_{P_n}[\phi_n^\ell-\psi_n^\ell]
\geq P_n\{\sqrt{n}\bar{W}_n\in O\}.$ If \(\sqrt n\,\mu(P_n)\to\theta\in\mathbf R^2\), Portmanteau's theorem and \eqref{eq:power-CLT} imply
\[
\liminf_{n\to\infty}
\mathbb E_{P_n}[\phi_n^\ell-\psi_n^\ell]
\geq
P\{N_2(\theta,I_2)\in O\}>0,
\]
where the last inequality follows from the everywhere positive Gaussian density. This proves part (i).

\smallskip
\noindent\textbf{Part (ii).}
Put \(h=1\) and \(d=k-h\). We first derive a necessary condition for two-step rejection and use it to bound the upper limit of its power as \(n\to\infty\). At \(B=0\), this bound is strictly below \(\alpha\) for sufficiently large \(A\), and the Gaussian-shift bound shows that the gap persists for sufficiently small \(B>0\). We then show that the LF limiting power approaches a rejection probability of at least \(\alpha\) as \(A\to\infty\). Comparing these bounds establishes that the LF test is asymptotically more powerful for sufficiently large finite \(A\).

\smallskip
\noindent\textit{Step 1: Necessary condition for two-step rejection.}
Choose \(\varepsilon>0\) such that
\(\gamma+\varepsilon<1/2\), and set
\begin{align}\label{eq:def-t_epsilon}
    t_\varepsilon:=a_d^\ell(\gamma+\varepsilon)>0.
\end{align}
We claim that there exists \(M>\sqrt{\chi^2_{k,1-\tau}}\) such that
\begin{align}\label{eq:power-necessary-event}
\{\psi^\ell(x)=1\}
\subseteq
\{x_H\leq M\}
\cup
\{\|y\|^2>\chi^2_{k,1-\tau}\}
\cup
\{V^\ell(y,I_d)>t_\varepsilon\}.
\end{align}
To verify the claim, we consider the Gaussian reference statistic at
the null nuisance value $(u,0_d)$. For almost every realization of
$Z=(Z_1,Z_D)$, its first coordinate $Z_1+u$ tends to $+\infty$ as
$u\to\infty$, while its remaining coordinates equal $Z_D$.
Lemma~\ref{lem:orthant-distance-reduction}, specifically
\eqref{eq:orthant-partial-divergence} with the retained coordinates
$\{2,\ldots,k\}$ and matrix $I_k$, therefore gives $T^\ell(Z+(u,0_d),I_k)
\overset{a.s.}{\longrightarrow} V^\ell(Z_D,I_d)$. Recall that $t_\varepsilon=a_d^\ell(\gamma+\varepsilon)>0$. The CDF of $V^\ell(Z_D,I_d)$ is continuous at this positive value,
and the definition of $a_d^\ell$ and \eqref{eq:def-t_epsilon} give $P\{V^\ell(Z_D,I_d)>t_\varepsilon\} =\gamma+\varepsilon.$ Consequently, $P\{T^\ell(Z+(u,0_d),I_k)>t_\varepsilon\}
\to\gamma+\varepsilon>\gamma.$ We may therefore choose $M>\sqrt{\chi^2_{k,1-\tau}}$ such that, for every $u>M$,
\[
J^\ell(t_\varepsilon,(u,0_d),I_k)
=1-P\{T^\ell(Z+(u,0_d),I_k)>t_\varepsilon\}
<1-\gamma.
\]
Since a CDF is right-continuous, for each such $u$ there is a
$\delta_u>0$ for which
$J^\ell(t_\varepsilon+\delta_u,(u,0_d),I_k)<1-\gamma$.
By the definition of the generalized quantile,
\begin{align}\label{eq:quantile-lowerbound}
    (J^\ell)^{-1}(1-\gamma,(u,0_d),I_k)
\geq t_\varepsilon+\delta_u
>t_\varepsilon.
\end{align}

Suppose \(x_H>M\) and
\(\|y\|^2\leq\chi^2_{k,1-\tau}\). Since \(x_H>M>0\),
the vector \((x_H,0_d)\) satisfies the null. Moreover,
\[
\|x-(x_H,0_d)\|^2=\|y\|^2\leq\chi^2_{k,1-\tau},
\]
so this vector belongs to the first-step confidence region.
Consequently, the confidence region meets the null, and the
first-step immediate-rejection rule does not apply. If the
confidence region lies entirely in the null, the test does not
reject. Otherwise, its second-step critical value is at least
the reference quantile evaluated at the feasible nuisance
vector \((x_H,0_d)\). By the definition in \eqref{eq:power-proof-critical} and the choice of \(M\) in \eqref{eq:quantile-lowerbound},
\[
\tilde c^\ell(x)
\geq
(J^\ell)^{-1}(1-\gamma,(x_H,0_d),I_k)
>
t_\varepsilon.
\]
Thus, in this remaining case, rejection requires
\(T^\ell(x,I_k)>t_\varepsilon\).

For \(\ell=s,q\), the identity covariance matrix gives $T^\ell(x,I_k)
=
\min\{\|x_+\|^2,\|x_-\|^2\}
\leq
\|x_-\|^2
=
\|y_-\|^2
=
V^\ell(y,I_d),$ where the equality uses \(x_H>0\). Hence
\(T^\ell(x,I_k)>t_\varepsilon\) implies \(V^\ell(y,I_d)>t_\varepsilon\). For \(\ell=m\), $T^m(x)
=
\min\{\max_j x_j,\max_j(-x_j)\}.$ If this minimum exceeds \(t_\varepsilon>0\), then \(\max_j(-x_j)>t_\varepsilon\). Since \(-x_H<0\), a coordinate in the remaining $d$ coordinates must satisfy \(-y_j>t_\varepsilon\). Therefore \(V^m(y)>t_\varepsilon\). We have shown that rejection, together with \(x_H>M\) and \(\|y\|^2\leq\chi^2_{k,1-\tau}\), implies \(V^\ell(y,I_d)>t_\varepsilon\). This proves \eqref{eq:power-necessary-event}.

\noindent\textit{Step 2: Bounding the rejection probability below \(\alpha\).}
The boundary of the event on the right-hand side of
\eqref{eq:power-necessary-event} is contained in the union
of the three level sets $x_H=M,$ $\|y\|^2=\chi^2_{k,1-\tau}$, and $V^\ell(y,I_d)=t_\varepsilon.$ The first two have probability zero under the Gaussian limit in \eqref{eq:power-CLT}, because that limit has a
density. The third also has probability zero by Lemma~\ref{lem:strictly_increasing_dist_T}, since \(t_\varepsilon>0\). Thus the containing event is a continuity set. Applying \eqref{eq:power-CLT} and then the union bound gives
\begin{align}\label{eq:power-necessary-probability}
\begin{aligned}
\limsup_{n\to\infty}\mathbb E_{P_n}[\psi_n^\ell]
&\leq
P\{Z_1+A\leq M\ \text{or}\
\|Y\|^2>\chi^2_{k,1-\tau}\ \text{or}\
V^\ell(Y,I_d)>t_\varepsilon\}\\
&\leq
\Phi(M-A)
+
P\{\|Y\|^2>\chi^2_{k,1-\tau}\ \text{or}\
V^\ell(Y,I_d)>t_\varepsilon\},
\end{aligned}
\end{align}
where \(P\{X_H\leq M\}=\Phi(M-A)\) because
\(X_H\sim N(A,1)\). Since this term tends to zero as
\(A\to\infty\), we next bound the residual probability
strictly below \(\alpha\) for sufficiently small \(B>0\).

Recall that \(d=k-1\). The independent variables
\(\|Z_D\|^2\) and \(Z_1^2\) have distributions
\(\chi_d^2\) and \(\chi_1^2\), respectively, and their
sum has distribution \(\chi_k^2\). Therefore,
\[
\begin{aligned}
\tau-P\{\chi_d^2>\chi^2_{k,1-\tau}\}
&=
P\{
\|Z_D\|^2\leq\chi^2_{k,1-\tau}
<
\|Z_D\|^2+Z_1^2
\}>0.
\end{aligned}
\]
The strict inequality holds because both independent variables have positive densities on \((0,\infty)\). Furthermore, We restrict the choice of \(\varepsilon\) in Step 1 so that
\[
0<\varepsilon<
\min\{
\tfrac12-\gamma,\,
\tfrac{\tau-P\{\chi_d^2>\chi^2_{k,1-\tau}\}}{8}
\},
\]
and use the corresponding fixed \(M\).

When \(B=0\), \(Y\) has the same distribution as \(Z_D\).
This implies,
\begin{align}\label{eq:upper-bound-at-zeroB}
\begin{aligned}
    P\{\|Z_D\|^2>\chi^2_{k,1-\tau}\ \text{or}\
    V^\ell(Z_D,I_d)>t_\varepsilon\} 
    &\leq
    P\{\chi_d^2>\chi^2_{k,1-\tau}\}
    +\gamma+\varepsilon\\
    &=
    \alpha-
    \left[
    \tau-P\{\chi_d^2>\chi^2_{k,1-\tau}\}
    \right]
    +\varepsilon
\end{aligned}
\end{align}
where the inequality holds also by the definition of \(t_\varepsilon=a_d^\ell(\gamma+\varepsilon)\) and continuity of the CDF at this positive quantile. We have the final expression by using $\gamma=\alpha-\tau$.  Therefore, at \(B=0\), this probability is strictly below \(\alpha\).

For \(B>0\), the laws of \(Y\) and \(Z_D\) have the same
covariance matrix \(I_d\), and their means differ by
\((-B\mathbf1_r,0_{d-r})\), whose Euclidean norm is
\(B\sqrt r\). Applying
Lemma~\ref{lem:LF-sharp-Gaussian-shift} in dimension \(d\)
to the same residual event gives
\begin{align}\label{eq:power-TV-bound}
\begin{aligned}
P\{\|Y\|^2>\chi^2_{k,1-\tau}\ \text{or}\
V^\ell(Y,I_d)>t_\varepsilon\}
&\leq
\Phi\left(
\Phi^{-1}\left(
P\{\|Z_D\|^2>\chi^2_{k,1-\tau}\ \text{or}\
V^\ell(Z_D,I_d)>t_\varepsilon\}
\right)
+B\sqrt r
\right)\\
&\leq
P\{\|Z_D\|^2>\chi^2_{k,1-\tau}\ \text{or}\
V^\ell(Z_D,I_d)>t_\varepsilon\}
+\frac{B\sqrt r}{\sqrt{2\pi}}\\
&\leq
 \alpha-
    \left[
    \tau-P\{\chi_d^2>\chi^2_{k,1-\tau}\}
    \right]
    +\varepsilon
+\frac{B\sqrt r}{\sqrt{2\pi}}.
\end{aligned}
\end{align}
The second inequality uses the fact that the standard normal density $\varphi$ is bounded by \(1/\sqrt{2\pi}\), so for any $u\in\mathbf{R}$ and $v\geq0$
$$\Phi(u+v)=\textstyle\int^{u+v}_{-\infty} \varphi(t) dt 
= \Phi(u) + \textstyle\int^{u+v}_{u} \varphi(t) dt 
\leq \Phi(u) + v/\sqrt{2\pi}.$$ 
The third inequality follows from \eqref{eq:upper-bound-at-zeroB}.

Choose
\begin{align}\label{eq:power-B0}
B_0:=
\frac{\sqrt{2\pi}}{2\sqrt{k-1}}
\left[
\tau-P\{\chi_d^2>\chi^2_{k,1-\tau}\}
\right]>0.
\end{align}
Then, for every \(0<B\leq B_0\) and
\(r\in\{1,\ldots,k-1\}\),
\[
\frac{B\sqrt r}{\sqrt{2\pi}}
\leq
\frac{B_0\sqrt{k-1}}{\sqrt{2\pi}}
=
\frac12
\left[
\tau-P\{\chi_d^2>\chi^2_{k,1-\tau}\}
\right].
\]
Combining this bound with \eqref{eq:power-TV-bound} and \eqref{eq:power-necessary-probability} yields 
\begin{align*}
\limsup_{n\to\infty}\mathbb E_{P_n}[\psi_n^\ell]
&\leq
\Phi(M-A)+\alpha
-\frac12
\left[
\tau-P\{\chi_d^2>\chi^2_{k,1-\tau}\}
\right]+\varepsilon.
\end{align*}
For the fixed \(\varepsilon\) and corresponding \(M\), take \(A\) sufficiently large that  $\Phi(M-A)
<
\tfrac18
[
\tau-P\{\chi_d^2>\chi^2_{k,1-\tau}\}
].$ Together with our choice of \(\varepsilon\), this gives
\[
\limsup_{n\to\infty}\mathbb E_{P_n}[\psi_n^\ell]
<
\alpha-
\frac14
\left[
\tau-P\{\chi_d^2>\chi^2_{k,1-\tau}\}
\right]
<
\alpha.
\]

\noindent\textit{Step 3: The LF limiting power remains at least \(\alpha\).}
Let \(Z_D\sim N_d(0_d,I_d)\) and couple
\(Y=Z_D+(-B\mathbf1_r,0_{d-r})\) with \(Z_D\). Moving any
coordinate in the negative direction weakly increases each
one-sided statistic \(V^\ell\). Hence, for every realization
of \(Z_D\), $V^\ell(Y,I_d)\geq V^\ell(Z_D,I_d),$ with the matrix argument omitted for \(\ell=m,s\). Since \(a_d^\ell(\alpha)\) is a positive continuity point of the distribution of \(V^\ell(Z_D,I_d)\),
\[
P\{V^\ell(Y,I_d)>a_d^\ell(\alpha)\}
\geq
P\{V^\ell(Z_D,I_d)>a_d^\ell(\alpha)\}
=
\alpha.
\]
Here \(d=k-1\). By \eqref{eq:power-LF-large-A}, \(\pi_L^\ell(A,B,1,r)\) converges to the probability on the left as \(A\to\infty\). It follows that, for all sufficiently large finite \(A\), the LF limiting power exceeds
\(\alpha- [\tau-P\{\chi^2_d>\chi^2_{k,1-\tau}\}]/8\). This proves
\eqref{eq:power-LF-wins}. 

\smallskip
\noindent\textbf{Part (iii).}
Let \(d=k-h\) and suppose
\[
a_d^\ell(\gamma)<a_{k-1}^\ell(\alpha).
\]
We first identify the limiting LF benchmark and then
construct a subset of the two-step rejection region whose
probability is strictly larger.

Under the distribution of \(Y\), define
\[
p_d:=P\{V^\ell(Y,I_d)>a_d^\ell(\gamma)\},
\qquad
p_{k-1}:=P\{V^\ell(Y,I_d)>a_{k-1}^\ell(\alpha)\}.
\]
Note that $p_{k-1}$ is the limiting rejection probability of the least favorable in \eqref{eq:power-LF-large-A}. Recall that the CDF of \(V^\ell(Y,I_d)\) is strictly increasing on \((0,\infty)\) by the one-sided conclusion of Lemma~\ref{lem:strictly_increasing_dist_T}; when \(d=1\),
this follows directly from the normal CDF. Therefore \(0<a_d^\ell(\gamma)<a_{k-1}^\ell(\alpha)\) implies $p_d>p_{k-1}>0.$ 

\noindent\textit{Step 1: Construct an event involving the remaining coordinates with probability exceeding the LF limit.}
For \(\varepsilon\in(0,\gamma)\), let
$$t_\varepsilon:=a_d^\ell(\gamma-\varepsilon).$$
As \(\varepsilon\downarrow0\), continuity of the positive
quantile gives
\(t_\varepsilon\downarrow a_d^\ell(\gamma).\) The absence of an
atom at this limit then implies $P\{V^\ell(Y,I_d)>t_\varepsilon\}\uparrow p_d.$ Since \(p_d>p_{k-1}\), we can choose \(\varepsilon\) sufficiently
small that this probability exceeds \(p_{k-1}\). For this fixed
\(\varepsilon\), the events
\(\{\|Y\|<L,\ V^\ell(Y,I_d)>t_\varepsilon\}\) increase to
\(\{V^\ell(Y,I_d)>t_\varepsilon\}\) as \(L\to\infty\).
We can therefore choose a finite \(L\) such that
\begin{align}\label{eq:power-residual-margin}
p_{\varepsilon,L}
:=
P\{\|Y\|<L,\ V^\ell(Y,I_d)>t_\varepsilon\}
>
p_{k-1}.
\end{align}

\noindent\textit{Step 2: two-step power envelope.} Suppose $\min_{j\leq h}x_j >M > \sqrt{\chi^2_{k,1-\tau}}$. We derive a bound on the second-step reference distribution.  For every \(t>0\), we claim that
\begin{align}\label{eq:power-uniform-tail-bound}
\sup_{\substack{
\eta\in\mathbf R_+^k\cup\mathbf R_-^k\\
\|\eta-x\|^2\leq\chi^2_{k,1-\tau}}}
P\{T^\ell(Z+\eta,I_k)>t\}
\leq
P\{V^\ell(Z_D,I_d)>t\}
+
h\{1-\Phi(M-\sqrt{\chi^2_{k,1-\tau}})\}.
\end{align}
To verify this bound, note first that the confidence ball cannot meet the negative orthant because its first \(h\) coordinates are more than \(\sqrt{\chi^2_{k,1-\tau}}\) above zero. Hence every feasible null nuisance satisfies \(\eta\geq0_k\) and
\(\eta_j\geq M-\sqrt{\chi^2_{k,1-\tau}}\) for \(j\leq h\).
On the event \(\{Z_H+\eta_H\geq0_h\}\), rejection of the sign-agreement null must be generated by the remaining negative coordinates. Thus, for \(\ell=s,q\), $T^\ell(Z+\eta,I_k) \leq V^\ell(Z_D+\eta_D,I_d),$  while for \(\ell=m\) the implication $T^m(Z+\eta)>t \Rightarrow V^m(Z_D+\eta_D)>t$
holds for every \(t>0\). Since \(\eta_D\geq0_d\), $V^\ell(Z_D+\eta_D,I_d)\leq V^\ell(Z_D,I_d).$ Consequently, for every feasible \(\eta\),
\[
P\{T^\ell(Z+\eta,I_k)>t\}
\leq
P\{V^\ell(Z_D,I_d)>t\}
+
P\{Z_H+\eta_H\not\geq0_h\}.
\]
Finally, the probability of the event $\{Z_H+\eta_H\not\geq0_h\}$ satisfies the following: \[
\begin{aligned}
P\{Z_H+\eta_H\not\geq0_h\}
=
P\{\cup_{j=1}^h\{Z_j<-\eta_j\}\} 
\leq
\textstyle\sum_{j=1}^h\Phi(-\eta_j)
\leq
h\{1-\Phi(M-\sqrt{\chi^2_{k,1-\tau}})\}.
\end{aligned}
\]
Since this bound does not depend on the feasible nuisance
vector \(\eta\), taking the supremum proves
\eqref{eq:power-uniform-tail-bound}.

\noindent\textit{Step 3: For sufficiently strong positive coordinates, the residual event forces two-step rejection.}
Choose
\(M>\max\{L,\sqrt{\chi^2_{k,1-\tau}}\}\) sufficiently large
that
\[
h\{1-\Phi(M-\sqrt{\chi^2_{k,1-\tau}})\}<\varepsilon.
\]
We claim that the event
\begin{align}\label{eq:power-sufficient-event}
\{\textstyle\min_{j\leq h}x_j>M,
~
\|y\|<L,
~
V^\ell(y,I_d)>t_\varepsilon\} \subset \{\psi^{\ell}(x)=1\}.
\end{align}
First, the positive value of \(V^\ell(y,I_d)\) implies that
\(y\) has a negative coordinate. Thus \(x\) has both positive
and negative coordinates and does not satisfy the null.
Moreover, on \eqref{eq:power-sufficient-event}, $$T^\ell(x,I_k)=V^\ell(y,I_d)>t_\varepsilon.$$ Indeed, for \(\ell=m\), the positive branch of \(T^m(x)\) is at least \(M\), while its negative branch is \(V^m(y)\leq\|y\|<L<M\). For \(\ell=s,q\), the positive squared-distance branch is at least \(M^2\), while the negative branch equals \(V^\ell(y,I_d)\leq\|y\|^2<L^2<M^2\).

If the first-step confidence ball misses the null, the test rejects immediately. Suppose instead that it meets the null. Because its center \(x\) does not satisfy the null, the ball cannot lie entirely in the null, so the second-step rule applies. Applying \eqref{eq:power-uniform-tail-bound} at \(t=t_\varepsilon\), for every feasible \(\eta\) we obtain
\[
\begin{aligned}
P\{T^\ell(Z+\eta,I_k)>t_\varepsilon\}
&\leq
P\{V^\ell(Z_D,I_d)>t_\varepsilon\}
+h\{1-\Phi(M-\sqrt{\chi^2_{k,1-\tau}})\}\\
&=
\gamma-\varepsilon
+h\{1-\Phi(M-\sqrt{\chi^2_{k,1-\tau}})\}<\gamma.
\end{aligned}
\]
where the $\gamma-\epsilon$ follows from the definition of $t_{\varepsilon}$ and the last strict inequality follows from the choice of $M.$ Thus the reference CDF at \(t_\varepsilon\) exceeds \(1-\gamma\) for every feasible \(\eta\). Every corresponding \((1-\gamma)\)-quantile is at most \(t_\varepsilon\), and so is their supremum. Since the observed statistic is strictly greater than \(t_\varepsilon\), the two-step test rejects. This proves the claim in \eqref{eq:power-sufficient-event}

Finally, the event in \eqref{eq:power-sufficient-event} is open.
Portmanteau's theorem and \eqref{eq:power-CLT} therefore
give a lower bound on its limiting probability. In the
Gaussian limit, the first \(h\) coordinates are independent
of \(Y\), and each exceeds \(M\) with probability
\(\Phi(A-M)\). Hence
\begin{align}\label{eq:power-sufficient-probability}
\liminf_{n\to\infty}\mathbb E_{P_n}[\psi_n^\ell]
\geq
\Phi(A-M)^h p_{\varepsilon,L}.
\end{align}
As \(A\to\infty\), the right-hand side converges to \(p_{\varepsilon,L}>p_{k-1}\) as in \eqref{eq:power-residual-margin}, whereas
\eqref{eq:power-LF-large-A} gives \(\pi_L^\ell(A,B,h,r)\to p_{k-1}\). The strict gap persists for all sufficiently large finite \(A\). Thus the lower limit of the two-step power exceeds the LF limiting power, proving
\eqref{eq:power-two-step-wins}.
\end{proof}

\section[\hspace{2mm}Supporting technical results]{Supporting technical results}
We begin with three auxiliary results from \citet{RomanoShaikh2012AoS}.

\begin{lemma}[Lemma S.6.1. in \cite{RomanoShaikh2012AoS}]\label{alem:consistency-sample-variance}
    Consider any sequence $\left\{P_n \in {\mathbf{P}}: n \geq 1\right\}$, where ${\mathbf{P}}$ satisfies \eqref{def:uniform-integrability}. Let $X_{i}, i=$ $1, \ldots, n$ be an i.i.d. sequence of real-valued random variables with distribution $P_n$. Then,
$$
\frac{S_n^2}{\sigma^2\left(P_n\right)} \stackrel{P_n}{\rightarrow} 1 .
$$
\end{lemma}

\begin{lemma}[Lemma S.7.1. in \cite{RomanoShaikh2012AoS}]\label{alem:consistency-sample-corr}
    Consider any sequence $\left\{P_n \in {\mathbf{P}}: n \geq 1\right\}$, where ${\mathbf{P}}$ satisfies \eqref{def:uniform-integrability}. Let $X_{ i}, i=1, \ldots, n$ be an i.i.d. sequence of random variables with distribution $P_n$. Then,
$$
\|\Omega(\hat{P}_n)-\Omega(P_n)\| \stackrel{P_n}{\rightarrow} 0,
$$
where the norm $\|\cdot\|$ is the component-wise maximum of the absolute value of all elements.
\end{lemma}

\begin{lemma}[{Lemma A.1(vi) in \cite{RomanoShaikh2012AoS}}]\label{alem:random-quantile-transfer}
{Let $T_n$ be a real-valued random variable with CDF
$F_n(x)=P\{T_n\le x\}$}, {and let $G_n$ be a random CDF on
$\mathbf R$}. {Put
$\Delta_n:=\sup_{x\in\mathbf R}|G_n(x)-F_n(x)|$.}
For {$\gamma\in(0,1)$} and {generalized} quantiles,
{\begin{align}\label{eq:random-quantile-transfer}
 P\{T_n>G_n^{-1}(1-\gamma)\}
 \le\gamma+\varepsilon+P\{\Delta_n>\varepsilon\},
 \qquad 0<\varepsilon<1-\gamma.
\end{align}}
\end{lemma}

{This} is {the upper-tail implication of the cited lemma, since
$\Delta_n$ bounds its one-sided CDF error}. {The same bound holds
when $G_n$ has mass at $+\infty$, with
$\inf\varnothing=+\infty$, because their quantile argument
applies unchanged.}

\subsection{Lemmas for Theorems \ref{thm:LFtest-uniform-validity} and \ref{thm:two-step-uniform-validity}}

The unconditional original-sample argument in the following lemma follows the proof of Theorem~3.1 of \citet{RomanoShaikh2012AoS}, {while the
conditional argument follows the proofs of their
Theorems~3.7--3.8, using Lemmas~S.12.1--S.12.2 of their
supplement}.

\begin{lemma}\label{lem:centered-sampling-bootstrap}
{Suppose Assumption \ref{ass:uniform-integrability} holds. Consider a sequence $\{P_n\in\mathbf P:n\ge1\}$} such that
{$\Omega(P_n)\to\Omega_*$ in $\bar{\mathbf O}$, allowing a singular limit}. {Put
$D(P):=\operatorname{diag}\{\sigma_1(P),\ldots,\sigma_k(P)\}$.
Let $\hat P_n^*$ denote probability
conditional on the original sample under $n$ independent draws
from $\hat P_n$. A star denotes a bootstrap sample quantity.} Then
\begin{align}\label{eq:centered-population-CLT}
D(P_n)^{-1}\sqrt n\{\bar W_n-\mu(P_n)\}
\Rightarrow N(0_k,\Omega_*),
\end{align}
and
\begin{align}\label{eq:centered-sampling-studentization}
\sqrt n S_n^{-1}\{\bar W_n-\mu(P_n)\}
\Rightarrow N(0_k,\Omega_*),
\qquad
S_n^{-1}D(P_n)\overset{P_n}{\longrightarrow}I_k,
\qquad
\hat\Omega_n\overset{P_n}{\longrightarrow}\Omega_*.
\end{align}
Moreover,
\begin{equation}\label{eq:bootstrap-centered-clt}
\begin{aligned}
d_{\mathrm{BL}}(
    \mathcal L_{\hat P_n^*}(
        \sqrt n\,S_n^{*-1}(\bar W_n^*-\bar W_n)
    ),
    N(0_k,\Omega_*)
)
&\overset{P_n}{\longrightarrow}0,\\
\hat P_n^*\{
    \|S_n^{*-1}D(P_n)-I_k\|>\varepsilon
\}
&\overset{P_n}{\longrightarrow}0,
    \qquad \text{for every }\varepsilon>0,\\
\hat P_n^*\{
    \|\hat\Omega_n^*-\Omega_*\|>\varepsilon
\}
&\overset{P_n}{\longrightarrow}0,
    \qquad \text{for every }\varepsilon>0.
\end{aligned}
\end{equation}
Here $d_{\mathrm{BL}}$ is the bounded-Lipschitz distance, and $\mathcal L_{\hat P_n^*}$ denotes the bootstrap law conditional on the original sample.
\end{lemma}

\begin{proof}
Write $V_{i,n}:=D(P_n)^{-1}\{W_i-\mu(P_n)\}$.
For every $M>0$,
\begin{align*}
\|V_{i,n}\|^2 I\{\|V_{i,n}\|>M\}
\le \textstyle\sum_{j=1}^k V_{j,i,n}^2
 \cdot \max_j I\{|V_{j,i,n}|>M/\sqrt k\}\leq k\textstyle\sum_{j=1}^k V_{j,i,n}^2
I\{|V_{j,i,n}|>M/\sqrt k\}
\end{align*}
where the inequality holds because $\|V_{i,n}\|^2=\sum_{j=1}^k V_{j,i,n}^2$ and 
$$ I\{\|V_{i,n}\|>M\}=I\{\textstyle\sum_j V_{j,i,n}^2>M^2\}\leq I\{ \max_j V_{j,i,n}^2>\tfrac{M^2}{k}\} \leq \max_j I\{V_{j,i,n}>\tfrac{M}{\sqrt{k}}\}.$$
Thus \eqref{def:uniform-integrability} implies the Lindeberg condition for every fixed linear combination of $n^{-1/2}V_{i,n}$ becuase for any $a\in\mathbf{R}^k$, Cauchy-Schwartz gives $|a'V_{i,n}|^2 \leq \|a\|^2 \|V_{i,n}\|^2$ and so
$$E_{P_n}[ (a'V_{i,n})^2 I\{|a'V_{i,n}|>\varepsilon\sqrt{n}\}] 
\leq \|a\|^2E_{P_n}[\|V_{i,n}\|^2 I\{\|V_{i,n}\|>\varepsilon\sqrt{n}/\|a\| \}]. $$
Its variance converges to the corresponding quadratic form in $\Omega_*$. The triangular-array CLT, with Chebyshev's inequality when this variance tends to zero, and the Cram\'er--Wold device give \eqref{eq:centered-population-CLT}.
Lemmas~\ref{alem:consistency-sample-variance} and
\ref{alem:consistency-sample-corr}, together with Slutsky's
theorem, then give \eqref{eq:centered-sampling-studentization}.

For the conditional conclusions in \eqref{eq:bootstrap-centered-clt}, use the empirical-distribution argument in the proofs of Lemma~S.12.1 and Theorem~3.8 of
\citet{RomanoShaikh2012AoS}, together with their Lemmas~S.6.1, S.7.1, and S.12.2. With $Q_n=\hat P_n$, this argument applies the conclusions above along almost-sure further subsequences. Since $\mu(\hat P_n)=\bar W_n$ and $D(\hat P_n)=S_n$, it gives the conditional centered CLT, correlation consistency, and
$S_n^{*-1}S_n\overset{\hat P_n^*}{\longrightarrow}I_k$,
in $P_n$-probability. Finally, $ S_n^{*-1}D(P_n)
 =(S_n^{*-1}S_n)(S_n^{-1}D(P_n)),$ together with \eqref{eq:centered-sampling-studentization}, gives the scale conclusion in \eqref{eq:bootstrap-centered-clt}.  \end{proof}

\begin{lemma}\label{lem:CDF-quantile-transfer}
{Let $F$ be a CDF supported on $[0,\infty)$} and continuous
{on $(0,\infty)$, and let $G_n$ be random CDFs supported on
$[0,\infty]$. If $G_n(x)\to F(x)$} in {probability for every
$x\ge0$}, {including $x=0$, then
$\sup_{x\in\mathbf R}|G_n(x)-F(x)|\to0$ in probability.}
The {same conclusion holds when $F$} is {continuous on $\mathbf R$}, {$G_n$ are CDFs on $\mathbf R\cup\{+\infty\}$, and convergence holds at every real $x$}.
\end{lemma}

\begin{proof}
{We adapt the finite-grid argument underlying P\'olya's theorem \citep[Theorem~11.2.9]{LehmannRomano2005testing} to convergence in probability}, {using convergence at zero to accommodate the possible atom there. Fix $\varepsilon>0$ and choose $\eta>0$ such that $F(\eta)-F(0)<\varepsilon$. Choose a finite increasing grid on $[\eta,\infty)$ whose successive $F$-increments and upper tail} are {less than $\varepsilon$}. {At its finitely many endpoints and at zero}, {the maximum difference between $G_n$} and {$F$ tends to zero in probability}. {Monotonicity bounds the difference between
endpoints}, {including the interval $[0,\eta]$, by} this {maximum plus $\varepsilon$. This} proves {uniform convergence. When $F$ is continuous on all of $\mathbf R$,} the {same finite-grid argument includes a lower tail instead of the point zero}.
\end{proof}

\begin{lemma}\label{lem:statistic-degenerate-case}
Consider a sequence $\{P_n\in\mathbf P:n\ge1\}$, where $\mathbf P$ is a set of distributions on $\mathbf R^k$ satisfying \eqref{def:uniform-integrability}. Let $W_i$, $i=1,\ldots,n$, be an i.i.d. sequence of random vectors with distribution $P_n$. Suppose
\[
\frac{\sqrt n\mu_j(P_n)}{\sigma_j(P_n)}\to-\infty
\quad\text{for all }1\le j\le k,
\qquad\text{or}\qquad
\frac{\sqrt n\mu_j(P_n)}{\sigma_j(P_n)}\to+\infty
\quad\text{for all }1\le j\le k.
\]
Then $T_n^m\overset{P_n}{\to}-\infty$ and $T_n^s\overset{P_n}{\to}0$. Furthermore, if Assumption \ref{ass:uniform-eigenvalue} holds, {$T_n^q\overset{P_n}{\to}0$}.
\end{lemma}

\begin{proof}
Let $A_n:=\{\min_{j\in[k]}S_{j,n}>0\}$. By Lemma~\ref{alem:consistency-sample-variance}, $P_n(A_n^c)\to0$. {All ratios} below {are evaluated} on {$A_n$; for} the {proof, set $X_n=0_k$ and $T_n^\ell=0$ on $A_n^c$.

For brevity} of {notation, consider
\begin{align*}
X_n:=\sqrt nS_n^{-1}\bar W_n
&=S_n^{-1}D(P_n)\cdot D(P_n)^{-1}\sqrt n\{\bar W_n-\mu(P_n)\}
+S_n^{-1}\sqrt n\mu(P_n).
\end{align*}
The first term} is {stochastically bounded}. {Indeed,
$S_n^{-1}D(P_n)=I_k+o_{P_n}(1)$ by Lemma~\ref{alem:consistency-sample-variance},} and {$D(P_n)^{-1}\sqrt n\{\bar W_n-\mu(P_n)\}=O_{P_n}(1)$ by Lemma~\ref{lem:centered-sampling-bootstrap} and compactness} of {$\bar{\mathbf O}$. Then, it follows from Slutsky's theorem. For} the {second term, marginal variance consistency gives
\[
\frac{\sqrt n\mu_j(P_n)}{S_{j,n}}
=\frac{\sigma_j(P_n)}{S_{j,n}}
\frac{\sqrt n\mu_j(P_n)}{\sigma_j(P_n)}.
\]
Since $\sigma_j(P_n)/S_{j,n}\overset{P_n}{\to}1$, this term diverges in probability with the same sign as the standardized local mean. Combining the two terms gives,} for {each $j\in[k]$,
\begin{align}\label{eq:statistic-degenerate-case-eq2}
X_{j,n}=\frac{\sqrt n\bar W_{j,n}}{S_{j,n}}
\overset{P_n}{\longrightarrow}
\begin{cases}
-\infty,&\text{if }\sqrt n\mu_j(P_n)/\sigma_j(P_n)\to-\infty,\\
+\infty,&\text{if }\sqrt n\mu_j(P_n)/\sigma_j(P_n)\to+\infty.
\end{cases}
\end{align}

Next, we show that $T_n^m$ diverges to $-\infty$ in probability}. {Suppose first that all standardized local means tend} to {$-\infty$. Since $T_n^m\le\max_{j\in[k]}X_{j,n}$, for any $M<0$,
\[
P_n\{T_n^m>M\}
\le P_n\{\textstyle\max_{j\in[k]}X_{j,n}>M\}
\le\textstyle\sum_{j=1}^kP_n\{X_{j,n}>M\}\to0
\]
by \eqref{eq:statistic-degenerate-case-eq2}. Thus $T_n^m\overset{P_n}{\to}-\infty$. If all standardized local means tend to $+\infty$,} the {same argument applies using $T_n^m\le\max_{j\in[k]}(-X_{j,n})$}.

{For $\ell=s$, the statistic is zero if} and {only if $X_n\in\mathbf R_+^k\cup\mathbf R_-^k$. Therefore}, {for any $\varepsilon>0$},
{\[
P_n\{|T_n^s|>\varepsilon\}
\le P_n\{T_n^s>0\}
=P_n\{X_n\notin\mathbf R_+^k\cup\mathbf R_-^k\}\to0,
\]}
where {the convergence follows from \eqref{eq:statistic-degenerate-case-eq2}}.

For $\ell=q$, let $E_n$ be the event that $A_n$ holds and the unregularized sample correlation matrix $\Omega(\hat P_n)$ is positive definite. Assumption \ref{ass:uniform-eigenvalue} and Lemma~\ref{lem:sample-corr-positive-definite} give $P_n(E_n^c)\to0$. On $E_n$, regularization is inactive and $T_n^q=0$ if and only if $X_n\in\mathbf R_+^k\cup\mathbf R_-^k$. Hence, for any $\varepsilon>0$,
\[
\begin{aligned}
P_n\{|T_n^q|>\varepsilon\}
&\le P_n(E_n^c)+P_n\{|T_n^q|>\varepsilon,E_n\}
\le P_n(E_n^c)
+P_n\{X_n\notin\mathbf R_+^k\cup\mathbf R_-^k,E_n\}\to0,
\end{aligned}
\]
again by \eqref{eq:statistic-degenerate-case-eq2}.
\end{proof}

\begin{lemma}\label{lem:orthant-distance-reduction}
For a {nonempty coordinate set $I$}, {write $x_I$ and $\Omega_{I\times I}$ for the corresponding subvector and principal submatrix. Define
\begin{align*}
V^m(x):=\max_j(-x_j),\qquad
V^s(x):=\sum_jx_j^2I\{x_j\le0\},\qquad
V^q(x,\Omega)&:=\inf_{u\in\mathbf R_+^k}(x-u)'\Omega^{-1}(x-u)
\end{align*}}
where {$\Omega\succ0$ and the second argument is suppressed for $\ell=m,s$}. {For $d\in\mathbf R_+^k$ and $\varnothing\ne I\subseteq[k]$},
{\begin{align}\label{eq:orthant-monotonicity-restriction}
V^\ell(x+d,\Omega)\le V^\ell(x,\Omega),
\qquad
V^\ell(x,\Omega)\ge V^\ell(x_I,\Omega_{I\times I}).
\end{align}
Suppose $\varnothing\ne I\subsetneq[k]$, $x_{I,n}\to x_I$, and $x_{j,n}\to+\infty$} for {$j\notin I$. For $\ell=q$}, {suppose also $\Omega_n\to\Omega\succ0$. Then
\begin{align}\label{eq:orthant-partial-divergence}
V^\ell(x_n,\Omega_n)\to V^\ell(x_I,\Omega_{I\times I}),
\qquad
T^\ell(x_n,\Omega_n)\to V^\ell(x_I,\Omega_{I\times I}).
\end{align}
The corresponding negative-divergence result follows by reversing signs.}
\end{lemma}

\begin{proof}
{For $\ell=m,s$, all claims follow directly from} the {definitions. For $\ell=q$}, {replacing a feasible $u$ by $u+d$ proves the first inequality in \eqref{eq:orthant-monotonicity-restriction}}. {Write $D=[k]\setminus I$ and, when $D$ is nonempty, put
\[
A:=\Omega_{I\times I},\qquad 
B=\Omega_{D\times I}\Omega_{I\times I}^{-1},\qquad
H=\Omega_{D\times D}-\Omega_{D\times I}\Omega_{I\times I}^{-1}\Omega_{I\times D}\succ0.
\]
For every residual $r=x-u$, the block quadratic-form identity} gives
{\begin{align}\label{eq:orthant-profile-identity}
r'\Omega^{-1}r
=r_I'\Omega_{I\times I}^{-1}r_I
+(r_D-Br_I)'H^{-1}(r_D-Br_I).
\end{align}
Taking the infimum over $u\ge0$ proves the  restriction inequality} in {\eqref{eq:orthant-monotonicity-restriction}. The case $I=[k]$} is {equality.

We prove \eqref{eq:orthant-partial-divergence} for $\ell=q$}
in {three steps}. {First}, {we show that}, {for all sufficiently
large $n$,
\begin{align}\label{eq:equivalence-reduced}
    V^q(x_n,\Omega_n)
=V^q(x_{I,n},(\Omega_n)_{I\times I}).
\end{align}
To establish this equality, we extend a minimizer of the reduced problem} to {a feasible point of the full problem that makes the second term in
\eqref{eq:orthant-profile-identity} equal to zero}.
{Second}, {we show that the reduced minimum converges to
$V^q(x_I,\Omega_{I\times I})$}. {Finally, we show that} the {distance} to the {negative orthant diverges} to {infinity}. {Since $T^q$ is the minimum of the distances to the two orthants, these claims give both limits.

To this end}, {write $D=[k]\setminus I$ and put
\[
A_n:=(\Omega_n)_{I\times I},
\qquad
B_n:=(\Omega_n)_{D\times I}A_n^{-1},
\qquad 
H_n:=(\Omega_n)_{D\times D}
      -B_n(\Omega_n)_{I\times D}.
\]
Since $\Omega_n\to\Omega\succ0$, there exist constants
$0<c<C<\infty$ such that the eigenvalues of $\Omega_n$
and $A_n$ lie in $[c,C]$ for all sufficiently large $n$}.
{Also},
{$B_n\to\Omega_{D\times I}A^{-1}$,
so $B_n$} is {bounded.

To verify the first claim} in {\eqref{eq:equivalence-reduced}}, {let $u_{I,n}$ attain
\[
V^q(x_{I,n},A_n)
=\inf_{u_I\in\mathbf R_+^{|I|}}
(x_{I,n}-u_I)'A_n^{-1}(x_{I,n}-u_I).
\]
Such a minimizer exists} because {the feasible set} is {closed
and} the {objective is continuous} and {tends to infinity as
$\|u_I\|\to\infty$}. {Since $u_I=0$} is {feasible,
\[
\frac{1}{C}\|x_{I,n}-u_{I,n}\|^2
\le (x_{I,n}-u_{I,n})'A_n^{-1}(x_{I,n}-u_{I,n})
=V^q(x_{I,n},A_n)
\le x_{I,n}'A_n^{-1}x_{I,n}
\le\frac{1}{c}\|x_{I,n}\|^2.
\]
Because $x_{I,n}\to x_I$,} the {right-hand side} is {bounded. Thus $x_{I,n}-u_{I,n}$ is bounded}, {and so is
$u_{I,n}$. Now set
\[
u_{D,n}:=x_{D,n}-B_n(x_{I,n}-u_{I,n}).
\]
The sequence $B_n(x_{I,n}-u_{I,n})$ is bounded because both $B_n$} and {$(x_{I,n}-u_{I,n})$ are bounded. Since every coordinate of
$x_{D,n}$ tends} to {$+\infty$}, we have {$u_{D,n}\ge0$
componentwise for all sufficiently large $n$}.
{Therefore,} the {vector $u_n$ with subvectors $u_{I,n}$
and $u_{D,n}$} is {eventually feasible for} the {full
minimization problem}. {Applying \eqref{eq:orthant-profile-identity}} with
{$\Omega_n$ gives
\[
\begin{aligned}
(x_n-u_n)'\Omega_n^{-1}(x_n-u_n)
&=(x_{I,n}-u_{I,n})'A_n^{-1}(x_{I,n}-u_{I,n})\\
&+((x_{D,n}-u_{D,n})-B_n(x_{I,n}-u_{I,n}))'H_n^{-1}
        ((x_{D,n}-u_{D,n})-B_n(x_{I,n}-u_{I,n}))\\
&=(x_{I,n}-u_{I,n})'A_n^{-1}(x_{I,n}-u_{I,n})
=V^q(x_{I,n},A_n),
\end{aligned}
\]
where the second equality holds by construction}. {Since $u_n$ is} feasible{, $V^q(x_n,\Omega_n)\le V^q(x_{I,n},A_n)$ eventually}. {The reverse inequality follows from \eqref{eq:orthant-monotonicity-restriction}.} Hence {\eqref{eq:equivalence-reduced} holds for all sufficiently large $n$}.

{We next verify convergence of} the {reduced minimum.
Let $u_I^*$ attain} the {limiting infimum
\[
V^q(x_I,A)
=\inf_{u_I\in\mathbf R_+^{|I|}}
(x_I-u_I)'A^{-1}(x_I-u_I).
\]
Evaluating the $n$th objective at this fixed feasible} point {gives $V^q(x_{I,n},A_n)
\le
(x_{I,n}-u_I^*)'A_n^{-1}(x_{I,n}-u_I^*).$ Since $x_{I,n}\to x_I$ and $A_n^{-1}\to A^{-1}$,
\[
\limsup_{n\to\infty}V^q(x_{I,n},A_n)
\le
(x_I-u_I^*)'A^{-1}(x_I-u_I^*)
=V^q(x_I,A).
\]
For} the {reverse inequality, choose a subsequence $n_l$ along which $V^q(x_{I,n_l},A_{n_l})$ converges to $\liminf_n V^q(x_{I,n},A_n)$}. {The minimizing vectors $u_{I,n_l}$ are bounded, as shown above. After passing to a further subsequence, we may}
therefore {assume} that {$u_{I,n_l}\to\bar u_I$.
The nonnegative orthant is closed, so $\bar u_I\in\mathbf R_+^{|I|}$. Consequently,
\[
\begin{aligned}
\liminf_{n\to\infty}V^q(x_{I,n},A_n)
&=\lim_{l\to\infty}
(x_{I,n_l}-u_{I,n_l})'A_{n_l}^{-1}
(x_{I,n_l}-u_{I,n_l})
=(x_I-\bar u_I)'A^{-1}(x_I-\bar u_I)
\ge V^q(x_I,A).
\end{aligned}
\]
Combining} the {upper and lower bounds yields
\[
V^q(x_{I,n},A_n)\to V^q(x_I,A).
\]
Together with \eqref{eq:equivalence-reduced},} this {proves the first limit in \eqref{eq:orthant-partial-divergence}}.

{Finally, we verify that the distance} to the negative
orthant {diverges}. {Fix any $j\in D$}, which {is possible
because $I\subsetneq[k]$.
Since $x_{j,n}\to+\infty$, we have $x_{j,n}>0$ eventually.
For every $u\in\mathbf R_-^k$} and {all sufficiently large $n$,
\[
(x_n-u)'\Omega_n^{-1}(x_n-u)
\ge\frac{1}{C}\|x_n-u\|^2
\ge\frac{1}{C}(x_{j,n}-u_j)^2
\ge\frac{x_{j,n}^2}{C}
\]}
where {last} inequality uses {$u_j\le0$ and $x_{j,n}>0$}.
{Taking the infimum over $u\in\mathbf R_-^k$ gives
\[
\inf_{u\in\mathbf R_-^k}
(x_n-u)'\Omega_n^{-1}(x_n-u)
\ge\frac{x_{j,n}^2}{C}\to+\infty.
\]
In contrast}, the {distance to the positive orthant
converges to the finite value $V^q(x_I,A)$}.
{Since $T^q(x_n,\Omega_n)$ is the minimum of these two
distances},
{\[
T^q(x_n,\Omega_n)
\to V^q(x_I,A)
=V^q(x_I,\Omega_{I\times I}).
\]}
This proves the {second limit in
\eqref{eq:orthant-partial-divergence}}.
\end{proof}

\begin{lemma}\label{lem:conv-ip-to-Jn}
Consider a sequence $\{P_n\in \mathbf{P}: n \geq 1\}$ where $\mathbf{P}$ is a set of distributions on $\mathbf{R}^k$ satisfying \eqref{def:uniform-integrability}. Let $W_i$, $i=1,2,...,n$, be an i.i.d. sequence of random vectors with distribution $P_n$. Suppose for some non-empty set $I\subset \{1,2, ..., k\}$ 
    \begin{align*}
        \frac{\sqrt{n} \mu_j(P_n)}{ \sigma_j(P_n)} \to \delta_j \in(-\infty, 0] \text{ for } j\in I\text{ and }\frac{\sqrt{n} \mu_j(P_n)}{ \sigma_j(P_n)} \to -\infty \text{ for }j \not\in I
    \end{align*}
    or
    \begin{align}\label{eq:conv-ip-to-Jn-condition2}
        \frac{\sqrt{n} \mu_j(P_n)}{ \sigma_j(P_n)} \to \delta_j \in [0,\infty) \text{ for } j\in I\text{ and }\frac{\sqrt{n} \mu_j(P_n)}{ \sigma_j(P_n)} \to \infty \text{ for }j \not\in I.
    \end{align}
Then, for $\ell\in\{m,s\}$, $J^{\ell}$ in \eqref{def:J} and  $J^{\ell}_n(x):=J^{\ell}_n(x,\sqrt{n}\mu(P_n),P_n)$ in \eqref{def:dist-test-stat} satisfy
    \begin{align}\label{eq:conv-ip-to-Jn}
        \sup_{x \in \mathbf{R}} | J^{\ell}(x, \sqrt{n}S^{-1}_n\mu(P_n), \Omega(\hat{P}_n)) - J^{\ell}_n(x)| \overset{P_n}{\to} 0 \text{ as }n\to\infty.
    \end{align}
Furthermore, under Assumption \ref{ass:uniform-eigenvalue}, \eqref{eq:conv-ip-to-Jn} holds for $\ell=q$.
\end{lemma}

\begin{proof}
From any subsequence, extract a further subsequence on which $\Omega(P_n)\to\Omega_*\in\bar{\mathbf O}$. For $\ell=q$, the eigenvalue condition gives $\Omega_*\succ0$. It is enough to prove the asserted convergence along every such subsequence. We give the proof under \eqref{eq:conv-ip-to-Jn-condition2}; the negative-orthant case follows by reversing signs.

Write $X_n:=\sqrt nS_n^{-1}\bar W_n$ and $Y_I:=(\Omega_*^{1/2}Z)_I+\delta_I$ for $Z\sim N(0_k,I_k)$. Lemma~\ref{lem:centered-sampling-bootstrap} gives
\begin{align}\label{eq:subvector-clt-lln}
     (X_n)_I\Rightarrow Y_I,
    \qquad
    X_{j,n}\overset{P_n}{\longrightarrow}+\infty
    \quad(j\notin I),
    \qquad
    \hat\Omega_n\overset{P_n}{\longrightarrow}\Omega_*.
\end{align}

Consider also the Gaussian vector $\hat\Omega_n^{1/2}Z+\sqrt nS_n^{-1}\mu(P_n)$. Conditional on the sample, $\sqrt nS_n^{-1}\mu(P_n)$ and $\hat\Omega_n$ are fixed, and
\begin{align}\label{eq:studentized-local-mean-conv}
     \hat\Omega_n^{1/2}Z+\sqrt nS_n^{-1}\mu(P_n)
    \mid(W_1,\ldots,W_n)
    \sim
    N\left(\sqrt nS_n^{-1}\mu(P_n),\hat\Omega_n\right).
\end{align}
Below, $P_Z$ and $\mathbb E_Z$ denote probability and expectation over $Z$ with the original sample held fixed. For each coordinate, 
\begin{align*}
    \frac{\sqrt n\,\mu_j(P_n)}{S_{j,n}}
    &=
    \frac{\sigma_j(P_n)}{S_{j,n}}
    \frac{\sqrt n\,\mu_j(P_n)}{\sigma_j(P_n)}
    \overset{P_n}{\longrightarrow}
    \begin{cases}
        \delta_j,&j\in I,\\
        +\infty,&j\notin I,
    \end{cases}
\end{align*}
by Lemma \ref{alem:consistency-sample-variance} and \eqref{eq:conv-ip-to-Jn-condition2}. We claim that the conditional law of $(\hat\Omega_n^{1/2}Z+\sqrt nS_n^{-1}\mu(P_n))_I$ converges weakly to the law of $Y_I=(\Omega_*^{1/2}Z)_I+\delta_I$ in $P_n$-probability. In other words, by Portmanteau lemma (Lemma 2.2(iii) of \cite{vandervaart2000asymptotic}), for every bounded Lipschitz function $f:\mathbf R^{|I|}\to\mathbf R$,
\begin{align}\label{eq:conditional-conv-in-dist}
    \mathbb E_Z\!\left[
        f\!\left(
            \bigl(
                \hat\Omega_n^{1/2}Z
                +\sqrt nS_n^{-1}\mu(P_n)
            \bigr)_I
        \right)
    \right]
    \overset{P_n}{\longrightarrow}
    \mathbb E_Z[f((\Omega_*^{1/2}Z)_I+\delta_I)].
\end{align}
This holds because for $d_n:=(\sqrt nS_n^{-1}\mu(P_n))_I$,
\begin{align*}
     \mathbb E_Z \left|f((\hat\Omega_n^{1/2}Z)_I+d_n) - f((\Omega_*^{1/2}Z)_I+\delta_I)\right|
     &\overset{(a)}{\leq} L \cdot \mathbb E_Z \|(\hat\Omega_n^{1/2}Z)_I+d_n-(\Omega_*^{1/2}Z)_I-\delta_I\|\\
     &\overset{(b)}{\leq} L \cdot \left\{\mathbb E\|Z\| \cdot\|\hat\Omega_n^{1/2} - \Omega_*^{1/2}\|_{\operatorname{op}} +\|d_n-\delta_I\| \right\} \overset{P_n}{\to} 0
\end{align*}
where (a) follows from the Lipschitz continuity of $f$ with Lipschitz constant $L$; (b) follows from the definition of the operator norm and the triangle inequality; and the last convergence follows from Lemma \ref{alem:consistency-sample-variance} and \eqref{eq:studentized-local-mean-conv} because $\mathbb E_Z\|Z\| <\infty.$ Furthermore, for each $j\notin I$ and every fixed $M\in\mathbf R$,
\begin{align}\label{eq:gaussian-diverge_nonI}
    P_Z\left\{
        (\hat\Omega_n^{1/2}Z)_j
        +\frac{\sqrt n\,\mu_j(P_n)}{S_{j,n}}
        \le M
    \right\}
    &=
    \Phi\left(
        M-\frac{\sqrt n\,\mu_j(P_n)}{S_{j,n}}
    \right)
    \overset{P_n}{\longrightarrow}0.
\end{align}
where the convergence follows from \eqref{eq:studentized-local-mean-conv}.

\medskip\noindent
\textbf{Case 1: All local means have finite limits.} Suppose $I=[k]$. We show that the CDF $J^{\ell}_n(x)$ and the Gaussian reference CDF $J^{\ell}(x, \sqrt{n}S^{-1}_n\mu(P_n), \Omega(\hat{P}_n))$ converge to the same limiting CDF at every real $x$. We first apply the continuous mapping theorem and then verify convergence at zero separately for $\ell=s,q$.

The convergence established in \eqref{eq:subvector-clt-lln} gives $X_n\Rightarrow Y$ and $\hat\Omega_n\overset{P_n}{\longrightarrow}\Omega_*$. The functions $T^m$ and $T^s$ are continuous in their vector arguments. For $\ell=q$, the eigenvalue condition gives $\Omega_*\succ0$,
and $(x,\Omega)\mapsto T^q(x,\Omega)$ is continuous whenever $\Omega\succ0$. Thus, by Slutsky's theorem and the continuous mapping theorem, $T_n^\ell\Rightarrow T^\ell(Y,\Omega_*)$, where the matrix argument is suppressed for $\ell=m,s$. The same argument applies to the Gaussian reference vector $\hat\Omega_n^{1/2}Z+\sqrt nS_n^{-1}\mu(P_n)$, conditionally on the original sample, with convergence in $P_n$-probability as established in \eqref{eq:conditional-conv-in-dist}. Consequently, at every continuity point $x$ of the CDF of
$T^\ell(Y,\Omega_*)$,
\begin{align*}
    J_n^\ell(x)
    \longrightarrow
    P\{T^\ell(Y,\Omega_*)\le x\},\qquad
    J^\ell\left(
        x,\sqrt nS_n^{-1}\mu(P_n),\hat\Omega_n
    \right)
    \overset{P_n}{\longrightarrow}
    P\{T^\ell(Y,\Omega_*)\le x\}.
\end{align*}
For $\ell=m$, Lemma~\ref{lem:strictly_increasing_dist_T} states that the limiting CDF is continuous on $\mathbf R$, so these convergences hold for every real $x$. For $\ell=s,q$, the same lemma gives continuity on $(0,\infty)$. Since these statistics are nonnegative, their CDFs are zero
for $x<0$. It therefore remains to verify convergence at $x=0$, where the limiting distribution may have an atom.

By definition, $T^s(y)=0$ if and only if $y\in\mathbf R_+^k\cup\mathbf R_-^k$. The same equivalence holds for $T^q(y,A)$ whenever $A\succ0$. Hence, for $\ell=s,q$,
\[
    P\{T^\ell(Y,\Omega_*)\le0\}
    =
    P\{T^\ell(Y,\Omega_*)=0\}
    =
    P\{Y\in\mathbf R_+^k\cup\mathbf R_-^k\}.
\]
Since $P\{Y\in\partial(\mathbf R_+^k\cup\mathbf R_-^k)\}=0$, applying the Portmanteau lemma to $X_n\Rightarrow Y$ gives
\[
    J_n^\ell(0)
    =
    P_n\{X_n\in\mathbf R_+^k\cup\mathbf R_-^k\}+o(1)
    \longrightarrow
    P\{Y\in\mathbf R_+^k\cup\mathbf R_-^k\},
    \qquad \ell=s,q,
\]
where the $o(1)$ accounts for the event that the sample variances are zero, whose probability tends to zero. For the Gaussian reference CDF, on the event that all the sample variances are positive,
\[
    J^\ell(
        0,\sqrt nS_n^{-1}\mu(P_n),\hat\Omega_n
    )
    =
    P_Z\{
        \hat\Omega_n^{1/2}Z
        +\sqrt nS_n^{-1}\mu(P_n)
        \in\mathbf R_+^k\cup\mathbf R_-^k
    \}.
\]
Since 
$P\{Y\in\partial(\mathbf R_+^k\cup\mathbf R_-^k)\}=0$, \eqref{eq:conditional-conv-in-dist} and the Portmanteau lemma, disregarding exceptional samples where the sample correlation is not positive definite, give
\[
    J^\ell(
        0,\sqrt nS_n^{-1}\mu(P_n),\hat\Omega_n
    )
    \overset{P_n}{\longrightarrow}
    P\{Y\in\mathbf R_+^k\cup\mathbf R_-^k\},
    \qquad \ell=s,q.
\]
This establishes convergence of both CDFs also at zero.

\medskip\noindent
\textbf{Case 2: Some local means diverge.}
Suppose $\varnothing\ne I\subsetneq[k]$. As in Case 1, we show that the statistics and their Gaussian references have the common limits
\begin{align}\label{eq:partial-divergence-common-limits}
    H_I^m&:=\max_{j\in I}(-Y_j),\qquad
    H_I^s:=\sum_{j\in I} Y_j^2 I\{Y_j\leq 0\},\qquad
    H_I^q:=\inf_{u\in\mathbf R_+^{|I|}}
    (Y_I-u)'(\Omega_{*,I\times I})^{-1}(Y_I-u).
\end{align}
Note that these are the reduced statistics $V^\ell(Y_I,\Omega_{*,I\times I})$ from Lemma~\ref{lem:orthant-distance-reduction}.

Because some elements of $X$ diverge as shown \ref{eq:subvector-clt-lln}, we cannot apply the continuous mapping theorem as we did in Case 1, to derive the asymptotic distribution of the test statistics. To deal with the diverging coordinates, define $h^{\ell}:\mathbf R^{|I|}\times\mathbf R^{k-|I|}
\times\bar{\mathbf O}\mapsto\mathbf R$ by
\[
h^{\ell}(x_I,u,\Omega)
:=
\begin{cases}
T^{\ell}\left(
    x_I,(\tan u_j)_{j\notin I},\Omega
\right),
&u\in\left(-\frac{\pi}{2},\frac{\pi}{2}\right)^{k-|I|},\quad \Omega\in\mathbf O,\\
V^{\ell}(x_I,\Omega_{I\times I}),
&u\notin\left(-\frac{\pi}{2},\frac{\pi}{2}\right)^{k-|I|},\quad \Omega\in\mathbf O,\\
0,
&\Omega\notin\mathbf O.
\end{cases}
\]
where we suppress third argument in $h^{\ell}$ for $\ell=m,s$. For $\ell=m,s$, the matrix argument and the conditions on $\Omega$ are omitted from the first two branches; the final branch applies only to $\ell=q$. The value in the last line is arbitrary.

This map is Borel measurable. Indeed, $\left(-\frac{\pi}{2},\frac{\pi}{2}\right)^{k-|I|}$ and $\mathbf O$ are open sets. The tangent map is continuous
on $\left(-\frac{\pi}{2},\frac{\pi}{2}\right)^{k-|I|}$, the statistics $T^m,T^s$ and the reduced statistics $V^m,V^s$ are continuous on their respective domains, and $T^q$ is continuous on $\mathbf R^k\times\mathbf O$. Moreover, $\Omega\mapsto\Omega_{I\times I}$ is continuous, every principal submatrix of a positive definite matrix is positive definite, and $V^q$ is continuous in its vector and positive definite matrix arguments. Thus each branch is continuous on
a Borel set, so the piecewise-defined maps are Borel measurable.

We next verify continuity at the relevant limiting points. For $\ell\in\{m,s\}$, consider any deterministic sequence satisfying
\[
(x_n)_I\to x_I,
\qquad
u_{j,n}\to\tfrac{\pi}{2},
\quad j\notin I.
\]
Along any subsequence on which $u_n\in\left(-\frac{\pi}{2},\frac{\pi}{2}\right)^{k-|I|}$, $\tan u_{j,n}\to+\infty$ for $j\notin I$, and Lemma~\ref{lem:orthant-distance-reduction} gives
\[
h^\ell((x_n)_I,u_n)=T^\ell\left(
    (x_n)_I,(\tan u_{j,n})_{j\notin I}
\right)
\longrightarrow
V^\ell(x_I).
\]
Along any subsequence on which $u_n\notin\left(-\frac{\pi}{2},\frac{\pi}{2}\right)^{k-|I|}$,
\[
h^\ell((x_n)_I,u_n)
=
V^\ell((x_n)_I)
\longrightarrow
V^\ell(x_I)
\]
by continuity of the reduced statistic. Since both types of
subsequences have the same limit. Hence $h^\ell$ is continuous at every point $\left(x_I,(\pi/2)_{j\notin I}\right).$ For $\ell=q$, suppose in addition that $\Omega_n\to\Omega\succ0. $ Because the positive definite cone is open, $\Omega_n\succ0$ eventually. Along any subsequence on which $u_n\in\left(-\frac{\pi}{2},\frac{\pi}{2}\right)^{k-|I|}$,
Lemma~\ref{lem:orthant-distance-reduction} gives
\[
h^q((x_n)_I,u_n,\Omega_n=T^q\left(
    (x_n)_I,
    (\tan u_{j,n})_{j\notin I},
    \Omega_n
\right)
\longrightarrow
V^q(x_I,\Omega_{I\times I}).
\]
Along any subsequence on which $u_n\notin\left(-\frac{\pi}{2},\frac{\pi}{2}\right)^{k-|I|}$,
\[
h^q((x_n)_I,u_n,\Omega_n)
=
V^q((x_n)_I,(\Omega_n)_{I\times I})
\longrightarrow
V^q(x_I,\Omega_{I\times I}).
\]
Therefore, $h^q((x_n)_I,u_n,\Omega_n)
\longrightarrow
V^q(x_I,\Omega_{I\times I}),$ and so $h^q$ is continuous at every point $(
x_I,(\pi/2)_{j\notin I},\Omega)$ with $\Omega\succ0.$

Next, we transform each coordinate outside $I$ by the arctangent function. Since $\arctan X_{j,n}\overset{P_n}{\to}\frac{\pi}{2}$ for every $j\notin I$, \eqref{eq:subvector-clt-lln} and Slutsky's theorem give
\[
    \left(
        (X_n)_I,\,
        (\arctan X_{j,n})_{j\notin I},\,
        \hat\Omega_n
    \right)
    \Rightarrow
    \left(
        Y_I,\,
        (\pi/2)_{j\notin I},\,
        \Omega_*
    \right).
\]
Since $h^{\ell}(x_I, u, \Omega)$ is continuous at $\left(x_I,(\pi/2)_{j\notin I},\Omega\right)$, we can apply the continuous mapping theorem to have
\[
    T_n^\ell\Rightarrow V^{\ell}(Y_I, \Omega_{*,I\times I})= H_I^\ell.
\]

For the Gaussian reference vector, \eqref{eq:conditional-conv-in-dist} gives conditional weak convergence of its coordinates in $I$ to $Y_I$. Equation~\eqref{eq:gaussian-diverge_nonI} implies that the arctangents of its coordinates outside $I$ converge conditionally to $\pi/2$, in $P_n$-probability. Together with $\hat\Omega_n\overset{P_n}{\longrightarrow}\Omega_*$ in \eqref{eq:subvector-clt-lln}, conditional Slutsky's theorem gives the same joint limit for the transformed inputs. Applying the continuous theorem as described above, conditionally on the original sample, gives the same limiting statistic $H_I^\ell$, in $P_n$-probability. Consequently, at every continuity point $x$ of its CDF,
\begin{gather*}
    J_n^\ell(x)
    \to P\{H_I^\ell\le x\},\qquad
    J^\ell(
        x,\sqrt nS_n^{-1}\mu(P_n),\hat\Omega_n
    )
    \overset{P_n}{\to}P\{H_I^\ell\le x\}.
\end{gather*}

For $|I|\ge2$, Lemma~\ref{lem:strictly_increasing_dist_T}, applied to the one-sided statistics on the coordinates in $I$, gives continuity of the CDF of $H_I^m$ on $\mathbf R$ and
of the CDFs of $H_I^s,H_I^q$ on $(0,\infty)$. If $I=\{j\}$, the same conclusions follow directly from
\[
    H_I^m=-Y_j,
    \qquad
    H_I^s=H_I^q=(Y_j)_-^2,
    \qquad Y_j\sim N(\delta_j,1).
\]
Therefore, the two CDF convergences hold for every real $x$ when $\ell=m$ and every positive $x$ when $\ell=s,q$. For $\ell=s,q$, the CDFs are zero for $x<0$. It remains to verify convergence at zero.

For $\ell=s,q$, the full statistic is zero precisely when its vector argument belongs to
$\mathbf R_+^k\cup\mathbf R_-^k$, apart from exceptional sample of vanishing probability, in which the sample correlation is not positive definite. When every coordinate outside $I$ is strictly positive with probability approaching to 1, membership in this union is equivalent to the coordinates
in $I$ being nonnegative. Hence
\[
    |J_n^\ell(0)-P_n\{(X_n)_I\in\mathbf R_+^{|I|}\}|
    \le
    \textstyle\sum_{j\notin I}P_n\{X_{j,n}\le0\}+o(1)
    \longrightarrow0.
\]
where the $o(1)$ accounts for the event that the sample variances are zero, whose probability tends to zero. Moreover, since $P\{Y_I\in\partial\mathbf R_+^{|I|}\}=0$, the Portmanteau lemma therefore gives
\[
    J_n^\ell(0)
    \longrightarrow
    P\{Y_I\in\mathbf R_+^{|I|}\}
    =
    P\{H_I^\ell=0\},
    \qquad \ell=s,q.
\]
Similarly, we can show 
\begin{align*}
    J^\ell(
        0,\sqrt nS_n^{-1}\mu(P_n),\hat\Omega_n
    )
    =P\{Y_I\in\mathbf R_+^{|I|}\} + o_{P_n}(1)
    =P\{H_I^\ell=0\} + o_{P_n}(1),
    \qquad \ell=s,q.
\end{align*}
This establishes convergence of both CDFs at zero, including when $\Omega_*$ is singular for $\ell=s$.

Let $F_I^\ell$ denote the limiting CDF in either case. Lemma~\ref{lem:CDF-quantile-transfer} now gives
\begin{gather}\label{eq:sampling-Gaussian-common-CDF}
    \sup_x|J_n^\ell(x)-F_I^\ell(x)|
    \longrightarrow0,\qquad
    \sup_x|
        J^\ell(
            x,\sqrt nS_n^{-1}\mu(P_n),\Omega(\hat P_n)
        )
        -F_I^\ell(x)
    |
    \overset{P_n}{\longrightarrow}0.
\end{gather}
For $q$, matrices on exceptional samples use the stated regularization; these samples have probability tending to zero. The triangle inequality proves \eqref{eq:conv-ip-to-Jn}.
\end{proof}

The next lemma is the nonparametric bootstrap counterpart of Lemma~\ref{lem:conv-ip-to-Jn}. With the bootstrap notation of Lemma~\ref{lem:centered-sampling-bootstrap}, its reference CDF is
\[
 J_n^\ell(x,\theta_1,\hat P_n)
 =\hat P_n^*\left\{
 T^\ell\left(\sqrt nS_n^{*-1}(\bar W_n^*-\bar W_n)
                   +S_n^{*-1}\theta_1,\hat\Omega_n^*\right)
 \le x\right\}.
\]
The second argument of $T^\ell$ is suppressed for $\ell=m,s$. For $q$, the bootstrap matrix uses the regularization in Remark~\ref{remark:regularized-correlation}.

\begin{lemma}\label{lem:bootstrap-conv-to-Jn}
Consider a sequence $\{P_n\in \mathbf{P}: n \geq 1\}$ where $\mathbf{P}$ is a set of distributions on $\mathbf{R}^k$ satisfying \eqref{def:uniform-integrability}. Let $W_i$, $i=1,2,...,n$, be an i.i.d. sequence of random vectors with distribution $P_n$. Suppose for some non-empty set $I\subset \{1,2, ..., k\}$ 
    \begin{align*}
        \frac{\sqrt{n} \mu_j(P_n)}{ \sigma_j(P_n)} \to \delta_j \in(-\infty, 0] \text{ for } j\in I\text{ and }\frac{\sqrt{n} \mu_j(P_n)}{ \sigma_j(P_n)} \to -\infty \text{ for }j \not\in I
    \end{align*}
    or
    \begin{align}\label{eq:bootstrap-drift-negative}
        \frac{\sqrt{n} \mu_j(P_n)}{ \sigma_j(P_n)} \to \delta_j \in [0,\infty) \text{ for } j\in I\text{ and }\frac{\sqrt{n} \mu_j(P_n)}{ \sigma_j(P_n)} \to \infty \text{ for }j \not\in I.
    \end{align}
Then, for $\ell\in\{m,s\}$, $J^\ell_{n}$ in \eqref{def:J_n-section3} satisfies
\begin{align}
\sup_{x\in\mathbf R}
\left|
J_n^\ell(x,\sqrt n\,\mu(P_n),\hat P_n)
-
P_n\{T_n^\ell\le x\}
\right|
\overset{P_n}{\longrightarrow}0.
\label{eq:bootstrap-conv-to-Jn}
\end{align}
Under Assumption \ref{ass:uniform-eigenvalue}, \eqref{eq:bootstrap-conv-to-Jn} also holds for $\ell=q$.
\end{lemma}

\begin{proof}
We follow the strategy of the proof of Lemma~\ref{lem:conv-ip-to-Jn}. From any subsequence, extract a further subsequence along which $\Omega(P_n)\to\Omega_*\in\bar{\mathbf O}$. For $\ell=q$, Assumption~\ref{ass:uniform-eigenvalue} gives $\Omega_*\succ0$. It is enough to prove the result along every such subsequence. We consider the positive-orthant version in \eqref{eq:bootstrap-drift-negative}; the negative-orthant case
follows by reversing signs.
 
On bootstrap samples for which all bootstrap sample variances are positive, write
\[
X_n^*
:=
\sqrt nS_n^{*-1}(\bar W_n^*-\bar W_n)
+
\sqrt nS_n^{*-1}\mu(P_n).
\]
Define $X_n^*$ arbitrarily on the remaining bootstrap samples. Their conditional probability converges to zero in $P_n$-probability by Lemma~\ref{lem:centered-sampling-bootstrap}. This is the vector entering $J_n^\ell(x,\sqrt n\mu(P_n),\hat P_n)$. For $\ell=q$, the regularized correlation matrix from
Remark~\ref{remark:regularized-correlation} is used.

Lemma~\ref{lem:centered-sampling-bootstrap} gives the conditional centered CLT, conditional variance consistency, and conditional correlation consistency. In particular, \[
    \frac{\sqrt n\,\mu_j(P_n)}{S_{j,n}^*}
    =
    \frac{\sigma_j(P_n)}{S_{j,n}^*}
    \frac{\sqrt n\,\mu_j(P_n)}{\sigma_j(P_n)} \overset{\hat P_n^*}{\longrightarrow}
    \left\{\begin{matrix}
        \delta_j, & j\in I\\
        +\infty,, & j\notin I
    \end{matrix}\right.
    \qquad \text{in $P_n$-probability.}
\]
The conditional Slutsky theorem \citep[Theorem~3.2]{XiongLi2008} furthermore gives the joint convergence
\begin{align}\label{eq:bootstrap-finite-coordinate-limit}
\left((X_n^*)_I,\hat\Omega_n^*\right)
\rightsquigarrow_{\hat P_n^*}
\left(Y_I,\Omega_*\right)
\quad\text{in $P_n$-probability},
\end{align}
where $Y_I:=(\Omega_*^{1/2}Z)_I+\delta_I.$ 

The coordinates outside $I$ have the same divergent behavior as in the proof of Lemma~\ref{lem:conv-ip-to-Jn}. To see this, note that, for every $K>0$ and fixed $M\in\mathbf R$,
\[
\{X_{j,n}^*\le M\}
\subseteq
\{\left|\tfrac{\sqrt n(\bar W_{j,n}^*-\bar W_{j,n})}{S_{j,n}^*}\right|>K\}
\cup
\{\tfrac{\sqrt n\,\mu_j(P_n)}{S_{j,n}^*}\le M+K\}.
\]
It follows that
\[
\hat P_n^*\{X_{j,n}^*\le M\}
\le
\hat P_n^*\{\left|\tfrac{\sqrt n(\bar W_{j,n}^*-\bar W_{j,n})}{S_{j,n}^*} \right|>K\}
+
\hat P_n^*\{\tfrac{\sqrt n\,\mu_j(P_n)}{S_{j,n}^*}\le M+K\}.
\]
The conditional centered CLT implies conditional tightness of
${\sqrt n(\bar W_{j,n}^*-\bar W_{j,n})}/{S_{j,n}^*}$, while the preceding drift convergence implies
that the second term converges to zero in $P_n$-probability
for every fixed $K$. First choosing $K$ sufficiently large and
then letting $n\to\infty$ gives
\begin{align}\label{eq:bootstrap-divergent-coordinate}
\hat P_n^*\{X_{j,n}^*\le M\}
\overset{P_n}{\longrightarrow}0,
\qquad j\notin I.
\end{align}
Equivalently, $\arctan X_{j,n}^*
\overset{\hat P_n^*}{\longrightarrow}\frac{\pi}{2}$ in $P_n$-probability for $j\notin I.$ 

We now apply the two cases in the proof of
Lemma~\ref{lem:conv-ip-to-Jn}.

\noindent\textbf{Case 1:} Suppose that $I=[k]$. The joint convergence in
\eqref{eq:bootstrap-finite-coordinate-limit} and the bootstrap continuous-mapping theorem \citep[Theorem~10.8]{Kosorok2008} give
\[
T^\ell(X_n^*,\hat\Omega_n^*)
\rightsquigarrow_{\hat P_n^*}
T^\ell(Y,\Omega_*)
\quad\text{in $P_n$-probability},
\]
where the matrix argument is suppressed for $\ell=m,s$. For $\ell=q$, the required continuity holds because $\Omega_*\succ0$.

\noindent\textbf{Case 2:}
Suppose next that $\varnothing\ne I\subsetneq[k]$.
For every $j\notin I$,
\eqref{eq:bootstrap-divergent-coordinate} implies $\arctan X_{j,n}^*
\overset{\hat P_n^*}{\longrightarrow}
\tfrac{\pi}{2}$ in $P_n$-probability. Consequently, the conditional Slutsky theorem
\citep[Theorem~3.2]{XiongLi2008}, applied to
\eqref{eq:bootstrap-finite-coordinate-limit} and
\eqref{eq:bootstrap-divergent-coordinate}, gives
\[
\left(
    (X_n^*)_I,\,
    (\arctan X_{j,n}^*)_{j\notin I},\,
    \hat\Omega_n^*
\right)
\rightsquigarrow_{\hat P_n^*}
\left(
    Y_I,\,
    (\pi/2)_{j\notin I},\,
    \Omega_*
\right)
\quad\text{in $P_n$-probability}.
\]
By its construction in Case~2 of
Lemma~\ref{lem:conv-ip-to-Jn}, $h^\ell$ is a Borel measurable map on the ambient finite-dimensional Euclidean space. Lemma~\ref{lem:orthant-distance-reduction} shows that it is continuous at every point $\left(x_I,(\pi/2)_{j\notin I},\Omega_* \right)$; for $\ell=q$, this uses $\Omega_*\succ0$. The conditional
continuous-mapping theorem \citep[Theorem~10.8]{Kosorok2008} therefore gives
\[
h^\ell\left(
    (X_n^*)_I,\,
    (\arctan X_{j,n}^*)_{j\notin I},\,
    \hat\Omega_n^*
\right)
\rightsquigarrow_{\hat P_n^*}
h^\ell\left(
    Y_I,\,
    (\pi/2)_{j\notin I},\,
    \Omega_*
\right)
=
H_I^\ell
\]
in $P_n$-probability, where $H_I^\ell$ is defined in
\eqref{eq:partial-divergence-common-limits}. The matrix
arguments are suppressed for $\ell=m,s$.

On the event $E_n^*$ of nonexceptional bootstrap samples, the
definition of $h^\ell$ gives
\[
h^\ell\left(
    (X_n^*)_I,\,
    (\arctan X_{j,n}^*)_{j\notin I},\,
    \hat\Omega_n^*
\right)
=
T^\ell(X_n^*,\hat\Omega_n^*).
\]
Moreover, $\hat P_n^*(E_n^{*c})
\overset{P_n}{\longrightarrow}0.$ Therefore, the conditional laws of the two statistics differ in bounded-Lipschitz distance by at most
$2\hat P_n^*(E_n^{*c})$, which converges to zero in
$P_n$-probability. It follows that
\[
T^\ell(X_n^*,\hat\Omega_n^*)
\rightsquigarrow_{\hat P_n^*}
H_I^\ell
\quad\text{in $P_n$-probability}.
\]

Let $F_I^\ell$ denote the limiting CDF used in
\eqref{eq:sampling-Gaussian-common-CDF}. The preceding
conditional weak convergence implies
\[
J_n^\ell(x,\sqrt n\mu(P_n),\hat P_n)
\overset{P_n}{\longrightarrow}
F_I^\ell(x)
\]
at every continuity point $x$ of $F_I^\ell$. More explicitly, from every subsequence one can extract a further subsequence along which the relevant conditional bounded-Lipschitz distance converges to zero almost surely. The Portmanteau lemma \citep[Lemma~2.2]{vandervaart2000asymptotic} then applies
pathwise to the conditional laws, and the subsequence characterization of convergence in probability gives the displayed conclusion.

The continuity properties established in the proof of
Lemma~\ref{lem:conv-ip-to-Jn} give convergence at every
$x\in\mathbf R$ when $\ell=m$ and every $x>0$ when
$\ell=s,q$. For $\ell=s,q$, the CDFs are zero for $x<0$, so it
remains to verify convergence at zero. To this end, let
\[
R_n^*
:=
\hat P_n^*\left\{\min_{1\le j\le k}S_{j,n}^*=0\right\}
+
I\{\ell=q\}
\hat P_n^*\{\hat\Omega_n^*\not\succ0\}\overset{P_n}{\longrightarrow}0
\]
where the convergence follows from Lemma~\ref{lem:centered-sampling-bootstrap} and $\Omega_*\succ0$. If $I=[k]$, the zero-set identity gives
\[
\left|
J_n^\ell(0,\sqrt n\mu(P_n),\hat P_n)
-
\hat P_n^*\{
X_n^*\in\mathbf R_+^k\cup\mathbf R_-^k
\}
\right|
\le R_n^*.
\]
The boundary of $\mathbf R_+^k\cup\mathbf R_-^k$ is contained in $\bigcup_{j=1}^k\{y:y_j=0\}$ and therefore has probability zero under the law of $Y$. Applying the Portmanteau lemma almost surely along further subsequences gives
\[
J_n^\ell(0,\sqrt n\mu(P_n),\hat P_n)
\overset{P_n}{\longrightarrow}
P\{Y\in\mathbf R_+^k\cup\mathbf R_-^k\}
=
F_I^\ell(0).
\]
If $\varnothing\ne I\subsetneq[k]$, then
\begin{align*}
&
\left|
J_n^\ell(0,\sqrt n\mu(P_n),\hat P_n)
-
\hat P_n^*\{(X_n^*)_I\in\mathbf R_+^{|I|}\}
\right|
\le
\textstyle\sum_{j\notin I}
\hat P_n^*\{X_{j,n}^*\le0\}
+
R_n^*
\overset{P_n}{\longrightarrow}0,
\end{align*}
where the convergence follows from
\eqref{eq:bootstrap-divergent-coordinate}. Moreover, since $P\{Y_I\in\partial\mathbf R_+^{|I|}\}=0$, applying the same Portmanteau argument to \eqref{eq:bootstrap-finite-coordinate-limit} gives
\[
\hat P_n^*\{(X_n^*)_I\in\mathbf R_+^{|I|}\}
\overset{P_n}{\longrightarrow}
P\{Y_I\in\mathbf R_+^{|I|}\}.
\]
Consequently,
\[
J_n^\ell(0,\sqrt n\mu(P_n),\hat P_n)
\overset{P_n}{\longrightarrow}
P\{Y_I\in\mathbf R_+^{|I|}\}
=
F_I^\ell(0).
\]

We have therefore established pointwise convergence at every
real $x$. Lemma~\ref{lem:CDF-quantile-transfer} gives
\begin{align}\label{eq:bootstrap-to-common-limit-bootstrap-proof}
\sup_{x\in\mathbf R}
\left|
J_n^\ell(x,\sqrt n\mu(P_n),\hat P_n)
-
F_I^\ell(x)
\right|
\overset{P_n}{\longrightarrow}0.
\end{align}
The sampling convergence in
\eqref{eq:sampling-Gaussian-common-CDF} and the triangle
inequality now give
\begin{align*}
&
\sup_{x\in\mathbf R}
\left|
J_n^\ell(x,\sqrt n\mu(P_n),\hat P_n)
-
P_n\{T_n^\ell\le x\}
\right|
\overset{P_n}{\longrightarrow}0.
\end{align*}
Since the initial subsequence was arbitrary, this proves
\eqref{eq:bootstrap-conv-to-Jn}. The negative-orthant version
follows by the same argument with all signs reversed.
\end{proof}

The next lemma establishes coverage of the first-step confidence regions used by the two-step tests. Singular bootstrap Wald draws retain the value $+\infty$ specified in the procedure.

\begin{lemma}\label{lem:two-step-first-step-coverage}
Suppose Assumptions~\ref{ass:uniform-integrability} and
\ref{ass:uniform-eigenvalue} hold. The parametric and
nonparametric confidence regions in
\eqref{def:two-step-normal-CS} and
\eqref{def:two-step-bootstrap-CS} satisfy
\begin{align}\label{eq:two-step-first-step-coverage}
\begin{aligned}
\limsup_{n\to\infty}\sup_{P\in\mathbf P}
 P\{\mu(P)\notin\mathbf C_n(1-\tau)\}&\le\tau,\\
\limsup_{n\to\infty}\sup_{P\in\mathbf P}
 P\{\mu(P)\notin\mathbf C_n^*(1-\tau)\}&\le\tau.
\end{aligned}
\end{align}
Moreover, for every sequence $\{P_n\in\mathbf P:n\ge1\}$, $r_{n,\tau}^*$ in \eqref{def:two-step-bootstrap-radius} satisfies
\begin{align}\label{eq:two-step-bootstrap-radius-limit}
 r_{n,\tau}^*\overset{P_n}{\longrightarrow}\chi^2_{k,1-\tau}.
\end{align}
\end{lemma}

\begin{proof}
Consider any sequence $\{P_n\in\mathbf P:n\ge1\}$.
From any subsequence, extract a further subsequence on which $\Omega(P_n)\to\Omega_*\succ0$.

On original samples with $\hat\Sigma_n\succ0$, define the Wald statistic below, and set $R_n=0$ otherwise.
Lemma~\ref{lem:centered-sampling-bootstrap}, together with
Slutsky's theorem and the continuous-mapping theorem, gives
\begin{align}\label{eq:two-step-Wald-chi-square}
R_n:=n(\bar W_n-\mu(P_n))'\hat\Sigma_n^{-1}
          (\bar W_n-\mu(P_n))
\Rightarrow\chi_k^2.
\end{align}
The original sample covariance is positive definite with
probability tending to one by marginal variance consistency
and Lemma~\ref{lem:sample-corr-positive-definite}.
Continuity of the chi-squared CDF therefore gives
\begin{align*}
P_n\{\mu(P_n)\notin\mathbf C_n(1-\tau)\}
\longrightarrow\tau.
\end{align*}

The bootstrap conclusions of
Lemma~\ref{lem:centered-sampling-bootstrap}, together with
conditional Slutsky's theorem
\citep[Theorem~3.2]{XiongLi2008} and the conditional
continuous-mapping theorem
\citep[Theorem~10.8]{Kosorok2008} applied to
$(u,A)\mapsto u'A^{-1}u$, give
\[
n(\bar W_n^*-\bar W_n)'(\hat\Sigma_n^*)^{-1}
  (\bar W_n^*-\bar W_n)
\rightsquigarrow_{\hat P_n^*}\chi_k^2
\quad\text{in $P_n$-probability}.
\]
The conditional probability of a singular bootstrap covariance
tends to zero, so retaining those draws at $+\infty$ does not
affect convergence of the conditional CDF at any real argument.
From every further subsequence, one can extract a subsequence
along which this CDF convergence holds almost surely at every
real argument. Continuity and strict increase of the limiting
CDF at $\chi^2_{k,1-\tau}$ then imply almost-sure convergence
of the corresponding $(1-\tau)$-quantiles. The subsequence
characterization of convergence in probability therefore gives
\eqref{eq:two-step-bootstrap-radius-limit} along the subsequence
under consideration. Combining it with
\eqref{eq:two-step-Wald-chi-square} gives
\[
P_n\{\mu(P_n)\notin\mathbf C_n^*(1-\tau)\}
\longrightarrow\tau.
\]

Every subsequence of the arbitrary sequence $\{P_n\}$ thus
admits a further subsequence satisfying these conclusions.
The subsequence characterization of convergence gives
\eqref{eq:two-step-bootstrap-radius-limit} and convergence
of both noncoverage probabilities to $\tau$ along the
full sequence.

To obtain the uniform coverage bounds, suppose that either
bound fails. Then there exist $\eta>0$, sample sizes
$n_j\to\infty$, and distributions $P_{n_j}\in\mathbf P$
such that the corresponding noncoverage probabilities exceed
$\tau+\eta$. This contradicts their convergence to $\tau$
along every such sequence, proving
\eqref{eq:two-step-first-step-coverage}.
\end{proof}

\begin{lemma}\label{lem:sample-corr-positive-definite}
Consider any sequence {$\{P_n\in\mathbf P:n\ge1\}$}, where {$\mathbf P$} satisfies {Assumptions \ref{ass:uniform-integrability} and \ref{ass:uniform-eigenvalue}}. Then,
{$$
P_n\{\hat\Omega_n\succ0\}\to1.
$$}
\end{lemma}
\begin{proof}
Let $\lambda^*
:=
\inf_{P\in\mathbf P}\lambda_{\min}(\Omega(P)).$ By Lemma \ref{alem:consistency-sample-corr},
$
\|\hat\Omega_n-\Omega(P_n)\|
\overset{P_n}{\to}0,
$
where $\|\cdot\|$ is the componentwise maximum norm. Applying Weyl's inequality \citep[Corollary~III.2.6]{Bhatia1997} to $\hat \Omega_n$ and $\Omega(P_n)$ gives
$$\lambda_{\min}(\hat\Omega_n)
\ge
\lambda_{\min}(\Omega(P_n))
-
\|\hat\Omega_n-\Omega(P_n)\|_{\mathrm{op}} \ge
\lambda^*
-
k\|\hat\Omega_n-\Omega(P_n)\|$$
where the second inequality holds because $|A|_{\mathrm{op}}\le k|A|$ for every $A\in\mathbf R^{k\times k}$ \citep[Section~I.2, p.~7]{Bhatia1997}. Since $\{\lambda_{\min}(\hat\Omega_n)>\lambda^*/2\} \subseteq\{\hat\Omega_n\succ0\},$ it holds that
$$
P_n\{\hat\Omega_n\not\succ0\}
\le
P_n\{\lambda_{\min}(\hat\Omega_n)\leq \lambda^*/2\}
\le
P_n\{
\|\hat\Omega_n-\Omega(P_n)\|
\ge{\lambda^*}/{2k}
\}
\to0.
$$
\end{proof}

\begin{lemma}\label{lem:strictly_increasing_dist_T}
Let {$k\ge2$, $\theta_1\in\mathbf R^k$, and $Z\sim N(0_k,I_k)$}.
\begin{enumerate}[i)]
\item {For $\theta_2\in\bar{\mathbf O}$,
$J^m(\cdot,\theta_1,\theta_2)$ is continuous and strictly
increasing on the interior} of {its support}.
\item {For $\ell=s$ and $\theta_2\in\bar{\mathbf O}$}, {and for
$\ell=q$ and $\theta_2\in\mathbf O$,} the {CDF of
$T^\ell(\theta_2^{1/2}Z+\theta_1,\theta_2)$} is {continuous on
$(0,\infty)$ and strictly increasing on} the interior of its
support. It may have an atom at zero, including a distribution
concentrated at zero. The matrix argument is omitted for $s$.
\item {The same conclusions hold for} the corresponding one-sided
statistics $\max_j(-x_j)$, $\sum_j(x_j)_-^2$, and
$\inf_{u\in\mathbf R_+^k}(x-u)'\theta_2^{-1}(x-u)$.
\end{enumerate}
\end{lemma}

\begin{proof}
Write $\theta_2=AA'$ where $A\in\mathbf R^{k\times r}$ has rank $r$, and let $\xi\sim N(0_r,I_r)$. Thus $\theta_2^{1/2}Z\overset{d}{=}A\xi$, and every row $a_j'$ of $A$ has norm one. The function $T^m(A\xi+\theta_1)$ is
continuous and takes, on finitely many pieces, one of the affine
forms $a_j'\xi+\theta_{1,j}$ or its negative. None is constant.
Each level set is therefore contained in finitely many hyperplanes and has Lebesgue measure zero. The same argument applies to the one-sided maximum, proving continuity of both CDFs.

For $s$, each strictly positive value comes from one of finitely
many quadratic polynomials of the form $\sum_{j\in A_0}(a_j'\xi+\theta_{1,j})^2$ with $A_0\ne\varnothing$. Its quadratic part is nonzero because every selected row has norm one, i.e. $\|a_j\|^2=(AA')_{jj}=(\theta_2)_{jj}=1$. Consequently every positive level set is contained in a finite union of zero sets of nonconstant polynomials and has Lebesgue measure zero. This argument applies to both the two-sided minimum and the one-sided sum. 

For $q$, positive definiteness gives $r=k$. Projection on each orthant is piecewise quadratic: on each active face its positive value is $x_{A_0}'(\theta_{2,A_0\times A_0})^{-1}x_{A_0}$ for a nonempty active set $A_0$. These are nonconstant quadratic polynomials after the invertible transformation $x=A\xi+\theta_1$. The same finite-union argument proves the absence of positive atoms for the one-sided distance and their two-sided minimum.

Finally, each statistic is a continuous function of $\xi$. Its image is an interval because $\mathbf R^r$ is connected, and its support is the closure of this image. Any open interval meeting the interior of the support has a nonempty open inverse image. The strictly positive Gaussian density assigns this inverse image positive probability. This proves
strict increase on the interior of the support. 
\end{proof}

\subsection{Proof of Corollary
\ref{cor:LFCtest-uniform-validity-equality} and its related lemmas}

\begin{proof}[Proof of Corollary \ref{cor:LFCtest-uniform-validity-equality}]
By Theorem~\ref{thm:LFtest-uniform-validity},
\[
\limsup_{n\to\infty}\sup_{P\in\mathbf P_0}
\mathbb E_P[\phi_n^\ell]\le\alpha
\]
for $\ell\in\{m,s\}$, and also for $\ell=q$ under
Assumption~\ref{ass:uniform-eigenvalue}. It remains to establish the
reverse inequality.

Let $Q$ be the mean-zero distribution in the corollary, and put
$D_Q=\operatorname{diag}(\sigma_1(Q),\ldots,\sigma_k(Q))$. Its location
family satisfies
\begin{align}\label{ass:location-submodel}
\mathcal L(X+D_Qa)\in\mathbf P
\qquad\text{for every }a\in\mathbf R^k,\quad X\sim Q.
\end{align}
Let $X_1,X_2,\ldots$ be i.i.d. with distribution $Q$, and let $P_n$
be the distribution of
\[
W_{i,n}:=X_i+\frac{D_Q\vartheta_n}{\sqrt n},
\qquad
\vartheta_n:=(n^{1/4},0,\ldots,0)'.
\]
By \eqref{ass:location-submodel}, $P_n\in\mathbf P$ for every $n$.
Moreover,
$\mu(P_n)=D_Q\vartheta_n/\sqrt n\in\mathbf R_+^k$, so
$P_n\in\mathbf P_0$. Since translation does not affect variances or correlations, $D(P_n)=D_Q,$  $\Omega(P_n)=I_k,$ and $\sqrt nD(P_n)^{-1}\mu(P_n)=\vartheta_n.$ Thus the first coordinate of the standardized local mean diverges to $+\infty$, while the remaining coordinates are zero. Let $\bar X_n$ denote the sample mean of $X_1,\ldots,X_n$.
The sample variances and correlations computed from the $W_{i,n}$
are the same as those computed from the $X_i$. Hence Lemma~\ref{lem:centered-sampling-bootstrap} gives
\[
S_n^{-1}D_Q\overset{P_n}{\longrightarrow}I_k,\qquad
\hat\Omega_n\overset{P_n}{\longrightarrow}I_k,\qquad
D_Q^{-1}\sqrt n\bar X_n\Rightarrow Z,
\quad Z\sim N(0_k,I_k).
\]
Since $\sqrt nS_n^{-1}\bar W_n
=
S_n^{-1}D_Q
\bigl(D_Q^{-1}\sqrt n\bar X_n+\vartheta_n\bigr),$ it follows that
\[
\frac{\sqrt n\bar W_{1,n}}{S_{1,n}}
\overset{P_n}{\longrightarrow}+\infty,
\qquad
\left(\frac{\sqrt n\bar W_{j,n}}{S_{j,n}}\right)_{j=2}^k
\Rightarrow (Z_j)_{j=2}^k.
\]

Define $V_{k-1}:=\sum_{j=2}^k Z_j^2I\{Z_j\le0\}$. The preceding convergences and Lemma~\ref{lem:orthant-distance-reduction}, retaining coordinates
$2,\ldots,k$, give
\begin{align}\label{location-submodel-stat-convergence}
T_n^m\Rightarrow\max_{2\le j\le k}(-Z_j),\qquad
T_n^s,\ T_n^q&\Rightarrow V_{k-1}.
\end{align}
For $m$ and $s$, the first coordinate makes the positive branch
diverge, while its contribution to the negative branch vanishes.
For $q$, the lemma shows that the distance to $\mathbf R_-^k$
diverges and the distance to $\mathbf R_+^k$ reduces to the squared
distance of $(Z_2,\ldots,Z_k)'$ from $\mathbf R_+^{k-1}$.
Moreover, $\hat\Omega_n\to I_k$ implies
$P_n\{\hat\Omega_n\succ0\}\to1$, so the convention for singular
sample correlation matrices does not affect this limit. By
independence of the coordinates of $Z$,
\[
\mathcal L(V_{k-1})
=
\sum_{r=0}^{k-1}
\binom{k-1}{r}2^{-(k-1)}\mathcal L(\chi_r^2),
\qquad \chi_0^2:=0.
\]

By Lemma~\ref{lem:continuity-critical-value-all},
$\hat\Omega_n\to I_k$ gives
\begin{align}\label{eq:location-submodel-cv-convergence}
\hat c_n^\ell
\overset{P_n}{\longrightarrow}
c^\ell(1-\alpha,I_k).
\end{align}
For $q$, this also uses the fact that $\hat\Omega_n$ is positive
definite with probability approaching one. By
\eqref{def:LF-cv-independent}, the limit in
\eqref{eq:location-submodel-cv-convergence} is the
$(1-\alpha)$-quantile of the corresponding limiting statistic in
\eqref{location-submodel-stat-convergence}. In particular, $c^m(1-\alpha,I_k)
=
\Phi^{-1}\bigl((1-\alpha)^{1/(k-1)}\bigr).$ Also,
$P\{V_{k-1}=0\}=2^{-(k-1)}\le\tfrac12<1-\alpha$, so the
$(1-\alpha)$-quantile of $V_{k-1}$ is strictly positive.
Each limiting CDF is continuous at its critical value and takes
value $1-\alpha$ there. Combining
\eqref{location-submodel-stat-convergence} and
\eqref{eq:location-submodel-cv-convergence} therefore gives
\[
P_n\{T_n^\ell>\hat c_n^\ell\}\to\alpha.
\]
Since $P_n\in\mathbf P_0$ for every $n$,
\[
\limsup_{n\to\infty}\sup_{P\in\mathbf P_0}
\mathbb E_P[\phi_n^\ell]
\ge
\lim_{n\to\infty}\mathbb E_{P_n}[\phi_n^\ell]
=\alpha.
\]
Together with the upper bound, this proves the result.
\end{proof}

\begin{lemma}\label{lem:parametric-cv-existence}
Consider the parametric least favorable critical value in \eqref{def:LF-cv-normal} and fix $\alpha\in(0,1)$. For $\ell\in\{m,s\}$ and any positive semidefinitecorrelation matrix
$\theta_2\in\bar{\mathbf O}$, $c^\ell(1-\alpha,\theta_2)$ is well-defined.
Moreover,
\begin{align}\label{eq:cv-upperbound1}
    c^m(1-\alpha,\theta_2)
\le
\Phi^{-1}\left(1-\frac{\alpha}{k}\right),\quad 
c^s(1-\alpha,\theta_2)
\le
k\left\{\Phi^{-1}\left(1-\frac{\alpha}{2k}\right)\right\}^2.
\end{align}
For $\ell=q$, the same conclusion holds for any $\theta_2\in\mathbf O$. In particular,
\begin{align}\label{eq:cv-upperbound2}
c^q(1-\alpha,\theta_2)
\le
\chi^2_{k,1-\alpha},
\end{align}
where $\chi^2_{k,1-\alpha}$ denotes the $(1-\alpha)$-quantile of the
chi-squared distribution with $k$ degrees of freedom.
\end{lemma}

\begin{proof}
Fix $\theta_1\in\mathbf R^k_-$ and
$\theta_2\in\bar{\mathbf O}$. Since
$T^\ell(\theta_2^{1/2}Z+\theta_1,\theta_2)$ is real-valued for
$\ell \in \{m,s\}$, $J^\ell(\cdot,\theta_1,\theta_2)$ in \eqref{def:J} is a proper CDF and
$(J^\ell)^{-1}(1-\alpha,\theta_1,\theta_2)=\inf\{x\in\mathbf{R}: P\{T^{\ell}(\theta_2^{1/2}Z +\theta_1, \theta_2)\leq x\}\geq 1-\alpha\}$ is finite. Moreover, the set over
which the supremum is taken in \eqref{def:LF-cv-normal} is non-empty since
$0_k\in\mathbf R^k_+\cup\mathbf R^k_-$. It remains to show that the quantiles are uniformly bounded from above.

Consider $\ell=m$. For any
$x\in\mathbf R$,
\begin{align*}
1-J^m(x,\theta_1,\theta_2)
&=
P\{T^m(\theta_2^{1/2}Z+\theta_1)>x\}
\overset{(a)}{\le}
P\left\{\max_{1\le j\le k}(\theta_2^{1/2}Z+\theta_1)_j>x\right\}
\\
&\overset{(b)}{\le}
P\left\{\max_{1\le j\le k}(\theta_2^{1/2}Z)_j>x\right\}  
\overset{(c)}{\le}
\sum_{j=1}^k P\{(\theta_2^{1/2}Z)_j>x\}
=
k\{1-\Phi(x)\}.
\end{align*}
where (a) holds because
$T^m(y)\le \max_{1\le j\le k}y_j$ for any $y\in\mathbf{R}^k$; (b) holds because
$\theta_1\in\mathbf R_-^k$, so
$(\theta_2^{1/2}Z+\theta_1)_j\le(\theta_2^{1/2}Z)_j$ for all $j$; (c) is the union bound. The last equality holds because
$Var((\theta_2^{1/2}Z)_j)=(\theta_2)_{jj}=1$, and hence
$(\theta_2^{1/2}Z)_j\sim N(0,1)$. Taking $x=\Phi^{-1}(1-\alpha/k)$ gives $J^m(\Phi^{-1}(1-\frac{\alpha}{k}),
\theta_1,\theta_2)
\ge 1-\alpha.$
Therefore, for all $\theta_1\in\mathbf{R}^k_-$, it holds $(J^m)^{-1}(1-\alpha,\theta_1,\theta_2)
\le
\Phi^{-1}\left(1-\frac{\alpha}{k}\right)$ and the upper bound of $c^m(1-\alpha,\theta_2)$ in \eqref{eq:cv-upperbound1} holds.

Next, let $\ell=s$. For any $x\ge0$,
\begin{align*}
1-J^s(x,\theta_1,\theta_2)
&\overset{(a)}{\le}
P\left\{
\sum_{j=1}^k
(\theta_2^{1/2}Z+\theta_1)_j^2
I\{(\theta_2^{1/2}Z+\theta_1)_j\ge0\}
>x
\right\} \\
&\overset{(b)}{\le}
P\left\{
\sum_{j=1}^k
\left((\theta_2^{1/2}Z)_j\right)^2
>x
\right\} 
\overset{(c)}{\le}
P\left\{
\max_{1\le j\le k}|(\theta_2^{1/2}Z)_j|>\sqrt{x/k}
\right\} \\
&\overset{(d)}{\le}
\sum_{j=1}^k
P\left\{|(\theta_2^{1/2}Z)_j|>\sqrt{x/k}\right\} 
\overset{(e)}{=}
2k\left\{1-\Phi\left(\sqrt{x/k}\right)\right\}
\end{align*}
where (a) holds by the definition of $T^s$; (b) holds because $\theta_1\in\mathbf R_-^k$: whenever
$(\theta_2^{1/2}Z+\theta_1)_j\ge0$, we have $0\le
(\theta_2^{1/2}Z+\theta_1)_j
\le
(\theta_2^{1/2}Z)_j$; (c) holds trivially; (d) holds by the union
bound; finally, (e) again uses
$(\theta_2^{1/2}Z)_j\sim N(0,1)$. Taking $x=
k\{\Phi^{-1}(1-\frac{\alpha}{2k})\}^2$ gives $J^s(
k\{\Phi^{-1}(1-\frac{\alpha}{2k})\}^2,
\theta_1,\theta_2
)
\ge 1-\alpha.$ Therefore, for all $\theta_1\in\mathbf{R}^k_-$, $(J^s)^{-1}(1-\alpha,\theta_1,\theta_2)
\le
k\left\{\Phi^{-1}\left(1-\frac{\alpha}{2k}\right)\right\}^2$ holds, so the upper bound of $c^s(1-\alpha,\theta_2)$ in \eqref{eq:cv-upperbound1} holds.

Finally consider $\ell=q$ and suppose $\theta_2\in\mathbf O$, so that
$\theta_2^{-1}$ exists. For any
$\theta_1\in\mathbf R^k_-$, $\mu=\theta_1$ is feasible in the
definition of $T^q$. Hence
\begin{align*}
T^q(\theta_2^{1/2}Z+\theta_1,\theta_2)
&=
\inf_{\mu\in\mathbf R^k_+\cup\mathbf R^k_-}
(\theta_2^{1/2}Z+\theta_1-\mu)'
\theta_2^{-1}
(\theta_2^{1/2}Z+\theta_1-\mu) 
\le
Z'Z.
\end{align*}
where the inequality holds because $\mu=\theta_1$ belongs to
$\mathbf R^k_+\cup\mathbf R^k_-$. Therefore, for any $x\in\mathbf R$,
\[
1-J^q(x,\theta_1,\theta_2)
\le
P\{Z'Z>x\}
=
P\{\chi^2_k>x\}.
\]
Taking $x=\chi^2_{k,1-\alpha}$ gives $(J^q)^{-1}(1-\alpha,\theta_1,\theta_2)
\le
\chi^2_{k,1-\alpha}.$ Taking the supremum over
$\theta_1\in\mathbf R^k_-$ yields \eqref{eq:cv-upperbound2}.
\end{proof}

For \(\theta_1\in[0,\infty]^k\), define
\[
\mathbf I(\theta_1):=\{j\in[k]:\theta_{1,j}=+\infty\},
\qquad
\mathbf I^c(\theta_1):=[k]\setminus \mathbf I(\theta_1).
\]
For a vector \(a\), write \(a_I\) for the subvector indexed by \(I\), and for a matrix
\(A\), write \(A_{I\times I}\) for the corresponding principal submatrix. Furthermore, for brevity, we use $\left(J^{\ell}\right)^{-1}$ and $J^{\ell,-1}$ interchangeably.

\begin{lemma}\label{lem:argmax-critical-value-all}
Fix \(\alpha\in(0,\tfrac12)\). For \(\ell\in\{m,s\}\), define the extended inverse
\(\bar J^{\ell,-1}(1-\alpha,\theta_1,\theta_2)\) where 
\((\theta_1, \theta_2)\in[0,\infty]^k\times\bar{\mathbf O}\) as follows.

For \(\ell=m\),
\[
\bar J^{m,-1}(1-\alpha,\theta_1,\theta_2)
:=
\begin{cases}
J^{m,-1}(1-\alpha,\theta_1,\theta_2),
& \mathbf I(\theta_1)=\varnothing,\\
\inf\{x:
P\{\max_{j\in\mathbf I^c(\theta_1)}(-Z_j-\theta_{1,j})\le x\}\ge1-\alpha\},
& \varnothing\subsetneq\mathbf I(\theta_1)\subsetneq[k],\\
-\infty,
& \mathbf I(\theta_1)=[k],
\end{cases}
\]
where \(Z=(Z_1, ..., Z_k)\sim N(0_k,\theta_2)\).

For \(\ell=s\), let $X^s:=\sum_{j\in\mathbf I^c(\theta_1)}(Z_j+\theta_{1,j})^2
I\{Z_j+\theta_{1,j}\le0\}$ and 
\[
\bar J^{s,-1}(1-\alpha,\theta_1,\theta_2)
:=
\begin{cases}
J^{s,-1}(1-\alpha,\theta_1,\theta_2),
& \mathbf I(\theta_1)=\varnothing,\\
\inf\{x:
P\{X^s\le x\}\ge1-\alpha\},
& \varnothing\subsetneq\mathbf I(\theta_1)\subsetneq[k],\\
0,
& \mathbf I(\theta_1)=[k].
\end{cases}
\]

For \(\ell=q\), if \(\theta_2\succ0\), define $Y:=(Z+\theta_1)_{\mathbf I^c(\theta_1)}
\sim
N(\theta_{1,\mathbf I^c(\theta_1)},\theta_{2,\mathbf I^c(\theta_1)\times\mathbf I^c(\theta_1)})$ and $X^q:=\inf_{\nu\in\mathbf R_+^{|\mathbf I^c(\theta_1)|}}
(Y-\nu)' (\theta_{2,\mathbf I^c(\theta_1)\times\mathbf I^c(\theta_1)})^{-1}(Y-\nu)$ and 
\[
\bar J^{q,-1}(1-\alpha,\theta_1,\theta_2)
:=
\begin{cases}
J^{q,-1}(1-\alpha,\theta_1,\theta_2),
& \mathbf I(\theta_1)=\varnothing,\\
\inf\{x:
P\{ X^q\le x \}
\ge1-\alpha\},
& \varnothing\subsetneq\mathbf I(\theta_1)\subsetneq[k],\\
0,
& \mathbf I(\theta_1)=[k].
\end{cases}
\]
Then, \((\theta_1,\theta_2)\mapsto\bar J^{\ell,-1}(1-\alpha,\theta_1,\theta_2)\) is continuous on $[0,\infty]^k\times \bar{\mathbf{O}}$ for $\ell=m,s$ and on $[0,\infty]^k\times \mathbf{O}$ for $\ell=q$.  

Furthermore, for each \(\theta_2\in\bar{\mathbf O}\), define
\[
\mathbf S^\ell(\theta_2)
:=
\arg\max_{\theta_1\in[0,\infty]^k}
\bar J^{\ell,-1}(1-\alpha,\theta_1,\theta_2),
\qquad \ell\in\{m,s\}.
\]
Then, \(\mathbf S^\ell(\theta_2)\) is nonempty and compact, and satisfies
\[
c^\ell(1-\alpha,\theta_2)
=
\max_{\theta_1\in[0,\infty]^k}
\bar J^{\ell,-1}(1-\alpha,\theta_1,\theta_2)
=
\sup_{\theta_1\in[0,\infty)^k}
J^{\ell,-1}(1-\alpha,\theta_1,\theta_2)
\]
for \(\ell\in\{m,s\}\). The same conclusion holds for \(\ell=q\) for
\(\theta_2\in\mathbf O\).
\end{lemma}

\begin{proof}
Fix $\ell\in\{m,s,q\}$ and consider any convergent sequence $(\theta_{1,n},\theta_{2,n})
\to(\theta_{1,\infty},\theta_{2,\infty})$ in the stated domain. For $\ell=q$, the limiting correlation matrix is positive definite, so the smallest eigenvalue of $\theta_{2,n}$
is bounded away from zero for all sufficiently large $n$. On one
probability space, let $Z\sim N(0_k,I_k)$ and put
$Y_n=\theta_{2,n}^{1/2}Z$ and
$Y_\infty=\theta_{2,\infty}^{1/2}Z$. Continuity of the principal
square-root map gives $Y_n\to Y_\infty$ almost surely.

Let $X_n^\ell$ denote the extended statistic whose law defines
$\bar J^\ell(\cdot,\theta_{1,n},\theta_{2,n})$. It equals $T^\ell$
when all coordinates of $\theta_{1,n}$ are finite, the one-sided
statistic $V^\ell$ on $\mathbf I^c(\theta_{1,n})$ when some but not
all coordinates are infinite, and $-\infty$ for $\ell=m$ or zero
for $\ell=s,q$ when all coordinates are infinite.

If $\mathbf I(\theta_{1,\infty})=\varnothing$, then all coordinates
of $\theta_{1,n}$ are eventually finite. Since
$Y_n+\theta_{1,n}\to Y_\infty+\theta_{1,\infty}$ almost surely,
continuity of $T^\ell$ on its domain gives $X_n^\ell \overset{a.s.}{\to} T^\ell(Y_\infty+\theta_{1,\infty},\theta_{2,\infty})$ where the matrix argument is omitted for $\ell=m,s$. 

If $\varnothing\subsetneq\mathbf I(\theta_{1,\infty})\subsetneq[k]$, the coordinates indexed by $\mathbf I(\theta_{1,\infty})$ diverge to $+\infty$, while those indexed by $\mathbf I^c(\theta_{1,\infty})$ have finite limits. We claim that
\begin{align*}
    \begin{aligned}
        X_n^m&\Rightarrow
\max_{j\in\mathbf I^c(\theta_{1,\infty})}
\{-(Y_\infty)_j-\theta_{1,\infty,j}\},\\
        X_n^s&\Rightarrow
\sum_{j\in\mathbf I^c(\theta_{1,\infty})}
\bigl((Y_\infty)_j+\theta_{1,\infty,j}\bigr)^2
I\{(Y_\infty)_j+\theta_{1,\infty,j}\le0\},\\
        X_n^q&\Rightarrow
\inf_{\nu\in\mathbf R_+^{|\mathbf I^c(\theta_{1,\infty})|}}
(Y-\nu)'
\bigl((\theta_{2,\infty})_{
\mathbf I^c(\theta_{1,\infty})\times
\mathbf I^c(\theta_{1,\infty})}\bigr)^{-1}
(Y-\nu),
    \end{aligned}
\end{align*}
where $Y:=(Y_\infty)_{\mathbf I^c(\theta_{1,\infty})}
  +(\theta_{1,\infty})_{\mathbf I^c(\theta_{1,\infty})}.$ 
When all coordinates of $\theta_{1,n}$ are finite, these limits
follow from \eqref{eq:orthant-partial-divergence} in
Lemma~\ref{lem:orthant-distance-reduction}, applied for almost
every realization of $Z$. To cover coordinates that may already
equal $+\infty$, take any subsequence and then a further
subsequence on which $\mathbf I(\theta_{1,n})$ is constant.
Convergence of $\theta_{1,n}$ implies that, for all sufficiently
large $n$, $\mathbf I(\theta_{1,n})
\subseteq\mathbf I(\theta_{1,\infty}).$ The coordinates in
$\mathbf I(\theta_{1,\infty})\setminus\mathbf I(\theta_{1,n})$
are finite at each $n$ and tend to $+\infty$. If
$\mathbf I(\theta_{1,n})=\varnothing$, apply the reduction for
$T^\ell$. If
$\varnothing\subsetneq\mathbf I(\theta_{1,n})
\subsetneq\mathbf I(\theta_{1,\infty})$, the statistic is already
one-sided on $\mathbf I^c(\theta_{1,n})$; the reduction for
$V^\ell$ in \eqref{eq:orthant-partial-divergence} then leaves
precisely the coordinates in
$\mathbf I^c(\theta_{1,\infty})$. If
$\mathbf I(\theta_{1,n})=\mathbf I(\theta_{1,\infty})$, the same
limit follows by continuity of the reduced statistic. Thus every
subsequence has a further subsequence with the displayed limits,
which proves the claim along the original sequence.

If $\mathbf I(\theta_{1,\infty})=[k]$, every coordinate that
remains finite in an approximating nuisance vector tends to
$+\infty$. Hence $X_n^m\to-\infty$ and $X_n^s,X_n^q\to0$ almost
surely. This holds whether $X_n^\ell$ is two-sided, already
one-sided, or assigned its value when every coordinate is
infinite. The limits in all three cases are exactly the laws
used to define $\bar J^{\ell,-1}$.

For $\ell=m$, if some limiting nuisance coordinate is finite,
the limiting CDF is continuous and strictly increasing at its
$(1-\alpha)$-quantile by
Lemma~\ref{lem:strictly_increasing_dist_T}. Therefore,
Lemma~11.2.1(i) of \cite{LehmannRomano2005testing} gives
\begin{align}\label{eq:quantile-convergence}
\bar J^{\ell,-1}(1-\alpha,\theta_{1,n},\theta_{2,n})
\to
\bar J^{\ell,-1}(1-\alpha,\theta_{1,\infty},\theta_{2,\infty}).
\end{align}
If every limiting nuisance coordinate is infinite, then
$\bar J^m(b,\theta_{1,n},\theta_{2,n})\to1$ for every finite $b$.
Thus the corresponding quantiles tend to $-\infty$, proving
\eqref{eq:quantile-convergence} in the extended-real topology.

For $\ell\in\{s,q\}$, first suppose that the limiting
$(1-\alpha)$-quantile is strictly positive. Its CDF is continuous
and strictly increasing at that quantile, so the same quantile
convergence lemma gives \eqref{eq:quantile-convergence}. The
conclusions about a reduced one-sided law also hold when only
one coordinate remains: for a normal variable $Y$ with variance
one, $V^m(Y)=-Y$ and $V^s(Y)=V^q(Y,1)=(Y_-)^2$.

Suppose instead that the limiting quantile for $\ell=s,q$ is
zero. The statistics are nonnegative, so their approximating
quantiles are at least zero. Choose a continuity point
$x_\varepsilon\in(0,\varepsilon)$. The limiting CDF at zero is
at least $1-\alpha$. If it is less than one,
Lemma~\ref{lem:strictly_increasing_dist_T} shows that its support
has no gap above zero, and hence its value at $x_\varepsilon$
is strictly greater than $1-\alpha$. The same strict inequality
is immediate if the limiting CDF at zero equals one.
Convergence of the CDFs at $x_\varepsilon$ therefore gives,
eventually,
\[
0\le
\bar J^{\ell,-1}(1-\alpha,\theta_{1,n},\theta_{2,n})
\le x_\varepsilon<\varepsilon.
\]
This proves \eqref{eq:quantile-convergence}. The strict CDF
comparisons used above also hold with $1-\alpha_n$ in place of
$1-\alpha$ whenever $\alpha_n\to\alpha$, proving the stated
varying-level conclusion.

Now fix $\theta_2\in\bar{\mathbf O}$ for $\ell=m,s$, or
$\theta_2\in\mathbf O$ for $\ell=q$. For $\ell=s,q$, the objective
$\theta_1\mapsto
\bar J^{\ell,-1}(1-\alpha,\theta_1,\theta_2)$ is a continuous
real-valued function on the compact set $[0,\infty]^k$. For
$\ell=m$, apply the Weierstrass theorem instead to
\[
\theta_1\mapsto
\exp\{\bar J^{m,-1}(1-\alpha,\theta_1,\theta_2)\},
\qquad \exp(-\infty):=0
\]
in order to preserve continuity of the mapping. Its maximum is positive because its value at any finite nuisance vector is positive, and it has the same maximizers as the original objective. Thus the maximizing set is nonempty and compact.
Finally, $[0,\infty)^k$ is dense in $[0,\infty]^k$, and
$\bar J^{\ell,-1}$ is a continuous extension of $J^{\ell,-1}$.
Therefore,
\[
\max_{\theta_1\in[0,\infty]^k}
\bar J^{\ell,-1}(1-\alpha,\theta_1,\theta_2)
=
\sup_{\theta_1\in[0,\infty)^k}
J^{\ell,-1}(1-\alpha,\theta_1,\theta_2).
\]
\end{proof}

\begin{lemma}\label{lem:continuity-critical-value-all}
Fix \(\alpha\in(0,\tfrac12)\). For \(\ell\in\{m,s\}\), the least favorable critical value
$c^\ell(1-\alpha,\theta_2)$ in \eqref{def:LF-cv-normal} is continuous on \(\bar{\mathbf O}\). For \(\ell=q\), the same conclusion holds on $\mathbf{O}$.
\end{lemma}

\begin{proof}
For $\ell=s,q$, Lemma~\ref{lem:argmax-critical-value-all}
expresses the LF critical value as the maximum of a continuous real-valued objective over the constant compact feasible set $[0,\infty]^k$. Berge's maximum theorem gives continuity on $\bar{\mathbf O}$ for $s$ and on $\mathbf O$ for $q$. For $\ell=m$, apply Berge's theorem to $H(\theta_1,\theta_2)=
\exp\{\bar J^{m,-1}(1-\alpha,\theta_1,\theta_2)\}$, with $\exp(-\infty)=0$, as in that lemma. Its maximum equals $\exp\{c^m(1-\alpha,\theta_2)\}$ and is strictly positive. Taking logarithms therefore proves continuity of $c^m$ on
$\bar{\mathbf O}$.
\end{proof}

\subsection{Lemmas for Proposition~\ref{prop:power-comparison}}

\begin{lemma}\label{lem:power-k2-calibration}
Suppose \(k=2\) and \(0<\tau<\alpha<1/2\). All three two-step tests have the same rejection region. If \(\mathcal R_{L,\xi}\) denotes the bivariate LF rejection region at significance level \(\xi\) and \(\mathcal R_2\) denotes the two-step rejection region, then
\begin{align}\label{eq:power-k2-sandwich}
\mathcal R_2
\subsetneq
\mathcal R_{L,\alpha}.
\end{align}
Moreover, \(\mathcal R_{L,\alpha}\setminus\mathcal R_2\) contains a nonempty open set.
\end{lemma}

\begin{proof}

We first show that the three two-step tests have the same rejection region when $k=2$.

If $x\in\mathbf R_+^2\cup\mathbf R_-^2$, the first-step feasible set contains $x$. For $\ell=s,q$, $T^\ell(x,I_2)=0$ and all reference quantiles are nonnegative, so neither test rejects. For $\ell=m$, by sign symmetry and permutation, write $x=(u,v)'$ with $u\ge v\ge0$. Then $T^m(x)=-v$, and
\[
\begin{aligned}
J^m(-v,x,I_2)
&=\Phi(-u-v)\Phi(-2v)
+\tfrac12\Phi(u-v)
\le \tfrac12\{\Phi(v-u)+\Phi(u-v)\}
=\tfrac12.
\end{aligned}
\]
Since $1-\alpha+\tau>1/2$, the reference quantile at the feasible nuisance value $\eta=x$ exceeds $T^m(x)$. Hence the maximum-statistic test does not reject either. For every $y\in\mathbf R^2$,
\[
T^s(y)=T^q(y,I_2)
=\{\max\{T^m(y),0\}\}^2.
\]
Because this transformation is continuous and nondecreasing, the corresponding reference quantiles, and hence their suprema over the common feasible set, satisfy the same relationship. For mixed-sign observations, $T^m(x)>0$, so all three tests have identical rejection decisions. The first-step rejection rule is also common to the three tests. This establishes the claimed equivalence. Below it suffices to  study the maximum statistic.

By permutation and simultaneous sign reversal, write a mixed-sign observation as \(x=(a,-b)\), where \(0<b\leq a\). Let
\(R:=\sqrt{\chi^2_{2,1-\tau}}=\sqrt{-2\log\tau}\), which is a first-step radius. If \(b>R\), the first-step confidence region does not intersect the null parameter space and the two-step test rejects immediately. We therefore consider \(b\leq R\).

\medskip
\noindent{\textit{Step 1: The survival probability over a rectangle is maximized at an axis corner.}
For a nonnegative nuisance vector \((u,v)\) and \(t>0\)}, {define
\[
S(u,v;t)
:=
P\{T^m(Z+(u,v))>t\}
=
\Phi(u-t)\Phi(-v-t)
+
\Phi(v-t)\Phi(-u-t),
\]}
where {\(Z\sim N_2(0_2,I_2)\)}. {We claim} that{, for every \(U\geq0\),
\[
\sup_{(u,v)\in[0,U]^2}S(u,v;t)=S(U,0;t)=S(0,U;t).
\]

To verify the claim}, {fix \(v\). Differentiation gives
\[
\frac{\partial S(u,v;t)}{\partial u}
=
\varphi(u+t)
\left\{
e^{2ut}\Phi(-v-t)-\Phi(v-t)
\right\}.
\]
The term in braces} is strictly increasing in {\(u\)}. {Since \(\varphi(u+t)>0\)}, the {derivative can change sign} at {most once}, and {only from negative to positive}. {Thus \(u\mapsto S(u,v;t)\)} is {either monotone or} first {decreasing} and {then increasing}, {so its maximum} on {a closed interval} is {attained at an endpoint. Applying the same argument} to {\(v\) reduces} the {maximization over \([0,U]^2\) to its four corners}.

{By symmetry}, {\(S(U,0;t)=S(0,U;t)\)}. {Moreover},
{\(S(U,0;t)=\Phi(-t)\{\Phi(U-t)+\Phi(-U-t)\}\)}
is {nondecreasing} in {\(U\geq0\)}, {because its derivative with respect to \(U\)} is
{\(\Phi(-t)\{\varphi(U-t)-\varphi(U+t)\}\geq0\)}.
{Hence \(S(U,0;t)\geq S(0,0;t)\).}

It remains to {compare} the {axis} and {diagonal corners}. {Consider
\begin{align*}
    S(U,0;t)-S(U,U;t)
    &= \Phi(-t)\{\Phi(U-t)+\Phi(-U-t)\} - 2 \Phi(U-t)\Phi(-U-t)\\
    &={2 \Phi(U-t)\Phi(-U-t) \Phi(-t)} \left\{\frac{1}{2\Phi(-U-t)} +\frac{1}{2\Phi(U-t)} - \frac{1}{\Phi(-t)}\right\}\geq0.
\end{align*}
This expression} is {non-negative because} the function {\(z\mapsto1/\Phi(z)\)} is strictly {convex}, {since
\[
\frac{d^2}{dz^2}\frac1{\Phi(z)}
=
\frac{\varphi(z)\{2\varphi(z)+z\Phi(z)\}}{\Phi(z)^3}>0,
\]
and applying convexity to \(U-t\) and \(-U-t\)}, {whose midpoint} is {\(-t\)}, {gives
\[
\frac1{\Phi(U-t)}
+
\frac1{\Phi(-U-t)}
\geq
\frac2{\Phi(-t)}.
\]
Therefore}, {we have \(S(U,U;t)\leq S(U,0;t)\)}.

\medskip
\noindent{\textit{Step 2: Within the first-step confidence region, the survival probability is maximized at its right axis endpoint.}
Define
\[
U:=a+\sqrt{R^2-b^2}.
\]}
We {claim that} the {supremum of} the {null survival probability over the first-step confidence region is attained at the nuisance} value {\((U,0)\)}. 

{To see} this, {consider first its intersection with \(\mathbf R_+^2\)}. {Any feasible \((u,v)\in\mathbf R_+^2\) satisfies
\((u-a)^2+(v+b)^2\leq R^2\). Hence \(0\leq u\leq U\)} and {\(0\leq v\leq R-b\leq U\)}. {Thus} the {entire positive-orthant intersection} is {contained in \([0,U]^2\)}, {while} the {point \((U,0)\) itself belongs} to the {confidence region}. {Step~1 therefore} implies that the {maximum over} this {intersection} is {attained at \((U,0)\)}.

It {remains} to {check} the {negative orthant. Such} an {intersection can exist only when \(a\leq R\)}. {Write} a {feasible nuisance value} as {\((-u,-v)\)} with {\(u,v\geq0\)}. {Then
\((a+u)^2+(v-b)^2\leq R^2\)}, {which} implies
{\(u\leq R-a\leq U\)} and {$v\leq b+\sqrt{R^2-a^2}
\leq
a+\sqrt{R^2-b^2}
=U,$ where} the second {inequality uses \(b\leq a\)}. Hence, {after simultaneous sign reversal}, the {negative-orthant intersection is also contained in \([0,U]^2\)}. By {Step~1} and {sign symmetry}, {its survival} probability {cannot exceed \(S(U,0;t)\)}. {This proves} the {claim}.

{It follows that} the {maximum-statistic critical value}, {denoted by \(c(U)\)}, is {determined by
\[
S(U,0;c(U))
=
\Phi(-c(U))
\{\Phi(U-c(U))+\Phi(-U-c(U))\}
=
\alpha-\tau.
\]}
Since {\(S(U,0;0)=1/2\) and \(\alpha-\tau<1/2\)}, {\(c(U)>0\)}. This {also verifies} the {earlier claim that} the three {test statistics yield} the {same two-step rejection region.}

\medskip
\noindent{\textit{Step 3: Every two-step rejection is also an LF rejection at level \(\alpha\).}
Let \(z_\xi:=\Phi^{-1}(1-\xi)\)}. {In} the {bivariate} case, the {three LF tests} have the same {rejection region,} and the {mixed-sign observation \(x=(a,-b)\)}, {\(0<b\leq a\)}, is {rejected at level \(\xi\) if} and {only if \(b>z_\xi\)}. {We show
\(\mathcal R_2\subseteq\mathcal R_{L,\alpha}\)
by proving} the {contrapositive: if} the {LF test} at {level \(\alpha\) does not reject}, then {neither does} the {two-step test}.

{Suppose therefore that \(b\leq z_\alpha\)}. {To see that the first-step does not reject}, {we use} the {normal-tail bound \(\Phi(-z)\leq \frac12e^{-z^2/2}\) for \(z\geq0\)}. {This implies that $z_{\alpha}$ solving \(\Phi(-z_\alpha)=\alpha\) satisfies $z_{\alpha}\leq\sqrt{-2\log(2\alpha)}$}. {Since \(R^2=-2\log\tau\)}, {this implies 
\begin{align}\label{eq:power-dominance-k2-radius-z_alpha}
    R^2-z_\alpha^2 \geq 2\log(2\alpha/\tau)>0,
\end{align}}
and {therefore \(R>z_{\alpha}\geq b\)}. {The confidence set overlaps with $\mathbf R^2_+$.

Next}, {to show} that the {$p$-value at $(a,-b)$}, {which is $S(U,0;b)$}, is {greater than its second step significance level $\alpha-\tau$}. {Consider
\begin{align*}
    S(U,0;b)
    &= \Phi(-b) \{\Phi(a+\sqrt{R^2-b^2}-b) + \Phi(-a-\sqrt{R^2-b^2}-b)\}\\
    &\overset{(1)}{\geq} \alpha \Phi(\sqrt{R^2-z_{\alpha}^2}) 
    \overset{(2)}{\geq} \alpha \Phi\big(\sqrt{2\log( \tfrac{2\alpha}{\tau})}\big) 
    = \alpha \big\{1-\Phi\big(-\sqrt{2\log( \tfrac{2\alpha}{\tau})}\big) \big\}\\
    &\overset{(3)}{\geq}  \alpha\{1-\tfrac12 \exp(-1/2\cdot 2\log( \tfrac{2\alpha}{\tau}))\} 
    =\alpha -\tau/4 > \alpha-\tau
\end{align*} 
where (1) holds because $z_{\alpha}\geq b$}, {$\Phi(-b) \geq \Phi(-z_{\alpha})=\alpha$} and {$a\geq b$; (2) follows from \eqref{eq:power-dominance-k2-radius-z_alpha};} and {(3) holds by} the {normal-tail bound again \(\Phi(-z)\leq \frac12e^{-z^2/2}\) for \(z\geq0\)}. {Thus} the {two-step test does not reject}, {proving  \(\mathcal R_2\subseteq\mathcal R_{L,\alpha}\)}.



\medskip
\noindent{\textit{Step 4: The inclusion is strict on an open set.}
At \(a=b=z_\alpha\)}, the {argument} in{\textit{ Step~3} gives \(S(U,0;b)>\alpha-\tau\)}. {Thus the two-step test does not reject} at {this} point. {Since \(z_\alpha<R\) by \eqref{eq:power-dominance-k2-radius-z_alpha}}, {continuity} implies that the {same inequality holds} in a {neighborhood of \((z_\alpha,z_\alpha)\)}. {This neighborhood contains} a nonempty {open} set {with \(b>z_\alpha\)}, {\(b\leq a\)}, and {\(b<R\)}. {Every such point} is {rejected by} the {LF test} at {level \(\alpha\) but not} by {the two-step test}. {Hence \(\mathcal R_2\subsetneq\mathcal R_{L,\alpha}\)}, and {\(\mathcal R_{L,\alpha}\setminus\mathcal R_2\) contains a nonempty open set}.
\end{proof}

\begin{lemma}\label{lem:consistency}
Suppose Assumption~\ref{ass:uniform-integrability} holds and fix $\alpha\in(0,1/2)$. For every fixed $P\in\mathbf P\setminus\mathbf P_0$, the LF test in \eqref{def:LFtest}, using either critical value, satisfies
\[
\mathbb E_P[\phi_n^\ell]\to1,
\qquad \ell\in\{m,s\}.
\]
If, in addition, Assumption~\ref{ass:uniform-eigenvalue} holds, this conclusion also holds for $\ell=q$, and, for every fixed $\tau\in(0,\alpha)$, the two-step tests satisfy
\[
\mathbb E_P[\psi_n^\ell]\to1,
\qquad
\mathbb E_P[\psi_n^{\ell *}]\to1,
\qquad \ell\in\{m,s,q\}.
\]
\end{lemma}

\begin{proof}
Fix $P\in\mathbf P\setminus\mathbf P_0$. Since $\mu(P)$ has both positive and negative components, the law of large numbers and Lemma~\ref{alem:consistency-sample-variance} imply $T_n^m,T_n^s\overset{p}{\longrightarrow}\infty$. The same holds for QLR because, on nonexceptional samples, $T_n^q\ge T_n^s/k$.

For the parametric LF test, Lemma~\ref{lem:parametric-cv-existence} bounds the critical value by a finite constant. For its nonparametric bootstrap counterpart, observe that, uniformly over $d\in\mathbf R_+^k\cup\mathbf R_-^k$,
\[
T^m(z+d)\le\|z\|_\infty,\qquad
T^s(z+d)\le\|z\|^2,\qquad
T^q(z+d,\Omega)\le z'\Omega^{-1}z.
\]
Applying these bounds to the centered bootstrap vector, Lemma~\ref{lem:centered-sampling-bootstrap} implies that the corresponding conditional quantiles are bounded in probability. Exceptional resamples have vanishing conditional probability. Thus both LF critical values are bounded above in probability, proving consistency of the LF tests.

For the two-step tests, covariance consistency and Lemma~\ref{lem:two-step-first-step-coverage} imply that both first-step confidence regions have radius $O_P(n^{-1/2})$ and centers converging to $\mu(P)$. Since $\mu(P)$ has positive distance from the closed set $\mathbf R_+^k\cup\mathbf R_-^k$, both regions are disjoint from this set with probability tending to one. The rejection rules therefore imply consistency of both two-step tests.
\end{proof}

\subsection{Proof of Corollary \ref{cor:general-estimator-uniform-validity}}
\label{subsec:supporting-regular-estimators}

Fix $\alpha\in(0,1/2)$ and $\tau\in(0,\alpha)$.
Consider an arbitrary sequence
$\{P_n\}\subset\mathbf P_{\beta,0}$ and write
\[
X_n^\beta:=\sqrt n\,\hat D_n^{-1}\hat\beta_n,
\qquad
d_n:=\sqrt n\,D(P_n)^{-1}\beta(P_n).
\]
From any subsequence, extract a further subsequence,
still indexed by $n$, such that
$\Omega(P_n)\to\Omega^*\in\bar{\mathbf O}$,
$\beta(P_n)$ belongs to the same null orthant,
and each coordinate of $d_n$ has an extended-real limit.
For the LF test with $\ell=q$ and the two-step test
with any $\ell$, condition
\eqref{eq:general-estimator-eigenvalue-condition}
ensures $\Omega^*\succ0$.
By sign symmetry, it suffices to consider
$\beta(P_n)\in\mathbf R_+^k$.

We first establish the Gaussian approximation.
By the definition of $\varphi_P$,
\[
\mathbb E_{P_n}[\varphi_{P_n}(U_i)]=0_k,
\qquad
\mathbb E_{P_n}[
\varphi_{P_n}(U_i)\varphi_{P_n}(U_i)'
]
=\Omega(P_n)\to\Omega^*.
\]
Since the marginal variances equal one and
\eqref{eq:UI-normalized-influence-function} ensures
uniform integrability, the triangular-array central
limit theorem and the Cram\'er--Wold device imply
\[
\frac{1}{\sqrt n}\sum_{i=1}^n
\varphi_{P_n}(U_i)
\Rightarrow N(0_k,\Omega^*).
\]
Together with
\eqref{eq:uniform-asymptotic-linearity-general}
and \eqref{eq:studentization-general-consistency},
this gives
\begin{align}
\sqrt n\,\hat D_n^{-1}
(\hat\beta_n-\beta(P_n))
&=
\hat D_n^{-1}D(P_n)
\left\{
\frac{1}{\sqrt n}\sum_{i=1}^n
\varphi_{P_n}(U_i)+o_{P_n}(1)
\right\}
\Rightarrow N(0_k,\Omega^*).
\label{eq:general-estimator-centered-limit}
\end{align}
Moreover,
$\hat\Omega_n\overset{P_n}{\to}\Omega^*$
by \eqref{eq:studentization-general-consistency}.

Note that
\[
X_n^\beta
=
\sqrt n\,\hat D_n^{-1}
(\hat\beta_n-\beta(P_n))
+
\hat D_n^{-1}D(P_n)d_n.
\]
The first term is $O_{P_n}(1)$ and
$\hat D_n^{-1}D(P_n)\overset{P_n}{\to}I_k$.
Consequently, following the finite-coordinate and
divergent-coordinate arguments in the proofs of
Lemmas~\ref{lem:statistic-degenerate-case} and
\ref{lem:conv-ip-to-Jn}, we obtain
\[
(X_n^\beta)_I
\Rightarrow N(\delta_I,\Omega^*_{I\times I}),
\qquad
X_{j,n}^\beta\overset{P_n}{\to}+\infty
\quad\text{for }j\notin I,
\]
where $I$ denotes the set of coordinates for which
$d_{j,n}\to\delta_j\in[0,\infty)$.
Indeed, for $j\in I$, the second term converges in
probability to $\delta_j$, jointly because $k$ is fixed.
For $j\notin I$, the corresponding standardized
population parameter diverges to $+\infty$ and
multiplication by a factor converging in probability
to one preserves this divergence. Adding the
$O_{P_n}(1)$ centered term does not affect it.

We now follow the two cases considered in the proof
of Theorem~\ref{thm:LFtest-uniform-validity}.

\medskip
\noindent
\textbf{Case 1: All standardized local parameters diverge.}
Suppose $d_{j,n}\to+\infty$ for every $j\in[k]$.
Then
\[
\min_{j\in[k]}X_{j,n}^\beta
\overset{P_n}{\to}+\infty.
\]
By the same argument as in
Lemma~\ref{lem:statistic-degenerate-case},
$T_{n,\beta}^m\overset{P_n}{\to}-\infty$
and $T_{n,\beta}^s\overset{P_n}{\to}0$.
Under the additional eigenvalue condition,
$T_{n,\beta}^q\overset{P_n}{\to}0$ as well.
The deterministic lower bounds for the Gaussian
LF critical values established in the proof of
Theorem~\ref{thm:LFtest-uniform-validity}
therefore imply $\mathbb E_{P_n}[\phi_{n,\beta}^\ell]\to0.$ 

For the two-step test, the eigenvalue condition and
\eqref{eq:studentization-general-consistency}
ensure $P_n\{\hat V_n\succ0\}\to1$. Every
$b\in\mathbf C_{n,\beta}(1-\tau)$ satisfies the following:
\begin{align*}
    \begin{aligned}
n(b_j-\hat\beta_{j,n})^2
&=
n\left\{
e_j'\hat V_n^{1/2}
\hat V_n^{-1/2}(b-\hat\beta_n)
\right\}^2\\
&\overset{(a)}{\le}
n\left(e_j'\hat V_ne_j\right)
\left\{
(b-\hat\beta_n)'\hat V_n^{-1}
(b-\hat\beta_n)
\right\}\\
&\overset{(b)}{\le}
(\hat V_n)_{jj}\chi^2_{k,1-\tau}
=(\hat D_n)_{jj}^2\chi^2_{k,1-\tau}.
\end{aligned}
\end{align*}
where (a) holds by the Cauchy--Schwarz inequality and (b) holds by the definition of $\mathbf C_{n,\beta}(1-\tau)$. Rearranging gives
\[
\frac{\sqrt n\,b_j}{(\hat D_n)_{jj}}
\ge
X_{j,n}^\beta-\sqrt{\chi^2_{k,1-\tau}},
\qquad j\in[k].
\]
Since $\min_{j\in[k]}X_{j,n}^\beta
\overset{P_n}{\to}+\infty$, it follows that $P_n\{\mathbf C_{n,\beta}(1-\tau)
\subseteq\mathbf R_+^k\}\to1.$ By the non-rejection rule in
\eqref{def:general-estimator-two-step-test},
$\mathbb E_{P_n}[\psi_{n,\beta}^\ell]\to0$.

\medskip
\noindent
\textbf{Case 2: At least one standardized local parameter has a finite limit.}
Suppose $I\neq\varnothing$.
The centered Gaussian convergence established above,
together with the finite-coordinate convergence and
divergence of the remaining coordinates, provides
exactly the convergence properties used in the proof
of Lemma~\ref{lem:conv-ip-to-Jn}.
Applying the same argument, including the treatment
of the atom at zero for $\ell=s,q$, gives
\[
\sup_{x\in\mathbf R}
\left|
P_n\{T_{n,\beta}^\ell\le x\}
-
J^\ell\left(
x,\sqrt n\,\hat D_n^{-1}\beta(P_n),
\hat\Omega_n
\right)
\right|
\overset{P_n}{\to}0.
\]
For $\ell=m,s$, this conclusion also allows
a singular limit $\Omega^*$; for $\ell=q$,
the eigenvalue condition ensures $\Omega^*\succ0$. By Lemma~\ref{alem:random-quantile-transfer}, for every fixed $\gamma\in(0,1)$,
\begin{align}
\limsup_{n\to\infty}
P_n\left\{
T_{n,\beta}^\ell>
(J^\ell)^{-1}\left(
1-\gamma,
\sqrt n\,\hat D_n^{-1}\beta(P_n),
\hat\Omega_n
\right)
\right\}
\le\gamma.
\label{eq:general-estimator-reference-tail}
\end{align}
For the LF test,
$\sqrt n\,\hat D_n^{-1}\beta(P_n)\in\mathbf R_+^k$,
so the Gaussian quantile in
\eqref{eq:general-estimator-reference-tail}
is no greater than
$c^\ell(1-\alpha,\hat\Omega_n)$
when $\gamma=\alpha$.
Hence
\[
\limsup_{n\to\infty}
\mathbb E_{P_n}[\phi_{n,\beta}^\ell]
\le\alpha.
\]

For the two-step test, the Gaussian convergence
in \eqref{eq:general-estimator-centered-limit},
$\hat\Omega_n\overset{P_n}{\to}\Omega^*\succ0$,
and Slutsky's theorem imply
\[
n(\hat\beta_n-\beta(P_n))'
\hat V_n^{-1}
(\hat\beta_n-\beta(P_n))
\Rightarrow\chi_k^2
\]
on the event $\hat V_n\succ0$, whose probability
tends to one. Thus
\[
P_n\{
\beta(P_n)\notin
\mathbf C_{n,\beta}(1-\tau)
\}\to\tau.
\]
Whenever the confidence region contains $\beta(P_n)$,
the two-step test either does not reject because the
region is contained in the null, or its critical value
is at least the Gaussian quantile evaluated at
$\sqrt n\,\hat D_n^{-1}\beta(P_n)$.
Consequently,
\[
\begin{aligned}
P_n\{\psi_{n,\beta}^\ell=1\}
&\le
P_n\{
\beta(P_n)\notin
\mathbf C_{n,\beta}(1-\tau)
\}
+
P_n\left\{
T_{n,\beta}^\ell>
(J^\ell)^{-1}\left(
1-\alpha+\tau,
\sqrt n\,\hat D_n^{-1}\beta(P_n),
\hat\Omega_n
\right)
\right\}.
\end{aligned}
\]
Applying \eqref{eq:general-estimator-reference-tail}
with $\gamma=\alpha-\tau$ yields
\[
\limsup_{n\to\infty}
\mathbb E_{P_n}[\psi_{n,\beta}^\ell]
\le\tau+(\alpha-\tau)=\alpha.
\]

The case $\beta(P_n)\in\mathbf R_-^k$ follows by
sign symmetry. Since every subsequence of every
null sequence admits a further subsequence covered
by one of the two cases, the subsequence argument
used in Theorems~\ref{thm:LFtest-uniform-validity}
and \ref{thm:two-step-uniform-validity}
establishes both uniform validity conclusions.\qed

\subsection{Lemmas for Proposition \ref{prop:LF-improved-numerical-accuracy}}

The proof also uses Lemmas~\ref{lem:parametric-cv-existence},
\ref{lem:argmax-critical-value-all}, and \ref{lem:strictly_increasing_dist_T}, which are established above. The following results collect the additional ingredients specific to the certified numerical procedure. 

\begin{lemma}\label{lem:LF-sharp-Gaussian-shift}
Let $P_0=N(\mu_0,\Sigma)$ and $P_1=N(\mu_1,\Sigma)$, where
$\Sigma\succ0$, and put $d
:=
\{(\mu_1-\mu_0)'\Sigma^{-1}(\mu_1-\mu_0)\}^{1/2}$. Then, for every measurable set $D\subseteq\mathbf R^k$,
\begin{align}
\Phi\{\Phi^{-1}(P_0(D))-d\}
\leq
P_1(D)
\leq
\Phi\{\Phi^{-1}(P_0(D))+d\},
\label{eq:LF-sharp-Gaussian-shift}
\end{align}
where the expressions are interpreted as zero or one when $P_0(D)$ equals
zero or one. 
\end{lemma}

\begin{proof}
The first claim is the Gaussian hypothesis-testing trade-off implied by the Neyman--Pearson lemma; see \citet[Section~2.2, equation (5)]{DongRothSu2022} and \citet{LehmannRomano2005testing}. We provide a direct proof for completeness. The result is immediate if $d=0$. Suppose $d>0$ and define
$S:={(\mu_1-\mu_0)'\Sigma^{-1}(X-\mu_0)}/{d}.$ Under $P_0$, $S\sim N(0,1)$, while under $P_1$, $\sim N(d,1)$, and $\frac{dP_1}{dP_0}(X)=\exp\{dS-d^2/2\}.$ Fix $p=P_0(D)$. By the Neyman--Pearson lemma, among all measurable sets having
$P_0$ probability $p$, the largest $P_1$ probability is attained by an upper
level set of the likelihood ratio. The value of $c^*$ solving $P_0\{S>c^*\}=p$ is $c^*=\Phi^{-1}(1-p).$ Therefore,
\begin{align*}
P_1(D)
\leq
P_1\{S>\Phi^{-1}(1-p)\} =
1-\Phi\{\Phi^{-1}(1-p)-d\} =
\Phi\{\Phi^{-1}(p)+d\}.
\end{align*}
The lower bound follows by applying the upper bound to $D^c$. 
\end{proof}

\begin{lemma}\label{lem:LF-boundary-reduction}
{Fix $\alpha\in(0,1/2)$ and $\Omega\in\mathbf O$.}
For $j\in[k]$, let $Z_{-j}\sim N(0_{k-1},\Omega_{-j,-j})$ and define
\begin{align*}
R_j^m
:=\max_{h\neq j}(-Z_h),~~~
R_j^s
:=\sum_{h\neq j}Z_h^2I\{Z_h\leq0\},~~~
R_j^q
:=\inf_{\nu\in\mathbf R_+^{k-1}}
(Z_{-j}-\nu)'\Omega_{-j,-j}^{-1}(Z_{-j}-\nu).
\end{align*}
Let $\bar J^{\ell,-1}$ denote the continuous compactified quantile in Lemma~\ref{lem:argmax-critical-value-all}. Then
\begin{align}
\max_{\theta_1\in[0,\infty]^k\setminus[0,\infty)^k}
\bar J^{\ell,-1}(1-\alpha,\theta_1,\Omega)
=
\max_{j\in[k]}q_{1-\alpha}(R_j^\ell).
\label{eq:LF-boundary-quantile-reduction}
\end{align}
\end{lemma}

\begin{proof}
Take a boundary point $\theta_1${,} and {write $D=\{j:\theta_{1,j}=\infty\}$} and $F=[k]\setminus D$. If $F$ is empty, the {extended} statistic {equals} $-\infty$ for {$m$} and zero for {$s,q$}, so {its quantile} cannot exceed the right-hand side of \eqref{eq:LF-boundary-quantile-reduction}.

{If} $F$ is nonempty{,} the {extended} statistic is {$V^\ell$} on $F$ at the finite {nonnegative shift} $\theta_{1,F}$. {By} Lemma{~\ref{lem:orthant-distance-reduction}}, its quantile is no {larger} than the corresponding quantile at {zero shift}. Choose $j\in D$ and put $G=[k]\setminus\{j\}$, {so} $F\subseteq G$. {Couple} the {vectors using $Z\sim N(0_k,\Omega)$}. {The restriction inequality in} the same {lemma} gives
{\[
V^\ell(Z_F,\Omega_{F\times F})
\le V^\ell(Z_G,\Omega_{G\times G})=R_j^\ell.
\]}
Therefore every boundary quantile is {bounded by $\max_jq_{1-\alpha}(R_j^\ell)$}. {The} reverse inequality follows because each singleton-infinity {vector $\infty e_j$ is a} boundary {point whose extended statistic is $R_j^\ell$}.
\end{proof}

\begin{lemma}\label{lem:LF-outside-box}
Fix $\alpha\in(0,1/2)$.
Let $V^\ell$ be the one-sided statistics in Lemma~\ref{lem:orthant-distance-reduction}.
For every $A>0$ and $\Omega\in\mathbf{O},$
\begin{align}
\sup_{\substack{\theta_1\in\mathbf R_+^k,~
\|\theta_1\|_\infty>A}}
J^{\ell,-1}(1-\alpha,\theta_1,\Omega)
\leq
\max_{j\in[k]}
q_{1-\alpha}\{V^\ell(\Omega^{1/2}Z+Ae_j,\Omega)\}.
\label{eq:LF-outside-box-quantile}
\end{align}
where $Z\sim N(0_k, I_k).$ Moreover, for every $j\in[k]$,
\begin{align}
q_{1-\alpha}\{V^\ell(\Omega^{1/2}Z+Ae_j,\Omega)\}
\to
q_{1-\alpha}(R_j^\ell)
\qquad\text{as }A\to\infty,
\label{eq:LF-outside-box-quantile-limit}
\end{align}
where $R_j^\ell$ is the singleton-infinity statistic defined in
Lemma~\ref{lem:LF-boundary-reduction}.
\end{lemma}

\begin{proof}
Fix $\theta_1\ge0$ with $\|\theta_1\|_\infty>A$, and choose $j$ {with} $\theta_{1,j}>A$. Then $\theta_1\ge Ae_j$. Lemma~\ref{lem:orthant-distance-reduction} gives
\[
T^\ell(\Omega^{1/2}Z+\theta_1,\Omega)
\le V^\ell(\Omega^{1/2}Z+\theta_1,\Omega)
\le V^\ell(\Omega^{1/2}Z+Ae_j,\Omega).
\]
Taking quantiles and maximizing over $j$ proves \eqref{eq:LF-outside-box-quantile}. The partial-divergence conclusion of the same lemma gives
$V^\ell(\Omega^{1/2}Z+Ae_j,\Omega)\to R_j^\ell$ almost surely. Each limiting CDF is continuous and strictly increasing at its $(1-\alpha)$-quantile by Lemma \ref{lem:strictly_increasing_dist_T}. For $s,q$, this quantile is positive because the probability that every remaining coordinate is nonnegative is at most $1/2<1-\alpha$. Quantile convergence proves \eqref{eq:LF-outside-box-quantile-limit}.
\end{proof}

\subsection{Proof of Proposition~\ref{prop:oulhaj-dominance}} \label{supp:oulhaj-strict-dominance}
\begin{lemma}\label{lem:oulhaj-strict-dominance}
Let $k\geq2$, $\alpha\in(0,1/2)$, and $\Omega\in\mathbf O$. Let $c^m(1-\alpha,\Omega)$ and
$c_{\mathrm{OEH}}(1-\alpha,\Omega)$ be the critical values
defined in \eqref{def:LF-cv-normal} and \eqref{eq:oulhaj-cv}, respectively, and let $T^m$ be the function in \eqref{def:T-ell-function}. Then
\[
c^m(1-\alpha,\Omega)
<
c_{\mathrm{OEH}}(1-\alpha,\Omega).
\]
Moreover, for every
$\theta_1\notin\mathbf R_+^k\cup\mathbf R_-^k$
and $W\sim N(\theta_1,\Omega)$,
\[
P\{T^m(W)>c^m(1-\alpha,\Omega)\}
>
P\{T^m(W)>c_{\mathrm{OEH}}(1-\alpha,\Omega)\}.
\]
\end{lemma}
\begin{proof}
Fix $\Omega\in\mathbf O$ and write
$c=c_{\mathrm{OEH}}(1-\alpha,\Omega)>0$.
Let $Y=\Omega^{1/2}Z$, where $Z\sim N(0_k,I_k)$.
We first claim that
\[
c^m(1-\alpha,\Omega)<c.
\]
To verify this claim, it suffices to consider
$\theta_1\in\mathbf R_+^k$, since Gaussian symmetry and
$T^m(-x)=T^m(x)$ give the same null distributions for the
negative orthant. We consider two cases.

\noindent\textit{Case 1: $\theta_{1,j}\leq1$ for every $j$.}
Since $0\leq\theta_{1,j}\leq1$, the positive maximum in
$T^m(Y+\theta_1)$ is at most $1+\max_jY_j$, while the
negative maximum is at most $\max_j(-Y_j)$. Hence
\[
T^m(Y+\theta_1)
\leq\min\{1+\max_jY_j,\max_j(-Y_j)\}.
\]
The right-hand side is at most $\max_j(-Y_j)$, whose upper
tail at $c$ equals $\alpha$ by \eqref{eq:oulhaj-cv} and
Gaussian symmetry. Moreover, on the event
\[
\{Y_\ell<-c-1\text{ for all }\ell\},
\]
we have $\max_j(-Y_j)>c$ but $1+\max_jY_j<-c<c$.
This event has positive probability because $\Omega\in\mathbf O$
implies that the density of $Y$ is positive everywhere.
Thus the right-hand side exceeds $c$ on a strict subset,
with a positive probability difference, of the event
$\{\max_j(-Y_j)>c\}$. Consequently,
\[
P\{\min\{1+\max_jY_j,\max_j(-Y_j)\}>c\}<\alpha.
\]

\noindent\textit{Case 2: $\theta_{1,j}>1$ for some $j$.}
For such a $j$, the negative maximum in $T^m(Y+\theta_1)$
is bounded by replacing $\theta_{1,j}$ with $1$ and all
other coordinates of $\theta_1$ with zero. Therefore,
\[
T^m(Y+\theta_1)
\leq\max\{-Y_j-1,\max_{\ell\ne j}(-Y_\ell)\}
\leq\max_\ell(-Y_\ell).
\]
On the event
\[
\{-c-1<Y_j<-c,\quad
Y_\ell>-c\text{ for all }\ell\ne j\},
\]
we have $\max_\ell(-Y_\ell)>c$, whereas both terms in the
middle maximum are strictly below $c$.
This event also has positive probability. The same argument
as in Case~1 therefore gives
\[
P\{\max\{-Y_j-1,\max_{\ell\ne j}(-Y_\ell)\}>c\}<\alpha
\qquad\text{for every }j.
\]

The bounding random variables in these two cases have
continuous distributions: equality to a fixed value requires
at least one of finitely many Gaussian coordinates to equal
a fixed number, an event of probability zero.
Since each upper tail at $c$ is strictly below $\alpha$,
continuity implies that each $(1-\alpha)$-quantile is
strictly below $c$.
There are only $k+1$ such distributions, and none depends
on $\theta_1$. The largest of their $(1-\alpha)$-quantiles
is therefore still strictly below $c$ and bounds the
corresponding quantile of $T^m(Y+\theta_1)$ for every
$\theta_1\in\mathbf R_+^k$.
Symmetry gives the same bound on $\mathbf R_-^k$.
Taking the supremum over $\mathbf R_+^k\cup\mathbf R_-^k$
in \eqref{def:LF-cv-normal} proves the claim.

We next claim that, for every
$\theta_1\notin\mathbf R_+^k\cup\mathbf R_-^k$ and
$W\sim N(\theta_1,\Omega)$,
\[
P\{c^m(1-\alpha,\Omega)<T^m(W)\leq c\}>0.
\]
To verify this claim, choose
$\max\{c^m(1-\alpha,\Omega),0\}<t<c$;
such a $t$ exists by the first claim and $c>0$.
At $x=(t,-t,0,\ldots,0)$, we have $T^m(x)=t$.
Continuity of $T^m$ then implies that
$c^m(1-\alpha,\Omega)<T^m(x)<c$ throughout an open
neighborhood of $x$. Since $\Omega\in\mathbf O$,
this neighborhood has positive probability under
$N(\theta_1,\Omega)$, proving the second claim.

Finally, both tests reject for large values of the same
statistic, so the probability in this claim is exactly
the LF rejection probability minus the OEH rejection
probability. The LF test therefore has strictly higher
power at every alternative in the stated Gaussian experiment.
\end{proof}

\end{document}